%% file: 0_Conditional_density_regression.tex
\documentclass[12pt,a4paper]{article} 
\usepackage{amsmath}
\usepackage{graphicx}
\usepackage{enumitem}

\makeatletter
\newcommand{\nitem}[1]{%
\item[#1]\protected@edef\@currentlabel{#1}%
}
\makeatother

\input{commands.tex}

\input{biber_setup}
\numberwithin{equation}{section} 
\numberwithin{table}{section}
\numberwithin{figure}{section}

\usepackage{hyperref}
\usepackage{multirow}
\usepackage{subcaption}
\usepackage{csquotes}

\usepackage{authblk} 

\input{01_2_authors}

\usepackage{ifthen}

\newcommand{\imagePathPaperTwo}{Images/}

\begin{document}

\def\spacingset#1{\renewcommand{\baselinestretch}%
{#1}\small\normalsize} \spacingset{1}

\title{\input{01_1_title}}

\maketitle

\bigskip
\begin{refsection}

\begin{abstract}
\noindent
\input{01_3_abstract}
\end{abstract}

\input{1_introduction.tex}
\input{2_poisson_regression.tex}
\input{4_application.tex}
\input{3_simulation.tex}
\input{5_conclusion.tex}

\subsection*{Acknowledgments}

\input{91_acknowledgements}

\subsection*{Funding}

\input{92_funding}




\addcontentsline{toc}{section}{References}
\printbibliography[heading=bibliography]
\end{refsection}

 \pagebreak
 \setcounter{page}{1}
 \begin{center}
     {\LARGE\bf APPENDIX}
 \end{center}
 \numberwithin{equation}{section}
 \numberwithin{table}{section}
 \numberwithin{figure}{section}
 \begin{refsection}
 \appendix
 \input{6_appendix.tex}
 \addcontentsline{toc}{section}{References}
 \printbibliography[heading=bibliography]
 \end{refsection}

\end{document}

%% file: commands.tex
\usepackage[utf8]{inputenc}
\usepackage[left=2.75cm,right=2.75cm,top=3cm,bottom=3cm]{geometry} 

\usepackage[english]{babel}

\usepackage{fancyhdr}

\usepackage{mathtools}
\usepackage{amsfonts}
\usepackage{amssymb}
\usepackage{amsthm}
\usepackage{esvect}

\usepackage{nicefrac}

\usepackage{dsfont}
\usepackage{bbm}

\usepackage{float}
\usepackage[rflt]{floatflt}
\usepackage{xcolor}

\newcommand{\B}{B^2(\mu)}
\newcommand{\Bl}{B^2(\lambda)}
\newcommand{\Bd}{B^2(\delta^\bullet)}
\newcommand{\Ln}{L^2_0(\mu)}
\newcommand{\Lnl}{L^2_0(\lambda)}
\newcommand{\Ll}{L^2(\lambda)}
\newcommand{\Lnd}{L^2_0(\delta^\bullet)}
\newcommand{\Ld}{L^2(\delta^\bullet)}

\DeclareMathOperator{\inv}{inv}
\DeclareMathOperator{\clr}{clr}
\DeclareMathOperator{\pen}{pen}
\DeclareMathOperator{\Pen}{Pen}

\newcommand{\Vertb}{\Vert_{\B}}

\newcommand{\fh}{\hat{f}}

\newcommand{\gh}{\hat{g}}

\newcommand{\fc}{f_{\mathrm{c}}}
\newcommand{\fd}{f_{\mathrm{d}}}
\newcommand{\Jtc}{\tilde{\iota}_{\mathrm{c}}}
\newcommand{\Jtd}{\tilde{\iota}_{\mathrm{d}}}
\newcommand{\ftc}{\tilde{f}_{\mathrm{c}}}
\newcommand{\ftd}{\tilde{f}_{\mathrm{d}}}

\newcommand{\gm}{s^{\max}}
\newcommand{\gc}{s_{\mathrm{c}}}

\newcommand{\Yc}{\mathcal{Y}_{\mathrm{c}}}
\newcommand{\Yd}{\mathcal{Y}_{\mathrm{d}}}
\newcommand{\Ycm}{\Ycal_{\mathrm{c}}^{\max}}

\newcommand{\Ycmflat}{\Ycal_{\mathrm{c}, \, -}^{\max}}
\newcommand{\Ycmpoint}{\Ycal_{\mathrm{c}, \, \bullet}^{\max}}
\newcommand{\yc}{y_{\mathrm{c}}}
\newcommand{\yd}{y_{\mathrm{d}}}

\newcommand{\OR}{O\! R}

\newcommand{\ootimes}{\unitlength1ex
\begin{picture}(3,2)
\put(1.5,0.6){\circle{2}}\put(1.5,0.6){\makebox(0,0){$\otimes$}}
\end{picture}}

\newtheorem{thm}{Theorem}[section]

\newtheorem{lem}[thm]{Lemma}

\newtheorem{prop}[thm]{Proposition}
\newtheorem{thme}[thm]{Theorem}

\newtheoremstyle{cited}{}{}{\itshape}{}{
}{\textbf{.}}{.5em}{\thmname{\textbf{#1}}\thmnumber{ \textbf{#2}} #3%
}

\theoremstyle{cited}

\theoremstyle{definition}

\newtheorem{rem}[thm]{Remark}

\newcommand{\ra}{\rightarrow}

\newcommand{\xra}{\xrightarrow}

\newcommand{\Ebb}{\mathbb{E}}

\newcommand{\Nbb}{\mathbb{N}}

\newcommand{\Pbb}{\mathbb{P}}

\newcommand{\Rbb}{\mathbb{R}}
\newcommand{\Sbb}{\mathbb{S}}

\newcommand{\Af}{\mathbf{A}}

\newcommand{\Afh}{\hat{\mathbf{A}}}
\newcommand{\bfe}{\mathbf{b}}

\newcommand{\Bf}{\mathbf{B}}

\newcommand{\Bfh}{\hat{\mathbf{B}}}

\newcommand{\eh}{\hat{e}}

\newcommand{\Ff}{\mathbf{F}}

\newcommand{\Ffh}{\hat{\mathbf{F}}}

\newcommand{\Hf}{\mathbf{H}}

\newcommand{\If}{\mathbf{I}}
\newcommand{\Id}{\mathbf{I}}

\newcommand{\Jf}{\mathbf{J}}

\newcommand{\nf}{\mathbf{n}}

\newcommand{\Pf}{\mathbf{P}}

\newcommand{\Sf}{\mathbf{S}}

\newcommand{\vf}{\mathbf{v}}

\newcommand{\Vf}{\mathbf{V}}

\newcommand{\Vfh}{\hat{\mathbf{V}}}
\newcommand{\wf}{\mathbf{w}}

\newcommand{\Wf}{\mathbf{W}}

\newcommand{\xf}{\mathbf{x}}

\newcommand{\Xf}{\mathbf{X}}

\newcommand{\Yf}{\mathbf{Y}}

\newcommand{\Zf}{\mathbf{Z}}

\newcommand{\At}{\tilde{A}}

\newcommand{\bt}{\tilde{b}}
\newcommand{\bft}{\tilde{\mathbf{b}}}

\newcommand{\Bft}{\tilde{\mathbf{B}}}
\newcommand{\bfb}{\bar{\mathbf{b}}}

\newcommand{\bb}{\bar{b}}

\newcommand{\ft}{\tilde{f}}

\newcommand{\hti}{\tilde{h}}

\newcommand{\Mt}{\tilde{M}}

\newcommand{\Nt}{\tilde{N}}

\newcommand{\vft}{\tilde{\mathbf{v}}}

\newcommand{\xft}{\tilde{\mathbf{x}}}

\newcommand{\yt}{\tilde{y}}

\newcommand{\betah}{\hat{\beta}}

\newcommand{\betat}{\tilde{\beta}}
\newcommand{\gammaf}{\boldsymbol{\gamma}}

\newcommand{\deltat}{\tilde{\delta}}

\newcommand{\lambdaf}{\boldsymbol{\lambda}}

\newcommand{\lambdah}{\hat{\lambda}}

\newcommand{\nuf}{\boldsymbol{\nu}}

\newcommand{\varsigmaf}{\boldsymbol{\varsigma}}

\newcommand{\thetaf}{\boldsymbol{\theta}}

\newcommand{\thetafh}{\hat{\boldsymbol{\theta}}}

\newcommand{\thetaft}{\tilde{\boldsymbol{\theta}}}

\newcommand{\varthetaf}{\boldsymbol{\vartheta}}

\newcommand{\tauf}{\boldsymbol{\tau}}

\newcommand{\taufh}{\hat{\boldsymbol{\tau}}}

\newcommand{\xif}{\boldsymbol{\xi}}

\newcommand{\dmu}{\mathrm{d}\mu}

\newcommand{\dnuf}{\mathrm{d}\nuf}
\newcommand{\dlamb}{\mathrm{d}\lambda}
\newcommand{\dlambf}{\mathrm{d}\lambdaf}

\newcommand{\ddel}{\mathrm{d}\delta}

\newcommand{\dt}{\mathrm{d}t}

\newcommand{\Var}{\mathrm{Var}}

\newcommand{\rMSE}{\mathrm{relMSE}}

\newcommand{\diag}{\mathrm{diag}}

\DeclareMathOperator{\GL}{GL}

\DeclareMathOperator{\diam}{diam}

\DeclareMathOperator{\dom}{dom}

\DeclareMathOperator{\spano}{span}

\DeclareMathOperator{\rank}{rank}

\newcommand{\Acal}{\mathcal{A}}
\newcommand{\Bcal}{\mathcal{B}}
\newcommand{\Ccal}{\mathcal{C}}
\newcommand{\Dcal}{\mathcal{D}}

\newcommand{\Ical}{\mathcal{I}}
\newcommand{\Jcal}{\mathcal{J}}

\newcommand{\Lcal}{\mathcal{L}}
\newcommand{\Mcal}{\mathcal{M}}

\newcommand{\Pcal}{\mathcal{P}}

\newcommand{\Tcal}{\mathcal{T}}
\newcommand{\Ucal}{\mathcal{U}}

\newcommand{\Wcalf}{\boldsymbol{\mathcal{W}}}
\newcommand{\Xcal}{\mathcal{X}}
\newcommand{\Ycal}{\mathcal{Y}}
\newcommand{\Zcal}{\mathcal{Z}}

\newcommand{\Ecalf}{\boldsymbol{\mathcal{E}}}

\DeclareMathOperator{\ocal}{\scriptstyle\mathcal{O}}

\newcommand{\Cfrak}{\mathfrak{C}}

\newcommand{\Dfrak}{\mathfrak{D}}



%% file: biber_setup.tex
\usepackage[
	backend=biber, 
    giveninits=true, 
    style=authoryear,	
	autocite=inline,
	mincitenames=1, 
	maxcitenames=2, 
	maxbibnames=99, 
	uniquename=false,
	uniquelist=minyear, 
	labeldateparts=true, 
	doi=false,
	url=false,
	eprint=false,
	isbn=false,
	natbib=true, 
	date=year,
    useprefix=false
	]{biblatex}
	
\DeclareNameAlias{sortname}{last-first} 

\makeatletter
\AtBeginDocument{\toggletrue{blx@useprefix}}
\makeatother

\renewbibmacro{in:}{} 
\DeclareFieldFormat{pages}{#1}
\DeclareFieldFormat{pagetotal}{#1~pages}
\DeclareFieldFormat[article]{volume}{\mkbibbold{#1}} 
\DeclareFieldFormat[article,inbook,incollection,inproceedings,patent,thesis,unpublished]{title}{#1} 

\AtEveryBibitem{\clearfield{number}}
\AtEveryBibitem{\clearfield{issue}}

%% file: 01_2_authors.tex
\author[1]{Eva-Maria Maier}
\author[2]{Alexander Fottner}
\author[3,*]{Sonja Greven}
\author[4,*]{Almond St\"ocker}

\affil[1,3]{
Humboldt-Universit\"at zu Berlin, 
Berlin, Germany
}
\affil[2]{
Augsburg University, Augsburg, Germany
}
\affil[4]{
EPFL, Lausanne, Switzerland
}
\affil[*]{\normalsize The last two authors, listed in alphabetical order, share senior authorship.}


%% file: 01_1_title.tex
Additive Density Regression

%% file: 01_3_abstract.tex
We present a structured additive regression approach to model conditional densities given scalar covariates, where only samples of the conditional distributions are observed. This links our approach to
distributional regression models for scalar data. 
The model is formulated in a Bayes Hilbert space -- preserving nonnegativity and integration to one under summation and scalar multiplication -- with respect to an arbitrary finite measure. 
This allows to consider, amongst others, continuous, discrete and mixed densities.
Our theoretical results include asymptotic existence, uniqueness, consistency, and asymptotic normality of the penalized maximum likelihood estimator, as well as confidence regions and inference for the (effect) densities.
For estimation, 
we propose to maximize the penalized log-likelihood corresponding to an appropriate  multinomial, or equivalently, Poisson regression model, which we show to approximate the original penalized maximum likelihood problem.
We apply our framework to a motivating gender economic data set from the German Socio-Economic Panel Study (SOEP), analyzing the distribution of the woman's share in a couple's total labor income given covariate effects for year, place of residence and age of the youngest child.
As the income share is a continuous variable having discrete point masses at zero and one for single-earner couples, the corresponding densities are of mixed type. 
\vspace{0.5cm}

\noindent%
{\textit{Keywords}:} 
Bayes Hilbert Space;
Distributional Regression; Structured Additive Model; Penalized Maximum Likelihood Estimation; Mixed Densities; Density-on-Scalar Regression. 

%% file: 1_introduction.tex
\section{Introduction}

We consider a regression setting with pairs of scalar response and covariate vector observations $(y_i, \xf_i)$, $i = 1, \ldots, N$.
Classical regression approaches model 
the mean of the conditional response distribution given the covariates $\xf_i$.
However, it can be very restrictive to account only for dependencies in the mean (and no further characteristics) of, typically, parametric distribution families.
\citet{kneib2023} in their `Rage Against the Mean'  review  different ``distributional regression'' approaches that overcome this limitation  by modeling the whole conditional response distribution. 
They focus on generalized additive models for location, scale and shape \citep[also sometimes  called distributional regression, e.g.,][ ]{rigby2005,klein2015aoas} working with parametric distribution families, conditional transformation models and distribution regression modeling the cumulative distribution function (cdf) of the response \citep[e.g.,][]{hothorn2014}, quantile and expectile regression \citep[e.g.,][]{koenker2005, newey1987}, and an approach 
\citep[e.g.,][]{dunson2007} 
modeling conditional response densities as a mixture of densities of some known distribution family, with parameters as well as mixture weights depending on the covariates (extended in \citet{rodriguez2025} 
to structured additive regression models for the mixture means).
While these methods allow for considerable flexibility, 
\citet{kneib2023} point out that:
\emph{``One particular difficulty encountered frequently when starting to apply distributional regression to a given data set is the challenge of interpreting the resulting regression effects. In particular, all models discussed in this review require the analyst to move beyond the simple and convenient ceteris paribus type of interpretation, where the effect of differences in one covariate of interest can be interpreted while keeping all other covariates fixed''} \citep[Section~7]{kneib2023}.

We aim to fill this gap, introducing flexible structured additive regression models for conditional response densities that allow for ceteris paribus interpretation of estimated covariate effects, while also developing inference and providing an implementation via our \texttt{R} package \texttt{DensityRegression} (developer version on \url{https://github.com/Eva2703/DensityRegression}).
Modeling the density (instead of say quantiles or the cdf) allows for straightforward interpretations and visualizations of bimodalities or shifts of probability mass depending on the covariates, and paves the way for future extensions to bi- or multivariate distributions. 
In addition, our approach does not require an assumption of (mixtures of) known parametric distributions. 

To respect nonnegativity and integration to one, we consider densities as 
elements of a Bayes Hilbert space \citep{egoz2006, vdb2014} -- an extension of the Aitchison geometry for compositional data \citep{aitchison1986}. 
The type of densities 
is specified via a reference measure, in particular allowing for mixed continuous-discrete densities (as well as the common continuous and discrete special cases).
Those occur in our motivating gender economic application, where we analyze the woman's share in a couple's total labor income in Germany depending on the couples' location in (West or East) Germany, its child status and the year:
While the conditional densities of the woman's income share can be assumed to be continuous on $(0, 1)$ for double-earner couples, they exhibit point masses at $0$ and $1$ reflecting single-earner couples.
In addition to their mixed structure, the conditional densities of the woman's income share are subject to further modeling challenges \citep{maier2021}:
The continuous component of such densities 
is known to be multi-modal for some covariate combinations, which requires flexible nonparametric modeling. 
Moreover, the categorical and continuous covariates require different types of partial effects in the model (group-specific, smooth, and interactions). 
Finally, it is of primary interest to interpret such effects on the women's income share distribution, an established object of gender economics \citep[e.g.,][]{bertrand2015,sprengholz2020}, rather than to only estimate/predict conditional distributions.
Our approach is designed to account for these challenges:
Formulating our structured additive regression models in Bayes Hilbert spaces allows for mixed densities and (odds ratio) ceteris paribus interpretation of effects building on \citet{maier2021};
representing the partial effects using a tensor product basis, building on spline bases for mixed densities developed in \citet{maier2021} as well as appropriate bases for the covariate effects, allows for flexible modeling of the shapes of the densities (in particular covering multimodality) and various types of partial effects.
The spline coefficients are estimated via penalized maximum likelihood estimation.
We show that the resulting estimator is asymptotically normally distributed and construct confidence regions and p-values for the estimated effect densities.
As the corresponding likelihood contains an integral term making maximization computationally expensive, 
we show that 
the likelihood, the penalized maximum likelihood estimator (PMLE), and its estimated covariance 
can be approximated via an appropriate multinomial, or equivalently, Poisson regression model obtained by binning the continuous component of the mixed density.

Interestingly, it turns out that modeling densities in Bayes Hilbert spaces -- effectively corresponding to a Hilbert space structure on the log-density-level -- together with  maximum likelihood-based estimation is closely related to conditional logspline density estimation \citep[e.g., ][]{stone1991,gu1995}, despite coming from a very different perspective.
%
\citet{stone1991} consider (unpenalized) maximum likelihood estimation for logspline density estimation, where continuous conditional densities given one continuous covariate are modeled as the normalized exponential function 
of a predictor given by a B-spline tensor product basis.
For identifiability, they choose a sum-to-zero constraint on the coefficients, similar to the integrate-to-zero constraint arising in the Bayes Hilbert space framework. 
In subsequent work, \citet{stone1994} also consider multivariate conditional densities with multiple covariates. 
Both manuscripts focus on theoretical results, deriving optimal rates of convergence,  
not on the (complicated) computation of the maximum likelihood estimator in practice nor on interpretation or inference. 
\citet{masse1999} address more practical problems for the setting of \citet{stone1991}, 
proposing to use a modified Newton algorithm to compute the (unpenalized) maximum likelihood estimator. 
They 
give pointwise confidence bounds, pointing out, however, that resulting bands are problematic as they do not only contain densities.
\citet{gu1995} propose penalized maximum likelihood estimation for continuous or discrete (and also multivariate) densities given scalar covariates, also based on the logistic density transform. 
They model the predictor via a (trimmed) analysis of variance (ANOVA) decomposition in a reproducing kernel Hilbert space, which is determined by the roughness functional used for penalization.
The ANOVA components are then modeled via tensor product smoothing splines, where the observations serve as knots, estimated via a Newton algorithm.
Regarding identifiabilty, \citet{gu1995} propose to set some value ``averaging'' the modeled functions  to zero, which our integrate-to-zero constraint is an example of.
They derive rates of convergence for the estimator of the conditional densities, but no asymptotic distribution and (thus) no confidence regions or p-values.
Such conditional logspline density estimation approaches are effectively covered as special cases by our density regression framework (except for differences in estimation details). The formulation in a Bayes Hilbert space adds, in particular, a novel more principled perspective on density modeling,  facilitating new modeling flexibility well beyond previous models,  
practical ways of estimation, valid  statistical inference, and means of (odds ratio-type ceteris paribus) interpretation. 
At the same time, we are able to draw on some of the previous theory.

Binning procedures for  
(conditional) 
density estimation  have been previously used:
\citet[Section 8]{eilers1996} discuss histogram smoothing for estimating continuous densities (without covariates), while \citet{gao2022} base conditional density estimation via boosted trees effectively on a 
Poisson histogram model, which they refer to as ``Lindsey's method'' \citep{lindsey1974comparison,lindsey1992}.
Our binning of the continuous density component's 
support is related, while we add a formal proof of convergence.

Bayes Hilbert spaces have been used to model conditional densities before \citep{maier2021,jeon2020additive,talska2017}. 
Different from this paper and all of the above citations, 
these approaches belong to the class of ``distributional response regression''. 
Such methods model the mean distribution/density given covariates 
with respect to a metric on probability distributions -- here the Bayes Hilbert space metric, in other cases, e.g., the $L^2$ metric after log-quantile/-hazard transforming the densities \citep{han2020} --,
typically assuming whole probability distributions/densities to be observed rather than samples thereof (with a few exceptions working with empirical distributions using an optimal transport metric, \citet{zhou2024}, or in Bayes Hilbert spaces, \citet{jeon2022}).
In the common case of only samples being available, 
these methods thus require a two-step procedure, first estimating the ``density observations'' from the sample, e.g., via kernel density estimation, before the actual analysis.
Such a two-step approach is particularly problematic if only few observations of each conditional density are available, which is usually the case in the presence of continuous covariates. 
Also, it carries the well-known problem of zero observations in parts of the domain \citep[compare, e.g.,][]{pawlowsky2015} for the case of distributional response regression in Bayes Hilbert spaces. 
We 
provide a natural solution for these problems in the following by proposing a one-step approach that models the latent conditional density building on \citet{maier2021}, while fitting the model on observed individual observations.

Our structured additive model framework formulated for densities in a Bayes Hilbert space avoids a two-step approach by working with samples of the conditional distributions, but marries distributional regression and distributional response perspectives, and yields a sound, flexible, and practically applicable framework, which to the best of our knowledge is unmatched by existing methods:
1) Our models allow for a variety of types of covariate effects (e.g., categorical, linear, smooth, interactions), 
which can be, 2), interpreted ceteris paribus.
3) Penalties may be added 
to obtain regularization.
4) We provide a unified framework for different types of densities (including discrete, continuous, and mixed ones).
5) Showing asymptotic normality, we derive valid confidence regions and p-values on the level of densities.
6) We reduce the non-standard estimation problem for mixed (and continuous) densities to an additive Poisson model over bins, showing asymptotic equivalence for decreasing bin sizes. 
7) We provide an implementation in our own \texttt{R} package \href{https://github.com/Eva2703/DensityRegression}{\texttt{DensityRegression}}, 
building on the numerous specification possibilities and estimation efficiency of the well-known \texttt{R} package \texttt{mgcv} \citep{wood2017} for estimating our Poisson models.

Our structured additive density regression model and estimation approach are introduced in Section~\ref{P2:chapter_regression_theory}, together with our theoretical results.
In Section~\ref{P2:chapter_application}, we use our method to analyze female income share distributions, before presenting a simulation study based thereon in Section~\ref{P2:chapter_simulation}
and concluding with a brief summary and discussion in Section~\ref{P2:chapter_conclusion}.

%% file: 2_poisson_regression.tex
\section{Additive Density Regression}\label{P2:chapter_regression_theory}

In this section, we present our penalized maximum-likelihood approach for estimating conditional densities from random samples.
Since densities are viewed as elements of a Bayes Hilbert space, we first introduce these spaces in Section~\ref{P2:chapter_bayes_hilbert_space}, 
before 
formulating our structured additive density-on-scalar regression model and deriving the corresponding log-likelihood function in Section~\ref{P2:chapter_bayes_space_regression}.
We show that the PMLE 
asymptotically exists and is unique, consistent and asymptotically normally distributed.
Furthermore, we construct confidence regions. 
In Section~\ref{P2:chapter_multinomial_regression}, we show that the log-likelihood, the PMLE and the penalized Fisher information can be approximated using a multinomial or equivalently Poisson distributional assumption.
Finally, we discuss interpretation of estimated effects in Section~\ref{P2:chapter_interpretation}.
Proofs are provided in appendix~\ref{P2a:chapter_proofs_main_section}.

\subsection{Bayes Hilbert spaces}\label{P2:chapter_bayes_hilbert_space}
We briefly summarize Bayes Hilbert spaces, following our more detailed introduction in appendix A of~\citet{maier2021} based on \citet{vdb2010, vdb2014}. 
These manuscripts also provide proofs for all statements.

Consider a measure space $(\Ycal, \Acal, \mu)$, where $\mu$ is a 
finite measure called \emph{reference measure}.
Denote the set of $\sigma$-finite measures with the same null sets as $\mu$ by $\Mcal (\mu) = \Mcal (\Ycal, \Acal, \mu)$.
Then,
there exists a $\mu$-almost everywhere ($\mu$-a.e.) positive and unique density with respect to $\mu$ (Radon-Nikodym derivative)
for every 
measure in $\Mcal (\mu)$. 
The addendum ``$\mu$-a.e.'' is omitted in the remaining work for the sake of readability.
We thus can identify each measure in $\Mcal (\mu)$ with its density. 
Note that densities in $\Mcal (\mu)$ 
not necessarily integrate to one.
However, proportionality with respect to positive constants 
defines an equivalence relation $\propto$ on the densities in $\Mcal (\mu)$. 
We call the corresponding set of $\propto$-equivalence classes \emph{Bayes space (with reference measure $\mu$)} and denote it by $\Bcal(\mu) = \Bcal(\Ycal, \Acal, \mu)$.
For an equivalence class $[f] \in \Bcal (\mu)$ containing densities with finite integral, the respective probability density $f / \int_\Ycal f \, \dmu$ is chosen as representative.
For the sake of readability, we omit 
the square brackets 
denoting equivalence classes in the following.
The \emph{Bayes Hilbert space (with reference measure $\mu$)} is the set 
$
B^2(\mu) = B^2(\Ycal, \mathcal{A}, \mu) 
:= \{ f \in \mathcal{B}(\mu) ~|~ \int_{\Ycal} \left(\log f \right)^2 \, \mathrm{d}\mu < \infty \}$.
It is a vector space with addition (\emph{perturbation}) and scalar multiplication (\emph{powering}) 
$
f_{1} \oplus f_{2} := f_{1} \, f_{2}$ 
and 
$
\alpha \odot f_{1} := (f_{1})^\alpha$
for 
$
f_{1}, f_{2} 
\in B^2(\mu)$ 
and $\alpha \in \mathbb{R}$. 
The additive neutral element is $0 := \frac1{\mu (\Ycal)} \in \B$, i.e., the equivalence class of constant functions, the additive inverse element of 
$
f \in B^2(\mu)$ is 
$\ominus 
f:= \frac{1}{
f}$, and the multiplicative neutral element is $1 \in \mathbb{R}$.
Consider the \emph{centered log-ratio (clr) transformation} 
$\clr [f] := \log f - \frac{1}{\mu (\Ycal)} \, \int_{\Ycal} \log f \, \mathrm{d}\mu$,
mapping an equivalence class of densities $f \in B^2(\mu)$ to a function in $L^2_0(\mu) = L^2_0(\Ycal, \mathcal{A}, \mu) := \{ \tilde{f} \in L^2(\mu) = L^2(\Ycal, \mathcal{A}, \mu) ~|~ \int_{\Ycal} \tilde{f} \, \mathrm{d}\mu = 0\}$, which is a closed subspace of $L^2(\mu)$.
The clr transformation is bijective with inverse transformation $\clr^{-1} [\ft] = \exp \ft$.
Furthermore, it is a linear transformation, i.e., it is an isomorphism.
Then, $B^2(\mu)$ is a Hilbert space with inner product 
$\langle 
f_1, f_2\rangle_{B^2(\mu)} := 
\int_{\Ycal} \operatorname{clr}[f_{1}] \cdot \operatorname{clr}[f_{2}] \,  \mathrm{d}\mu$ 
for 
$
f_1, f_2 \in \B$.
Note that per definition 
$
\langle f_{1}, f_{2} \rangle_{B^2(\mu)} = \langle \clr [f_{1}], \clr [f_{2}] \rangle_{L^2(\mu)}$, i.e., the clr transformation is isometric.

While Bayes Hilbert spaces are defined for arbitrary measurable spaces $(\Ycal, \Acal)$ as long as the reference measure $\mu$ is finite, 
we focus on $\Ycal \subset \Rbb$ with three common cases~\citep{maier2021}:
A \emph{continuous} Bayes Hilbert space corresponds to $\Ycal = \Yc = [a, b]$ for $a, b \in \Rbb$ with $a < b$,
where $\Acal = \Bcal$ is the Borel $\sigma$-algebra restricted to $\Yc$, and $\mu = \lambda$ is the Lebesgue measure.
For a \emph{discrete} Bayes Hilbert space, we have $\Ycal = \Yd = \{ t_1, \ldots , t_D\}$ with $\Acal = \Pcal (\Yd)$ the power set of $\Yd$ and $\mu = \delta := \sum_{d=1}^D w_d \delta_{t_d}$ 
a weighted sum of Dirac measures with $w_d > 0$ for all $d = 1, \ldots , D$.
A \emph{mixed} Bayes Hilbert space combines both of the aforementioned cases, with $\Ycal = \Yc \cup \Yd$, $\Acal$ 
chosen to be the smallest $\sigma$-algebra containing all closed subintervals of $\Yc$ and all points of $\Yd$, and 
$\mu = \delta 
+ \lambda$. 
If $\Yd \subset \Yc$, we obtain $\Acal = \Bcal$.
Allowing either $\Yd$ or $\Yc$ to be empty 
includes continuous and discrete Bayes Hilbert spaces as special cases, thus it suffices to consider mixed Bayes Hilbert spaces in the following. 

\subsection{Penalized Maximum Likelihood Estimation for Additive Density Regression}\label{P2:chapter_bayes_space_regression}


In this section, we introduce our penalized maximum-likelihood approach to estimate conditional densities 
given scalar covariates.
For this purpose, let $\B$ be a mixed Bayes Hilbert space and  $\Xcal = \bigtimes_{p=1}^P \Xcal_p \subset \Rbb^P, \, P \in \Nbb$, a covariate space, such that each $\Xcal_p$ is either a finite discrete set (categorical covariate) 
or a compact non-trivial interval (continuous covariate).
We consider independent observations $(y_i, \xf_i) \in \Ycal \times \Xcal$ for $i = 1, \ldots, N, ~ N \in \Nbb$, where the $y_i$ are realizations of the conditional distribution of $Y_i ~|~ 
\xf_i$, whose density $
f_{\xf_i} 
\in \B$ is not observed directly.
Assuming $(\xf, y) \mapsto f_{\xf} (y)$ is continuous on $\Xcal \times (\Yc \setminus \Yd)$ and  $\xf \mapsto f_{\xf} (y)$ is continuous on $\Xcal$ for every $y \in \Ycal$ (which is both trivially fulfilled for categorical covariates)
we consider a structured additive density-on-scalar regression model 
\begin{align}
f_{\xf_i} 
= \bigoplus_{j=1}^J h_j (\xf_i), 
\label{P2:equation_model}
\end{align}
where $J \in \Nbb$ and $h_j(\xf) \in \B$ are partial effects depending on no, one, or more covariates in $\xf = (x_1, \ldots , x_P)$.
E.g., $h_j(\xf)$ can be an intercept $\beta_0$, a linear effect $x_k \odot \beta$, a smooth effect $g_1(x_k)$, a group-specific intercept $\beta_{x_m}$ or smooth effect $g_{x_m}(x_k)$, or a smooth interaction $g_2(x_k, x_l)$, all in $\B$, where $x_k, x_l$ denote continuous covariates 
and $x_m$ a categorical covariate (grouping variable), 
for $k, l,m \in \{1, \ldots , P\}$ pairwise distinct.
For 
more possible partial effects, see also Table~B.1 in appendix B of \citet{maier2021}.
To obtain identifiable models, we usually need to include constraints on the $h_j$.
For a model including an intercept, we thus center the partial effects, i.e., 
$
\frac1{N} \odot \bigoplus_{i=1}^N h_j(\xf_i) = 0
$ for all $j = 1, \ldots, J$.
Analogously, centering of interaction effects around the main effects can be obtained.
For more details see \citet{wood2017}, in particular Section~1.8.1 for inclusion of such constraints. 
Each partial effect is represented via a tensor product basis as
\begin{align*}
h_j(\xf_i) = \left( \bfe_{\Xcal, \, j}(\xf_i) \ootimes \bfe_{\Ycal} \right)^\top \thetaf_j 
= \bigoplus_{n=1}^{K_{\Xcal, \, j}} \bigoplus_{m=1}^{K_{\Ycal}}  \theta_{j, n, m} \odot  b_{\Xcal, \, j, n}(\xf_i) \odot b_{\Ycal, \,m}
, 
\end{align*}
where $\bfe_{\Xcal, \, j} = (b_{\Xcal, \, j, 1}, \ldots , b_{\Xcal, \, j, K_{\Xcal, \, j}})^\top : \Xcal \ra \Rbb^{K_{\Xcal, \, j}}$ 
and $\bfe_{\Ycal} = ( b_{\Ycal, \,1}, \ldots, b_{\Ycal, \,K_{\Ycal}})^\top \in \B^{K_{\Ycal}}$ 
are vectors of basis functions over the covariates and over $\Ycal$, respectively, $\thetaf_j = (\theta_{j, 1, 1}, \ldots ,\theta_{j, K_{\Xcal, \, j}, K_{\Ycal}})^\top \in \Rbb^{K_{\Xcal, \, j} \, K_{\Ycal}}$ is the unknown coefficient vector, and $\ootimes$ denotes the Kronecker product of a real-valued matrix with a $\B$-valued matrix, defined by replacing multiplication with $\odot$ in the usual Kronecker product of two real-valued matrices.
Accordingly, matrix multiplication of a real-valued with a $\B$-valued matrix is defined by using $\oplus$ instead of sums and $\odot$ instead of products in the usual matrix multiplication.
Note that different (numbers of) basis functions over $\Ycal$ per $j$ are possible, but not carried out for simplicity of notation.
The basis functions $\bfe_{\Xcal, \, j}$ are chosen according to the presumed partial effect $h_j$.
E.g., a linear effect of a continuous covariate $x_k, \, k \in \{1, \ldots , P \}$, is obtained by $\bfe_{\Xcal, \, j} 
: \Xcal \ra \Rbb, \xf \mapsto x_k$, a smooth effect by a B-spline basis. 
Regularization can be obtained by including penalty matrices $\Pf_{\Xcal, \, j} \in \Rbb^{K_{\Xcal, \, j} \times K_{\Xcal, \, j}}$ 
(e.g., a difference penalty for a smooth effect)
as discussed below. 
%
Regarding the choice of $\bfe_{\Ycal} \in \B^{K_{\Ycal}}$, we first consider the discrete and continuous special cases, i.e, $\Ycal = \Yc$ and $\mu = \lambda$, or $\Ycal = \Yd$ and $\mu = \delta$.
In both cases, we choose an appropriate basis in $L^2(\mu)^{K_{\Ycal} + 1}$, transform it to $L^2_0(\mu)^{K_{\Ycal}}$ (as in appendix D of \citealp{maier2021}), 
and apply the inverse clr transformation component-wise.
As initial basis in $L^2(\mu)^{K_{\Ycal} + 1}$, we suggest a B-spline basis with pairwise distinct knots and degree at least $2$ in the continuous case and a basis consisting of the indicator functions $\mathbbm{1}_{\{t_d\}}$ for $d = 1, \ldots , D = K_{\Ycal} + 1$ in the discrete case.
Transformed difference penalties $\Pf_{\Ycal} 
\in \Rbb^{K_{\Ycal} \times K_{\Ycal}}$
\citep[as in appendix D of][]{maier2021}
can be added as discussed below to obtain smoother estimates.
In the mixed case, we construct $\bfe_{\Ycal} \in \B^{K_{\Ycal}}$ using both cases:
\citet{maier2021} show that a mixed Bayes Hilbert space $\B$, 
can be orthogonally decomposed into a continuous one, $\Bl = B^2\left( \Yc, \Bcal \cap \Yc, \lambda \right)$, and a discrete one, $\Bd = B^2\left( \Yd^\bullet, \Pcal \left( \Yd^\bullet \right), \delta^\bullet \right)$, where 
$\Yd^\bullet := \Yd \cup \{t_{D + 1}\}$ for some $t_{D+1} \in \Rbb \setminus \Yd$ and $\delta^\bullet := \sum_{d = 1}^{D + 1} w_d \, \delta_{t_d}$ with $w_{D + 1} := \lambda(\Yc)$.
Both are subspaces of $\B$ via the embeddings
$\iota_{\mathrm{c}} : \Bl \hookrightarrow \B$ with $\iota_{\mathrm{c}} (\fc) = \fc$ on $\Yc \setminus \Yd$ and $\iota_{\mathrm{c}} (\fc) = \exp ( \frac{1}{\lambda(\Yc)}\int_{\Yc} \log \fc \, \dlamb )$ 
on $\Yd$ and $\iota_{\mathrm{d}}: \Bd \hookrightarrow \B$ with $\iota_{\mathrm{d}} (\fd) = \fd \left( t_{D+1} \right)$ on $\Yc \setminus \Yd$ and $\iota_{\mathrm{d}} (\fd) = \fd$ on $\Yd$.
Thus, applying $\iota_{\mathrm{c}}$ and $\iota_{\mathrm{d}}$ to bases $\bfe_{\Yc} \in \Bl^{K_{\Yc}}$ and $\bfe_{\Yd^\bullet} \in \Bd^{K_{\Yd^\bullet}}$ component-wise yields a basis in $\B^{K_{\Ycal}}$ with $K_{\Ycal} = K_{\Yc} + K_{\Yd^\bullet}$.
The penalty matrix $\Pf_{\Ycal} 
\in \Rbb^{K_{\Ycal} \times K_{\Ycal}}$
is obtained from the ones in the special cases as the block diagonal matrix $\diag(\Pf_{\Yc}, \Pf_{\Yd})$
(yielding separate regularization in the continuous and discrete components).

To emphasize the dependence on the coefficient vector $\thetaf$, we write $f_{\xf_i, \thetaf}$ for $f_{\xf_i}$ in the following.
Applying the clr transformation to~\eqref{P2:equation_model} and using its linearity yields
\begin{align}
\ft_{\xf_i, \thetaf} 
:= \clr[f_{\xf_i, \thetaf}] 
&= \sum_{j=1}^J \sum_{n=1}^{K_{\Xcal, \, j}} \sum_{m=1}^{K_{\Ycal}}  \theta_{j, n, m} \, b_{\Xcal, \, j, n}(\xf_i) \, \tilde{b}_{\Ycal, m}  
= \sum_{j=1}^J \left( \bfe_{\Xcal, \, j}(\xf_i) \otimes \tilde{\bfe}_{\Ycal} \right)^\top \thetaf_j 
\notag
\\
&= \left( \bfe_{\Xcal}(\xf_i) \otimes \tilde{\bfe}_{\Ycal} \right)^\top \thetaf
= \bft(\xf_i)^\top \thetaf
,
\label{P2:equation_model_clr}
\end{align}
where $\tilde{\bfe}_{\Ycal} := (\tilde{b}_{\Ycal, \, 1}, \ldots , \tilde{b}_{\Ycal, \, K_{\Ycal}})^\top$ with 
$\tilde{b}_{\Ycal, m} := \clr[b_{\Ycal, \,m}]$, 
$\bfe_{\Xcal} = ( \bfe_{\Xcal, \, 1}^\top, \ldots , \bfe_{\Xcal, \, J}^\top )^\top : \Xcal \ra \Rbb^{K_{\Xcal}}$ for $K_{\Xcal} := \sum_{j=1}^J K_{\Xcal, \, j}$, 
$\thetaf = ( \thetaf_1^\top, \ldots , \thetaf_J^\top )^\top \in \Rbb^K$ for $K := K_{\Xcal} \, K_{\Ycal} 
$,
and $\bft(\xf) := \bfe_{\Xcal}(\xf) \otimes \tilde{\bfe}_{\Ycal} \in L_0^2(\mu)^K$ for $\xf \in \Xcal$.
To estimate the coefficient vector~$\thetaf$, 
consider the likelihood 
$ 
\Lcal (\thetaf) 
:= 
\prod_{i=1}^N f_{\xf_i, \thetaf} (y_i) 
= \prod_{i=1}^N\frac{\exp(\tilde{f}_{\xf_i, \thetaf}(y_i))}{\int \exp(\tilde{f}_{\xf_i, \thetaf}) \, \dmu},
$ 
which yields the log-likelihood
\begin{align}
\ell(\thetaf) 
&= \sum_{i=1}^N \Bigl( \bft(\xf_i)(y_i)^\top \thetaf 
- \log \int_{\Ycal} \exp \bigl[ \bft(\xf_i)^\top \thetaf \bigr]\, \dmu \Bigr) 
.
\label{P2:equation_bayes_log_likelihood}
\end{align}
Regularization 
is obtained by adding a penalty term 
$\pen(\thetaf) 
:= - \sum_{j=1}^J \thetaf_j^\top \Pf_{j}\thetaf_j
= - \thetaf^\top \Pf \thetaf$, 
resulting in the penalized log-likelihood
$\ell_{\mathrm{pen}} (\thetaf) := \ell (\thetaf) + \pen (\thetaf)$.
Here, 
$
\Pf_{j} = \xi_{\Xcal, \, j} (\Pf_{\Xcal, \, j} \otimes \If_{K_{\Ycal}}) + \xi_{\Ycal, \,j} (\If_{K_{\Xcal, \, j}} \otimes \Pf_{\Ycal}) 
$
penalizes the $j$-th partial effect with 
smoothing parameters $\xi_{\Xcal, \, j},\xi_{\Ycal, \, j} \geq 0$ in the respective directions, $j = 1, \ldots , J$, and
$\Pf$ is the 
block diagonal matrix $\diag (\Pf_1, \ldots , \Pf_J) \in \Rbb^{K \times K}$.
Simplifications like 
$\xi_{\Ycal, \, j} = \xi_{\Ycal}$ or $\xi_{\Xcal, \, j} = \xi_{\Ycal, \, j} = \xi_j$ (also known as isotropic penalty) for all $j \in \{1, \ldots , J\}$ are possible.
The smoothing parameters can be determined, e.g., by restricted maximum likelihood (REML) optimization \citep[Section~6.5.2]{wood2017}.
%
In the following, the special case 
$\Pf = \mathbf{0}$ (corresponding to $\xi_{\Xcal, \, j} = 0 = \xi_{\Ycal, \, j}$ for all $j = 1, \ldots , J$),
yielding the unpenalized log-likelihood, is always included,
the considered model is assumed to be identifiable,
and 
basis functions and penalty matrices are chosen as proposed above 
(basic properties necessary for the validity of the theorems/lemmas but fulfilled for standard basis choices are given in appendix~\ref{P2a:chapter_arbitrary_bases_and_penalties}).
Note that then, 
$\ell_{\pen}$ is defined on 
$
\Rbb^{K}$.

In the remainder of this section, we derive theoretical properties of the PMLE $\thetafh \in \Rbb^{K}$ of $\ell_{\pen}$.
For this purpose, 
let $\nuf$ be the product measure on $\Xcal$ consisting of Lebesgue measures for continuous covariates and sums of Dirac measures at the finitely many possible values for categorical covariates, 
let $\Sbb^{K-1}$ be the unit sphere in $\Rbb^{K}$, and
$s_{\vf, \xf} : \Ycal \ra \Rbb, y \mapsto \bft(\xf)(y)^\top \vf 
$
for $\xf \in \Xcal$ and $\vf \in \Sbb^{K-1}$.
Consider the following assumptions:
\begin{enumerate}
\setlength{\itemindent}{0.5cm}
\nitem{(BX)}\label{P2:assumption_covariate_basis_nonsingular}
For $\vf_{\Xcal} \in \Rbb^{K_{\Xcal}}$: 
If $\bfe_{\Xcal} (\xf_i)^\top \vf_{\Xcal} = 0$ for all $i = 1, \ldots , N$, then $\vf_{\Xcal} = \mathbf{0}$.
\nitem{(BX')}\label{P2:assumption_stone_1}
There exist $M' > 0$ 
and $N' \in \Nbb$ such that for all $N \geq N'$ and all $\vf_{\Xcal} 
\in \Rbb^{K_{\Xcal}}$:
$ M' \, N \int_{\Xcal} \left( \bfe_{\Xcal} (\xf)^\top \vf_{\Xcal} \right)^2 \dnuf (\xf) \leq \sum_{i=1}^N \left( \bfe_{\Xcal} (\xf_i)^\top \vf_{\Xcal} \right)^2 $.
\nitem{(S$<$)}\label{P2:assumption_scalar_product_non_constant}
For all $\vf \in \Sbb^{K-1}$ there exists $i \in \{1, \ldots , N\}$ with 
$s_{\vf, \xf_i}(y_i) < \sup_{y \in \Ycal} s_{\vf, \xf_i} (y)$.
\end{enumerate}

All of them are rather mild 
(see also appendix~\ref{P2a:chapter_assumptions}, which particularly includes the remarks/lemmas referred to in the following):
Assumption~\ref{P2:assumption_covariate_basis_nonsingular} 
is fulfilled, if the matrix
$\Bf_{\Xcal} 
\in \Rbb^{N \times K_{\Xcal}}$
containing $\bfe_{\Xcal} (\xf_i)$ in the $i$-th row, $i = 1, \ldots , N$, 
has rank $K_{\Xcal}$, see Lemma \ref{P2a:lem_assumptions}~\ref{P2a:lem_assumption_covariate_basis_nonsingular}.
This corresponds to the (marginal) design matrix having full column rank, which is a very common assumption for regression models.
In our case of $\Bf_{\Xcal}$, 
it 
can be fulfilled by suitably choosing the covariate basis for given observations:
For $j \in \{ 1, \ldots , J \}$, the basis functions $\bfe_{\Xcal, \, j}$ have to be chosen such that 
(i) their number $K_{\Xcal, \, j}$ is at most the number of unique observed values of the subvector of covariates, on which the $j$-th partial effect depends,
and
(ii) all functions contained in $\bfe_{\Xcal, \, j}$ have at least one observation in their support. 
See Remark~\ref{P2a:rem_assumptions_lem}~\ref{P2a:rem_assumption_covariate_basis_nonsingular} 
for a more detailed discussion and derivation of (i) and (ii). 
Note that, for $N \geq N'$, 
assumption~\ref{P2:assumption_covariate_basis_nonsingular} is furthermore implied by assumption~\ref{P2:assumption_stone_1}, which itself is fulfilled, if the observations of the covariates are spread reasonably over their domain with growing sample size (compare assumption~\ref{P2a:assumption_appendix_stone_2} in appendix~\ref{P2a:chapter_assumptions} 
for the formal statement), see Lemma \ref{P2a:lem_assumptions}~\ref{P2a:lem_assumption_stone_implication}.
Moreover, note that~\ref{P2:assumption_covariate_basis_nonsingular} is equivalent to the nonsingular-condition for logspline regression models \citep[e.g.,][]{stone1991}, while assumptions~\ref{P2:assumption_stone_1} and~\ref{P2a:assumption_appendix_stone_2} are inspired by~(1) and~(2) of \citet{stone1991}. 
Finally, regarding assumption~\ref{P2:assumption_scalar_product_non_constant}, note that under assumption~\ref{P2:assumption_covariate_basis_nonsingular}, for every $\vf \in \Sbb^{K-1}$ there exists an $i \in \{ 1, \ldots , N\}$ such that 
$
s_{\vf, \xf_i}$ is not constant (Lemma~\ref{P2a:lem_scalar_product_constant_implies_0}) and thus, the assumption can be fulfilled. 
Furthermore, 
under assumption \ref{P2a:assumption_appendix_stone_2}, 
not only does \ref{P2:assumption_covariate_basis_nonsingular} hold, but also the probability that~\ref{P2:assumption_scalar_product_non_constant} holds tends to one for $N \ra \infty$, see Lemma \ref{P2a:lem_assumptions}~\ref{P2a:lem_assumption_appendix_scalar_product_non_constant}. 
The following theorem then yields asymptotic existence of the PMLE $\thetafh$.

\begin{thme}\label{P2:thm_existence_uniqueness_bayes_PMLE}
Under~\ref{P2:assumption_covariate_basis_nonsingular}, $\thetafh$ is unique, if it exists.
If~\ref{P2:assumption_scalar_product_non_constant} holds as well,
$\hat{\thetaf}$ exists.
\end{thme}

Asymptotic existence is also contained as byproduct in the following Theorem~\ref{P2:thm_MLE_asymptotic_normal} on consistency and asymptotic normality.
While these two properties of the PMLE are well-known for ``regular log-likelihoods'', this term is often not specified further in the literature, particularly in our case of non-identically distributed (but independent) observations. 
In contrast, we clearly state all necessary assumptions and provide proof in appendix~\ref{P2a:chapter_proofs}, which also includes a brief discussion of literature. 
Assuming the true functions $h_j$ 
are in the span of the chosen basis functions, $j = 1, \ldots , J$, denote the unknown true parameter vector with $\varthetaf = (\varthetaf_1^\top, \ldots, \varthetaf_J^\top)^\top$.
Furthermore, define the penalized Fisher information $\Ff_{\mathrm{pen}}(\thetaf) := - D^2 \ell_{\mathrm{pen}} (\thetaf)$,
set $\xi_{\max} := \max \{\xi_{\Xcal, \, 1}, \ldots , \xi_{\Xcal, \, J}, \xi_{\Ycal, \, 1}, \ldots , \xi_{\Ycal, \, J} \}$,
and let $\Zf_N = \ocal_{\Pbb}(a_N)$ denote $\frac{\Zf_N}{a_N}$ 
converging in probability to $\mathbf{0}$ for random vectors $\Zf_N$ and real numbers $a_N, \, N \in \Nbb$.

\begin{thm}\label{P2:thm_MLE_asymptotic_normal}
Assume that~\ref{P2:assumption_stone_1} holds.
\begin{enumerate}[label=\arabic*)]
\item\label{P2:thm_item_PMLE_cosistent}
If
$\xi_{\max} 
= \ocal_{\Pbb}(N)$,
the probability that 
$\thetafh$ 
exists in an arbitrary small neighborhood around $\varthetaf$ converges to one for $N \ra \infty$ and $\thetafh$ is a consistent estimator for $\varthetaf$.
\item\label{P2:thm_item_PMLE_asymptotically_normal}
If 
$\xi_{\max} 
= \ocal_{\Pbb}(\sqrt{N})$, 
then
$\thetafh 
\overset{\text{a}}{\sim} \mathcal{N} (\varthetaf, \Ff_{\mathrm{pen}}(\varthetaf)^{-1}) 
$.
\end{enumerate}
\end{thm}

Regarding~\ref{P2:thm_item_PMLE_asymptotically_normal}, note that $\Ff_{\pen}(\thetaf)$ is invertible and equals its expectation 
for all $\thetaf \in \Rbb^K$ (see Proposition~\ref{P2a:prop_appendix_fisher_informations_positive_definite} in appendix~\ref{P2a:chapter_auxiliary_statements}).
Concerning the assumptions on the limiting behavior of
$\xi_{\max}$, 
note that the one
in \ref{P2:thm_item_PMLE_asymptotically_normal} implies the one 
in \ref{P2:thm_item_PMLE_cosistent}.
Assumptions like these are common in the literature, with similar (but more restrictive) ones, e.g., in \citet[Assumption 3]{kauermann2009} or \citet[Theorem 1]{claeskens2009}. 

Based on this asymptotic distribution, 
we construct confidence 
regions for (clr transformed) sums 
of partial effects, 
denoting the $(1- \alpha)$-quantile of the $\chi^2 (K_\Ycal)$ distribution with $\chi^2_{1- \alpha} (K_\Ycal)$ for $\alpha \in (0, 1)$ and its cumulative distribution function with $F_{\chi^2(K_{\Ycal})}$:

\begin{lem}\label{P2:lemma_confidence_regions}
Set
$ 
h_{\Jcal} (\xf)
:= \bigoplus_{j \in \Jcal} 
h_{j} (\xf) 
= \bigoplus_{j \in \Jcal} 
( \bfe_{\Xcal , \, j}(\xf) \ootimes \bfe_{\Ycal} )^\top \varthetaf_{
j}
$ 
for $\xf \in \Xcal$ and $\Jcal \subseteq \{ 1, \ldots ,J\}$.
Denote the matrix mapping the vector $\thetafh
$ to its $j$-th subvector $\thetafh_{
j}$ by $\Sf_j \in \Rbb^{K_{\Xcal, \, j} K_\Ycal \times K}$ and set
$\Af := \sum_{j \in \Jcal} 
(\bfe_{\Xcal , \, j} (\xf)^\top \otimes \Id_{K_\Ycal}) \Sf_{j} \in \Rbb^{K_\Ycal \times K}$.
Set
$\varthetaf_\Ycal := \Af \varthetaf
$,
$\thetafh_\Ycal := \Af \thetafh
$,
and
$\Vfh_\Ycal := \Af \Ff_{\mathrm{pen}}(\thetafh) ^{-1} 
\Af^\top$.
Assume~\ref{P2:assumption_stone_1} and $\xi_{\max} = \ocal_{\Pbb}(\sqrt{N})$.
Then, 
$CR_{\varthetaf_\Ycal} := \{ \thetaf_\Ycal \in \Rbb^{K_\Ycal} ~|~ (\thetaf_\Ycal - \thetafh_\Ycal)^\top \Vfh_\Ycal^{-1} (\thetaf_\Ycal - \thetafh_\Ycal) \leq \chi^2_{\alpha} (K_\Ycal)\}$, 
is an asymptotic $\alpha \cdot 100\%$ confidence 
region for $\varthetaf_\Ycal$,
$ 
\{ \bfe_\Ycal^\top \thetaf_\Ycal 
~|~ \thetaf_\Ycal \in CR_{\varthetaf_\Ycal} \} \subseteq \B
$ 
is an asymptotic $\alpha \cdot 100\%$ confidence 
region for $h_{\Jcal}(\xf)$, and 
$
\{ \bft_\Ycal^\top \thetaf_\Ycal 
~|~ \thetaf_\Ycal \in CR_{\varthetaf_\Ycal} \} \subseteq L_0^2 (\mu)
$ is an asymptotic $\alpha \cdot 100\%$ confidence 
region for 
$
\clr[h_{\Jcal}(\xf)]$.
Furthermore, $1 - F_{\chi^2(K_{\Ycal})} (\thetafh_{\Ycal}^\top \Vfh_{\Ycal}^{-1} \thetafh_{\Ycal})$ is a valid asymptotic p-value for 
$H_0: h_{\Jcal}(\xf) = 0 \in \B$ and for $H_0: \clr[h_{\Jcal}(\xf)] = 0 \in L^2_0(\mu)$.
\end{lem}
The special cases $\Jcal = \{j\}, \, j \in \{1, \ldots , J\}$, and $\Jcal = \{ 1, \ldots , J\}$ with $\xf = \xf_i$ for an $i \in \{ 1, \ldots , N\}$
yield confidence 
regions for the $j$-th partial effect and the density $f_{\xf_i, \varthetaf}$. 
Note that 
the confidence 
regions constructed in Lemma~\ref{P2:lemma_confidence_regions} are simultaneous over the domain $\Ycal$, but point-wise regarding the covariates.
Confidence regions and p-values, which are also simultaneous in the covariates for the $j$-th partial effect, can be obtained directly from the asymptotic distribution of $\thetafh_j$, compare Lemma~\ref{P2a:lemma_confidence_regions_simultaneos} in appendix~\ref{P2a:chapter_proofs}.

\subsection{Relation to Multinomial and Poisson Regression}\label{P2:chapter_multinomial_regression}
The integral term in \eqref{P2:equation_bayes_log_likelihood} makes maximizing $\ell_{\mathrm{pen}}$ 
computationally expensive and requires specialized implementations. 
%
We thus show 
that~$\ell_{\mathrm{pen}} (\thetaf)$ 
and $\Ff_{\mathrm{pen}}(\thetaf)$
can be approximated by a (shifted) penalized log-likelihood/Fisher information obtained from multinomially distributed data, which is constructed from the original response observations $y_1, \ldots , y_N$ by combining all observations of the same conditional distribution (i.e., all observations sharing identical values in all covariates) into a vector of counts via a histogram on $\Yc \setminus \Yd$ and counts on~$\Yd$.
We show convergence for decreasing bin width and that this also transfers to the PMLEs 
under mild assumptions.
Furthermore, we derive that Poisson regression can be used equivalently to multinomial regression.
This is very appealing, 
since Poisson regression is well-known and flexible additive implementations are available.
Thus, in our \texttt{R} package \href{https://github.com/Eva2703/DensityRegression}{\texttt{DensityRegression}} 
we build on \texttt{mgcv} \citep{wood2017} for fitting Poisson models. 

To formally construct the multinomial data, let $\xf^{(1)}, \ldots , \xf^{(L)}\in\Xcal$ be the unique combinations of observed covariate values and let $\Ical^{(l)} := \{ i \in \{ 1, \ldots , N \} ~|~ \xf_i = \xf^{(l)} \}$ for $l = 1, \ldots , L \leq N$ denote the corresponding partition of indices.
Thus, $\{ y_i ~|~ i \in \Ical^{(l)} \}$ forms an independent and identically distributed random sample of $Y^{(l)} ~|~ \xf^{(l)}$.
Note that in the presence of continuous covariates, we usually have distinct observed covariate values, i.e., $\Ical^{(l)} = \{l\}$ 
for $l = 1, \ldots , L = N$, with each observation sampled from a different conditional distribution.
In the following, let $l \in \{ 1, \ldots , L\}$.
The set $\Ical^{(l)}$ is partitioned depending on the corresponding response observation belonging to the discrete component or not, i.e.,
$\Ical_{\mathrm{d}}^{(l)} := \{ i \in \Ical^{(l)} ~|~ y_i \in \Yd 
\}$ 
and 
$\Ical_{\mathrm{c}}^{(l)} := \{ i \in \Ical^{(l)} ~|~  y_i \in \Yc \setminus \Yd \} 
$. 
From the latter, we construct a histogram for the continuous component of the mixed Bayes Hilbert space, approximating the conditional density of $Y^{(l)} ~|~ \xf^{(l)}$ corresponding to the $l$-th covariate combination.
Thus, if $\Yc \neq \emptyset$, let $G^{(l)} \in \Nbb$ 
be the numbers of histogram bins. 
Consider a partition 
$\Zcal^{(l)} = (a_0^{(l)}, a_1^{(l)}, \ldots, a_{G^{(l)}}^{(l)}) \in \Xi := \{ (a_0, a_1, \ldots , a_G) ~|~ a = a_0 < a_1 < \ldots  < a_G = b, ~ G \in \Nbb \}$ of the interval $\Yc = [a, b]$, 
yielding disjoint bins $U_g^{(l)} := [a_{g-1}^{(l)}, a_{g}^{(l)})$ for $g= 1, \dots, G^{(l)} - 1$ and $U_{G^{(l)}}^{(l)} := [a_{G^{(l)}-1}^{(l)}, a_{G^{(l)}}^{(l)}]$.
Denote the width of bin $U_g^{(l)}$ by 
$\Delta_g^{(l)} := 
\mu (U_g^{(l)}) =
a_{g}^{(l)} - a_{g-1}^{(l)} > 0$. 
If there exists a unique $y_{g}^{(l)} \in \{y_i ~|~ i \in \Ical_{\mathrm{c}}^{(l)} \} \cap U_g^{(l)}$, 
choose $u_g^{(l)} = y_{g}^{(l)} 
$, otherwise let $u_g^{(l)} \in U_g^{(l)}$ be 
the bin midpoint.
The values of the histogram 
then are $n_g^{(l)} := \sum_{i \in \Ical_{\mathrm{c}}^{(l)}} \mathbbm{1}_{U_g^{(l)}}(y_i)$ for $g = 1, \ldots , G^{(l)}$.
Otherwise, if $\Yc = \emptyset$, 
set $G^{(l)} = 0$ and $\Zcal^{(l)} = \emptyset$. 
For the discrete component of the mixed Bayes Hilbert space, we 
similarly count the number of response observations being equal to each discrete value:
For $g = G^{(l)} + 1, \ldots , G^{(l)} + D =: \Gamma^{(l)}$, set $U_g^{(l)} := \{t_{g - G^{(l)}}\}$, 
$\Delta_g^{(l)} := 
\mu (U_{g}^{(l)}) = 
w_{g - G^{(l)}}$, 
$u_g^{(l)} := t_{g - G^{(l)}}$, and $n_g^{(l)} := \sum_{i \in \Ical_{\mathrm{d}}^{(l)}} \mathbbm{1}_{U_g^{(l)}}(y_i)$.

We view the vectors $(n_1^{(l)}, \ldots , n_{\Gamma^{(l)}}^{(l)})$ as realizations of $L$ independent multinomially distributed random variables 
with probabilities $p_1^{(l)}, \ldots , p_{\Gamma^{(l)}}^{(l)}$, conditional on $n^{(l)} := \sum_{g = 1}^{\Gamma^{(l)}} n_g^{(l)} = | \Ical^{(l)}|$.
We model them via
$p_g^{(l)} (\thetaf)
= \frac{\Delta_g^{(l)}\exp [ \bft(\xf^{(l)})(u_g^{(l)})^\top \thetaf]}{\sum_{g'=1}^{\Gamma^{(l)}} \Delta_{g'}^{(l)}\exp [ \bft(\xf^{(l)})(u_{g'}^{(l)}) ^\top \thetaf ]}
$
for an unknown coefficient vector
$\thetaf
\in \Rbb^{K}$, where $\bft$ is the tensor product basis vector 
as in~\eqref{P2:equation_model_clr}.
Note that 
the integrate-to-zero constraint of $\tilde{\bfe}(\xf^{(l)}) \in L^2_0(\mu)^K$ preserves identifiability despite the fraction. 
%
%
Omitting additive constants with respect to $\thetaf$ (indicated by $\propto$), the multinomial log-likelihood is
\begin{align}
&\sum_{l=1}^L \sum_{g=1}^{\Gamma^{(l)}} n_g^{(l)} \log \frac{\Delta_g^{(l)}\exp \bigl[ \bft(\xf^{(l)})(u_g^{(l)}) ^\top \thetaf \bigr]}{\sum_{g'=1}^{\Gamma^{(l)}} \Delta_{g'}^{(l)}\exp \bigl[ \bft(\xf^{(l)})(u_{g'}^{(l)})^\top \thetaf \bigr]} \notag 
\\
\propto & \sum_{l=1}^L \sum_{g=1}^{\Gamma^{(l)}} n_g^{(l)} \Bigl( \bft(\xf^{(l)})(u_g^{(l)})^\top \thetaf - \log  
\sum_{g'=1}^{\Gamma^{(l)}} \Delta_{g'}^{(l)}\exp \bigl[ \bft(\xf^{(l)})(u_{g'}^{(l)})^\top \thetaf \bigr]
\Bigr) 
=: \ell^{\mathrm{mn}}_\Zcal (\thetaf) ,
\label{P2:equation_multinomial_log_likelihood}
\end{align}
where the dependence on $\Zcal := ( \Zcal^{(1)} , \ldots , \Zcal^{(L)} ) 
$ is implicit via $\Gamma^{(l)} = G^{(l)} + D, 
n_g^{(l)}, u_g^{(l)},$ and $\Delta_g^{(l)}$ for $l = 1, \ldots , L$ and $g = 1, \ldots , G^{(l)} 
$.
Its penalized version, using the same penalty term as in Section~\ref{P2:chapter_bayes_space_regression}, is $\ell^{\mathrm{mn}}_{\Zcal, \mathrm{pen}} (\thetaf) := \ell^{\mathrm{mn}}_\Zcal (\thetaf) + \pen (\thetaf)$.
Consider the maximal histogram bin width $\Delta = \Delta(\Zcal) := \max \{ \Delta_g^{(l)} ~|~ l = 1, \ldots , L, ~ g = 1, \ldots , G^{(l)} \}$ and the penalized Fisher information 
$\Ff_{\Zcal, \mathrm{pen}}^{\mathrm{mn}}(\thetaf) := - D^2 \ell_{\Zcal, \mathrm{pen}}^{\mathrm{mn}}(\thetaf)$. 

\begin{thme}\label{P2:thm_approximation_bayes_multinomial}
For all $\thetaf \in \Rbb^K$, 
$\lim_{\Delta \ra 0} \ell^{\mathrm{mn}}_{\Zcal, \mathrm{pen}} (\thetaf) = \ell_{\mathrm{pen}} (\thetaf)$,
$\lim_{\Delta \ra 0} \Ff^{\mathrm{mn}}_{\Zcal, \mathrm{pen}} (\thetaf) = \Ff_{\mathrm{pen}} (\thetaf)$,
and (if the inverses exist) 
$\lim_{\Delta \ra 0} \Ff^{\mathrm{mn}}_{\Zcal, \mathrm{pen}} (\thetaf)^{-1} = \Ff_{\mathrm{pen}} (\thetaf)^{-1}$.
If $\Yc = \emptyset$, all equalities hold without limit.
\end{thme}

The proof in appendix~\ref{P2a:chapter_proofs} uses convergence of Riemann sums and that for small enough bin width, each bin contains at most one observed value.
We furthermore show in Proposition~\ref{P2a:prop_appendix_fisher_informations_positive_definite} in appendix~\ref{P2a:chapter_auxiliary_statements} that $\Ff^{\mathrm{mn}}_{\Zcal, \mathrm{pen}} (\thetaf)$ is invertible and equals its expectation, 
if~\ref{P2:assumption_covariate_basis_nonsingular} and the 
following assumption hold:
\begin{enumerate}
\setlength{\itemindent}{0.5cm}
\nitem{(BU)}\label{P2:assumption_domain_basis_nonsingular}
For $\Zcal \in \Xi^L 
$, 
$l \in \{ 1, \ldots , L\}$,
$\vf_{\Ycal} \in \Rbb^{K_\Ycal}$, and $c \in \Rbb$:
If $\bft_\Ycal (u_{g}^{(l)})^\top \vf_{\Ycal} = c 
$ for all $g = 1, \ldots , \Gamma^{(l)}$, then
$\vf_{\Ycal} = \mathbf{0}$ (and $c = 0$). 
\end{enumerate}

This assumption is 
fulfilled, if the matrices
$(\Bft_{\Ycal}^{(l)} : \mathbf{1}_{\Gamma^{(l)}}) 
\in \Rbb^{\Gamma^{(l)} \times (K_\Ycal + 1)}$ have rank $K_\Ycal + 1$ for all $l = 1, \ldots , L$, 
where
$\Bft_{\Ycal}^{(l)}$ 
is the matrix containing $\bft_\Ycal (u_g^{(l)})$ in the $g$-th row (given $\Zcal^{(l)}$), 
$\mathbf{1}_{\Gamma^{(l)}}$ is the vector of length $\Gamma^{(l)}$ containing $1$ 
in every entry,
and $(\cdot : \cdot)$ denotes matrix concatenation.
If $\Yc = \emptyset$ (discrete special case), this rank condition is immediately fulfilled.
If $\Yc \neq \emptyset$, 
it is sufficient that (i) 
the number of basis functions over the continuous component plus $1$ is at most the length of the vector of bins, $K_{\Yc} + 1 \leq G^{(l)}$ for all $l = 1, \ldots , L$, and (ii) 
for all functions contained in the (unconstrained) B-spline basis yielding $\bfe_{\Yc}$ and for all $l \in \{ 1, \ldots , L \}$, there is at least one $g \in \{ 1, \ldots , G^{(l)}\}$ such that $u_{g}^{(l)}$ is in its support.
Both, (i) and (ii), and thus also~\ref{P2:assumption_domain_basis_nonsingular}, hold for small enough bin width.
See Lemma~\ref{P2a:lem_assumptions}~\ref{P2a:lem_assumption_appendix_domain_basis_nonsingular}
in appendix~\ref{P2a:chapter_assumptions} for proof.
Hence, a sequence $\Zcal_m \in \Xi^L$, $m \in \Nbb$, such that~\ref{P2:assumption_domain_basis_nonsingular} holds for every $\Zcal_m, m \in \Nbb,$ and $\lim_{m \ra \infty} \Delta (\Zcal_m) = 0$ (as considered in 
the following theorem transferring convergence to the PMLEs) exists.

\begin{thme}\label{P2:thm_approximation_bayes_multinomial_MLE}
Assume that 
\ref{P2:assumption_covariate_basis_nonsingular} holds.
\begin{enumerate}[label=\arabic*)]
\item\label{P2:thm_approximation_multinomial_MLE_unique}
Under~\ref{P2:assumption_domain_basis_nonsingular}, 
$\hat{\thetaf}_\Zcal^{\mathrm{mn}}$ is unique, 
if it exists.
\item\label{P2:thm_approximation_bayes_multinomial_main_statement}
Assume 
that~\ref{P2:assumption_scalar_product_non_constant}
holds and $\Yc \neq \emptyset$.
Then, 
for any sequence $\Zcal_m \in \Xi^L$, $m \in \Nbb$, such that~\ref{P2:assumption_domain_basis_nonsingular} holds for every $\Zcal_m, m \in \Nbb$, and $\lim_{m \ra \infty} \Delta (\Zcal_m) = 0$, there exists an $M \in \Nbb$ such that $\thetafh_{\Zcal_m}^{\mathrm{mn}}$ exists for all $m \geq M$ 
and
$\lim_{m \ra \infty} \hat{\thetaf}_{\Zcal_m}^{\mathrm{mn}} = \hat{\thetaf}$. 
\end{enumerate}
\end{thme}

Equivalently to the multinomial model constructed above, we can consider a Poisson model.
Compared to earlier work on this well-known equivalence \citep[e.g.,][]{baker1994}, 
we include a penalty term for the coefficient vector $\thetaf$.
Denote the upper left $K \times K$ submatrix of a matrix $\Af$ with $\Af_{[K]}$. 


\begin{thme}\label{P2:proposition_equivalence_multniomial_poisson}
%
For observations $\mathbf{n}^{(l)} = (n_1^{(l)}, \ldots , n_{\Gamma^{(l)}}^{(l)}) \in \Nbb_0^{\Gamma^{(l)}}$ 
with $l = 1, \ldots , L$, consider two different distributional assumptions:
\begin{enumerate}[label=\arabic*)]
\item\label{P2:proposition_equivalence_multniomial_poisson_assumption_poisson}
$n_g^{(l)}$ are realizations of Poisson variables $N_g^{(l)} \sim Po ( \lambda_g^{(l)} ( \thetaf, \tauf 
) )$, 
independent for $g = 1, \ldots , \Gamma^{(l)}, \, l = 1, \ldots , L$, where the Poisson rates are modeled via 
$ 
\lambda_g^{(l)} ( \thetaf , \tauf ) 
= 
\exp ( \eta_g^{(l)} (\thetaf) + \tau^{(l)} )
,
$ 
with a predictor $\eta_g^{(l)} (\thetaf) \in \Rbb$ for an unknown coefficient vector 
$\thetaf \in 
\Rbb^{K}$ and 
$\tauf = (\tau^{(1)}, \ldots , \tau^{(L)} )^\top \in \Rbb^L $.
Set $\lambda^{(l)} := \sum_{g=1}^{\Gamma^{(l)}} \lambda_g^{(l)} (\thetaf , \tauf)$, $n^{(l)} := \sum_{g=1}^{\Gamma^{(l)}} n_g^{(l)}$,  $\lambdaf = (\lambda^{(1)} , \ldots , \lambda^{(L)})$, and $\nf = (n^{(1)}, \ldots , n^{(L)})$ for $l = 1, \ldots , L$.
Given an additive penalty $\pen(\thetaf)$ on the log-likelihood, denote the PMLE for $(\thetaf, \tauf)$ with $(\hat{\thetaf}^{\mathrm{po}}, \taufh)$, 
the 
penalized Fisher information at $(\thetafh^{\mathrm{po}}, \taufh)$
with $\Ffh_{\mathrm{pen}}^{\mathrm{po}}$
and the 
expected penalized Fisher information (with respect to the Poisson distributional assumption given $(\thetaf, \tauf)$) 
on the restricted parameter space 
$\{ (\thetaf , \tauf) \in \Rbb^K 
\times \Rbb^L ~|~ \lambdaf = \nf \}$ with $\Ebb^{\mathrm{po}}_{\thetaf, \tauf} ( \Ff_{\mathrm{pen}}^{\mathrm{po}} ~|~ \lambdaf = \nf)$. 
\item\label{P2:proposition_equivalence_multniomial_poisson_assumption_multinom}
$\mathbf{n}^{(l)}$ are realizations of multinomial vectors 
$\mathbf{N}^{(l)} \sim M ( n^{(l)}, p_1^{(l)}(\thetaf), \ldots , p_{\Gamma^{(l)}}^{(l)}(\thetaf) \, | $ 
$n^{(l)} = \sum_{g=1}^{\Gamma^{(l)}} n_g^{(l)} )$, independent for $l = 1, \ldots , L$, where the probabilities are modeled via 
$ 
p_g^{(l)}(\thetaf) 
= $
$\frac{
\exp [ \eta_g^{(l)} (\thetaf) ]}{\sum_{g'=1}^{\Gamma^{(l)}} 
\exp [ \eta_{g'}^{(l)} (\thetaf) ]}
\,,
$ 
with $\eta_g^{(l)} (\thetaf) \in \Rbb$ the same predictor as in \ref{P2:proposition_equivalence_multniomial_poisson_assumption_poisson}.
Given the penalty $\pen(\thetaf)$ from \ref{P2:proposition_equivalence_multniomial_poisson_assumption_poisson}, denote the PMLE for $\thetaf$ with $\hat{\thetaf}^{\mathrm{mn}}$,
the penalized Fisher information at $\hat{\thetaf}^{\mathrm{mn}}$ with $\Ffh_{\mathrm{pen}}^{\mathrm{mn}}$,
and 
the expected penalized Fisher information (with respect to the multinomial distributional assumption given $\thetaf$) with $\Ebb^{\mathrm{mn}}_{\thetaf} ( \Ff_{\mathrm{pen}}^{\mathrm{mn}})$.
\end{enumerate}
Then, 
the PMLEs for $\thetaf$ and the inverse (expected) penalized Fisher informations corresponding to $\thetaf$ are identical: 
$\hat{\thetaf}^{\mathrm{po}} = \hat{\thetaf}^{\mathrm{mn}}$, $((\Ffh_{\mathrm{pen}}^{\mathrm{po}})^{-1})_{[K]} = (\Ffh_{\mathrm{pen}}^{\mathrm{mn}})^{-1}$, and 
$((\Ebb^{\mathrm{po}}_{\thetaf, \tauf} ( \Ff_{\mathrm{pen}}^{\mathrm{po}} \, |$ 
$\lambdaf = \nf))^{-1})_{[K]}
= (\Ebb^{\mathrm{mn}}_{\thetaf} ( \Ff_{\mathrm{pen}}^{\mathrm{mn}}))^{-1}$.
\end{thme}

Note that in general $(\Af^{-1})_{[K]}\neq (\Af_{[K]})^{-1}$, and hence the identities regarding the (sub-matrices of) Fisher informations do not hold without inversion.
Furthermore, note that the intercepts $\tau^{(l)}$ indirectly model the absolute counts $n^{(l)}$ for the $l$-th covariate combination (see the proof in appendix~\ref{P2a:chapter_proofs} for details on the relation), which are of interest in a Poisson distribution, while the probabilities modeled in the multinomial distribution only carry relative information, which is contained completely in $\thetaf$.
The predictor should be chosen in such a way that the coefficients $\thetaf$ (and intercepts $\tau^{(l)}$ in the Poisson distribution assumption) are identifiable.
When using Theorem~\ref{P2:proposition_equivalence_multniomial_poisson} for our regression approach, we choose the predictor as
$ 
\eta_g^{(l)} (\thetaf)
= \log \Delta_g^{(l)} + \bft(\xf^{(l)})(u_g^{(l)}) ^\top \thetaf,
$ 
for $l = 1, \ldots , L$ and $g = 1, \ldots , \Gamma^{(l)}$.
%
The integrate-to-zero constraint of each component of $\bft (\xf^{(l)}) \in L_0^2(\mu)^{K}$  
ensures that the coefficients are identifiable regarding additive constants. 
As before, we use the penalty
$\pen(\thetaf) 
= - \thetaf^\top \Pf \thetaf $.
However, Theorem~\ref{P2:proposition_equivalence_multniomial_poisson} is also valid for a general penalty (or prior taking a Bayesian view). 

Theorems~\ref{P2:thm_approximation_bayes_multinomial_MLE} and~\ref{P2:proposition_equivalence_multniomial_poisson} imply that we can use Poisson regression to 
estimate the parameter vector in 
$\ell_{\mathrm{pen}}$ 
and thus build our approach's implementation on existing ones.
Furthermore, justified by our theorems, 
we may replace $\thetafh$ with $\thetafh^{\mathrm{po}}$ and 
$\Ff_\mathrm{pen}(\thetafh)^{-1}$ 
with $((\Ffh_\mathrm{pen}^{\mathrm{po}})^{-1})_{[K]}$ in Lemma~\ref{P2:lemma_confidence_regions} to construct confidence regions in practice.

Note that there might be numerous unique covariate combinations, 
in particular if metric covariates are involved. 
Thus, the number of intercepts $\tau^{(l)}$ to estimate for $l = 1, \ldots , L$ in the Poisson regression model can get very large.
In practice, this problem might be circumvented by modeling the intercepts smoothly in the covariate basis as well, i.e., $\tau^{(l)} = \bfe_{\Xcal}(\xf^{(l)})^\top \tauf$, where $\tauf \in \Rbb^{K_{\Xcal}}$.
However, as in any additive Poisson regression model like here for the absolute counts $n^{(l)}$, one needs to carefully inspect the counts to decide whether this smoothness assumption is reasonable.

%
In the beginning of this section, we constructed one 
vector of counts per unique observed covariate combination $\xf^{(l)}$, $l \in \{1, \ldots , L\}$, combining all observations of the same conditional distribution. 
However, the construction described above can also be carried out for partitions of the corresponding set of indices $\Ical^{(l)}$, without changing estimation results, 
compare Lemma~\ref{P2a:lemma_splitting_histograms} in appendix~\ref{P2a:chapter_proofs}.
In the most extreme case, 
this yields one vector of counts per observation (with one count of $1$ and all others $0$).
While this is usually not desirable with respect to computation time in practice, since it inflates the size of the data set artificially, it shows the estimation does not depend on aggregation of observations and the presented approach is thus overall coherent.

\subsection{Interpretation}\label{P2:chapter_interpretation}

An advantage of formulating our structured additive regression models in Bayes Hilbert spaces is 
that it allows for ceteris paribus interpretation of partial effects as proposed by \citet{maier2021}.
In the following, we briefly summarize two of their interpretation methods. 
For this purpose, let $h_j = h_j(\xf) \in \B,j\in\{ 1, \ldots , J\}$ be a partial effect as in model~\eqref{P2:equation_model} and $h_0 = h_0(\xf) \in \B$ the sum of the remaining partial effects, omitting the covariates $\xf \in \Xcal$ in the notation for the sake of readability.

For the first interpretation method, 
consider the infinitesimal odds ratio 
$\OR_j(s,t) := \frac{(h_0(s)\oplus h_j(s))/(h_0(t) \oplus h_j(t))}{h_0(s)/h_0(t)} = h_j(s)/h_j(t)$ for $s, t \in \Ycal$.
As then $\log \OR_j(s,t) 
= \clr[h_j](s) - \clr[h_j](t)$, the log odds ratio  can be easily read off as vertical differences from plots of the clr transformed effects. 
As $\OR_j(s, t)$ is independent of $h_0$, these odds ratios can be interpreted ceteris paribus.
In addition to the infinitesimal interpretation,
\citet[Proposition~3.1~(b)]{maier2021} show that $\OR_j(s,t)$ corresponds to the limit of the usual odds ratio in the vicinity of $s$ and~$t$. 
Furthermore, Proposition~3.1~(a) in \citet{maier2021} implies that if $\OR_j (s, t) > 1$ for all $s \in A, \, t \in B$, where $A, B \in \{C \in \Acal ~|~ \mu (C) > 0 \}$, we have $\Pbb(A) / \Pbb(B) > \Pbb_0(A) / \Pbb_0(B)$ 
where $\Pbb$ and $\Pbb_0$ denote the probability measures corresponding to the densities $h_0 \oplus h_j$ and $h_0$.

The second approach 
examines shifts of probability masses under summation: 
For $\alpha \in \Rbb$, let 
$A_+ := \{ h_j 
\geq \alpha \}$ and 
$A_- := 
\{ h_j 
< \alpha \}$.
Then,
$\int_{A_+} h_0 \oplus h_j 
\, \dmu \geq \int_{A_+} h_0 
\, \dmu 
$
and 
$\int_{A_-} h_0 \oplus h_j 
\, \dmu \leq \int_{A_-} h_0 
\, \dmu 
$ \citep[Theorem~F.1]{maier2021}, i.e., any threshold $\alpha$ decomposes $\Ycal$ 
into two sets -- one where the probability mass of any density 
$h_0 
\in \B$ 
increases under perturbation with $h_j$ 
and one where it decreases.
Note the result also holds, when constructing $A_+$ and $A_-$ based on $\clr [h_j]$ instead of $h_j$, since the clr transformation retains a density's monotonic behavior.


%% file: 4_application.tex
\section{Application}\label{P2:chapter_application}
The density of the woman's share in a couple's total labor income is considered in several gender economic contributions \citep[e.g.,][]{bertrand2015, sprengholz2020}.
A brief discussion thereof is available in \citet{maier2021}, who (in contrast to the aforementioned gender economic literature) analyze the dependence of such income share densities on covariates.
Their density-on-scalar regression approach requires preliminary estimation of densities from samples of the conditional distributions, which then serve as response densities for the regression.
Our newly developed method avoids this two-step approach by directly working with the samples 
and additionally provides inference for the effect functions.
In the following, we apply our modeling approach to the same data set considered by \citet{maier2021}, which is derived from the German Socio-Economic Panel (SOEP) 
(version 33, doi:10.5684/soep.v33, see \citealp{soep2019}).
It contains $N = 154,924$ 
observations of opposite-sex couples living together in one household (not necessarily married), with at least one partner reporting positive labor income. 

We aim to model conditional densities of the woman's 
share $s$ in a couple's total gross labor income, given three covariates.
The binary covariate \emph{West\_East} indicates whether the couple's household is located in \emph{West} or \emph{East} Germany (including Berlin).
The covariate \emph{c\_age} categorizes the couples via the age of the youngest child living in their household into three groups: \emph{0-6} (zero to six years), \emph{7-18} (seven to 18 years), and \emph{other} (no minor children). 
The third covariate \emph{year} ranges from 1984 for \emph{West} Germany and from 1991 for \emph{East} Germany to 2016.
Given these covariates, we consider the conditional densities $f_{\text{\emph{West\_East, 
c\_age, year}}}: [0, 1] \ra \Rbb^+
$ 
of the share $s$ as elements of a mixed Bayes Hilbert space $\B = B^2([0, 1], \Bcal, \mu)$ with reference measure $\mu := \delta_0 + \lambda + \delta_1$, since the boundary values $0$ and~$1$ of the share, corresponding to single-earner households, have positive probability mass.
We estimate the model
\begin{align}
f_{\text{\emph{West\_East, c\_age, year}}} &= \beta_0 \oplus \beta_{\text{\emph{West\_East}}} \oplus \beta_{c\_age} \oplus \beta_{\text{\emph{c\_age, West\_East}}} \notag \\
&\hspace{0.45 cm} \oplus g(year) \oplus g_{\text{\emph{West\_East}}} (year) \oplus g_{c\_age} (year) \notag \\
&\hspace{0.45 cm} \oplus g_{\text{\emph{c\_age, West\_East}}}(year)
,
\label{P2:soep_model}
\end{align}

where all summands are elements of the Bayes Hilbert space $B^2 (\mu)$, i.e., densities of~
$s$.
We use reference coding with the same reference categories as \citet{maier2021}, namely $\text{\emph{West\_East}} = \text{\emph{West}}, \, \text{\emph{c\_age}} = \text{\emph{other}}$, and $\text{\emph{year}} = 1991$ 
(with the latter, for a continuous covariate, not directly specifiable in \texttt{mgcv}'s function \texttt{gam()}, which our \texttt{R} package \href{https://github.com/Eva2703/DensityRegression}{\texttt{DensityRegression}} 
uses for fitting Poisson models, requiring a transformation as discussed in appendix~\ref{P2a:chapter_reference_coding_application}). 
The intercept $\beta_0$ corresponds to the density given these covariate values.
The other effects then describe the deviation from the reference for different values of the respective covariate, e.g., the smooth effect $g(\text{\emph{year}})$ gives the deviation for each \emph{year} from the reference 1991 (for \emph{West} Germany and \emph{c\_age} \emph{other}).
The model also contains several interaction terms, namely, a group-specific intercept $\beta_{\text{\emph{c\_age, West\_East}}}$ as well as group-specific smooth effects $g_{\text{\emph{West\_East}}} (\text{\emph{year}})$, $g_{\text{\emph{c\_age}}} (\text{\emph{year}})$, and $g_{\text{\emph{c\_age, West\_East}}}(\text{\emph{year}})$.
Note that more covariates might be included (at the expense of longer computation time), however for comparability purposes we consider (almost) the same model as in \citet{maier2021}.
The only difference of model~\eqref{P2:soep_model} to theirs is that they include an additional group specific intercept corresponding to a fourth covariate \emph{region}, which contains a finer disaggregation of the covariate \emph{West\_East} into a total of six regions.
Since they discovered almost no variation within the regions corresponding to \emph{West} and \emph{East} Germany, respectively, we decided to not include this additional covariate for the sake of model parsimony. 
The basis functions $\bfe_{\Ycal} \in \B^{K_{\Ycal}}$ are obtained from cubic B-splines for the continuous and indicator functions for the discrete component as described in Section~\ref{P2:chapter_bayes_space_regression}.
Like \citet{maier2021} we use no penalty, setting $\xi_{\Ycal, \, j} = 0$ for all $j = 1, \ldots , J$. 
Since their two-step approach includes pre-smoothing by estimating the response densities, the number of basis functions 
used in the actual regression model is not comparable and we use a lower number ($K_{\Ycal} = 15$) to obtain a similar level of smoothness.
For the (group-specific) smooth effects, we use cubic B-splines with second order difference penalty for $\bfe_j$ and $\Pf_j$
and REML optimization to determine the smoothing parameters $\xi_{\Xcal, \, j}$ (given in appendix~\ref{P2a:chapter_application_smoothing_parameters}).
For histogram construction on $(0, 1)$, we use $G = 100$ equidistant bins for each covariate combination.

Note that the data set contains a sample weight 
per household, 
which we basically incorporate 
by using weighted counts as response data. 
However, these sample weights and thus the weighted counts are not natural numbers, technically not valid as response observations in a Poisson model. 
Thus, we construct an equivalent log-likelihood 
by weighting the respective log-likelihood contribution 
with the quotient of the weighted and the unweighted count 
and subtracting the quotient's logarithm from the predictor. 
See appendix~\ref{P2a:chapter_weights} for details. 

For the level $\alpha = 5\%$ (even for $\alpha = 1\%$), all effects except for the group-specific smooth effects $g_{\text{\emph{West\_East}}} (year)$ and $g_{\text{\emph{c\_age, West\_East}}} (year)$ 
are significant, compare the (simultaneous) p-values shown in Table~\ref{P2a:table_pvalues_simultan} in appendix~\ref{P2a:chapter_application_pvalues}.
There, we also give p-values which are pointwise in the covariates in Table~\ref{P2a:table_pvalues_pointwise}.
In the following, we discuss the main effects of model~\eqref{P2:soep_model}. 
All effects are shown in appendix~\ref{P2a:chapter_application_effects}, in particular the interaction effects, for which we cannot go into detail here, due to space constraints.
In Figures~\ref{P2:fig_estimated_West_East} to~\ref{P2:fig_estimated_year}, 
the left side illustrates the sum 
of the intercept and the respective main effect (i.e., the conditional densities given the different values of the considered covariate and the reference values for the remaining covariates), 
while the right side shows the clr transformed effects. 
The latter is useful for interpreting effects via log odds ratios or shifts of probability mass, see Section~\ref{P2:chapter_interpretation}. 
Note that for the discrete values $0$ and~$1$, the density values are visualized as dashes and shifted slightly outwards for better distinction.
%

The estimated effect of \emph{West\_East} is shown in Figure~\ref{P2:fig_estimated_West_East} via the colored lines.
The grey ones illustrate samples (of size $100$) from the respective $95\%$ confidence regions constructed as in Lemma~\ref{P2:lemma_confidence_regions} (simultaneous over $\Ycal = [0, 1]$ and point-wise regarding the covariates).
Note that $\betah_0 \oplus \betah_{West} = \betah_0$ and $\clr [\betah_{West}] = 0$ due to reference coding.
Considering double-earner couples ($s \in (0, 1)$) in the left part of the figure, we observe both estimated conditional densities have a mode around $0.45$.
For couples in \emph{West} Germany (without minor children in 1991) the density has a tendency towards a second mode around $s = 0.1$. 
For the former communist \emph{East} Germany, there is a shift of probability mass to the right yielding a more symmetric density.
This is also indicated by $\clr [\betah_{East}]$ on the right: 
Using the value $\alpha \approx -0.65$, 
we obtain that the probability mass of the estimated densities in \emph{East} Germany decreases on $
\{ \clr [\betah_{East}] < \alpha \} = (0, 0.2)$ compared to \emph{West} Germany (ceteris paribus).
Regarding the boundary values, 
nonworking women ($s = 0$) are more frequent in \emph{West} than in \emph{East} Germany, while single-earner women ($s = 1$) are more frequent in \emph{East} Germany.
We can quantify this via the log odds ratio of $\betah_{\text{\emph{East}}}$ and $\betah_{\text{\emph{West}}}$ for $s = 1$ compared to $t = 0$, which is 
$\clr [\betah_{\text{\emph{East}}}](s) - \clr [\betah_{\text{\emph{East}}}] (t) - (\clr [\betah_{\text{\emph{West}}}](s) - \clr [\betah_{\text{\emph{West}}}] (t)) 
= 0.23 - (-0.09) = 0.32$, which means that the odds for single-earner 
versus nonworking women in \emph{East} Germany are $\exp(0.32) \approx 1.38$ times the odds in \emph{West} Germany.
The (clr transformed) densities sampled from the simultaneous (over $[0, 1]$) confidence regions show similar shapes to the estimated densities over large parts of the domain, in particular regarding our findings.
However, towards the boundaries, i.e., for $s \ra 0$ and $s \ra 1$ with $s \in (0, 1)$, the functions spread out, reflecting estimation uncertainty due to few observations in these areas.

\begin{figure}[H]
\begin{center}
\includegraphics[width=0.49\textwidth]{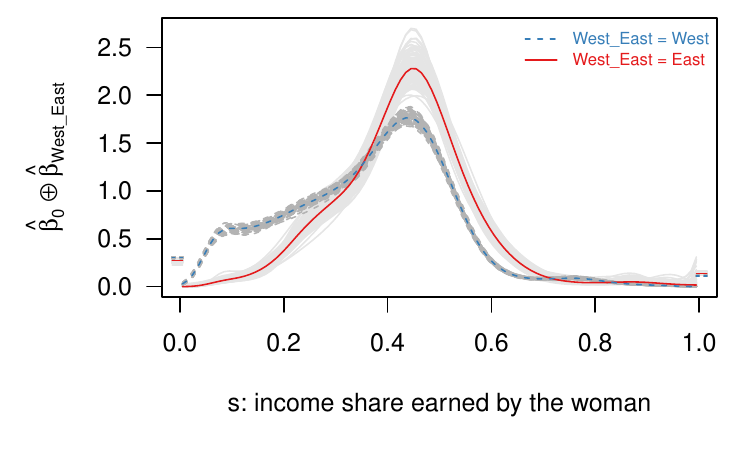}
\includegraphics[width=0.49\textwidth]{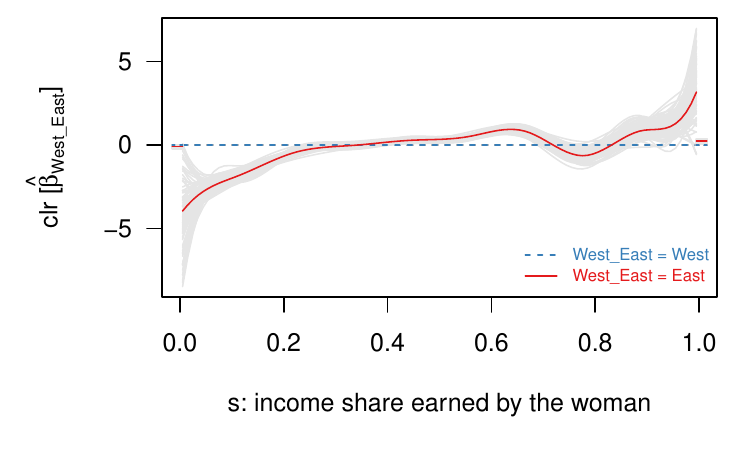}
\end{center}
\vspace{-0.5cm}
\caption{Estimated conditional densities for couples without minor children in 1991 for \emph{West} vs.\ \emph{East} Germany [left] and clr transformed estimated effects of \emph{West\_East} [right] with $100$ draws sampled uniformly from the respective $95\%$ simultaneous (over $[0, 1]$) confidence region. \label{P2:fig_estimated_West_East}}
\end{figure}

Figure~\ref{P2:fig_estimated_c_age} shows the analogous plot for the effect of \emph{c\_age}.
To avoid overplotting, we do not display the whole sample of the confidence regions, but only the pointwise minimum and maximum thereof (grey). See Figure~\ref{P2a:fig_appendix_estimated_c_age} in appendix~\ref{P2a:chapter_application_effects} for the respective plots with the whole sample.
The tendency towards a second mode observed in Figure~\ref{P2:fig_estimated_West_East} 
for couples without minor children living in \emph{West} Germany
amplifies for couples with minor children, as the densities for \emph{c\_age 0-6} and \emph{7-18} are both bimodal.
Both have one mode around $s = 0.1$ with their second and higher modes further right around $s = 0.35$ for \emph{c\_age 0-6} and $s = 0.25$ for \emph{c\_age 7-18}, but clearly smaller than that for \emph{c\_age other}.
While the shapes of the densities for the two different groups of couples with minor children are similar, they differ in scale, reflecting the smaller share of double-earners ($s \in (0, 1)$) among couples with small children (\emph{c\_age 0-6}) compared to the other groups.
The clr transformed effects are decreasing for most parts of $(0, 1)$, i.e., odds (ratios) $\betah_{\text{\emph{c\_age}}} (s) / \betah_{\text{\emph{c\_age}}} (t) = \frac{\betah_{\text{\emph{c\_age}}} (s) / \betah_{\text{\emph{c\_age}}} (t)}{\betah_{\text{\emph{other}}} (s) / \betah_{\text{\emph{other}}} (t)}$ for larger $s$ vs.\ smaller $t \in (0, 1)$ tend to be small for $\text{\emph{c\_age}} \in \{ \text{\emph{0-6}} , \text{\emph{7-18}}\}$, indicating probability mass shifts towards the left compared to couples without minor children (ceteris paribus).
Note that regarding the estimated conditional densities for couples living in \emph{East} Germany (in 1991, including interaction terms), there is (almost) no such shift, with the densities unimodal (and close to symmetric) for all three values of \emph{c\_age}, compare Figure~\ref{P2a:fig_appendix_estimated_West_East_c_age} in appendix~\ref{P2a:chapter_application_effects}, reflecting the different social norms in the two parts of Germany with formerly different political systems.
For both groups of couples with minor children (in \emph{West} Germany), nonworking women ($s = 0$) are more frequent than among couples without minor children and vice-versa for single-earner women ($s = 1$).
The log odds ratio of $\betah_{\text{\emph{0-6}}}$ (respectively $\betah_{\text{\emph{7-18}}}$) and $\betah_{\text{\emph{other}}}$ for $s = 1$ compared to $t = 0$ is 
$
-1.11 - 1.37 = -2.48$
(respectively $-0.64 - 0.55 = -1.19$), which means that the odds for single-earner 
versus nonworking women of couples with children aged zero to six (respectively seven to 18) are $\exp(-2.48) \approx 0.08$ (respectively $\exp(-2.48) \approx 0.3$) times the odds of couples without minor children.
Again, the (clr transformed) densities sampled from the simultaneous (over $[0, 1]$) confidence regions support our findings, having mostly similar shapes as the estimated densities.
As for \emph{West\_East}, estimation uncertainty shows in a spread towards the boundaries, in particular for $s \ra 1$, where observations are very rare.

\begin{figure}[H]
\begin{center}
\includegraphics[width=0.49\textwidth]{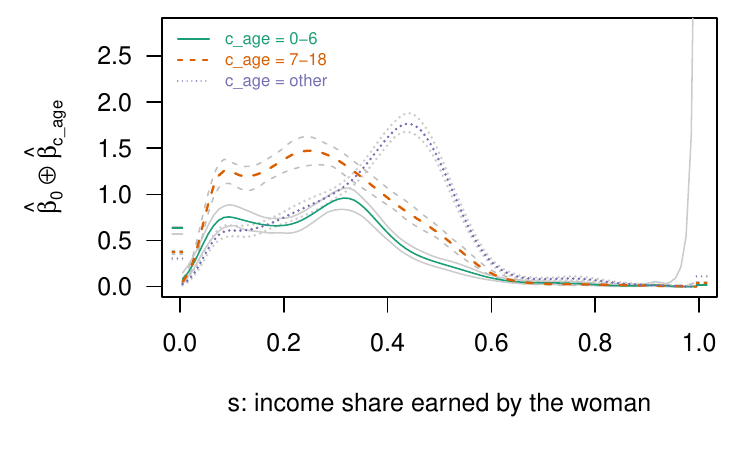}
\includegraphics[width=0.49\textwidth]{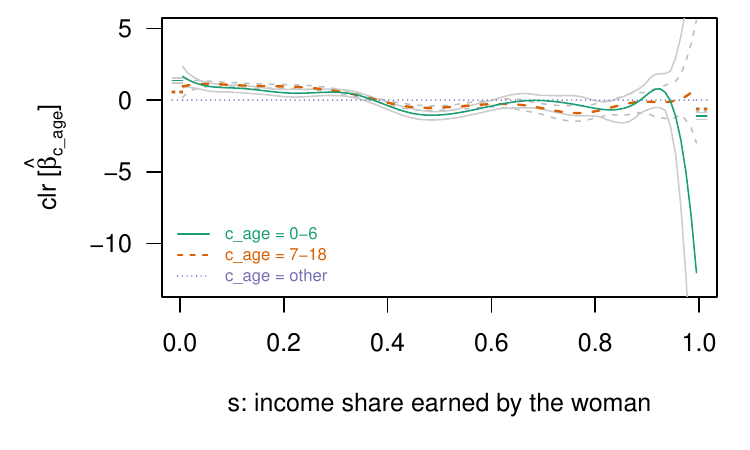}
\end{center}
\vspace{-0.5cm}
\caption{Estimated conditional densities for couples living in \emph{West} Germany in 1991 for all three values of \emph{c\_age} [left] and clr transformed estimated effects of \emph{c\_age} [right] with pointwise minimum/maximum of $100$ draws sampled uniformly from  the respective $95\%$ simultaneous (over $[0, 1]$) confidence region.
Note that the y-axis is restricted to reasonable limits regarding the estimated functions, not spanning the whole range of functions sampled from the confidence regions. \label{P2:fig_estimated_c_age}}
\end{figure}

The smooth effect of \emph{year} is illustrated in Figure~\ref{P2:fig_estimated_year}, without samples of confidence regions to avoid overplotting. 
See Figures~\ref{P2a:fig_appendix_estimated_year} and~\ref{P2a:fig_appendix_estimated_year_clr} in appendix~\ref{P2a:chapter_application_effects} for plots including samples of confidence regions.
Note that as for the other main effects, the functions sampled from the confidence region have similar shapes as the estimated ones with dispersion towards the boundaries.
A color gradient and different line types indicate the different \emph{years}.
Considering the left part of the figure, in each \emph{year}, the density (for couples without minor children living in \emph{West} Germany) reaches its maximum around $s = 0.45$ with a second mode around $s = 0.1$ developing over time. 
Furthermore, the estimated densities tend to disperse within $s \in (0, 1)$ over time, with probability mass particularly increasing for low positive income shares, while the share of nonworking women ($s = 0$) decreases.
This is a sign of women switching from no paid work to part-time employment, emphasizing the need for considering mixed densities, since an analysis of double-earner couples only ($s \in (0, 1)$) might misinterpret this development as a shift from larger to smaller income shares.
Moreover, the frequency of single-earner women ($s = 1$) increases over time, 
in the most recent years reaching a similar level as the (strongly decreasing) frequency of nonworking women ($s = 0$).
The right part of the figure also indicates the increasing dispersion on ($s \in (0, 1)$), as the clr transformed effects tend to be smaller for low and high income shares (e.g., $s \in A = (0, 0.2) \cup (0.8, 1)$) than for income shares in between (e.g.,  $t \in B = (0.2, 0.6)$) before the reference \emph{year} 1991.
This reverses afterwards. 
Hence,
the odds of the probabilities of $A$ (outer area) versus $B$ (central area) are smaller in earlier than in later \emph{years}. 
This implies that over time the probability of $A$ increases and the probability of $B$ decreases.

\begin{figure}[H]
\begin{center}
\includegraphics[width=0.49\textwidth]{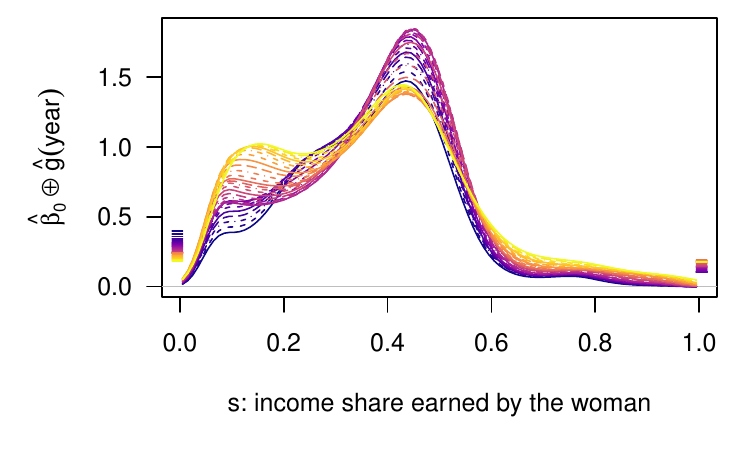}
\includegraphics[width=0.49\textwidth]{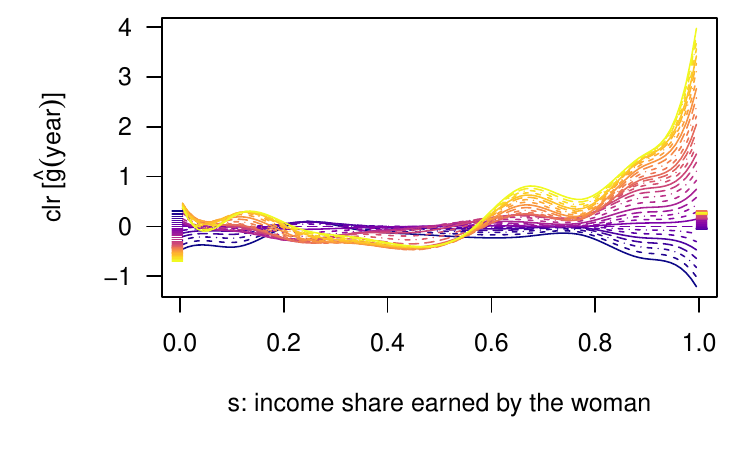}
\includegraphics[width=0.7\textwidth]{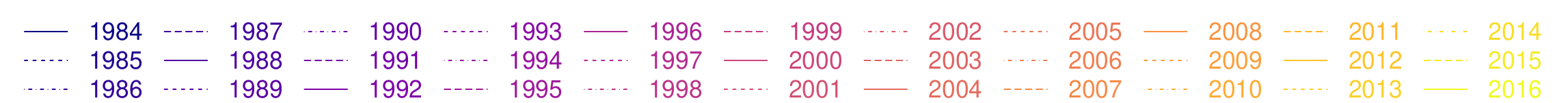}
\end{center}
\vspace{-0.5cm}
\caption{Estimated conditional densities for couples without minor children living in \emph{West} Germany over time [left] and clr transformed estimated effects of \emph{year} [right].\label{P2:fig_estimated_year}}
\end{figure}

All in all, our new approach finds similar effects as the one of \citet{maier2021}, 
but 
by working directly with the samples of the conditional distributions more appropriately uses the available information and additionally provides uncertainty quantification.
Note that we also illustrate and briefly discuss all estimated conditional densities in appendix~\ref{P2a:chapter_application_predictions}, again confirming findings of \citet{maier2021} for their pre-smoothed density-observations and predictions.

%% file: 3_simulation.tex
\section{Simulation study}\label{P2:chapter_simulation}
In this section, we present a simulation study to validate our approach.
We focus on the setting of our application with the estimated conditional densities obtained in Section~\ref{P2:chapter_application} serving as true conditional densities $f_{\emph{\text{West\_East}}, \emph{\text{c\_age}}, \emph{\text{year}}}$, where
$\emph{\text{West\_East}} \in \{\emph{\text{West}}, \emph{\text{East}}\}$,
$\emph{\text{c\_age}} \in \{\emph{\text{0-6}}, \emph{\text{7-18}}, \emph{\text{other}}\}$,
and $\emph{\text{year}} \in \{ 1984, \ldots, 2016\}$ for $\emph{\text{West\_East}} = \emph{\text{West}}$
and $\emph{\text{year}} \in \{ 1991, \ldots, 2016\}$ for $\emph{\text{West\_East}} = \emph{\text{East}}$.
We consider four different sample sizes, $N \in \{ 50 000, 150 000, 500 000, 1 000 000\}$, where $150 000$ corresponds to the rounded sample size of the data set underlying the application, and three different bin numbers for the continuous component, $G \in \{ 50, 100, 200\}$, using equidistant bins that are equal for all covariate combinations.
We thus omit the superscript indicating the covariate combination in the notation compared to Section~\ref{P2:chapter_multinomial_regression}. 
For $i = 1, \ldots , N$, the $i$-th covariate combination $\xf_i = (\emph{\text{West\_East}}_i, \emph{\text{c\_age}}_i, \emph{\text{year}}_i)$ is drawn with equal probabilities from all possible covariate combinations. 
Regarding the sampling of the women's income share $s_i \in [0, 1]$, note that due to the binning conducted to create the count data, the exact value within one bin of the continuous component is not relevant and it suffices to draw $s_i$ from the set of bin midpoints united with the discrete values $\{0, 1\}$, i.e., from $\{ u_1, \ldots, u_{\Gamma} \}$, where $\Gamma = G + 2$. 
The probability of sampling $u_g$ is given as
$
\int_{U_g} f_{\emph{\text{West\_East}}_i, \emph{\text{c\_age}}_i, \emph{\text{year}}_i} \, \dmu$ 
for $g = 1, \ldots, \Gamma$.
Since the different bin sizes are chosen such that the bins are nested, for fixed $N$ it suffices to sample $s_1, \ldots , s_N$ for $G = 200$, which corresponds to the finest partition. 
The same simulated data is then also used for the two smaller values of $G$. 
We estimate model~\eqref{P2:soep_model} based on the simulated data, using the same basis functions as in Section~\ref{P2:chapter_application} (including identifiability constraints).
For each combination of $N$ and $G$, we perform $200$ replications.
Motivated by \citet{maier2021}, 
we evaluate the quality of the estimated conditional densities (prediction) and the estimated partial effects using the relative mean squared error (relMSE), defined as
$ 
\rMSE (\eh)
:= 
\frac{\frac{1}{N} \, \sum_{i=1}^N \Vert e_i \ominus \eh_i \Vertb^2}{\frac{1}{N} \, \sum_{i=1}^N \Vert e_i \Vertb^2}
$, 
where $e_i \in \B$ is either the true conditional density 
$f_i = f_{\emph{\text{West\_East}}_i, \emph{\text{c\_age}}_i, \emph{\text{year}}_i}$
or a true partial effect
$h_j (\xf_i)$, $j = 1, \ldots J$,
and $\eh_i$ is its estimated version. 

\begin{figure}[H]
\vspace{-0.3cm} 
\begin{center}
\includegraphics[width=0.75\textwidth]{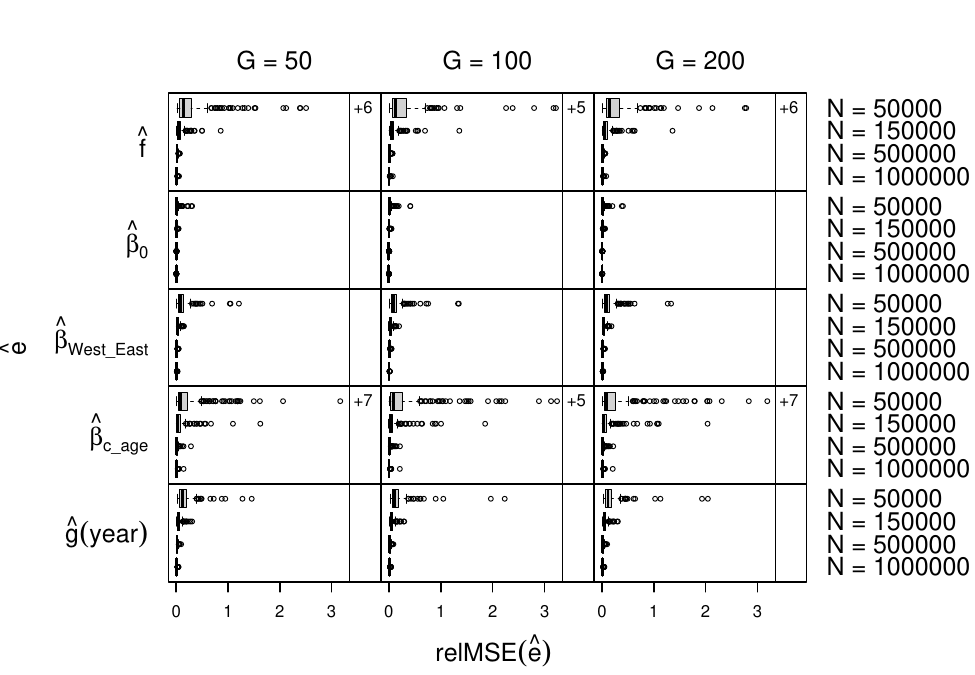}
\end{center}
\vspace{-0.5cm} 
\caption{RelMSE for prediction 
and main effects for different values of $G$ [columns] and $N$ [rows]. Numbers following a plus-sign at the right of a panel give the number of values larger than $3.35$ (vertical line).
Empty space means all values are visible.\label{P2:fig_relMSE_main}}
\end{figure}

Figure~\ref{P2:fig_relMSE_main} shows boxplots of the relMSEs for prediction and main effects, obtained based on $200$ replications for each combination of $N$ and $G$. 
All effects are illustrated in Figure~\ref{P2a:fig_relMSE_all} in appendix~\ref{P2a:chapter_simulation}.
For the smallest sample size, $N = 50 000$, while the other effects are well estimated, some iterations yield very large relMSEs for predictions and for the effect of \emph{c\_age} (also for interaction effects including \emph{c\_age}, compare Figure~\ref{P2a:fig_relMSE_all}), 
which are cut off with their count indicated in the plot.
In these iterations, the simulated data contains very few observations $s_i \in (0.9, 1)$ with $\text{\emph{c\_age}}_i = \text{\emph{0-6}}$, 
corresponding to areas where the respective true conditional densities we simulate from are close to zero due to the sharp drop for $s \rightarrow 1$ on clr-transformed level  (compare Figure \ref{P2:fig_estimated_c_age}). 
This 
leads to the continuous component of the density being underestimated for $s \ra 1$ (corresponding to an even sharper descent on clr-transformed level), which yields large squared norms $\Vert f_i \ominus \fh_i \Vertb^2$ and 
$\Vert \beta_{\text{\emph{c\_age}}_i} \ominus \betah_{\text{\emph{c\_age}}_i} \Vertb^2
$
for the corresponding observations, drastically deteriorating the relMSE. 
The problem disappears and estimation quality improves overall with increasing sample size $N$, while the three different values of $G$ yield relMSEs at a similar level for given $N$.
Considering $\rMSE (\fh)$, we obtain a median of $14\%$ ($G = 100$) to $15\%$ ($G \in \{50, 200\}$) for $N = 50000$, which decreases to $4.8\%$ ($G \in \{ 100, 200\}$) to $4.9\%$ ($G = 50$) for $N = 150000$, further to $1.4 \%$ ($G \in \{50, 100, 200\}$) for $N = 500000$, and down to $0.6 \%$ ($G = 100$), $0.7\%$ ($G = 200$), and $0.9 \%$ ($G = 50$) for $N = 1 000 000$.
%
%
%
Among the main effects, we obtain the smallest relMSEs for $\betah_0$ 
(medians: $N = 50000$: $0.9\%$ to $1\%$; 
$N = 150000$: $0.3\%$; 
$N = 500000$: $0.1\%$; 
$N = 1 000 000$: $0.05\%$ to $0.06 \%$), 
while the other effects $\betah_{\text{\emph{West\_East}}}$, $\betah_{\text{\emph{c\_age}}}$, and $\gh (\text{\emph{year}})$ yield larger values, whose medians are on a similar level
($N = 50000$: $9\%$, $10\%$, and $12\%$; 
$N = 150000$: $3.2\%$, $3.2\%$ to $3.4\%$, and $4.8\%$; 
$N = 500000$: $0.9\%$, $0.9\%$, and $1.6\%$; 
$N = 1000000$: $0.4\%$ to $0.5 \%$, $0.4\%$ to $0.6 \%$, and $0.9\%$).
For the interaction effects illustrated in Figure~\ref{P2a:fig_relMSE_all}, the relMSEs tend to be larger, in particular for effects involving \emph{c\_age} due to underestimation of the effects for large $s < 1$ as discussed above.

%
%
%
%

\begin{figure}[H]
\vspace{-0.3cm}
\begin{center}
\includegraphics[width=0.75\textwidth]{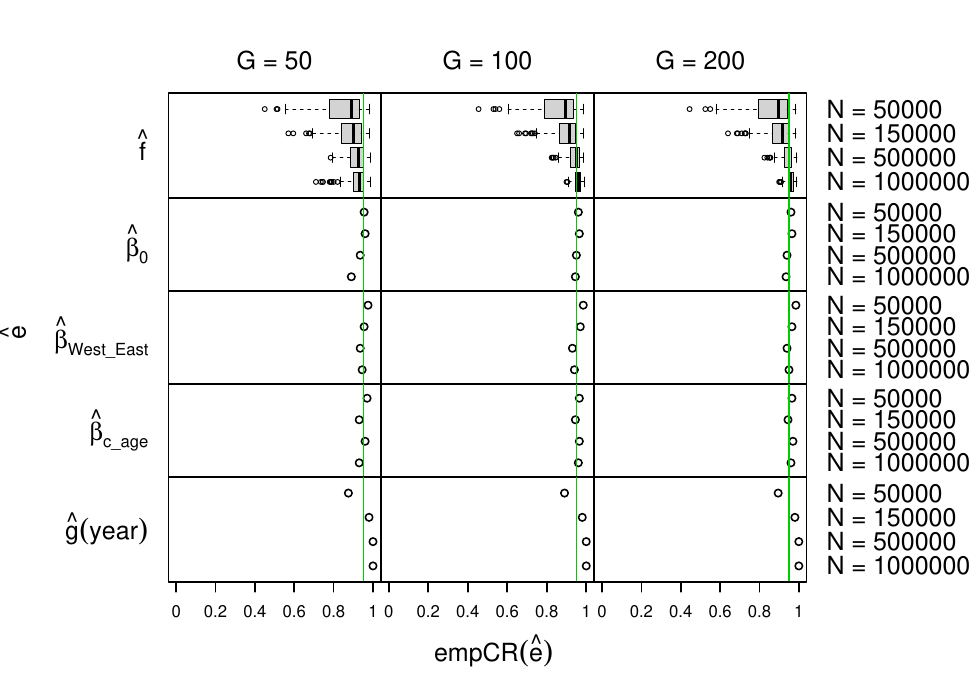}
\end{center}
\vspace{-0.5cm}
\caption{Empirical coverage rates (empCR) for prediction 
and main effects (simultaneous in covariates) for different values of $G$ [columns] and $N$ [rows].
The boxplots for predictions summarize coverage for the conditional densities corresponding to the $177$ unique covariate combinations.\label{P2:fig_empCR_smt_main}}
\end{figure}

The empirical coverage rates for predictions (as constructed in Lemma~\ref{P2:lemma_confidence_regions}) and main effects (as constructed in Lemma~\ref{P2a:lemma_confidence_regions_simultaneos}, i.e., simultaneous in covariates) are illustrated in Figure~\ref{P2:fig_empCR_smt_main}, where the green vertical line indicates the true theoretical value $\alpha = 95\%$.
Appendix~\ref{P2a:chapter_simulation} contains the analogous figure showing all effects (Figure~\ref{P2a:fig_empCR_smt_all}), as well as empirical coverage rates for partial effects based on point-wise confidence regions as constructed in Lemma~\ref{P2:lemma_confidence_regions} (Figure~\ref{P2a:fig_empCR_pw_all}).
For predictions, the range of the empirical coverage over the $177$ conditional distributions mostly reduces and converges to the theoretical value $95\%$ with increasing $N$. 
The only exception is the combination of smallest bin number, $G = 50$, with largest sample size 
$N = 1 000 000$. While  the median improves slightly, from $92.5\%$ to $93\%$ compared to $N = 500 000$, the range also increases, possibly indicating that too-small bin numbers may induce undercoverage especially for larger sample sizes, where variance decreases and even small approximation biases thus may be non-negligible.
Similarly, for the intercept, the empirical coverage rate for the largest sample size $N = 1 000 000$ is worst for $G = 50$, only $89\%$, while 
coverage rates are close to the theoretical value with $94.5\%$ ($G = 100$) and $93.5\%$ ($G = 200$) for more bins, and only slightly below the coverage rates of $95\%$ ($G = 100$) and $94\%$ ($G = 200$) obtained for $N = 500 000$.
For the group-specific intercepts $\betah_{\text{\emph{West\_East}}}$ and $\betah_{\text{\emph{c\_age}}}$ the empirical coverage rates are close to the true value. 
For $\gh (\text{\emph{year}})$, we observe some undercoverage for the smallest sample size $N = 50 000$ (again most severe for the smallest bin size: $87.5\%$ for $G = 50$, $89\%$ for $G = 100$ and $89.5\%$ for $G = 200$), while for larger values of $N$, there is slight overcoverage. 
This trend is similar for the interaction effects $\gh_{\text{\emph{West\_East}}}(\text{\emph{year}})$ and $\gh_{\text{\emph{c\_age}}}(\text{\emph{year}})$ (with even more severe undercoverage for $N = 50 000$), while for the three-way interaction $\gh_{\text{\emph{c\_age, West\_East}}}(\text{\emph{year}})$, we observe slight overcoverage already for the smallest sample size, $N = 50 000$ (compare Figure~\ref{P2a:fig_empCR_smt_all}).
The results based on point-wise confidence regions illustrated in Figure~\ref{P2a:fig_empCR_pw_all} are similar, with empirical coverage rates (mostly) close to the theoretical value for (group-specific) intercepts (with sometimes slightly worse results for $G = 50$ compared to $G \in \{ 100, 200\}$ given the same $N$), while the (group-specific) smooth effects of year yield more varied coverage rates (over the different years) for the smallest sample size, $N = 50 000$, contracting around $95\%$ (or a larger value for the three-way interaction) with increasing $N$.

All in all, 
the smallest considered sample size, $N = 50 000$, yields large relMSEs in some iterations -- mostly for $\beta_{c\_age}$, which exhibits a sharp drop on clr-level -- and tends to yield some undercoverage for predictions and \emph{year}-effects.
The available data size in the application, $N = 150 000$, already yields good results -- especially given the complexity of the considered model with several interaction effects. Estimation quality unsurprisingly improves further with more data available.
Regarding the bin number $G$, we observe almost no influence on the relMSE values within the reasonable ranges considered here, while the empirical coverage rates in some cases suffer for the lowest bin number compared to good coverage for the larger two considered bin numbers.


%% file: 5_conclusion.tex
\section{Conclusion}\label{P2:chapter_conclusion}

We presented a penalized maximum-likelihood approach to estimate structured additive regression models for densities given scalar covariates from random samples of the underlying conditional distribution. 
To respect their nature, densities are considered as elements of a Bayes Hilbert space $\B$, providing a unified framework for arbitrary finite measure spaces, in particular including mixed densities (and discrete and continuous special cases), which we focused on for estimation.
We derived the asymptotic distribution of the penalized maximum likelihood estimator (PMLE), which we used to construct confidence regions and p-values for estimated partial effects. Our framework is the first distributional regression approach to allow ceteris paribus interpretations for covariate effects.
Furthermore, we approximated the log-likelihood by a multinomial or equivalently Poisson log-likelihood, with the respective observations computed via histograms on the continuous and counts on the discrete part of the domain, and showed that the PMLE  and its estimated covariance (the inverse Fisher information) converge to the ones of the limiting (Bayes Hilbert space) log-likelihood.
This simplifies estimation, allowing to 
take advantage of the well-established \texttt{R} package \texttt{mgcv} \citep{wood2017} for our implementation in the  \texttt{R} package \href{https://github.com/Eva2703/DensityRegression}{\texttt{DensityRegression}}.
We used the new additive density regression approach to analyze the motivating (mixed) densities of the woman's income share for couples in Germany, observing similar covariate effects as \citet{maier2021}, but avoiding a two-step estimation approach and additionally providing inference. In particular, we found significant main effects for West vs. East Germany, presence vs. absence of minor children in the household, and time.
A simulation study based on our application further validated our approach.

Formulating our regression models directly for densities in a Bayes Hilbert space not only allows to consider continuous, discrete, and mixed densities in a unified framework, and to include a variety of covariate effects, but also to interpret these effects ceteris paribus.
Thereby, our approach stands out from other distributional regression approaches modeling the conditional distribution based on individual observations, which lack such interpretability \citep{kneib2023}.
In contrast, distributional response regression approaches benefiting from the Bayes Hilbert space framework \citep[e.g.,][]{talska2017, maier2021} 
assume whole densities to be observed as response objects, which is usually not possible 
and thus necessitates a two-step procedure. 
This leads to additional estimation uncertainty and requires enough observations to estimate conditional densities for each covariate combination, which typically cannot be expected in the presence of continuous covariates.
Our maximum-likelihood approach modeling the conditional densities directly based on the individual observations overcomes these issues.
In addition to these strong advantages, the Bayes Hilbert space framework also has limitations in certain settings. 
As Bayes Hilbert spaces contain ($\mu$-a.e.) positive densities, densities that are zero on parts of the domain 
cannot be modeled. 
However, our new estimation approach  based on individual data avoids zero density values that often occur in two-step approaches due to presmoothing (or over-presmoothing to prevent zero values), providing a natural solution to this problem except in settings with structurally different density supports.
Furthermore, the reference measure needs to be finite, excluding the Lebesgue measure as reference for $\Ycal = \Rbb$.
However, the probability measure of a reference distribution, such as the standard normal distribution, provides an alternative reference in this case \citep{vdb2014}, and the Lebesgue measure can be used for the finite domain case commonly observed in practice.


In contrast to existing literature \citep[e.g.,][]{masse1999}, we provided simultaneous confidence regions,  containing complete densities. 
The asymptotic distribution of the PMLE these are based on was derived assuming a true parameter value to exist for the given finite-dimensional basis expansion. 
As the MLE is known to converge under standard assumptions to the minimizer of the Kullback-Leibler divergence within the modeled class of distributions, 
here spanned by the considered basis \citep{knight1999, stone1991},
our semi-parametric conditional distribution estimates can be interpreted as best approximations to the true conditional distributions  within a finite but potentially high dimensional model space. Given the large model flexibility compared to existing  approaches, and similarly to other penalized spline approaches \citep{ruppert2002}, 
this should not be much of a limitation 
in practice.   
Future work, motivated by \citet{stone1991},  might also study convergence and inference as the number of basis coefficients goes to infinity.
Furthermore, we aim to extend the presented approach from univariate to multivariate densities in future.

%% file: 91_acknowledgements.tex
We would like to thank Simon Wood for fruitful discussions on   
penalized spline estimation and confidence band coverage.

%% file: 92_funding.tex
E.-M. Maier and S. Greven were funded by the Deutsche Forschungsgemeinschaft (DFG, German Research Foundation) - project number 513634041.
A. St\"ocker acknowledges support from SNSF Grant 200020\_207367.

%% file: 6_appendix.tex
\section{Proofs} \label{P2a:chapter_proofs_main_section}
This section contains proofs for all statements formulated in our main manuscript.
We start with preliminaries in Section~\ref{P2a:chapter_proofs_preliminaries}.
More in detail, in Section~\ref{P2a:chapter_notation}, we introduce some notation and show a first auxiliary statement (Lemma~\ref{P2a:lem_characterization_constrined_basis}), which is used in the following subsections.
In Section~\ref{P2a:chapter_arbitrary_bases_and_penalties}, we specify necessary properties of basis functions and penalty matrices.
All of them are actually very common 
and in particular fulfilled for our suggestions in Section~\ref{P2:chapter_bayes_space_regression}.
Thus, these properties are always assumed to hold, without mentioning them explicitly in statements afterwards.
In Section~\ref{P2a:chapter_assumptions}, we recall the assumptions introduced in Sections~\ref{P2:chapter_bayes_space_regression} and~\ref{P2:chapter_multinomial_regression} and specify the additional assumption~\ref{P2a:assumption_appendix_stone_2}, which was only briefly described, but not given explicitly in Section~\ref{P2:chapter_bayes_space_regression} as it requires further notation (introduced in Section~\ref{P2a:chapter_notation}) and is rather lengthy.
Furthermore, Section~\ref{P2a:chapter_assumptions} contains a more detailed discussion of the assumptions and proof of the statements made in Sections~\ref{P2:chapter_bayes_space_regression} and~\ref{P2:chapter_multinomial_regression} regarding the assumptions (Lemma~\ref{P2a:lem_assumptions}).
Auxiliary statements, which are used several times in later proofs, are collected in Section~\ref{P2a:chapter_auxiliary_statements}.
Finally, Section~\ref{P2a:chapter_proofs} contains the proofs for all theorems/lemmas of our main manuscript, and also Lemma~\ref{P2a:lemma_confidence_regions_simultaneos}, where we construct confidence regions and p-values, which are simultaneous in the covariates (and over $\Ycal$), as mentioned in Section~\ref{P2:chapter_bayes_space_regression} after Lemma~\ref{P2:lemma_confidence_regions}.
Note that the proof of Theorem~\ref{P2:thm_MLE_asymptotic_normal} is rather extensive and requires several auxiliary statements (namely, Lemmas~\ref{P2a:lem_assumption_smoothing_parameters_implies_convergence_penalty}, \ref{P2a:lem_score_convergesP_zero}, and~\ref{P2a:lem_lower_bound_quadratic_form_Fisher_info}), which are not used in any other proof.
Thus, they are arranged after the proof of Theorem~\ref{P2:thm_existence_uniqueness_bayes_PMLE}, in the same section as the proof of Theorem~\ref{P2:thm_MLE_asymptotic_normal}.
This section also contains a brief discussion of the proof strategy and literature before formulating and proving the auxiliary statements.


\subsection{Preliminaries}\label{P2a:chapter_proofs_preliminaries}

\subsubsection{Notation and Lemma~\ref{P2a:lem_characterization_constrined_basis}}\label{P2a:chapter_notation}
In the remainder of Section~\ref{P2a:chapter_proofs_main_section}, we use the following notation (partially already introduced in Section~\ref{P2:chapter_regression_theory}):
\begin{enumerate}[label=\arabic*)]
\item\label{P2a:notation_covariates}
\emph{Covariates:}
Recall that the support of the covariates is a compact set $\Xcal = \bigtimes_{p=1}^P \Xcal_p \subseteq \Rbb^P$ such that each $\Xcal_p$ is either a finite discrete set or a compact interval.
Set $\Dfrak = \{p \in \{ 1, \ldots , P \} ~|~ \Xcal_p ~ \text{is discrete}\}$ and $\Cfrak := \Dfrak^\complement = \{p \in \{ 1, \ldots , P \} ~|~ \Xcal_p ~ \text{is an interval}\}$,
set $\nuf = \bigtimes_{p=1}^P \nu_p$, where $\nu_p := \delta_p = \sum_{x \in \Xcal_p} \delta_x 
$, if $p \in \Dfrak$ and $\nu_p := \lambda$, if $p \in \Cfrak$,
and set $\Xcal^{\Dfrak} := \bigtimes_{p \in \Dfrak} \Xcal_p$ and $\Xcal^{\Cfrak} := \bigtimes_{p \in \Cfrak} \Xcal_p$.
For $\xf \in \Xcal$, we write $\xf = \xf^{\Dfrak}, \xf^{\Cfrak}$ with $\xf^{\Dfrak} \in \Xcal^{\Dfrak}$ and $\xf^{\Cfrak} \in \Xcal^{\Cfrak}$ denoting the discrete and continuous covariates contained in $\xf$.
\item\label{P2a:notation_spline_functions}
\emph{Spline functions for normalized coefficient vector:}
Recall that $\Ycal = \Yc \cup \Yd$, where $\Yc = [a, b]$ for $a, b \in \Rbb$ with $a < b$ and $\Yd = \{t_1, \ldots , t_D \} \subset \Rbb$. 
Denote the unit sphere in $\Rbb^K$ with $\Sbb^{K-1}$.
For $\vf \in \Sbb^{K-1}$ and $\xf \in \Xcal$, 
let 
\begin{align*}
s_{\vf, \xf} : \Ycal \ra \Rbb, &&y \mapsto \bft(\xf)(y)^\top \vf = ( \bfe_{\Xcal}(\xf) \otimes \tilde{\bfe}_{\Ycal} (y))^\top \vf
,
\end{align*}
and
let $s_{\vf, \xf, \mathrm{c}}$ be the unique continuous function on $\Yc$ with $s_{\vf, \xf, \mathrm{c}} (y) = s_{\vf, \xf} (y)$ for all $y \in \Yc \setminus \Yd$.
Set $s_{\vf, \xf}^{\max} := \sup_{y \in \Ycal} s_{\vf, \xf} (y)$, $s_{\vf, \xf, \mathrm{d}}^{\max} := \sup_{y \in \Yd} s_{\vf, \xf} (y)$, and $s_{\vf, \xf, \mathrm{c}}^{\max} := \sup_{y \in \Yc \setminus \Yd} s_{\vf, \xf} (y) = \sup_{y \in \Yc} s_{\vf, \xf, \mathrm{c}} (y)$.
\item\label{P2a:notation_mixed_basis_clr}
\emph{Mixed basis functions in $L^2_0 (\mu)$:}
Recall the construction of a basis $\bfe_{\Ycal} \in \B^{K_{\Ycal}}$ in the mixed case described in Section~\ref{P2:chapter_bayes_space_regression}, 
via the orthogonal decomposition of $\B$ into $\Bl$ and $\Bd$ of \citet{maier2021}.
Equivalently, the basis $\bft_{\Ycal}$ (containing the clr transformation of the functions in $\bfe_{\Ycal}$) can be obtained via the orthogonal decomposition of $\Ln$ into $\Lnl$ and $\Lnd$ by \citet[Proposition C.2 in appendix C]{maier2021}:
We obtain the bases $\bt_{\Yc, \, 1}, \ldots , \bt_{\Yc, \, K_{\Yc}} \in \Lnl$ and $\bt_{\Yd^\bullet, \, 1}, \ldots , \bt_{\Yd^\bullet, \, K_{\Yd^\bullet}} \in \Lnd$ from appropriate bases $\bb_{\Yc, \, 1}, \ldots , \bb_{\Yc, \, K_{\Yc} + 1} \in \Ll$ and $\bb_{\Yd^\bullet, \, 1}, \ldots , \bb_{\Yd^\bullet, \, K_{\Yd^\bullet} + 1} \in \Ld$.\footnote{Note that $K_{\Yd^\bullet} = D$ is the number of discrete points in the mixed space $\Ln$ (while $K_{\Yd^\bullet} + 1 = D + 1$ is the number of discrete points in the discrete space $\Lnd$) and 
$\bar{b}_{\Yd^\bullet, \, d} = \mathbbm{1}_{\{t_{d}\}}$ for all $d = 1, \ldots , K_{\Yd^\bullet} + 1$. 
For the sake of consistent notation, we write $K_{\Yd^\bullet}$ and $\bar{b}_{\Yd^\bullet, \, d}$, when the (number of) basis functions are concerned. For the number of discrete points, we write $D$.}
In Section~\ref{P2:chapter_bayes_space_regression} we propose to use a B-spline basis with pairwise distinct knots and degree at least $2$ for the former and indicator functions at the discrete values $t_1 , \ldots , t_{D+1}$ for the latter.
More precisely, the basis functions for the continuous component are obtained as
$\bt_{\Yc, \, m} := \sum_{j = 1}^{K_{\Yc} + 1} \bb_{\Yc, \, j} z^{\Yc}_{j m}$, where
$(z^{\Yc}_{1 m}, \ldots, z^{\Yc}_{(K_{\Yc} + 1) m})$ is the $m$-th column of a transformation matrix $\Zf_{\Yc} \in \Rbb^{(K_{\Yc} + 1) \times K_{\Yc}}$ for $m = 1, \ldots , K_{\Yc}$.
Analogously for $\bt_{\Yd, \, m}$ with $m = 1, \ldots , K_{\Yd}$ and a transformation matrix $\Zf_{\Yd} \in \Rbb^{(K_{\Yd} + 1) \times K_{\Yd}}$ for the discrete component.
Both transformation matrices are obtained via QR decompositions of the vectors of integrals of the untransformed bases (viewed as a matrix with one row) as described in appendix B of \citet{maier2021} \citep[based on][Section~1.8.1]{wood2017} and have full rank per construction.
The transformed B-spline basis 
$\bft_{\Yc} = (\bt_{\Yc, \, 1}, \ldots , \bt_{\Yc, \, K_{\Yc}})^\top \in \Lnl^{K_{\Yc}}$ 
and the transformed indicator functions 
$\bft_{\Yd^\bullet} = (\bt_{\Yd^\bullet, \, 1}, \ldots , \bt_{\Yd^\bullet, \, K_{\Yd^\bullet}})^\top \in \Lnd^{K_{\Yd^\bullet}}$
are combined to a vector of basis functions $\bft_{\Ycal} \in \Ln^{K_{\Ycal}}$ by applying the embeddings
\begin{align}
&\Jtc : \Lnl \hookrightarrow \Ln && \ftc \mapsto
\begin{cases}
\ftc & \mathrm{on}~ \Yc \setminus \Yd\\
0 & \mathrm{on}~ \Yd
\end{cases} 
\label{P2a:eq_embedding_L20_c}
\\
&\Jtd: \Lnd \hookrightarrow \Ln && \ftd \mapsto
\begin{cases}
\ftd \left( t_{D+1} \right) & \mathrm{on}~ \Yc \setminus \Yd \\
\ftd & \mathrm{on}~ \Yd
\end{cases} \, . \label{P2a:eq_embedding_L20_d}
\end{align}
of \citet[Proposition C.2 in appendix C]{maier2021} component-wise to the functions contained in $\bft_{\Yc}$ and $\bft_{\Yd^\bullet}$ (equivalently to the approach described to obtain the basis functions in $\B$ in Section~\ref{P2:chapter_bayes_space_regression}).
\end{enumerate}

Related to the notation introduced in~\ref{P2a:notation_mixed_basis_clr}, we show a brief auxiliary statement, which is used for discussing the properties of penalty matrices/basis functions in Section~\ref{P2a:chapter_arbitrary_bases_and_penalties} and in the proof of Lemma~\ref{P2a:lem_assumptions} in Section~\ref{P2a:chapter_assumptions}.

\begin{lem}\label{P2a:lem_characterization_constrined_basis}
Let $\bb_{\Tcal, \, 1}, \ldots , \bb_{\Tcal, \, K_{\Tcal} + 1}$ be linearly independent real-valued functions defined on $\Tcal \subset \Rbb$, let $1_{\Tcal} : \Tcal \ra \{1 \}, t \mapsto 1$, let 
$\Zf_{\Tcal} 
= (z_{j, m}^{\Tcal})_{j = 1, \ldots , K_{\Tcal} + 1, m = 1, \ldots , K_{\Tcal}}$ have rank $K_{\Tcal}$, and for $m = 1, \ldots , K_{\Tcal}$, let
$\bt_{\Tcal, \, m} := \sum_{j = 1}^{K_{\Tcal} + 1} \bb_{\Tcal, \, j} z^{\Tcal}_{j m}$.
If $1_{\Tcal} \in \spano (\bfb_{\Tcal}) := \spano (\bb_{\Tcal, 1}, \ldots , \bb_{\Tcal, K_{\Tcal} + 1})$ 
and
$1_{\Tcal} \notin \spano (\bft_{\Tcal}) := \spano (\bt_{\Tcal, 1}, \ldots , \bt_{\Tcal, K_{\Tcal}})$, 
then the functions $\bt_{\Tcal, \, 1}, \ldots , \bt_{\Tcal, \, K_{\Tcal}}, 1_{\Tcal}$ are linearly independent and
\begin{align*}
\spano (\bfb_{\Tcal}) := \spano (\bb_{\Tcal, 1}, \ldots , \bb_{\Tcal, K_{\Tcal} + 1}) 
= \spano (\bt_{\Tcal, 1}, \ldots , \bt_{\Tcal, K_{\Tcal}}, 1_{\Tcal}).
\end{align*}
In particular, this holds for 
$\Tcal = \Yc$ with $\bb_{\Tcal, 1} = \bb_{\Yc, 1}, \ldots , \bb_{\Tcal, K_{\Tcal} + 1} = \bb_{\Yc, K_{\Ccal} + 1}$ a B-spline basis and $\Zf_{\Tcal} = \Zf_{\Yc}$, 
and for 
$\Tcal = \Yd$ with $\bb_{\Tcal, 1} = \bb_{\Yd, 1}, \ldots , \bb_{\Tcal, K_{\Tcal} + 1} = \bb_{\Yd, K_{\Yd} + 1}$ the basis of indicator functions at the discrete values $t_1, \ldots , t_D$  and $\Zf_{\Tcal} = \Zf_{\Yd}$ as in notation~\ref{P2a:notation_mixed_basis_clr} above.
\end{lem}


\begin{proof}
Since the functions $\bb_{\Tcal, \, 1}, \ldots , \bb_{\Tcal, \, K_{\Tcal} + 1} \in \Ll$ are linearly independent and the transformation matrix $\Zf_{\Tcal}$ has full rank,
the transformed basis $\bt_{\Tcal, \, 1}, \ldots , \bt_{\Tcal, \, K_{\Tcal}} \in \Lnl$ consists of linearly independent functions \citep[Lemma 9.23]{liesen2015}.
Furthermore, per construction via linear combinations, the basis functions contained in $\bft_{\Tcal}$ are elements of
$\spano (\bfb_{\Tcal})$. 
Since $1_{\Tcal} \notin \spano (\bft_{\Tcal})$, the functions $\bt_{\Tcal, \, 1}, \ldots , \bt_{\Tcal, \, K_{\Tcal}}, 1_{\Tcal}$ are linearly independent.
Furthermore, as $1_{\Tcal} \in \spano (\bfb_{\Tcal})$, the $K_{\Tcal} + 1$ functions $\bt_{\Tcal, \, 1}, \ldots , \bt_{\Tcal, \, K_{\Tcal}}, 1_{\Tcal}$ span the same space as the $K_{\Tcal} + 1$ functions $\bb_{\Tcal, \, 1}, \ldots , \bb_{\Tcal, \, K_{\Tcal} + 1}$, i.e., 
$\spano (\bfb_{\Tcal}) = \spano (\bt_{\Tcal, 1}, \ldots , \bt_{\Tcal, K_{\Tcal}}, 1_{\Tcal})$.

Consider a B-spline basis $\bb_{\Tcal, \, 1} = \bb_{\Yc, 1}, \ldots , \bb_{\Tcal, \, K_{\Tcal} + 1} = \bb_{\Yc, K_{\Yc} + 1}$, which consists of linearly independent functions and fulfills $1_{\Yc} \in \spano (\bfb_{\Tcal})$ due to \citet[Section~IX, (36) B-spline Property (iv)]{deboor2001}.
Since the functions contained in $\bft_{\Yc}$ all integrate to zero, we have $1_{\Yc} \notin \spano (\bft_{\Yc})$.
The basis of indicator functions $\bb_{\Tcal, \, 1} = \bb_{\Yd, 1}, \ldots , \bb_{\Tcal, \, K_{\Tcal} + 1} = \bb_{\Yd, K_{\Yd} + 1}$ fulfills the linear independence assumption and $1_{\Yd} \in \spano (\bfb_{\Yd})$ straightforwardly, while $1_{\Yd} \notin \spano (\bft_{\Yd})$ follows analogously with the in\-te\-grate-to-zero constraint.
Thus, in both cases the assumptions are fulfilled.
\end{proof}

\subsubsection{Necessary Properties of Bases and Penalties}\label{P2a:chapter_arbitrary_bases_and_penalties}
For penalty matrices and basis functions not chosen as in Section~\ref{P2:chapter_bayes_space_regression}, we assume:
\begin{enumerate}
\setlength{\itemindent}{0.5cm}
\nitem{(P)}\label{P2a:assumption_penalties_positive}
All penalty matrices $\Pf_j$ are positive semi-definite, $j = 1, \ldots , J$.
\nitem{(SPP)}\label{P2a:assumption_piecewise_polynomial}
There exist $R \in \Nbb$ and $\kappa \geq 2$ such that 
every function contained in $\spano (\bft_{\Yc})$ is a $(\kappa - 1)$-times continuously differentiable piecewise polynomial consisting of $R$ polynomials of degree at most $\kappa$.
\nitem{(B)}\label{P2a:assumption_domain_basis}
The basis functions contained in $\bft_{\Ycal}$ as well as the basis functions contained in $\bfe_{\Xcal}$ 
are linearly independent, continuous in any $\yc \in \Yc \setminus \Yd$ respectively in any $\xf \in \Xcal$,
and bounded.
\end{enumerate}

Since these are very common properties of penalty matrices/(spline) basis functions, 
these are not mentioned explicitly in the remaining work to not overload the statements.
Now, consider the choices of penalty matrices and basis functions recommended in Section~\ref{P2:chapter_bayes_space_regression}, included in an identifiable model.
We briefly discuss, that the above assumptions are fulfilled:

\begin{itemize}
\item
\ref{P2a:assumption_penalties_positive}:
As marginal penalty matrices $\Pf_{\Xcal, \, j}, j = 1, \ldots , J,$ and $\Pf_{\Ycal}$ in Section~\ref{P2:chapter_bayes_space_regression}, we propose to use difference penalties and identity matrices, which are positive semi-definite.
Since Kronecker products and linear combinations (for non-negative scalars) of positive semi-definite matrices are positive semi-definite (see \citealp[Corollary 4.2.13]{horn1994}, and \citealp[Observation 7.1.3]{horn2012}, respectively),
$
\Pf_{j} = \xi_{\Xcal, \, j} (\Pf_{\Xcal, \, j} \otimes \If_{K_{\Ycal}}) + \xi_{\Ycal, \,j} (\If_{K_{\Xcal, \, j}} \otimes \Pf_{\Ycal}) 
$ is positive semi-definite for all $j = 1, \ldots , J$ i.e., assumption~\ref{P2a:assumption_penalties_positive} is fulfilled.
\item
\ref{P2a:assumption_piecewise_polynomial}: 
Let $\kappa \geq 2$ denote the degree of the (unconstrained) B-splines contained in $\bfb_{\Yc}$.
Every function contained in $\spano (\bfb_{\Yc})$ is a $(\kappa -1)$ times continuously differentiable piecewise polynomial with the breakpoints corresponding to the knots defining the B-spline basis
\citep{deboor2001}.
By Lemma~\ref{P2a:lem_characterization_constrined_basis}, we have $\spano (\bft_{\Yc}) \subset \spano (\bfb_{\Yc})$ and thus every function contained in $\spano (\bft_{\Yc})$ is also such a piecewise polynomial 
and
we obtain~\ref{P2a:assumption_piecewise_polynomial}. 
\item
\ref{P2a:assumption_domain_basis}:
We first consider $\bft_{\Ycal}$. In the continuous or discrete special case, linear independence of the functions contained in $\bft_{\Ycal}$ was already shown in Lemma~\ref{P2a:lem_characterization_constrined_basis}.
In the mixed case, the basis in $\Ln$ is obtained by applying the embedding 
$\tilde{\iota}_{\mathrm{c}}$ given in \eqref{P2a:eq_embedding_L20_c}
to the linear independent functions $\bt_{\Yc, \, 1}, \ldots , \bt_{\Yc, \, K_{\Yc}} \in \Lnl$
and the embedding $\tilde{\iota}_{\mathrm{d}}$ given in \eqref{P2a:eq_embedding_L20_d} 
to the linear independent functions $\bt_{\Yd^\bullet, \, 1}, \ldots , \bt_{\Yd^\bullet, \, K_{\Yd^\bullet}} \in \Lnd$.
Since these embeddings are injective, they retain linear independence and thus the functions contained in the vector $\bft_{\Ycal}$ are linearly independent.
Continuity and boundedness of the functions contained in $\bft_{\Ycal}$ follows per construction from continuity and boundedness of the B-spline basis $\bb_{\Yc, \, 1}, \ldots , \bb_{\Yc, \, K_{\Yc} + 1}$.
\\
Regarding $\bfe_{\Xcal}$, 
the linear independence, continuity, and boundedness assumptions are fulfilled for usual choices of basis functions $\bfe_{\Xcal, \, j}, j \in \{ 1, \ldots , J\}$, in particular the ones we recommend in Section~\ref{P2:chapter_bayes_space_regression} (recalling the covariate domain $\Xcal$ is compact).
Note that across different partial effects, i.e., for $j \neq j' \in \{ 1, \ldots , J\}$, identifiability constraints ensure linear independence, where necessary (e.g., if both partial effects are an effect of the same smooth covariate -- one a main effect, the other an interaction effect).
Hence, assumption~\ref{P2a:assumption_domain_basis} holds.

\end{itemize}

\subsubsection{Assumptions~\ref{P2a:assumption_appendix_covariate_basis_nonsingular}, \ref{P2a:assumption_appendix_stone_1}, \ref{P2a:assumption_appendix_stone_2}, \ref{P2a:assumption_appendix_scalar_product_non_constant}, and \ref{P2a:assumption_appendix_domain_basis_nonsingular}
}\label{P2a:chapter_assumptions}
We now recall the assumptions, which were used to formulate the statements in Section~\ref{P2:chapter_regression_theory} and introduce the additional assumption~\ref{P2a:assumption_appendix_stone_2}.
For this purpose, 
denote the diameter of a set $X \subset \Xcal^{\Cfrak}$ with $\diam(X) := \sup \{\Vert \xf - \xft \Vert_2 ~|~ \xf, \xft \in X \}$.

\begin{enumerate}
\setlength{\itemindent}{0.5cm}
\nitem{(BX)}\label{P2a:assumption_appendix_covariate_basis_nonsingular}
For $\vf_{\Xcal} \in \Rbb^{K_{\Xcal}}$: 
If $\bfe_{\Xcal} (\xf_i)^\top \vf_{\Xcal} = 0$ for all $i = 1, \ldots , N$, then $\vf_{\Xcal} = \mathbf{0}$.
\nitem{(BX')}\label{P2a:assumption_appendix_stone_1}
There exist $M' > 0$ (independent of $N$) and $N' \in \Nbb$ such that for all $N \geq N'$ and all $\vf_{\Xcal} 
\in \Rbb^{K_{\Xcal}}$,
\vspace{-0.3cm}
\begin{align}
M' \, N \int_{\Xcal} \left( \bfe_{\Xcal} (\xf)^\top \vf_{\Xcal} \right)^2 \dnuf (\xf) \leq \sum_{i=1}^N \left( \bfe_{\Xcal} (\xf_i)^\top \vf_{\Xcal} \right)^2 ,
\label{P2a:eq_stone_1}
\end{align}
\vspace{-0.6cm}
\nitem{(BX'')}\label{P2a:assumption_appendix_stone_2}
There exist $\alpha \in (0, 1)$,
$M'' > 0$, and $N'' \in \Nbb$ such that for all $N \geq N''$, all $\xf^{\Dfrak} \in \Xcal^{\Dfrak}$, and all subsets $X \subseteq \Xcal^{\Cfrak}$ with $\diam (X) \geq N^{\alpha - 1}$,
\vspace{-0.1cm}
\begin{align}
M'' \, N \, \diam(X)^{|\Cfrak|}
\leq \sum_{i=1}^N \mathbbm{1}_{\{ \xf^{\Dfrak}\} \times X} (\xf^\Dfrak_i , \xf^\Cfrak_i)
.
\label{P2a:eq_stone_2}
\end{align}
\vspace{-0.6cm}
\nitem{(S$<$)}\label{P2a:assumption_appendix_scalar_product_non_constant}
For all $\vf \in \Sbb^{K-1}$ there exists $i \in \{1, \ldots , N\}$ with 
$s_{\vf, \xf_i}(y_i) < 
s_{\vf, \xf_i}^{\max} 
$.
\nitem{(BU)}\label{P2a:assumption_appendix_domain_basis_nonsingular}
For $\Zcal \in \Xi^L 
$, 
$l \in \{ 1, \ldots , L\}$,
$\vf_{\Ycal} \in \Rbb^{K_\Ycal}$, and $c \in \Rbb$:
If $\bft_\Ycal (u_{g}^{(l)})^\top \vf_{\Ycal} = c 
$ for all $g = 1, \ldots , \Gamma^{(l)}$, then
$\vf_{\Ycal} = \mathbf{0}$ (and $c = 0$). 
\end{enumerate}

In the following remark, we discuss 
these assumptions, amplifying the respective discussion paragraphs after introducing the assumptions in Sections~\ref{P2:chapter_bayes_space_regression} and~\ref{P2:chapter_multinomial_regression}.
For this purpose, recall that 
$\Bf_{\Xcal} 
\in \Rbb^{N \times K_{\Xcal}}$ denotes the matrix containing $\bfe_{\Xcal} (\xf_i)$ in the $i$-th row, $i = 1, \ldots , N$,
$\Bft_{\Ycal}^{(l)} 
\in \Rbb^{\Gamma^{(l)} \times K_\Ycal}$ the matrix containing $\bft_\Ycal (u_g^{(l)})$
in the $g$-th row, $g = 1, \ldots , \Gamma^{(l)}$, 
$\mathbf{1}_{\Gamma^{(l)}}$ the vector of length $\Gamma^{(l)}$ containing $1$ 
in every entry,
and $( \Af : \Bf )$ concatenation of matrices $\Af, \Bf$ with equal number of rows.

\begin{rem}\label{P2a:rem_assumptions_lem}
\begin{enumerate}[label=\arabic*)]
\item\label{P2a:rem_assumption_covariate_basis_nonsingular}
We will show in Lemma~\ref{P2a:lem_assumptions}~\ref{P2a:lem_assumption_covariate_basis_nonsingular} that assumption~\ref{P2a:assumption_appendix_covariate_basis_nonsingular} holds, if $\Bf_{\Xcal}$ has rank $K_{\Xcal}$.
Here, we thus focus on fulfilling this rank condition, which can actually be achieved by choosing the basis functions $\bfe_{\Xcal, \, j}, j \in \{ 1, \ldots , J\}$ appropriately for given observations $\xf_1, \ldots , \xf_N$. 
More precisely, for $j \in \{ 1, \ldots , J\}$, the basis functions $\bfe_{\Xcal, \, j}$ have to be chosen such that 
\begin{enumerate}[label=(\roman*)]
\item\label{P2a:item_basis_function_number}
their number $K_{\Xcal, \, j}$ is at most the number of unique observed values of the subvector of covariates, on which the $j$-th partial effect depends,\footnote{The subvector might contain only one covariate. 
If it depends on more than one covariate, the number of unique observed values of the subvector corresponds to the unique observed combinations of the respective covariates.
}
\item\label{P2a:item_observations_in_support}
all functions contained in $\bfe_{\Xcal, \, j}$ have at least one observation in their support (using B-splines to model smooth effects of continuous covariates).
\end{enumerate}
Let $\Bf_{\Xcal , \, j} \in \Rbb^{N \times K_{\Xcal, \, j}}$ denote the matrix containing $\bfe_{\Xcal , \, j} (\xf_i)$ in the $i$-th row.
First we derive that under conditions~\ref{P2a:item_basis_function_number} and~\ref{P2a:item_observations_in_support}, $\Bf_{\Xcal , \, j}$ has rank $K_{\Xcal, \, j}$.
To see this, note that the rank of $\Bf_{\Xcal, j}$ equals the number of linearly independent rows, which is at most the number of unique rows.
Thus, \ref{P2a:item_basis_function_number} ensures that $K_{\Xcal, \, j}$ is at most the number of unique rows.
The second condition \ref{P2a:item_observations_in_support} then ensures that $\Bf_{\Xcal , \, j}$ has linearly independent columns, or equivalently, rank $K_{\Xcal, \, j}$.
For a linear effect of a continuous covariate or a group specific intercept of a grouping variable this is straightforward. 
For smooth effects modeled via a B-spline basis see \citet[Section XIV, Lemmas (23) and (25)]{deboor2001}.\footnote{\citet[Section XIV, (23) Lemma]{deboor2001} actually states that our condition of interest, i.e., any B-spline function having at least one observation in its support, is equivalent to a function $\Vert \cdot \Vert_2$ on the span of the B-spline basis, which \citet{deboor2001} denote with $\$_{k, t}$ (where $k$ is the order of the B-spline basis and $t$ corresponds to the knot sequence), being a norm on $\$_{k, t}$. 
However, $\Vert \cdot \Vert_2$ being a norm on $\$_{k, t}$ implies that $\Bf_{\Xcal, j}$ has full column rank, which is stated as part of \citet[Section XIV, (25) Lemma]{deboor2001}: ``Assume that $\Vert \cdot \Vert_2$ is a norm on $\$_{k, t}$ (that is, $\rank A = n$) [...]'', where $A$ corresponds to our matrix $\Bf_{\Xcal, j}$ and $n$ corresponds to the number of B-spline functions, ultimately yielding the referenced statement.\label{P2a:footnote_bspline_design_matrix_full_rank}}
Note that in practice, these bases might be transformed to enforce necessary identifiability constraints. 
Per construction this retains linear independence of columns in $\Bf_{\Xcal, j}$ for each $j \in \{ 1, \ldots , J\}$ and yields linear independence of columns in the concatenated matrix across different partial effects $(\Bf_{\Xcal, j} : \Bf_{\Xcal, j'})$ for $j \neq j' \in \{ 1, \ldots , J\}$. 
\\
All in all, the full column rank condition and thus assumption~\ref{P2a:assumption_appendix_covariate_basis_nonsingular} is not restrictive and actually requiring full column rank of the design matrix is very common for regression models.
Furthermore, note that~\ref{P2a:assumption_appendix_covariate_basis_nonsingular} is similar to the nonsingular-condition for logspline regression models \citep[e.g.,][]{stone1991}. 
\item
We will show in Lemma~\ref{P2a:lem_assumptions}~\ref{P2a:lem_assumption_stone_implication} that assumption~\ref{P2a:assumption_appendix_stone_2} implies~\ref{P2a:assumption_appendix_stone_1}, which implies~\ref{P2a:assumption_appendix_covariate_basis_nonsingular} for $N \geq N'$. 
Assumptions~\ref{P2a:assumption_appendix_stone_1} and \ref{P2a:assumption_appendix_stone_2} generalize the lower bounds of Equations~(1) and (2) in \citet{stone1991}, where only one continuous covariate is considered and no proof is provided for the implication of the first assumption by the second one.
We simplify the second assumption, demanding it to hold for one $\alpha \in (0,1)$, while the corresponding assumption (2) in \citet{stone1991} is formulated to hold for all $\alpha > 0$ (with constants $M''$ and $N''$ depending on $\alpha$), which is not necessary for the implication of~\ref{P2a:assumption_appendix_stone_1}.
\footnote{Noteworthy, if \ref{P2a:assumption_appendix_stone_2} holds for an $\alpha \in (0, 1)$, it immediately holds for all $\alpha' > \alpha$, using the same constants $M'' >0$ and $N'' \in \Nbb$, since $N^{\alpha'-1} > N^{\alpha - 1}$ for all $N \geq N''$.}
While some statements only assume \ref{P2a:assumption_appendix_stone_1} (e.g., Theorem~\ref{P2:thm_MLE_asymptotic_normal}), 
this assumption is not very tangible.
However, assumption \ref{P2a:assumption_appendix_stone_2} is: It basically requires that the covariate observations 
are spread reasonably over the whole domain, i.e., the fraction of observations in large enough subregions is lower-bounded by a function of the subregion's volume. 
In particular, assumption \ref{P2a:assumption_appendix_stone_2} also implies that for any $\xf^{\Dfrak} \in \Xcal^{\Dfrak}$, the number of observations taking this particular $\xf^{\Dfrak}$ grows with $N$.
\item
Assumption~\ref{P2a:assumption_appendix_scalar_product_non_constant} is always satisfiable under assumption~\ref{P2a:assumption_appendix_covariate_basis_nonsingular}, 
as we will show in Lemma~\ref{P2a:lem_scalar_product_constant_implies_0} that for all $\vf \in \Sbb^{K-1}$ there exists an $i \in \{ 1, \ldots , N\}$ such that 
the function $s_{\vf, \xf_i}$ is not constant. 
Thus, there exists a subset of $\Ycal$, on which $s_{\vf, \xf_i}$ does not reach its maximum.
We will furthermore show in Lemma~\ref{P2a:lem_assumptions}~\ref{P2a:lem_assumption_appendix_scalar_product_non_constant} that under assumption~\ref{P2a:assumption_appendix_stone_2}, the probability that~\ref{P2a:assumption_appendix_scalar_product_non_constant} holds tends to one for $N \ra \infty$.
\item
We will show in Lemma~\ref{P2a:lem_assumptions}~\ref{P2a:lem_assumption_appendix_domain_basis_nonsingular} that assumption~\ref{P2a:assumption_appendix_domain_basis_nonsingular} is fulfilled, if 
$(\Bft_{\Ycal}^{(l)} : \mathbf{1}_{\Gamma^{(l)}} ) \in \Rbb^{\Gamma^{(l)} \times (K_\Ycal + 1)}$ has rank $K_{\Ycal}+1$ for any $l \in \{1, \ldots , L\}$, 
which for $\bft_{\Ycal}$ chosen as in Section~\ref{P2:chapter_bayes_space_regression} is fulfilled for small enough bin width.
\end{enumerate}
\end{rem}

Before formulating and proving Lemma~\ref{P2a:lem_assumptions}, we show two auxiliary statements (Lemmas~\ref{P2a:lem_scalar_product_constant_implies_0} and~\ref{P2a:lem_sup_continuous}). 

\begin{lem}\label{P2a:lem_scalar_product_constant_implies_0}
Assume that \ref{P2a:assumption_appendix_covariate_basis_nonsingular} holds 
and let $\vf \in \Sbb^{K-1}$.
If the function $s_{\vf, \xf_i} = \bft (\xf_i)^\top \vf$ is constant over $\Ycal$ for every $i \in \{ 1, \ldots , N\}$, then 
$\vf = \mathbf{0}$.
\end{lem}

\begin{proof}
Let $i \in \{ 1, \ldots , N\}$.
Since $\bft (\xf_i) \in L^2_0(\mu)^K$, we obtain that $\bft (\xf_i)^\top \vf$ being constant over $\Ycal$ implies $\bft (\xf_i)^\top \vf = 0$.
Per assumption, this holds for every $i \in \{ 1, \ldots , N\}$.
Let $\vf_{\Xcal} = (\bft_{\Ycal}^\top \vf_{1,1}, \ldots , \bft_{\Ycal}^\top \vf_{J, K_{\Xcal, \,J}})$ with $\vf_{j, n} = (v_{j, n, 1} , \ldots , v_{j, n, K_{\Ycal}})^\top$, where $v_{j, n, m}$ for $j = 1, \ldots , J,~n = 1, \ldots , K_{\Xcal, \, j}$, and $m = 1, \ldots , K_{\Ycal}$, corresponds to the respective entry of $\vf$, where the indices are arranged in the same order as the vector $\thetaf \in \Rbb^K$ containing the values $\theta_{j, n, m}$ in Section~\ref{P2:chapter_bayes_space_regression}.
It is straightforward to show $\bft (\xf_i)^\top \vf = \bfe_{\Xcal} (\xf_i)^\top \vf_{\Xcal}$ 
and from assumption~\ref{P2a:assumption_appendix_covariate_basis_nonsingular} we obtain $\vf_{\Xcal} = \mathbf{0}$, i.e., $\bft_{\Ycal}^\top \vf_{j, n} = 0$ for all $j = 1, \ldots , J,~n = 1, \ldots , K_{\Xcal, \, j}$.
Since 
the functions contained in $\bft_{\Ycal}$ are linearly independent, this implies $\vf_{j, n} = \mathbf{0}$ for all $j = 1, \ldots , J,~n = 1, \ldots , K_{\Xcal, \, j}$, i.e., $\vf = \mathbf{0}$.
\end{proof}

\begin{lem}\label{P2a:lem_sup_continuous}
The function
$\Sbb^{K-1} \times \Xcal \ra \Rbb, (\vf, \xf) \mapsto s_{\vf, \xf}^{\max}$
is (uniformly) continuous.
\end{lem}

\begin{proof}
%
Let $\varepsilon > 0$ and denote the maximum norm, defined as the maximum of the absolute values of the entries of a vector, by $\Vert \cdot \Vert_{\max}$.
We show that there exists $\delta > 0$ such that for all $(\vf, \xf), (\vf', \xf') \in \Sbb^{K-1} \times \Xcal$, we have $|s_{\vf, \xf}^{\max} - s_{\vf', \xf'}^{\max} | < \varepsilon$, if $\Vert (\xf, \vf) - (\xf', \vf') \Vert_{\max} < \delta$.
For this purpose, we use uniform continuity of $s_{\vf, \xf, \mathrm{c}} (y)$ on $\Sbb^{K-1} \times \Xcal \times \Yc$ and continuity of $s_{\vf, \xf} (\yd)$ on $\Sbb^{K-1} \times \Xcal$ for each $\yd$ in the finite set $\Yd$.

Since $s_{\vf, \xf, \mathrm{c}} (y)$ is continuous and $\Sbb^{K-1} \times \Xcal \times \Yc$ is compact, $s_{\vf, \xf, \mathrm{c}} (y)$ is uniformly continuous on $\Sbb^{K-1} \times \Xcal \times \Yc$.
Thus, there exists a $\delta_{\mathrm{c}} > 0$ such that $\Vert (\vf, \xf, y) - (\vf', \xf', y')\Vert_{\max} < \delta_{\mathrm{c}}$ (or equivalently, $\Vert (\vf, \xf) - (\vf', \xf')\Vert_{\max} < \delta_{\mathrm{c}}$ and $| y - y'| < \delta_{\mathrm{c}}$) implies $|s_{\vf, \xf, \mathrm{c}} (y) - s_{\vf', \xf', \mathrm{c}} (y)| < \varepsilon$. 
Furthermore, as $s_{\vf, \xf} (\yd)$ is continuous on the compact set $\Sbb^{K-1} \times \Xcal$ for every $\yd \in \Yd$, there exists a $\delta_{\yd} > 0$ such that $\Vert (\vf, \xf) - (\vf', \xf')\Vert_{\max} < \delta_{\yd}$ implies $|s_{\vf, \xf} (\yd) - s_{\vf', \xf'} (\yd)| < \varepsilon$. 
Set $\delta := \min (\{ \delta_{\mathrm{c}} \} \cup \{ \delta_{\yd} ~|~ \yd \in \Yd \})$.
Note that this minimum exists as $\Yd$ is a finite set.
In the following, let $(\vf, \xf), (\vf', \xf') \in \Sbb^{K-1} \times \Xcal$ with $\Vert (\vf, \xf) - (\vf', \xf')\Vert_{\max} < \delta$.
Since $s_{\vf, \xf, \mathrm{c}}$ is continuous on the compact set $\Yc$, there exists $y_{\vf, \xf, \mathrm{c}} \in \Yc$ such that $s_{\vf, \xf, \mathrm{c}} (y_{\vf, \xf, \mathrm{c}}) = s_{\vf, \xf, \mathrm{c}}^{\max}$.
Furthermore, since $\Yd$ is a finite set, there exists $y_{\vf, \xf, \mathrm{d}} \in \Yd$ such that $s_{\vf, \xf} (y_{\vf, \xf, \mathrm{d}}) = s_{\vf, \xf, \mathrm{d}}^{\max}$.
Since $s_{\vf, \xf}^{\max} = \max \{ s_{\vf, \xf, \mathrm{c}}^{\max}, s_{\vf, \xf, \mathrm{d}}^{\max}\}$, we have 
(i) $s_{\vf, \xf, \mathrm{c}} (y_{\vf, \xf, \mathrm{c}}) = s_{\vf, \xf}^{\max}$
or 
(ii) $s_{\vf, \xf} (y_{\vf, \xf, \mathrm{d}}) = s_{\vf, \xf}^{\max}$.
In case (i), 
we have
$s_{\vf, \xf, \mathrm{c}} (y_{\vf, \xf, \mathrm{c}}) - \varepsilon < s_{\vf', \xf', \mathrm{c}} (y_{\vf, \xf, \mathrm{c}})$
and thus
\begin{align*}
s_{\vf', \xf'}^{\max}
\geq
s_{\vf', \xf' , \mathrm{c}}^{\max}
\geq s_{\vf', \xf', \mathrm{c}} (y_{\vf, \xf, \mathrm{c}})
> s_{\vf, \xf, \mathrm{c}} (y_{\vf, \xf, \mathrm{c}}) - \varepsilon
= s_{\vf, \xf}^{\max} - \varepsilon.
\end{align*}
Similarly, in case (ii), 
we have
$s_{\vf, \xf} (y_{\vf, \xf, \mathrm{d}}) - \varepsilon < s_{\vf', \xf'} (y_{\vf, \xf, \mathrm{d}})$
and thus
\begin{align*}
s_{\vf', \xf'}^{\max}
\geq s_{\vf', \xf'} (y_{\vf, \xf, \mathrm{d}})
> s_{\vf, \xf} (y_{\vf, \xf, \mathrm{d}}) - \varepsilon
= s_{\vf, \xf}^{\max} - \varepsilon.
\end{align*}
Analogously, we obtain
$ s_{\vf, \xf}^{\max} >  s_{\vf', \xf'}^{\max} - \varepsilon$
and
thus, $|s_{\vf, \xf}^{\max} - s_{\vf', \xf'}^{\max} | < \varepsilon$.
\end{proof}


\begin{lem}\label{P2a:lem_assumptions}
\begin{enumerate}[label=\arabic*)]
\item\label{P2a:lem_assumption_covariate_basis_nonsingular}
Assumption~\ref{P2a:assumption_appendix_covariate_basis_nonsingular} is fulfilled, if 
$\Bf_{\Xcal} 
\in \Rbb^{N \times K_{\Xcal}}$ 
has rank $K_{\Xcal}$,
\item\label{P2a:lem_assumption_stone_implication}
Assumption \ref{P2a:assumption_appendix_stone_2} implies \ref{P2a:assumption_appendix_stone_1} which implies \ref{P2a:assumption_appendix_covariate_basis_nonsingular} for all $N \geq N'$.
\item\label{P2a:lem_assumption_appendix_scalar_product_non_constant}
If assumption~\ref{P2a:assumption_appendix_stone_2} holds, then the probability that assumption~\ref{P2a:assumption_appendix_scalar_product_non_constant} holds tends to $1$ for $N \ra \infty$.
\item\label{P2a:lem_assumption_appendix_domain_basis_nonsingular}
Assumption~\ref{P2a:assumption_appendix_domain_basis_nonsingular} is fulfilled, if 
$(\Bft_{\Ycal}^{(l)} : \mathbf{1}_{\Gamma^{(l)}} ) \in \Rbb^{\Gamma^{(l)} \times (K_\Ycal + 1)}$ has rank $K_{\Ycal}+1$ for any $l \in \{1, \ldots , L\}$.
Assume the vector of basis functions $\bft_{\Ycal}$ 
is chosen as proposed in Section~\ref{P2:chapter_bayes_space_regression}.
If $\Yc = \emptyset$, the rank condition holds. 
If $\Yc \neq \emptyset$,
the rank condition holds, if
$K_{\Yc} + 1 \leq G^{(l)}$ and 
for each $m = 1, \ldots, K_{\Yc} + 1$, there exists a $g \in \{ 1, \ldots , G^{(l)}\}$ such that
$\bar{b}_{\Yc, \, m} (u_g^{(l)}) > 0$.
This is fulfilled for small enough bin width.
\end{enumerate}
\end{lem}

\begin{proof}
\begin{enumerate}[label=\arabic*)]
\item\label{P2a:lem_assumption_covariate_basis_nonsingular_proof}
Since $\Bf_{\Xcal}$ has full column rank $K_{\Xcal}$, the equation system $\Bf_{\Xcal} \vf_{\Xcal} = \mathbf{0}$ (which is simply the matrix notation of the equation system in \ref{P2a:assumption_appendix_covariate_basis_nonsingular}) has a unique solution \citep[Corollary~6.6~(2)]{liesen2015}.
Hence, the obvious solution $\vf_{\Xcal} = \mathbf{0}$ is the unique solution. 
This shows \ref{P2a:assumption_appendix_covariate_basis_nonsingular}.
\item\label{P2a:lem_stone_assumptions_implication}
We first show the second implication.
Thus, let $M' > 0 $ and $N' \in \Nbb$ be the constants from \ref{P2a:assumption_appendix_stone_1}.
It is easily seen that then also \ref{P2a:assumption_appendix_covariate_basis_nonsingular} holds for $N \geq N'$: 
Let $\vf_{\Xcal} \in \Rbb^{K_{\Xcal}}$ such that $\bfe_{\Xcal}(\xf_i)^\top \vf_{\Xcal} = 0$ for all $i = 1, \ldots , N$.
Then, $\sum_{i=1}^N \left( \bfe_{\Xcal} (\xf_i)^\top \vf_{\Xcal} \right)^2 = 0$ and we obtain
$M' \, N \int_{\Xcal} \left( \bfe_{\Xcal} (\xf)^\top \vf_{\Xcal} \right)^2 \dnuf (\xf) = 0$ with \ref{P2a:assumption_appendix_stone_1} and using that the integrand is nonnegative.
This implies $\bfe_{\Xcal} (\xf)^\top \vf_{\Xcal} = 0$ for all $\xf \in \Xcal$ and linear independence of the functions contained in $\bfe_{\Xcal}$ yields $\vf_{\Xcal} = \mathbf{0}$, i.e., \ref{P2a:assumption_appendix_covariate_basis_nonsingular}.

Now, we show that \ref{P2a:assumption_appendix_stone_2} implies \ref{P2a:assumption_appendix_stone_1}.
For this purpose, first set 
$\lambdaf = \lambda^{|\Cfrak|}$ and let $\alpha \in (0, 1)$, $M'' >0$, and $N'' \in \Nbb$ be the constants from \ref{P2a:assumption_appendix_stone_2}.
The main ideas are to use multidimensional Riemann sums for the continuous covariates, and to use that the mean of the function values at the observations whose continuous covariates are contained in some connected subset of $\Xcal^{\Cfrak}$ and whose discrete covariates are $\xf^{\Dfrak}$ lies between the function's infimum and supremum on this set.
In the following, when considering partitions, each member is assumed to be a connected set. 
\\
We first show that for every $\tilde{M} > 1$, there exists a $\delta > 0$ such that if a partition $X_1, \ldots , X_T$ of $\Xcal^{\Cfrak}$ fulfills $\max_{t \in \{1, \ldots , T\}} \diam (X_t) < \delta$, we have for any $(\xif_1^{\Cfrak}, \ldots , \xif_T^{\Cfrak}) \in \bigtimes_{t=1}^T X_t$, any $\xf^{\Dfrak} \in \Xcal^{\Dfrak}$, and any $\vf_{\Xcal} \in \Rbb^{K_{\Xcal}}$,
\begin{align}
 \int_{\Xcal} (\bfe_{\Xcal} ^\top \vf_{\Xcal})^2 \, \dnuf
\leq 
\tilde{M} \, 
\sum_{\xf^{\Dfrak} \in \Xcal^{\Dfrak}} \sum_{t=1}^T \lambdaf (X_t) \, \bigl( \bfe_{\Xcal} (\xf^{\Dfrak}, \xif^{\Cfrak}_{t}) ^\top \vf_{\Xcal} \bigr)^2 
\label{P2a:eq_stone_implication_ingetral_bound}
\end{align}
%
Thus, let $\tilde{M} > 1$. 
First note that 
$\langle \vf_{\Xcal}, \wf_{\Xcal} \rangle_{\bfe_{\Xcal}}
:= \int_{\Xcal} ( \bfe_{\Xcal}^\top \vf_{\Xcal} ) \, ( \bfe_{\Xcal}^\top \wf_{\Xcal} ) \, \dnuf$ defines an inner product on $\Rbb^{K_{\Xcal}} \times \Rbb^{K_{\Xcal}}$, which is straightforwardly verified, using the linear independence of the functions contained in $\bfe_{\Xcal}$. 
Then, $\Vert \vf_{\Xcal} \Vert_{\bfe_{\Xcal}} := [\int_{\Xcal} (\bfe_{\Xcal}^\top \vf_{\Xcal})^2 \, \dnuf]^{\frac12}$ defines a norm on $\Rbb^{K_{\Xcal}}$ and thus there exists $M_{\bfe_{\Xcal}} > 0$ such that
$\Vert \vf_{\Xcal} \Vert_2
\leq
M_{\bfe_{\Xcal}} \, \Vert \vf_{\Xcal} \Vert_{\bfe_{\Xcal}} 
$
\citep[e.g.,][Corollary 5.4.5]{horn2012}.
For $\xf^{\Dfrak} \in \Xcal^{\Dfrak}$, consider the real-valued continuous function $\xf^{\Cfrak} \mapsto b_{k, k'} (\xf^{\Dfrak}, \xf^{\Cfrak}) := \bfe_{\Xcal} (\xf^{\Dfrak}, \xf^{\Cfrak})_{[k]} \bfe_{\Xcal} (\xf^{\Dfrak}, \xf^{\Cfrak})_{[k']}$, where the index ${[k]}$ denotes the $k$-th entry of the vector. 
Then, 
there exist $\delta_{\xf^{\Dfrak}, k, k'} > 0$ such that if a partition $X_1, \ldots , X_T$ of $\Xcal^{\Cfrak}$ fulfills $\max_{t \in \{1, \ldots , T\}} \diam (X_t) < \delta_{\xf^{\Dfrak}, k, k'}$, we have
\begin{align*}
\Bigl\vert \int_{\Xcal^{\Cfrak}} b_{k, k'} (\xf^{\Dfrak}, \xf^{\Cfrak}) \, \dlambf (\xf^{\Cfrak}) - \sum_{t=1}^T \lambdaf (X_t) \, b_{k, k'} (\xf^{\Dfrak}, \xif_t^{\Cfrak}) \Bigr\vert 
< \frac{1- \tilde{M}^{-1}}{M_{\bfe_{\Xcal}} \, | \Xcal^{\Dfrak} | \, K_{\Xcal}} 
,
\end{align*}
for any $(\xif_1^{\Cfrak}, \ldots , \xif_T^{\Cfrak}) \in \bigtimes_{t=1}^T X_t$ \citep[e.g.,][p. 376: \emph{Integral of a Continuous Function as Limit}]{fusco2022}.
Let $\delta := \min \{ \delta_{\xf^{\Dfrak}, k, k'} ~|~ \xf^{\Dfrak} \in \Xcal^\Dfrak , \, k, k' \in \{1 , \ldots , K_{\Xcal} \} \} > 0$ and $X_1, \ldots , X_T$ a partition of $\Xcal^{\Cfrak}$ with $\max_{t \in \{1, \ldots , T\}} \diam (X_t)$ $< \delta$.
Then, for any $\vf_{\Xcal} = (v_{\Xcal , \, 1}, \ldots , v_{\Xcal, \, K_{\Xcal}}) \in \Rbb^{K_{\Xcal}}$ and for any $(\xif_1^{\Cfrak}, \ldots , \xif_T^{\Cfrak}) \in \bigtimes_{t=1}^T X_t$, we have
\begin{align*}
&~\Bigl| \int_{\Xcal} (\bfe_{\Xcal}^\top \vf_{\Xcal})^2 \, \dnuf
- \sum_{\xf^{\Dfrak} \in \Xcal^{\Dfrak}} \sum_{t=1}^T \lambdaf (X_t) \, \bigl( \bfe_{\Xcal} (\xf^{\Dfrak}, \xif^{\Cfrak}_{t}) ^\top \vf_{\Xcal} \bigr)^2 \Bigr|
\\
= &~
\Bigl| \sum_{k, k' = 1}^{K_{\Xcal}} v_{\Xcal , \, k} v_{\Xcal , \, k'} \Bigl( \sum_{\xf^{\Dfrak} \in \Xcal^{\Dfrak}} \Bigl( \int_{\Xcal^{\Cfrak}} b_{k, k'}(\xf^{\Dfrak}, \xf^{\Cfrak}) \, \dlambf (\xf^{\Cfrak}) 
- \sum_{t=1}^T \lambdaf (X_t) \, b_{k, k'}(\xf^{\Dfrak}, \xif_t^{\Cfrak}) \Bigr) \Bigr)
\Bigr|
\\
\overset{(\ast)}{\leq} \hspace{-0.1cm} & ~
\Bigl( \sum_{k, k' = 1}^{K_{\Xcal}} v_{\Xcal , \, k}^2 v_{\Xcal , \, k'}^2 \Bigr)^{\frac12} 
\\
& ~
\Bigl( \sum_{k, k' = 1}^{K_{\Xcal}} \Bigl( \sum_{\xf^{\Dfrak} \in \Xcal^{\Dfrak}} \int_{\Xcal^{\Cfrak}} b_{k, k'}(\xf^{\Dfrak}, \xf^{\Cfrak}) \, \dlambf (\xf^{\Cfrak}) - \sum_{t=1}^T \lambdaf (X_t) \, b_{k, k'}(\xf^{\Dfrak}, \xif_t^{\Cfrak}) \Bigr)^2 \Bigr)^{\frac12}
\\
< & ~
\Bigl( \sum_{k = 1}^{K_{\Xcal}} v_{\Xcal , \, k}^2 \Bigr) \, \Bigl( K_{\Xcal}^2 \, | \Xcal^{\Dfrak} |^2 \, \frac{(1- \tilde{M}^{-1})^2}{M_{\bfe_{\Xcal}}^2 \, | \Xcal^{\Dfrak} |^2 \, K_{\Xcal}^2} \Bigr)^{\frac12}
\\
= & ~ \Vert \vf_{\Xcal} \Vert_2^2 \, \frac{1- \tilde{M}^{-1}}{M_{\bfe_{\Xcal}}}
\leq 
(1- \tilde{M}^{-1}) \, \int_{\Xcal} (\bfe_{\Xcal}^\top \vf_{\Xcal})^2 \, \dnuf 
,
\end{align*}
where we used the Cauchy-Schwarz inequality \citep[e.g.,][Theorem~5.1.4]{horn2012} in $(\ast)$.
This 
implies~\eqref{P2a:eq_stone_implication_ingetral_bound}.
\\
In the following, let $\tilde{M} > 1$ and $\delta > 0$ such that~\eqref{P2a:eq_stone_implication_ingetral_bound} holds for partitions of fineness $\delta$.
Since 
$\alpha \in (0, 1)$, we have $\lim_{N\ra \infty} N^{\alpha -1} = 0$, and thus 
there exists $N_\alpha \in \Nbb$ such that 
$N^{\alpha -1} 
< \delta$ for $N \geq N_\alpha$. 
In the following, let $N \geq \max \{N_\alpha, N''\} =: N'$
and let
$X_1, \ldots , X_T$ 
be a partition
of $\Xcal^{\Cfrak}$ such that 
$N^{\alpha -1} 
< \diam (X_t) < \delta$ for all $t \in \{1, \ldots , T\}$. 
Then,
\eqref{P2a:eq_stone_2} holds for 
every $X_t,~ t \in \{ 1, \ldots , T\}$.
%
Furthermore, 
for any such $X_t$, 
any 
$\vf_{\Xcal} 
\in \Rbb^{K_{\Xcal}}$,
and any $\xf^{\Dfrak} \in \Xcal^{\Dfrak}$, the mean of the function values of $\xf^\Cfrak \mapsto (\bfe_{\Xcal} (\xf^{\Dfrak}, \xf^\Cfrak) ^\top \vf_{\Xcal})^2$ at the observations whose continuous covariates are contained in $X_t$ and whose discrete covariates are $\xf^{\Dfrak}$ lies between the function's infimum and supremum: 
\begin{align*}
\inf_{\xf^{\Cfrak} \in X_t} (\bfe_{\Xcal} (\xf^{\Dfrak}, \xf^{\Cfrak}) ^\top \vf_{\Xcal})^2 
&\leq \frac{\sum_{i=1}^N \mathbbm{1}_{\{ \xf^{\Dfrak}\} \times X_t} (\xf^\Dfrak_i , \xf^\Cfrak_i) \, (\bfe_{\Xcal}(\xf^\Dfrak_i, \xf^\Cfrak_i) ^\top \vf_{\Xcal})^2}{\sum_{i=1}^N \mathbbm{1}_{\{ \xf^{\Dfrak}\} \times X_t} (\xf^\Dfrak_i , \xf^\Cfrak_i)} 
\\
&\leq \sup_{\xf^{\Cfrak} \in X_t} (\bfe_{\Xcal} (\xf^{\Dfrak}, \xf^{\Cfrak}) ^\top \vf_{\Xcal})^2.
\end{align*}
With the Intermediate Value Theorem \citep[e.g.,][p.93]{fusco2022}, there exists a $\xif^{\Cfrak}_{\xf^{\Dfrak}, \, t}$ in the interior $X_t^\circ$ of $X_t$, such that 
\begin{align*}
(\bfe_{\Xcal}(\xf^{\Dfrak}, \xif^{\Cfrak}_{\xf^{\Dfrak}, \, t}) ^\top \vf_{\Xcal})^2 = \frac{\sum_{i=1}^N \mathbbm{1}_{\{ \xf^{\Dfrak}\} \times X_t} (\xf^\Dfrak_i , \xf^\Cfrak_i) \, (\bfe_{\Xcal}(\xf^\Dfrak_i, \xf^\Cfrak_i) ^\top \vf_{\Xcal})^2}{\sum_{i=1}^N \mathbbm{1}_{\{ \xf^{\Dfrak}\} \times X_t} (\xf^\Dfrak_i , \xf^\Cfrak_i)} .
\end{align*}
Then, 
\begin{align*}
& ~ \int_{\Xcal} (\bfe_{\Xcal} ^\top \vf_{\Xcal})^2 \, \dnuf
\\
\overset{\eqref{P2a:eq_stone_implication_ingetral_bound}}{\leq} & ~
\tilde{M} \, 
\sum_{\xf^{\Dfrak} \in \Xcal^{\Dfrak}} \sum_{t=1}^T \lambdaf (X_t) \, \bigl( \bfe_{\Xcal} (\xf^{\Dfrak}, \xif^{\Cfrak}_{\xf^{\Dfrak}, \, t}) ^\top \vf_{\Xcal} \bigr)^2 
\\
\overset{(\ast)}{\leq} \hspace{0.25cm} & ~ \tilde{M} \, \sum_{\xf^{\Dfrak} \in \Xcal^{\Dfrak}} \sum_{t=1}^T \diam (X_t)^{|\Cfrak|} \, \frac{\sum_{i=1}^N \mathbbm{1}_{\{ \xf^{\Dfrak}\} \times X_t} (\xf^\Dfrak_i , \xf^\Cfrak_i) \, (\bfe_{\Xcal}(\xf^\Dfrak_i, \xf^\Cfrak_i) ^\top \vf_{\Xcal})^2}{\sum_{i=1}^N \mathbbm{1}_{\{ \xf^{\Dfrak}\} \times X_t} (\xf^\Dfrak_i , \xf^\Cfrak_i)} 
\\
\overset{\eqref{P2a:eq_stone_2}}{\leq} & ~ \tilde{M} \, \sum_{\xf^{\Dfrak} \in \Xcal^{\Dfrak}} \sum_{t=1}^T \diam (X_t)^{|\Cfrak|} \, \frac{\sum_{i=1}^N \mathbbm{1}_{\{ \xf^{\Dfrak}\} \times X_t} (\xf^\Dfrak_i , \xf^\Cfrak_i) \, (\bfe_{\Xcal}(\xf^\Dfrak_i, \xf^\Cfrak_i) ^\top \vf_{\Xcal})^2}{M'' \, N \, \diam (X_t)^{|\Cfrak|}}
\\
= \hspace{0.25cm} & ~ \frac{\tilde{M}}{M'' \, N} \, \sum_{i=1}^N  \, (\bfe_{\Xcal}(\xf_i) ^\top \vf_{\Xcal})^2 ,
\end{align*}
where $(\ast)$ holds since for each $t \in \{1, \ldots , T\}$, there exists a ${|\Cfrak|}$-dimensional cube $Q_t$ with edge length $\diam(X_t)$ such that $X_t \subseteq Q_t$ and due to monotonicity of measures $\lambdaf(X_t) \leq \lambdaf (Q_t) = \diam(X_t)^{|\Cfrak|}$.
Setting $M' := \tilde{M}^{-1} \, M''$, we obtain \ref{P2a:assumption_appendix_stone_1} as desired.
\item
Consider the function
\begin{align*}
\rho 
: \Sbb^{K-1} \times \Xcal \times \Ycal \ra \Rbb^+_0 ,
&& (\vf, \xf, y) 
\mapsto s_{\vf, \xf}^{\max} - s_{\vf, \xf} (y)
.
\end{align*}
We aim to show
$\lim_{N \ra \infty} \Pbb (A_N) = 1$,
where
\[
A_N := \bigl\{(y_1, \ldots , y_N) \in \Ycal^N ~|~ \forall \vf \in \Sbb^{K-1} \exists i \in \{1, \ldots , N\}: 
0 < \rho (\vf, \xf_i, y_i)
\bigr\}.
\]
This probability is hard to compute directly due to $\Sbb^{K-1}$ being uncountable.
Thus, we construct a subset of $A_N$ (which will be denoted as $\At_N$), whose probability tends to $1$ for $N \ra \infty$, based on a finite covering of $\Sbb^{K-1}$, which does not depend on $N$.
More precisely, we define subsets of $\Sbb^{K-1} \times \Xcal \times \Ycal$, on which $\rho$ is positive (i.e., fulfilling the condition of $A_N$ such that $\At_N \subset A_N$ holds), whose projection onto $\Sbb^{K-1}$ yields a finite covering, whose projection onto $\Ycal$ has positive measure $\mu$, and whose projection onto $\Xcal$ allows to apply~\ref{P2a:assumption_appendix_stone_2}.
The latter two properties particularly will help us to prove that the probability of $\At_N$ tends to one.
To construct those subsets, we first introduce some notation.
Denote the maximum norm, defined as the maximum of the absolute values of the entries of a vector, by $\Vert \cdot \Vert_{\max}$.
For $\vf \in \Sbb^{K-1}, \xf \in \Xcal$, 
$y \in \Ycal 
$, 
and
$\varepsilon > 0$, set
\begin{align*}
\upsilon_{\varepsilon} (\vf, \xf, y)
&:= \{ (\vft, \xft, \yt) \in \Sbb^{K-1} \times \Xcal \times \Ycal 
~|~ \Vert (\vft, \xft, \yt) - (\vf, \xf, y) \Vert_{\max} < \varepsilon \}
.
\end{align*} 
Analogously, define sets $\upsilon_{\varepsilon} (\cdot)$, where $(\cdot)$ corresponds to a sub-tupel of $(\vf, \xf, y)$
e.g., 
$\upsilon_{\varepsilon} (\vf)
:= \{ \vft \in \Sbb^{K-1} ~|~ \Vert \vft - \vf \Vert_{\max} < \varepsilon \}$.

Let $\vf \in \Sbb^{K-1}$. 
Then, $\vf \neq \mathbf{0}$ and thus by Lemma~\ref{P2a:lem_scalar_product_constant_implies_0}, for all observations $\xf_1, \ldots , \xf_N \in \Xcal$ fulfilling assumption~\ref{P2a:assumption_appendix_covariate_basis_nonsingular}, there exists an $i \in \{ 1, \ldots , N\}$ such that $\bft (\xf_i)^\top \vf = s_{\vf, \xf_i}$ is not constant.
As by~\ref{P2a:lem_assumption_stone_implication} there exist observations fulfilling~\ref{P2a:assumption_appendix_covariate_basis_nonsingular} (for large enough $N$), there then in particular exists an $\xf_{\vf} \in \Xcal$ such that $s_{\vf, \xf_{\vf}}$ is not constant.
Thus, there exists an 
$y_{\vf} \in \Ycal
$ with $\rho (\vf, \xf_{\vf}, y_{\vf}) > 0$. 
In the following, we consider $\xf_{\vf}$ and $y_{\vf}$ given $\vf$ to be fixed (i.e., for each $\vf$, we choose one particular pair $(\xf_{\vf}, y_{\vf})$ from the set $\{ (\xf_{\vf}, y_{\vf}) \in \Xcal \times \Ycal ~|~ \rho (\vf, \xf_{\vf}, y_{\vf}) > 0\}$).
To define the desired set $\At_N \subset A_N$, we first construct 
a subset of $\Sbb^{K-1} \times \Xcal \times \Ycal$ containing $\{ (\vf, \xf_{\vf}, y_{\vf})\}$ and fulfilling the properties mentioned above.
For this purpose, we distinguish the two cases $y_{\vf} \in \Yc \setminus \Yd$ and $y_{\vf} \in \Yd$.
In the first case, we have $\mu(\{y_{\vf}\}) = \lambda (\{y_{\vf}\}) = 0$ and use continuity of the function $s_{\vf, \xf_{\vf}}$ on $\Yc \setminus \Yd$ to construct a neighborhood of $y_{\vf}$ with positive measure.
In the second case, we have $\mu(\{y_{\vf}\}) = \delta (\{y_{\vf}\}) > 0$ and thus can use $\{y_{\vf}\}$ directly. 
Note that in this case, $s_{\vf, \xf_{\vf}}$ is not continuous in $y_{\vf}$ and thus the construction from the first case is not applicable, making separate treatment of the cases necessary.
Nevertheless, in both cases, we proceed similarly:
\begin{enumerate}
\item\label{P2a:item_scalar_product_non_constant_c}
If $y_{\vf} \in \Yc \setminus \Yd$,
as $(\vf, \xf, y) \mapsto s_{\vf, \xf} (y)$ is continuous on 
$\Sbb^{K-1} \times \Xcal \times (\Yc \setminus \Yd)$,\footnote{Note that for categorical covariates this is trivially fulfilled as any function defined on a finite set is continuous.\label{P2a:footnote_continuity_on_finite_set}}
we obtain that $\rho$ is continuous in $(\vf, \xf_{\vf}, y_{\vf})$ using Lemma~\ref{P2a:lem_sup_continuous}.
Thus, there exists an $\varepsilon_{\vf} > 0$, 
such that for all $(\vft , \xft, \yt) \in \upsilon_{\varepsilon_{\vf}} (\vf, \xf_{\vf}, y_{\vf})$, we have 
$\rho (\vft, \xft, \yt) > 
0$.
Set 
$\Upsilon_{\vf} := \upsilon_{\varepsilon_{\vf}} (\xf_{\vf}, y_{\vf})$ for  $y_{\vf} \in \Yc \setminus \Yd$.
\item\label{P2a:item_scalar_product_non_constant_d}
If $y_{\vf} \in \Yd$,
as $(\vf, \xf) \mapsto s_{\vf, \xf} (y_{\vf})$ is continuous on 
$\Sbb^{K-1} \times \Xcal$,\textsuperscript{\ref{P2a:footnote_continuity_on_finite_set}} 
we obtain that
$(\vf, \xf) \mapsto \rho (\vf, \xf, y_{\vf})$ 
is continuous in $(\vf, \xf_{\vf})$ using Lemma~\ref{P2a:lem_sup_continuous}.
Thus, there exists an $\varepsilon_{\vf} > 0$, 
such that for all $(\vft , \xft) \in \upsilon_{\varepsilon_{\vf}} (\vf, \xf_{\vf})$, 
we have 
$\rho
(\vft, \xft, y_{\vf}) > 
0$.
Set 
$\Upsilon_{\vf} := \upsilon_{\varepsilon_{\vf}} (\xf_{\vf}) \times \{ y_{\vf} \}$ for  $y_{\vf} \in \Yd$.
\end{enumerate}
Since $\Sbb^{K-1}$ is compact, there exist $Z \in \Nbb$ and $\vf_{1}, \ldots , \vf_{Z} \in \Sbb^{K-1}$ such that 
$\Sbb^{K-1} = \bigcup_{z=1}^Z \upsilon_{\varepsilon_{\vf_z}} (\vf_z)$.
Note that the construction of this covering does not depend on $N$ and is considered to be fix in the remaining proof.
Consider the set
\begin{align*}
\At_N &:= 
\bigl\{(y_1, \ldots , y_N) \in \Ycal^N ~|~ \forall z \in \{ 1, \ldots , Z\} 
\exists i \in \{ 1, \ldots , N \}: 
(\xf_i, y_i) \in \Upsilon_{\vf_z} 
\bigr\}
.
\end{align*}
We show that
$\At_N \subseteq A_N$.
For this purpose, let 
$(y_1, \ldots , y_N) \in \At_N$ and let $\vf \in \Sbb^{K-1}$.
We need to show that there exists an $i \in \{ 1, \ldots , N\}$ with $s_{\vf, \xf_i} (y_i) < s_{\vf, \xf_i}^{\max}$.
As $\vf \in \Sbb^{K-1} = \bigcup_{z=1}^Z \upsilon_{\varepsilon_{\vf_z}} (\vf_z)$, 
there exists a $z \in \{1, \ldots , Z\}$ with
$\vf \in \upsilon_{\varepsilon_{\vf_z}} (\vf_z)$.
Furthermore, as $(y_1, \ldots , y_N) \in \At_N$, there exists
$i \in \{ 1, \ldots , N \}$ for this $z$ with
$(\xf_{i}, y_{i}) \in \Upsilon_{\vf_z}$. 
Let $\xf_{\vf_z} \in \Xcal$ and $y_{\vf_z} \in \Ycal$ be as above, i.e., such that $\rho (\vf_z, \xf_{\vf_z}, y_{\vf_z}) > 0$.
We again distinguish based on $y_{\vf_z}$:
\begin{enumerate}
\item
If $y_{\vf_z} \in \Yc \setminus \Yd$, 
we have
$(\xf_{i}, y_{i}) \in \Upsilon_{\vf_z} = \upsilon_{\varepsilon_{\vf_z}} (\xf_{\vf_z}, y_{\vf_z})$,
which together with 
$\vf \in \upsilon_{\varepsilon_{\vf_z}} (\vf_z)$ 
yields
$(\vf, \xf_{i
}, y_{i
}) \in \upsilon_{\varepsilon_{\vf_z}} (\vf_z, \xf_{\vf_z}, y_{\vf_z}) 
$.
By (a) 
above, we obtain
$\rho
(\vf, \xf_{i}, y_{i}) > 
0$ and thus 
$s_{\vf, \xf_{i}} (y_{i}) < s_{\vf, \xf_{i}}^{\max}$.
\item
If $y_{\vf_z} \in \Yd$, 
we have
$(\xf_{i}, y_{i}) \in \Upsilon_{\vf_z} = \upsilon_{\varepsilon_{\vf_z}} (\xf_{\vf_z}) \times \{y_{\vf_z}\}$,
which together with 
$\vf \in \upsilon_{\varepsilon_{\vf_z}} (\vf_z)$ 
yields
$(\vf, \xf_{i}) \in \upsilon_{\varepsilon_{\vf_z}} (\vf_z, \xf_{\vf_z})$.
By (b) 
above, we obtain
$\rho
(\vf, \xf_{i}, y_{i}) 
= (\vf, \xf_{i}, y_{\vf_z})> 
0$ and thus
$s_{\vf, \xf_{i}} (y_{i}) < s_{\vf, \xf_{i}}^{\max}$.
\end{enumerate}
This shows $(y_1, \ldots , y_N) \in A_N$ and thus $\At_N \subseteq A_N$, which yields 
\begin{align*}
&~ \Pbb ( A_N)
 \geq \Pbb ( 
\At_N )
\\
\geq & ~ 1 - \sum_{z=1}^Z \Pbb \bigl(\left\{(y_1, \ldots , y_N) \in \Ycal^N ~|~ \forall i \in \{ 1, \ldots , N\}:
(\xf_i, y_i) \notin \Upsilon_{\vf_z} 
\right\} \bigr)
.
\intertext{As the defintion of $\Upsilon_{\vf_z}$ depends on whether $y_{\vf_z} \in \Yc \setminus \Yd$ or $y_{\vf_z} \in \Yd$, we split the sum over $\{ 1, \ldots , Z\}$ into two sums: 
One over 
$\Zcal_{\mathrm{d}} := \{ z \in \{ 1, \ldots , Z\} ~|~ y_{\vf_z} \in \Yd\}$ 
and one over
$\Zcal_{\mathrm{c}} := \{ z \in \{ 1, \ldots , Z\} ~|~ y_{\vf_z} \in \Yc \setminus \Yd\}$.
Then, the last line equals
}
& ~ 1 - 
\sum_{z \in \Zcal_{\mathrm{d}}} \Pbb \bigl(\bigl\{(y_1, \ldots , y_N) \in \Ycal^N ~|~ \forall i \in \{ 1, \ldots , N\}:
\\
&\hspace{6.1cm}
\xf_i \notin \upsilon_{\varepsilon_{\vf_z}} (\xf_{\vf_z}) \vee y_{i} \neq y_{\vf_z}
\bigr\} \bigr)
\\
& \hspace{0.35cm}
- \sum_{z \in \Zcal_{\mathrm{c}}} \Pbb \bigl(\bigl\{(y_1, \ldots , y_N) \in \Ycal^N ~|~ \forall i \in \{ 1, \ldots , N\}:
\\
&\hspace{6.1cm}
\xf_i \notin \upsilon_{\varepsilon_{\vf_z}} (\xf_{\vf_z}) \vee y_{i} \notin \upsilon_{\varepsilon_{\vf_z}} (y_{\vf_z})
\bigr\} \bigr)
\\
= & ~ 1 -
\sum_{z \in \Zcal_{\mathrm{d}}} \prod_{i = 1}^N \Bigl( 1 - \Pbb \bigl(\bigl\{y_i \in \Ycal ~|~ \xf_i \in \upsilon_{\varepsilon_{\vf_z}} (\xf_{\vf_z}) \wedge y_{i} = y_{\vf_z} \bigr\} \bigr) \Bigr)
\\
&\hspace{0.35cm} -
\sum_{z \in \Zcal_{\mathrm{c}}} \prod_{i = 1}^N \Bigl( 1 - \Pbb \bigl(\bigl\{y_i \in \Ycal ~|~ \xf_i \in \upsilon_{\varepsilon_{\vf_z}} (\xf_{\vf_z}) \wedge y_{i} \in \upsilon_{\varepsilon_{\vf_z}} (y_{\vf_z}) \bigr\} \bigr) \Bigr)
,
\intertext{where we used independence of $y_1, \ldots , y_N$ in the last equality.
For $z \in \{ 1, \ldots , Z\}$, set $\Ical_{z} := \{i \in \{ 1, \ldots , N \} ~|~ \xf_i \in \upsilon_{\varepsilon_{\vf_z}} (\xf_{\vf_z}) \}$. 
Then, $\xf_i \notin \upsilon_{\varepsilon_{\vf_z}} (\xf_{\vf_z})$ for any $i \in \Ical_{z}^\complement$ and thus, the corresponding factors in each of the products are $1$. 
Thus, the last line equals}
& ~ 1 - \sum_{z \in \Zcal_{\mathrm{d}}} \prod_{i \in \Ical_{z} }
\Bigl( 1 - \Pbb \bigl(\bigl\{y_i \in \Ycal ~|~ y_{i} = y_{\vf_z} \bigr\} \bigr) \Bigr)
\\
& - \sum_{z \in \Zcal_{\mathrm{c}}} \prod_{i \in \Ical_{z} }
\Bigl( 1 - \Pbb \bigl(\bigl\{y_i \in \Ycal ~|~ y_{i} \in \upsilon_{\varepsilon_{\vf_z}} (y_{\vf_z}) \bigr\} \bigr) \Bigr)
\\
= & ~ 1 - \sum_{z \in \Zcal_{\mathrm{d}}} \prod_{i \in \Ical_{z} }
\Bigl( 1 - \int_{\{ y_{\vf_z} \}} f_{\xf_i} \, \dmu \Bigr)
- \sum_{z \in \Zcal_{\mathrm{c}}} \prod_{i \in \Ical_{z} }
\Bigl( 1 - \int_{\upsilon_{\varepsilon_{\vf_z}} (y_{\vf_z})} f_{\xf_i} \, \dmu \Bigr)
,
\end{align*}
Since $f_{\xf}$ is ($\mu$-a.e.) positive for all $\xf \in \Xcal$,
$\Xcal$ and $\Ycal$ are compact, 
and $f_{\xf} (y)$ is continuous on $\Xcal \times \Ycal$,
$\min \{ f_{\xf}(y) ~|~ \xf \in \Xcal, y \in \Ycal\} > 0$ exists.
Furthermore,
we have
$\mu (\{ y_{\vf_z} \}) = \delta (\{ y_{\vf_z} \}) > 0$ for all $z \in \Zcal_{\mathrm{d}}$
and
$\mu (\upsilon_{\varepsilon_{\vf_z}} (y_{\vf_z})) \geq \lambda (\upsilon_{\varepsilon_{\vf_z}} (y_{\vf_z})) = 2\, \varepsilon_{\vf_z} > 0$ for all $z \in \Zcal_{\mathrm{c}}$.
Thus, 
\begin{align*}
\delta 
&:= \min \Bigl(
\Bigl\{ \int_{\{ y_{\vf_z} \}} f_{\xf} \, \dmu ~\big|~ z \in \Zcal_{\mathrm{d}}, x \in \Xcal \Bigr\} 
\cup 
\Bigl\{ \int_{\upsilon_{\varepsilon_{\vf_z}} (y_{\vf_z})} f_{\xf} \, \dmu ~\big|~ z \in \Zcal_{\mathrm{c}}, x \in \Xcal \Bigr\}  
\Bigr)
\end{align*}
exists and is in $(0, 1)$.
Hence,
$\Pbb ( A_N)
\geq
1 - \sum_{z = 1}^Z ( 1 - \delta )^{|\Ical_{z}|}$.
Finally, we show $\lim_{N \ra \infty} |\Ical^{(z)}| = \infty$ for all $z = 1, \ldots , Z$, which immediately implies the desired result $\lim_{N \ra \infty} \Pbb (A_N) = 1$.
For this purpose, let $\Mt \in \Rbb$ be an arbitrary constant.
We show that there exists an $\Nt \in \Nbb$ such that $\tilde{M} \leq |\Ical^{(z)}|$ for all $N \geq \Nt$.
Let $\alpha \in (0, 1), M'' > 0$, and $N'' \in \Nbb$ be the constants from~\ref{P2a:assumption_appendix_stone_2}.
Since $\varepsilon_{\vf_z}$ is independent of $N$ and $\lim_{N \ra \infty} N^{\alpha -1} = 0$, there exists an $N_{\alpha} \in \Nbb$ such that $N^{\alpha -1} < \diam \left( \upsilon_{\varepsilon_{\vf_z}} (\xf_{\vf_z}^{\Cfrak}) \right)$ for all $N \geq N_{\alpha}$.
Setting
$\Nt := \max \Bigl\{ N_{\alpha}, N'', \frac{\Mt}{M'' \, \diam\left( \upsilon_{\varepsilon_{\vf_z}} (\xf_{\vf_z}^{\Cfrak}) \right)^{|\Cfrak|}} \Bigr\}$,
we obtain for any
$ 
N \geq 
\Nt
,
$ 
by assumption~\ref{P2a:assumption_appendix_stone_2},
\begin{align*}
\Mt \leq
M'' \, N \, \diam\left( \upsilon_{\varepsilon_{\vf_z}} (\xf_{\vf_z}^{\Cfrak}) \right)^{|\Cfrak|}
\leq \sum_{i=1}^N \mathbbm{1}_{\{ \xf_{\vf_z}^{\Dfrak}\} \times [\upsilon_{\varepsilon_{\vf_z}} (\xf_{\vf_z}^{\Cfrak})]} (\xf^\Dfrak_i , \xf^\Cfrak_i)
\leq | \Ical^{(z)} |
.
\end{align*}
This shows 
$\lim_{N \ra \infty} |\Ical^{(z)}| = \infty$ and thus $\lim_{N \ra \infty} \Pbb (A_N) = 1$.
\item\label{P2a:rem_assumption_appendix_domain_basis_nonsingular_proof}
We first show that~\ref{P2a:assumption_appendix_domain_basis_nonsingular} holds, if 
$(\Bft_{\Ycal}^{(l)} : \mathbf{1}_{\Gamma^{(l)}} )$ has rank $K_{\Ycal} + 1$
for all $l \in \{1, \ldots , L\}$.
Thus, let $l \in \{1, \ldots , L\}$ and assume that the matrix 
$(\Bft_{\Ycal}^{(l)} : \mathbf{1}_{\Gamma^{(l)}} )$ has rank $K_{\Ycal} + 1$.
Then, $\Bft_{\Ycal}^{(l)}$ has rank $K_{\Ycal} < K_{\Ycal} + 1$ and there thus is no $\vf_{\Ycal} \in \Rbb^{K_{\Ycal}}$ such that $\Bft_{\Ycal}^{(l)} \vf_{\Ycal} = \mathbf{1}_{\Gamma^{(l)}}$ \citep[Corollary~6.6~(1)]{liesen2015}.
Equivalently, for $c \in \Rbb \setminus \{ 0 \}$, there is no $\vf_{\Ycal} \in \Rbb^{K_{\Ycal}}$ such that $\Bft_{\Ycal}^{(l)} \vf_{\Ycal} = c \, \mathbf{1}_{\Gamma^{(l)}}$. 
However, for $c = 0$, we obtain $\vf_\Ycal = \mathbf{0}$ as the unique solution 
by \citet[Corollary~6.6~(2)]{liesen2015}.
This shows assumption~\ref{P2a:assumption_appendix_domain_basis_nonsingular}.
\\
Now assume that $\bft_{\Ycal}$ is chosen as proposed in Section~\ref{P2:chapter_bayes_space_regression}.
If $\Yc \neq \emptyset$, furthermore assume that $K_{\Yc} + 1 \leq G^{(l)}$ and that for each unconstrained B-spline basis $\bb_{\Yc, \, m}$, there exists a $u_g^{(l)}$ in its support.
We show that then, the matrix 
$( 
\Bft_{\Ycal}^{(l)} : \mathbf{1}_{\Gamma^{(l)}}
) 
\in \Rbb^{\Gamma^{(l)} \times (K_\Ycal + 1)}$ has rank $K_{\Ycal}+1$.
We first consider the continuous and discrete special cases (where $\Gamma^{(l)} = G^{(l)}$ and $\Gamma^{(l)} = D$, respectively). 
By Lemma~\ref{P2a:lem_characterization_constrined_basis} and \citet[Lemmas 9.22 and 9.23]{liesen2015},
there exist transformation matrices 
$\Af_{\Yc} \in \Rbb^{(K_{\Yc} + 1) \times (K_{\Yc} + 1)}$ with rank $K_{\Yc} + 1$
and 
$\Af_{\Yd} \in \Rbb^{(K_{\Yd} + 1) \times (K_{\Yd} + 1)}$ with rank $K_{\Yd} + 1$
such that
$(\bt_{\Yc, \, 1}, \ldots , \bt_{\Yc, \, K_{\Yc}}, 1_{\Yc}) = \bfb_{\Yc} 
\Af_{\Yc}$
and
$(\bt_{\Yd, \, 1}, \ldots , \bt_{\Yd, \, K_{\Yd}}, 1_{\Yd}) = \bfb_{\Yd} 
\Af_{\Yd}$.
\footnote{Actually, the first $K_{\Yc}$ ($K_{\Yd}$) columns of $\Af_{\Yc}$ ($\Af_{\Yd}$) equal the transformation matrix $\Zf_{\Yc}$ ($\Zf_{\Yd}$), 
while the last column is $\mathbf{1}_{K_{\Yc} + 1}$ ($\mathbf{1}_{K_{\Yd} + 1}$).} 
Furthermore, the matrix $\bar{\Bf}_{\Yc}^{(l)} \in \Rbb^{G^{(l)} \times (K_{\Yc} + 1)}$ containing $\bfb_{\Yc} (u_g^{(l)})$ in the $g$-th row, $g = 1, \ldots , G^{(l)}$, has rank $K_{\Yc} + 1$, since for every unconstrained B-spline basis function $\bar{b}_{\Yc, \, 1}, \ldots , \bar{b}_{\Yc, \, K_{\Yc} + 1}$ there is at least one $g \in \{ 1, \ldots , G^{(l)}\}$ such that $u_g^{(l)}$ is in its support. 
This follows from \citet[Section XIV, Lemmas~(23) and (25)]{deboor2001}, compare also footnote~\ref{P2a:footnote_bspline_design_matrix_full_rank} in Remark~\ref{P2a:rem_assumptions_lem}~\ref{P2a:rem_assumption_covariate_basis_nonsingular}.
Similarly, the matrix $\bar{\Bf}_{\Yd} \in \Rbb^{D \times (K_{\Yd} + 1)}$ containing $\bfb_{\Yd} (t_{d})$ in the $d$-th row, $d = 1, \ldots , D$, is simply the identity matrix of dimension $K_{\Yd} + 1 = D$ and thus has rank $K_{\Yd} + 1$.
Thus,
$( 
\Bft_{\Yc}^{(l)} : \mathbf{1}_{G^{(l)}}
) 
= \bar{\Bf}_{\Yc}^{(l)} \Af_{\Yc}$
has rank $K_{\Yc} + 1$
and
$
( 
\Bft_{\Yd} : \mathbf{1}_{D}
) 
= \bar{\Bf}_{\Yd} \Af_{\Yd} 
$
has rank $K_{\Yd} + 1$ \citep[Lemma 9.23]{liesen2015}.
\\
In the mixed case, using the embeddings $\Jtc$ and $\Jtd$ as given in~\eqref{P2a:eq_embedding_L20_c} and~\eqref{P2a:eq_embedding_L20_d}
and recalling $u_g^{(l)} = t_{g - G^{(l)}}$ for $g = G^{(l)} + 1 , \ldots , \Gamma^{(l)}$, we obtain
\begin{align*}
\bigl( 
\Bft_{\Ycal}^{(l)} : \mathbf{1}_{\Gamma^{(l)}}
\bigr) 
= 
\begin{pmatrix}
\Bft_{\Yc}^{(l)} & \bft_{\Yd^\bullet}(t_{D + 1})^\top \otimes \mathbf{1}_{G^{(l)}} & \mathbf{1}_{G^{(l)}}
\\
\mathbf{0}_{D \times K_{\Yc}} & \Bft_{\Yd^\bullet} (t_1, \ldots , t_D) & \mathbf{1}_D
\end{pmatrix},
\end{align*}
where $\mathbf{0}_{K_{\Yd^\bullet} \times K_{\Yc}}$ is the $(K_{\Yd^\bullet} \times K_{\Yc})$-dimensional matrix containing zero in every entry and $\Bft_{\Yd^\bullet} (t_1, \ldots , t_D) \in \Rbb^{D \times K_{\Yd^\bullet}}$ is the matrix containing $\bft_{\Yd^\bullet} (t_{d})$ in the $d$-th row, $d = 1, \ldots , D$, i.e., the matrix $\Bft_{\Yd^\bullet}$ without the last row $\bft_{\Yd^\bullet} (t_{D + 1})$.
Note that this last row is, however, contained $G^{(l)}$ times in the block above.
Since we already showed the statement for the continuous case, we in particular have that $\Bft_{\Yc}^{(l)}$ has rank $K_{\Yc}$.
Thus, using elementary row operations, the upper left matrix $\Bft_{\Yc}^{(l)}$ can be transformed to
a $(G^{(l)} \times K_{\Yc})$-dimensional matrix containing the $K_{\Yc}$-dimensional identity matrix $\Id_{K_{\Yc}}$ in the upper $K_{\Yc}$ rows and $G^{(l)} - K_{\Yc}\geq 1$ rows containing zeros below \citep[Section 5]{liesen2015}. 
Applying these row operations to the whole matrix, the lower blocks are unaffected, while the remaining matrix in the upper block sharing the same rows as the upper left block $\Bft_{\Yc}^{(l)}$ 
contains the same vector $(\bt_{\Yd^\bullet, \, 1}(t_{D + 1}), \ldots , \bt_{\Yd^\bullet, \, K_{\Yd^\bullet}}(t_{D + 1}), 1)$ in each row.
Thus, performing the row operations on this block simply results in each row containing a multiple of the original vector (with the multiplicity depending on the row).
The first $K_{\Yc}$ rows in this upper middle-right block can now be eliminated via elementary column operations using the identity matrix in the left block.
Since the blocks below the identity matrix contains only zeros, the remaining rows stay unchanged, when applying these column operations to the whole matrix.
After these transformations, the matrix has the form
\begin{align*}
\begin{pmatrix}
\Id_{K_{\Yc}} & \mathbf{0}_{K_{\Yc} \times K_{\Yd^\bullet}} & \mathbf{0}_{K_{\Yc} \times 1}
\\
\mathbf{0}_{(G^{(l)} - K_{\Yc}) \times K_{\Yc}} & \bft_{\Yd^\bullet}(t_{D + 1})^\top \otimes \vf_{G^{(l)} - K_{\Yc}} & \vf_{G^{(l)} - K_{\Yc}}
\\
\mathbf{0}_{D \times K_{\Yc}} & \Bft_{\Yd^\bullet} (t_1, \ldots , t_D) & \mathbf{1}_{D}
\end{pmatrix},
\end{align*}
where $\vf_{G^{(l)} - K_{\Yc}} \in \Rbb^{G^{(l)} - K_{\Yc}}$ is the vector containing the multiplicity of the original row $(\bt_{\Yd^\bullet, \, 1}(t_{D + 1}), \ldots , \bt_{\Yd^\bullet, \, K_{\Yd^\bullet}}(t_{D + 1}), 1)$ for each row.
Thus, the last of these rows can be used to eliminate all of the rows above in the middle-row-blocks.
Now, the left block stays unchanged as it contains only zero.
Finally, we divide the last row by its multiplicity and obtain
\begin{align*}
\begin{pmatrix}
\Id_{K_{\Yc}} & \mathbf{0}_{K_{\Yc} \times K_{\Yd^\bullet}} & \mathbf{0}_{K_{\Yc} \times 1}
\\
\mathbf{0}_{(G^{(l)} - K_{\Yc} - 1) \times K_{\Yc}} & \mathbf{0}_{(G^{(l)} - K_{\Yc} - 1) \times K_{\Yd^\bullet}} & \mathbf{0}_{(G^{(l)} - K_{\Yc} - 1) \times 1}
\\
\mathbf{0}_{(D + 1) \times K_{\Yc}} & \Bft_{\Yd^\bullet} & \mathbf{1}_{D + 1}
\end{pmatrix}.
\end{align*}
Using the results above for the discrete case, we have that 
$( 
\Bft_{\Yd^\bullet} : \mathbf{1}_{D + 1}
) 
$ has rank $K_{\Yd^\bullet} + 1$ 
and thus
$( 
\Bft_{\Ycal}^{(l)} : \mathbf{1}_{\Gamma^{(l)}}
) 
$ has rank $K_{\Yc}+K_{\Yd^\bullet} + 1 = K_{\Ycal}+1$.
\\
Finally, the assumptions for the case $\Yc \neq \emptyset$ 
are fulfilled for small enough bin width:
The first assumption $K_{\Yc} + 1 \leq G^{(l)}$ means that the number of bins needs to be at least the number of unconstrained B-spline functions is straightforwardly fulfilled for small enough bin widths.
The second assumption that for each unconstrained B-spline basis $\bb_{\Yc, \, m}$, there exists a $u_g^{(l)}$ in its support, is fulfilled, if the maximal bin width is so small, that each pair of adjacent knots underlying the unconstrained B-spline basis contains at least one bin. 
Thus, both conditions are fulfilled for small enough bin width.
\qedhere
\end{enumerate}
\end{proof}

\subsubsection{Auxiliary statements used in multiple proofs in Section~\ref{P2a:chapter_proofs}}\label{P2a:chapter_auxiliary_statements}

Before we start formulating and proving the auxiliary statements, we give a brief overview.
In Proposition~\ref{P2a:prop_appendix_fisher_informations_positive_definite}, we show negative (positive) definiteness of the Hessian matrices (penalized Fisher informations) of $\ell_{\pen}$ and $\ell_{\Zcal, \pen}^{\mathrm{mn}}$, which also implies invertibility of those matrices.
Lemma~\ref{P2a:lem_matrix_inversion_continuous} states the technical result that matrix inversion is a continuous function. 
The remaining theorems,~\ref{P2a:thm_generalized_laplace_approximation},~\ref{P2a:thm_loglikelihood_coercive} and~\ref{P2a:thm_appendix_loglikelihood_level_bounded} actually form a cascade of statements: 
We derive a generalization of Laplace's approximation in Theorem~\ref{P2a:thm_generalized_laplace_approximation}, 
which is used to show $\lim_{\Vert \thetaf \Vert_2 \ra \infty} \ell_{\pen} (\thetaf) = - \infty$ in Theorem~\ref{P2a:thm_loglikelihood_coercive},
which helps us to prove that $- \ell_{\pen}$ is level-bounded in Theorem~\ref{P2a:thm_appendix_loglikelihood_level_bounded}.
Ultimately, of these last three theorems in this section, only Theorem~\ref{P2a:thm_appendix_loglikelihood_level_bounded} is referred to in the proofs in Section~\ref{P2a:chapter_proofs}.

Denote the (multivariate) expectation with respect to the conditional distributions $\Yf_i ~|~ \xf_i, ~ i = 1, \ldots , N$, with $\Ebb$ and the (multivariate) expectation with respect to the corresponding multinomial distributions with respect to partitions $\Zcal$ with $\Ebb_\Zcal^{\mathrm{mn}}$.

\begin{prop}\label{P2a:prop_appendix_fisher_informations_positive_definite}
Assume that \ref{P2a:assumption_appendix_covariate_basis_nonsingular} holds. Then, for any $\thetaf \in \Rbb^K$,
\begin{enumerate}[label=\arabic*)]
\item\label{P2a:prop_appendix_bayes_fisher_information_positive_definite}
the Hessian matrix $\Hf_{\mathrm{pen}}(\thetaf)$ of $\ell_{\mathrm{pen}} (\thetaf)$ is negative definite and thus invertible, 
and $\Hf_{\mathrm{pen}}(\thetaf) = \Ebb (\Hf_{\mathrm{pen}} (\thetaf))$,
\item\label{P2a:prop_appendix_multinomial_fisher_information_positive_definite}
under 
assumption \ref{P2a:assumption_appendix_domain_basis_nonsingular}, 
the Hessian matrix $\Hf_{\Zcal, \mathrm{pen}}^{\mathrm{mn}}(\thetaf)$ of $\ell_{\Zcal, \mathrm{pen}}^{\mathrm{mn}} (\thetaf)$ is negative definite and thus invertible, 
and $\Hf_{\Zcal, \mathrm{pen}}^{\mathrm{mn}}(\thetaf) = \Ebb_\Zcal^{\mathrm{mn}} (\Hf_{\Zcal, \mathrm{pen}}^{\mathrm{mn}}(\thetaf))$.
\end{enumerate}
The corresponding penalized Fisher informations are positive definite, invertible, and equal their expected versions.
\end{prop}

\begin{proof}
We have $\Hf_{\mathrm{pen}} (\thetaf) = \Hf (\thetaf) + D^2 \pen (\thetaf)$ and $\Hf_{\Zcal, \mathrm{pen}}^{\mathrm{mn}} (\thetaf) = \Hf_{\Zcal}^{\mathrm{mn}} (\thetaf) + D^2 \pen (\thetaf)$.
Since the matrices $\Pf_j$ are positive semi-definite for all $j \in \{1, \ldots , J\}$, the matrix $D^2 \pen (\thetaf) = - 2 \Pf$ is negative semi-definite.
Thus, it suffices to show that the Hessian matrices $\Hf (\thetaf)$ and $\Hf_{\Zcal}^{\mathrm{mn}} (\thetaf)$ are negative definite, since a sum of a negative definite and a negative semi-definite matrix is negative definite 
%
%
and in particular invertible.
The statements for the penalized Fisher informations follows immediately.

In the following, we 
use the notation $\gammaf_{[k]}$ to denote the $k$-th entry of a vector $\gammaf$.

\begin{enumerate}[label=\arabic*)]
\item
Setting $C_i(\thetaf) := - \log \int_{\Ycal} \exp [ \bft (\xf_i)^\top \thetaf ]\, \dmu$ for $i = 1, \ldots , N$, we have
$\ell (\thetaf)
= \sum_{i=1}^N [ \bft (\xf_i)(y_i)^\top \thetaf - C_i(\thetaf) ]
$
and thus $\Hf(\thetaf) = \sum_{i=1}^N \Hf_i (\thetaf)$,
where $\Hf_i (\thetaf) \in \Rbb^{K \times K}$ denotes the Hessian matrix of 
$C_i(\thetaf)$ for $i = 1, \ldots , N$.
For its computation, we briefly show that $\frac{\partial}{\partial \theta_k} \int_\Ycal \exp [ \bft (\xf_i)^\top \thetaf ] \, \dmu = \int_\Ycal \frac{\partial}{\partial \theta_k} \exp [ \bft (\xf_i)^\top \thetaf ] \, \dmu$ by verifying the three properties given in Corollary 16.3 in \citet{bauer2001} for the considered function $\exp [ \bft (\xf_i)^\top \thetaf ]$. 
The first two required properties are straightforward:
For all $\thetaf \in \Rbb^K 
$, $\exp [ \bft (\xf_i)^\top \thetaf ]$ is $\mu$-inte\-grable
and for all $y \in \Ycal$ and $k \in \{ 1, \ldots , K \}$,
the partial derivative 
$\frac{\partial}{\partial \theta_k} \exp [ \bft (\xf_i)(y)^\top \thetaf ] 
= \exp [ \bft (\xf_i)(y)^\top \thetaf ] \, (\bft (\xf_i))_{[k]}
$ exists.
To show the third required property, note that since the basis functions contained in $\bft_{\Ycal}$ are bounded and $\mu$ is finite, $\vert \frac{\partial}{\partial \theta_k} \exp [ \bft (\xf_i)(y)^\top \thetaf ] \vert$ is bounded on an open ball 
$\upsilon_\varepsilon(\thetaf) := \{\thetaf' \in \Rbb^K 
~|~ \Vert \thetaf' - \thetaf \Vert_2 < \varepsilon \}$ for each 
$\thetaf \in \Rbb^K 
$ and $\varepsilon > 0$ by a nonnegative finite constant,
defining a constant and thus $\mu$-integrable function on $\Ycal$.
Thus, $\frac{\partial}{\partial \theta_k} \int_\Ycal \exp [ \bft (\xf_i)^\top \thetaf ] \, \dmu = \int_\Ycal \frac{\partial}{\partial \theta_k} \exp [ \bft (\xf_i)^\top \thetaf ] \, \dmu$ for every 
$\thetaf \in \Rbb^K 
$ \citep[Corollary 16.3]{bauer2001}.
Hence, for each $\thetaf \in \Rbb^K 
$, $k = 1, \ldots , K$ and $i = 1, \ldots , N$
, we have
\begin{align}
\frac{\partial}{\partial \theta_k} C_i(\thetaf)
&= - \frac{\int_\Ycal \exp [ \bft (\xf_i)^\top \thetaf ] \, (\bft (\xf_i))_{[k]} \, \dmu}{\int_\Ycal \exp [ \bft (\xf_i)^\top \thetaf ] \, \dmu },
\label{P2a:eq_first_derivative_C_i}
\end{align}
Using the same arguments as above, we show for $k, k' = 1, \ldots , K$, we have $\frac{\partial}{\partial \theta_{k'}} \int_\Ycal \exp [ \bft (\xf_i)^\top \thetaf ] \, (\bft (\xf_i))_{[k]} \, \dmu = \int_\Ycal \frac{\partial}{\partial \theta_{k'}} \exp [ \bft (\xf_i)^\top \thetaf ] \, (\bft (\xf_i))_{[k]} \, \dmu$.
Thus, the $(k, k')$-th element of $\Hf_i (\thetaf), ~i = 1, \ldots , N$, is
\begin{align}
\frac{\partial^2}{\partial \theta_k \partial \theta_{k'}} C_i(\thetaf)
&= - \frac{ \int_\Ycal \exp [ \bft (\xf_i)^\top \thetaf ] \, (\bft (\xf_i))_{[k]} \, (\bft (\xf_i))_{[k']} \, \dmu}{\int_\Ycal \exp [ \bft (\xf_i)^\top \thetaf ] \, \dmu} \notag
\\
&\hspace{0.5cm}+ 
\frac{\int_\Ycal \exp [ \bft (\xf_i)^\top \thetaf ] \cdot (\bft (\xf_i))_{[k]} \, \dmu }{\int_\Ycal \exp [ \bft (\xf_i)^\top \thetaf ] \, \dmu} \notag
\\
&\hspace{1cm} \cdot
\frac{\int_\Ycal \exp [ \bft (\xf_i)^\top \thetaf ] \cdot (\bft (\xf_i))_{[k']} \, \dmu}{\int_\Ycal \exp [ \bft (\xf_i)^\top \thetaf ] \, \dmu}
\, . \label{P2a:eq_hessian_bayes}
\end{align}
Using 
$f_{\xf_i, \thetaf} = \frac{\exp[\bft (\xf_i)^\top \thetaf]}{\int_{\Ycal} \exp[\bft (\xf_i)^\top \thetaf] \, \dmu} 
$, 
we obtain
\begin{align}
\frac{\partial^2}{\partial \theta_k \partial \theta_{k'}} C_i(\thetaf)
&=
- \int_\Ycal f_{\xf_i, \thetaf} \cdot (\bft (\xf_i))_{[k]} \, (\bft (\xf_i))_{[k']} \, \dmu \notag
\\
&\hspace{0.5cm}
+ \int_\Ycal f_{\xf_i, \thetaf} \cdot (\bft (\xf_i))_{[k]} \, \dmu \,
\int_\Ycal f_{\xf_i, \thetaf} \cdot (\bft (\xf_i))_{[k']} \, \dmu .
\label{P2a:eq_hesse_matrix_log_likelihood}
\end{align}
Then, 
with $\int_\Ycal f_{\xf_i, \thetaf} \, \dmu = 1$ for all $i = 1, \ldots , N$,
we have for $\vf \in \Rbb^K$,
\begin{align*}
\vf^\top \Hf_i (\thetaf) \vf
&= - \int_\Ycal f_{\xf_i, \thetaf} \cdot \bigl(\bft (\xf_i) ^\top \vf \bigr)^2 \, \dmu
+ \Bigl( \int_\Ycal f_{\xf_i, \thetaf} \cdot \bft (\xf_i)^\top \vf \, \dmu \Bigr)^2
\\
&= - \Bigl[ \int_\Ycal f_{\xf_i, \thetaf} \cdot \bigl(\bft (\xf_i) ^\top \vf \bigr)^2 \, \dmu
- 2 \Bigl( \int_\Ycal f_{\xf_i, \thetaf} \cdot \bft (\xf_i)^\top \vf \, \dmu \Bigr)^2 \Bigr.
\\
&\hspace{0.5cm}
\Bigl. + \Bigl( \int_\Ycal f_{\xf_i, \thetaf} \cdot \bft (\xf_i)^\top \vf \, \dmu \Bigr)^2 \cdot \int_\Ycal f_{\xf_i, \thetaf} \, \dmu\Bigr]
\\
&=- \int_\Ycal f_{\xf_i, \thetaf} \cdot \Bigl( \bft (\xf_i)^\top \vf - \int f_{\xf_i, \thetaf} \cdot \bft (\xf_i)^\top \vf \, \dmu \Bigr)^2 \, \dmu ,
\end{align*}
which gives us
\begin{align}
\vf^\top \Hf (\thetaf) \vf 
= - \sum_{i=1}^N \int_\Ycal f_{\xf_i, \thetaf} \cdot \Bigl( \bft (\xf_i)^\top \vf - \int f_{\xf_i, \thetaf} \cdot \bft (\xf_i)^\top \vf \, \dmu \Bigr)^2 \, \dmu .
\label{P2a:eq_quatratic_form_Hessematrix}
\end{align}
Since $f_{\xf_i, \thetaf} > 0$ on $\Ycal$ for all $i= 1, \ldots , N$, all summands are non-negative and thus $\vf^\top \Hf (\thetaf) \vf \leq 0$.
Now, assume $\vf^\top \Hf (\thetaf) \vf = 0$, which then implies $\bft (\xf_i)^\top \vf = \int f_{\xf_i, \thetaf} \cdot \bft (\xf_i)^\top \vf \, \dmu$ for $i = 1, \ldots , N$.
This means, $\bft (\xf_i)^\top \vf$ is constant over $\Ycal$ for every $i = 1, \ldots , N$ and 
by Lemma~\ref{P2a:lem_scalar_product_constant_implies_0}, we obtain $\vf = \mathbf{0}$.
Hence, we have $\vf^\top \Hf (\thetaf) \vf < 0$ for all $\vf \in \Rbb^K \setminus \{ \mathbf{0} \}$, i.e., $\Hf (\thetaf)$ is negative definite, which implies it is invertible.
As $\Hf_i (\thetaf) + D^2 \pen (\thetaf)$ is constant regarding $\Ebb$ (compare \eqref{P2a:eq_hessian_bayes}), we have $\Ebb (\Hf_{\mathrm{pen}} (\thetaf)) = \Hf_{\mathrm{pen}}(\thetaf)$.
The results for the penalized Fisher information follow immediately with $\Ff_{\mathrm{pen}}(\thetaf) = - \Hf_{\mathrm{pen}}(\thetaf)$.
\item
Setting 
$C_{\Zcal}^{\text{mn} \, (s)}(\thetaf) := - \log \sum_{g'=1}^{\Gamma^{(l)}} \Delta_{g'}^{(l)}\exp [ \bft(\xf^{(l)}) (u_{g'}^{(l)}) ^\top \thetaf ]$ 
for $l = 1, \ldots , L$, we have
$
\ell^{\text{mn}}_{\Zcal} (\thetaf) 
=
\sum_{l=1}^L \sum_{g=1}^{\Gamma^{(l)}} n_g^{(l)} [ \bft(\xf^{(l)}) (u_g^{(l)}) ^\top \thetaf 
- C_{\Zcal}^{\text{mn} \, (s)}(\thetaf) ]
$.
Then,
we obtain
$\Hf^{\text{mn}}_{\Zcal} (\thetaf)
= \sum_{l=1}^L \sum_{g=1}^{\Gamma^{(l)}} n_g^{(l)} \Hf_{\Zcal}^{\text{mn} \, (s)}(\thetaf) 
$, where $\Hf_{\Zcal}^{\text{mn} \, (s)}(\thetaf)$ denotes the Hessian matrix of $C_{\Zcal}^{\text{mn} \, (s)}(\thetaf)$.
The $(k, k')$-th element of $\Hf_{\Zcal}^{\text{mn} \, (s)}(\thetaf)$ is
\begin{align}
& ~\frac{\partial^2}{\partial \theta_k \partial \theta_{k'}} C_{\Zcal}^{\text{mn} \, (s)}(\thetaf) \notag
\\
= &~- \frac{ \sum_{g'=1}^{\Gamma^{(l)}} \bigl[ \Delta_{g'}^{(l)}\exp [ \bft(\xf^{(l)}) (u_{g'}^{(l)}) ^\top \thetaf ] \, (\bft(\xf^{(l)})(u_{g'}^{(l)}))_{[k]} 
\, (\bft (\xf^{(l)})(u_{g'}^{(l)}))_{[k']} \bigr]}{\sum_{g'=1}^{\Gamma^{(l)}} \Delta_{g'}^{(l)}\exp [ \bft(\xf^{(l)}) (u_{g'}^{(l)}) ^\top \thetaf ] } \notag
\\
&~+ \frac{\sum_{g'=1}^{\Gamma^{(l)}} \bigl[ \Delta_{g'}^{(l)}\exp [ \bft(\xf^{(l)}) (u_{g'}^{(l)}) ^\top \thetaf ] \, (\bft(\xf^{(l)})(u_{g'}^{(l)}))_{[k]} \bigr]}{\sum_{g'=1}^{\Gamma^{(l)}} \Delta_{g'}^{(l)}\exp [ \bft(\xf^{(l)}) (u_{g'}^{(l)}) ^\top \thetaf ]} \notag
\\
&~ \cdot \frac{\sum_{g'=1}^{\Gamma^{(l)}} \bigl[ \Delta_{g'}^{(l)}\exp [ \bft(\xf^{(l)}) (u_{g'}^{(l)}) ^\top \thetaf ] \, (\bft(\xf^{(l)})(u_{g'}^{(l)}))_{[k']} \bigr] \,}{\sum_{g'=1}^{\Gamma^{(l)}} \Delta_{g'}^{(l)}\exp [ \bft(\xf^{(l)}) (u_{g'}^{(l)}) ^\top \thetaf ]} \label{P2a:eq_hessian_multinomial}
. 
\end{align}
Using 
$p_{g'}^{(l)} (\thetaf)
= \frac{\Delta_{g'}^{(l)}\exp [ \bft(\xf^{(l)})(u_{g'}^{(l)})^\top \thetaf]}{\sum_{g=1}^{\Gamma^{(l)}} \Delta_g^{(l)}\exp [ \bft(\xf^{(l)})(u_g^{(l)})^\top \thetaf ]} 
$, 
we obtain,
\begin{align*}
&~ \frac{\partial^2}{\partial \theta_k \partial \theta_{k'}} C_{\Zcal}^{\text{mn} \, (s)}(\thetaf)
\\
=
&~- \sum_{g'=1}^{\Gamma^{(l)}} p_{g'}^{(l)}(\thetaf) \, (\bft(\xf^{(l)})(u_{g'}^{(l)}))_{[k]} \, (\bft (\xf^{(l)})(u_{g'}^{(l)}))_{[k']}
\\
&~
+ \sum_{g'=1}^{\Gamma^{(l)}} \Bigl[ p_{g'}^{(l)}(\thetaf) \, (\bft(\xf^{(l)})(u_{g'}^{(l)}))_{[k]} \Bigr] \, 
\sum_{g'=1}^{\Gamma^{(l)}} \Bigl[ p_{g'}^{(l)}(\thetaf) \, (\bft(\xf^{(l)})(u_{g'}^{(l)}))_{[k']} \Bigr] ,
\end{align*}
Then, 
with $\sum_{g'=1}^{\Gamma^{(l)}} p_{g'}^{(l)} (\thetaf) = 1$ for all $l = 1, \ldots , L$, 
we have for $\vf \in \Rbb^K$
\begin{align*}
\vf^\top \Hf_{\Zcal}^{\text{mn} \, (s)}(\thetaf) \vf
= 
- \sum_{g'=1}^{\Gamma^{(l)}} p_{g'}^{(l)} (\thetaf) \, \Bigl( \bft(\xf^{(l)})(u_{g'}^{(l)})^\top \vf - \sum_{g''=1}^{\Gamma^{(l)}} p_{g''}^{(l)} (\thetaf) \, \bft(\xf^{(l)})(u_{g''}^{(l)})^\top \vf \Bigr)^2 ,
\end{align*}
which gives us
\begin{align*}
&\vf^\top \Hf^{\text{mn}}_{\Zcal} (\thetaf) \vf 
= \sum_{l=1}^L \sum_{g=1}^{\Gamma^{(l)}} n_g^{(l)} \vf^\top \Hf_{\Zcal}^{\text{mn} \, (s)}(\thetaf) \vf \notag
\\
= &- \sum_{l=1}^L \sum_{g=1}^{\Gamma^{(l)}} n_g^{(l)} \sum_{g'=1}^{\Gamma^{(l)}} p_{g'}^{(l)} (\thetaf) \, \Bigl( \bft(\xf^{(l)})(u_{g'}^{(l)})^\top \vf - \sum_{g''=1}^{\Gamma^{(l)}} p_{g''}^{(l)} (\thetaf)\, \bft(\xf^{(l)})(u_{g''}^{(l)})^\top \vf \Bigr)^2 .
\end{align*}
Since $n_g^{(l)} \geq 0$ and $p_{g}^{(l)} (\thetaf) > 0$ for all $l = 1, \ldots , L ,~ g, k = 1, \ldots , \Gamma^{(l)}$, all summands are non-negative and thus $\vf^\top \Hf^{\text{mn}}_{\Zcal} (\thetaf) \vf \leq 0$.
Assuming $\vf^\top \Hf^{\text{mn}}_{\Zcal} (\thetaf) \vf = 0$ implies $\bft(\xf^{(l)})(u_{g'}^{(l)})^\top \vf = \sum_{g''=1}^{\Gamma^{(l)}} p_{g''}^{(l)} (\thetaf) \, \bft(\xf^{(l)})(u_{g''}^{(l)})^\top \vf$ for all $l = 1, \ldots , L, ~g' = 1, \ldots , \Gamma^{(l)}$, since for all $l = 1, \ldots , L$, there is a $g \in \{ 1, \ldots , \Gamma^{(l)} \}$ with $n_g^{(l)} > 0$.
This means, $\bft(\xf^{(l)})(u_1^{(l)})^\top \vf = \ldots = \bft(\xf^{(l)})(u_{\Gamma^{(l)}}^{(l)})^\top \vf$.
Let $\vf_{\Ycal}^{(l)} = (\bfe_{\Xcal}(\xf^{(l)})^\top \vf_{1}$, $\ldots, \bfe_{\Xcal}(\xf^{(l)})^\top \vf_{K_{\Ycal}})$ 
with $\vf_{m} = (v_{1, 1, m} , \ldots , v_{J, K_{\Xcal, \, J}, m})^\top$.
It is straightforward to show that then $\bft (\xf^{(l)})^\top \vf = \bft_{\Ycal}^\top \vf_{\Ycal}^{(l)}$. 
Thus, by assumption~\ref{P2a:assumption_appendix_domain_basis_nonsingular}, we have $\vf_{\Ycal}^{(l)} = \mathbf{0}$, 
i.e., $\bfe_{\Xcal}(\xf^{(l)})^\top \vf_{m} = 0$
for all $l = 1, \ldots , L, m = 1, \ldots , K_\Ycal$. 
This is equivalent to $\bfe_{\Xcal} (\xf_i)^\top \vf_{m} = 0$ for all $i = 1, \ldots , N$ and thus, by assumption~\ref{P2a:assumption_appendix_covariate_basis_nonsingular}, $\vf_{m} = 0$ for $ m = 1, \ldots , K_\Ycal$,
i.e., $\vf = \mathbf{0}$.
Hence, we have $\vf^\top \Hf^{\text{mn}}_{\Zcal} (\thetaf) \vf < 0$ for all $\vf \in \Rbb^K \setminus \{ \mathbf{0} \}$, i.e., $\Hf^{\text{mn}}_{\Zcal} (\thetaf)$ is negative definite and 
thus invertible. 
\\
Note that 
$\Hf_{\Zcal}^{\text{mn} \, (s)}(\thetaf) 
$ is constant 
regarding $\Ebb_\Zcal^{\mathrm{mn}}$, 
compare~\eqref{P2a:eq_hessian_multinomial}.
Replacing the observations $n_g^{(l)}$ with the underlying random variables $N_g^{(l)}$, the random vectors $(N_1^{(l)}, \ldots , N_{\Gamma^{(l)}}^{(l)})$ are conditionally multinomially distributed given $\sum_{g=1}^{\Gamma^{(l)}} N_g^{(l)} = n ^{(l)}$, $g = 1, \ldots , \Gamma^{(l)}, \, l = 1, \ldots , L$.
Thus, 
\begin{align*}
\Ebb_\Zcal^{\mathrm{mn}} (\Hf_\Zcal^{\mathrm{mn}}(\thetaf)) 
&= \Ebb_\Zcal^{\mathrm{mn}} \Bigl (\sum_{l=1}^L \sum_{g=1}^{\Gamma^{(l)}} N_g^{(l)} \Hf_{\Zcal}^{\text{mn} \, (s)}(\thetaf) \Bigr)
\\
&= \Ebb_\Zcal^{\mathrm{mn}} \Bigl( \sum_{l=1}^L n^{(l)} \Hf_{\Zcal}^{\text{mn} \, (s)}(\thetaf) \Bigr)
= \sum_{l=1}^L n^{(l)} \Hf_{\Zcal}^{\text{mn} \, (s)}(\thetaf)
= \Hf_\Zcal^{\mathrm{mn}}(\thetaf),
\end{align*}
and adding the constant $D^2 \pen (\thetaf)$ yields
$\Hf_{\Zcal, \mathrm{pen}}^{\mathrm{mn}}(\thetaf) = \Ebb_\Zcal^{\mathrm{mn}} (\Hf_{\Zcal, \mathrm{pen}}^{\mathrm{mn}}(\thetaf))$.
\\
Since $\Ff_{\Zcal, \mathrm{pen}}^{\mathrm{mn}}(\thetaf) = - \Hf_{\Zcal, \mathrm{pen}}^{\mathrm{mn}}(\thetaf)$,
the results for the penalized Fisher information follow immediately.
\qedhere
\end{enumerate}
\end{proof}

\begin{lem}\label{P2a:lem_matrix_inversion_continuous}
Denote the set of invertible $(K \times K)$-matrices with $\GL(K , \Rbb)$.
Then, $\inv : \GL(K , \Rbb) \ra \GL(K , \Rbb), \, \Ff \mapsto \Ff^{-1}$ is differentiable and thus continuous.
\end{lem}

\begin{proof}
For any matrix $\Af \in \Rbb^{K \times K}$ with 
spectral norm
$\Vert \Af \Vert_{\mathrm{op}} 
:= \sup_{\thetaf \in \Rbb^K \setminus \{ \mathbf{0} \}} \frac{ \Vert \Af \thetaf \Vert_2}{\Vert \thetaf \Vert_2}
< 1
$ 
we have the Neumann series \citep[Statement~5.7]{alt2016} 
\begin{align}
(\If + \Af)^{-1} 
&= \sum_{m=0}^{\infty} (-1)^m \Af^m
= \If - \Af + O(\Vert \Af \Vert_{\mathrm{op}}^2).
\label{P2a:eq_geometric_series}
\end{align}
For $\Ff \in \GL(K, \Rbb)$ consider some matrix $\Bf \in \Rbb^{K \times K}$, such that $\Vert \Bf \Vert_{\mathrm{op}} < \frac{1}{\Vert \Ff^{-1} \Vert_{\mathrm{op}}}$.
Then, $\Vert \Ff^{-1} \Bf \Vert_{\mathrm{op}} \leq \Vert \Ff^{-1} \Vert_{\mathrm{op}} \, \Vert \Bf \Vert_{\mathrm{op}} < 1$ and with \eqref{P2a:eq_geometric_series}, $\Ff^{-1} \Bf$ is invertible.
Thus, also $(\Ff + \Bf) = \Ff (\If + \Ff^{-1} \Bf)$ is invertible with
\begin{align*}
(\Ff + \Bf)^{-1}
= (\If + \Ff^{-1} \Bf)^{-1} \, \Ff^{-1}
\overset{\eqref{P2a:eq_geometric_series}}{=} \Ff^{-1} - \Ff^{-1} \Bf  \, \Ff^{-1} + O(\Vert \Ff^{-1} \Bf \Vert_{\mathrm{op}}^2)
\end{align*}
and the matrix inversion function $\inv$ is differentiable with the derivative given as $D \inv(\Ff) : \Rbb^{K \times K} \ra \Rbb^{K \times K}, \,  \Bf \mapsto - \Ff^{-1} \Bf  \, \Ff^{-1}$.
Thus, $\inv$ is also continuous.
\end{proof}

\begin{thm}\label{P2a:thm_generalized_laplace_approximation}
For $\kappa \geq 2$, let $s_{\mathrm{c}}: \Ycal_{\mathrm{c}} \ra \Rbb $ be a $(\kappa-1)$ times continuously differentiable piecewise polynomial consisting of finitely many polynomials of degree at most $\kappa$
and let $s : \Ycal \ra \Rbb$ such that $s(y) = s_{\mathrm{c}} (y)$ for all $y \in \Ycal_{\mathrm{c}} \setminus \Ycal_{\mathrm{d}}$. 
For 
$s^{\max} := \sup_{y \in \Ycal} s(y) 
$ consider the sets 
$\Dcal^{\max} := \Dcal^{\max} (s) := \{ d \in \{ 1, \ldots , D\} ~|~ s(t_d) = s^{\max} \}$
and
$\Ycal_{\mathrm{c}}^{\max} := \Ycal_{\mathrm{c}}^{\max} (s) := \{ y \in \Yc ~|~ s_{\mathrm{c}}(y) = s^{\max}\}$.
Then, at least one of these sets is not empty.
Furthermore, 
$\Ycm = \Ycmflat \sqcup \Ycmpoint$, where
$\Ycmflat = \Ycmflat (s) = 
\bigsqcup_{m=1}^M [a_m, b_m]
$
for $a_m < b_m, m = 1, \ldots, M \in \Nbb_0,$
and
$\Ycmpoint = \Ycmpoint (s)$ is a finite (possibly empty) set containing isolated values of $\Ycm$. 
Then, for each 
$y \in \Ycal_{\mathrm{c}, \, \bullet}^{\max} 
$, there exist $c_y > 0$ and $\alpha_y \in \Nbb$ with $\alpha_y \leq \kappa$ such that
\begin{align}
\lim_{z \ra \infty} \frac{\int_{\Ycal} \exp [ z \, s ] \, \dmu}{\exp [ z \, s^{\max} ] \, 
\Bigl( \sum_{d \in \Dcal^{\max}} w_d 
+ \lambda \bigl( \Ycal_{\mathrm{c}, \, -}^{\max} \bigr) 
+ \sum_{y \in \Ycal_{\mathrm{c}, \, \bullet}^{\max}} c_y \, z^{-\alpha_y^{-1}} \Bigr)
} = 1 .
\label{P2a:eq_laplace_approx_generalized}
\end{align}
\end{thm}

\begin{proof}
We first show that at least one of $\Yc^{\max}$ and $\Dcal^{\max}$ is not empty.
Since $\Yc$ is a compact interval and $\gc$ is continuous, there exists a $\yc \in \Yc$ with $\gc (\yc) = \max_{y \in \Yc} \gc (y) =: s_{\mathrm{c}}^{\max}$.
Furthermore, since $\Yd$ is a finite set, there exists a $\yd \in \Yd$ with $s (\yd) = \max_{y \in \Yd} s (y) =: s_{\mathrm{d}}^{\max}$.
Noting that $s^{\max} = \max \{ s_{\mathrm{c}}^{\max}, s_{\mathrm{d}}^{\max}\}$, we obtain the result. 

Now, let $\Ycmpoint \subseteq \Ycm$ be the set of all isolated values of $\Ycm$.
We show that $\Ycm = \Ycmflat \sqcup \Ycmpoint$, where $\Ycmflat$ is a disjoint union of $M \in \Nbb_0$ nontrivial closed intervals.
For this purpose, we assume $\Ycm \neq \Ycmpoint$ (otherwise $M = 0$ and there is nothing left to show) and show that any $\yc \in \Ycm \setminus \Ycmpoint$ is contained in a nontrivial closed interval.
Thus, let $\yc \in \Ycm \setminus \Ycmpoint$ and recall that $\gc$ is a piecewise polynomial.
Since $\yc$ might be a breakpoint of $\gc$, we distinguish between the polynomial to the left and the one to the right of $\yc$. 
However, if $\yc$ is not a breakpoint, both polynomials coincide.
If $\yc \in \Ycm \cap (a, b]$, there exists a polynomial
$
p_< 
$ with degree 
$\kappa_< 
\leq \kappa$ 
and $\varepsilon > 0$,
such that $\gc (y) = p_< (y)$ for all $y \in [\yc - \varepsilon , \yc]$. 
Let $\varepsilon_< = \varepsilon_< (\yc) > 0$ be maximal with this property (i.e., $\yc - \varepsilon_<$ is a breakpoint of the piecewise polynomial) and set
$\overline{\Ucal (\yc , \varepsilon_<)}_< := [\yc - \varepsilon_< , \yc]$.
Analogously, if $\yc \in \Ycm \cap [a, b)$, there exists a polynomial
$p_>$ with degree 
$\kappa_> \leq \kappa$ 
and a maximal $\varepsilon_>  = \varepsilon_> (\yc) > 0$
such that $\gc (y) = p_> (y)$ for all $y \in [\yc , \yc + \varepsilon_> ]=: \overline{\Ucal (\yc , \varepsilon_>)}_>$.
In the following, we treat both polynomials simultaneously with the subscript ``$_\lessgtr$'' either corresponding to ``$_<$'' or to ``$_>$''.
For the boundary values, only consider the respective existent case, i.e., if $\yc = a$, only consider the objects with subscript ``$_>$'', if $\yc = b$ only the objects with subscript ``$_<$''.
Since $\yc \in \Ycm \setminus \Ycmpoint$ is not isolated, it is possible to find further points belonging to $\Ycm$ in any neighborhood of $\yc$ and thus by successively diminishing the size of the neighborhoods we obtain arbitrary many values in $\Ycm$.
More precisely, let $y_{0, <} := \yc - \varepsilon_<$ and $y_{0, >} := \yc + \varepsilon_>$.
Recursively, there exist $y_{m, <} \in (y_{m - 1, <}, \yc)$ and $y_{m, >} \in (\yc, y_{m - 1, >})$ with 
$y_{m, \lessgtr} \in \Ycm$, 
i.e., with
$p_\lessgtr (y_{m, \lessgtr}) 
= \gm$, for 
$m = 1, \ldots , \kappa_\lessgtr + 1$.
These $\kappa_\lessgtr + 1$ values are thus pairwise distinct roots of the polynomial $p_\lessgtr - \gm$, which \citep[e.g., using][Lemma~11.1]{laczkovich2015} implies $p_\lessgtr (y) = \gm$ for all $y \in \overline{\Ucal (\yc , \varepsilon_\lessgtr)}_\lessgtr$.
Thus, setting $\varepsilon_< := 0$ if $\yc = a$ and $\varepsilon_> := 0$, if $\yc = b$, we obtain that if $\yc \in \Ycm$ is not isolated, it is contained in the nontrivial closed interval $[\yc - \varepsilon_<, \yc + \varepsilon_>]$ on which $\gc$ is constantly $\gm$. 
Note that this interval cannot be contained in a larger open interval, since $\overline{\Ucal (\yc , \varepsilon_\lessgtr)}_\lessgtr$ was constructed as the maximal neighborhood to the left/right of $\yc$, on which $p_\lessgtr = \gc$.
In particular, the boundaries of this interval then correspond to two adjacent breakpoints of the piecewise polynomial $\gc$.
As there are only finitely many breakpoints, 
we obtain that there exist only finitely many nontrivial closed intervals on which $\gc$ is constantly $\gm$ and thus $\Ycm = \Ycmflat \sqcup \Ycmpoint$ with $\Ycmflat = \bigsqcup_{m=1}^M [a_m, b_m]$, where $a_m < b_m, m = 1, \ldots, M \in \Nbb_0$.
Furthermore, we have $| \Ycmpoint| < \infty$ as each polynomial can take the value $\gm$ at most $\kappa$ times (or even less, if its degree is smaller than $\kappa$) and $\gc$ consists of a finite number of polynomials.

Finally, we show equation~\eqref{P2a:eq_laplace_approx_generalized}, first proving that the limit of the fraction on the left is at least $1$ and afterwards that it is at most $1$.
In both cases, we proceed similarly:
The idea is to bound the integral in the numerator of the fraction by a term that mimics the denominator.
More precisely, these bounds will differ from the denominator in the terms corresponding to $\Ycmpoint$ (and $\Ycmflat$ in case of the upper bound)
by some small value $\varepsilon > 0$, but converge to the denominator for $\varepsilon \ra 0$, 
which then yields the result.
As we ultimately consider $z \ra \infty$, we always assume $z > 0$ in the following.
Furthermore, note that 
$\int_{\Ycal} \exp [ z \, s ] \, \dmu
= \int_{\Ycal_{\mathrm{c}}} \exp [ z \, s ] \, \dlamb + 
\int_{\Ycal_{\mathrm{d}}} \exp [ z \, s ] \, \ddel $
and
$\int_{\Ycal_{\mathrm{c}}} \exp [ z \, s ] \, \dlamb
= \int_{\Ycal_{\mathrm{c}}} \exp [ z \, s_{\mathrm{c}} ] \, \dlamb$.

To show that the limit in~\eqref{P2a:eq_laplace_approx_generalized} is at least $1$,
let $\varepsilon > 0$ be small enough such that 
$\Ycmflat (\varepsilon) 
\cap \Ycmpoint = \emptyset$, where $\Ycmflat (\varepsilon) := \bigcup_{m=1}^M (a_m - \varepsilon, b_m + \varepsilon) \cap \Yc$.
We have
\begin{align}
\int_{\Ycmflat (\varepsilon)} \exp [ z \, \gc ] \, \dlamb
\geq
\int_{\Ycmflat} \exp [ z \, \gc ] \, \dlamb
= \exp [z \,  s^{\max}] \, \lambda \bigl( \Ycmflat \bigr) .
\label{P2a:eq_laplace_lower_bound_flat}
\end{align}

For $\yc \in \Ycmpoint$, consider $p_<$ and $p_>$ as above (if existent). 
If $\yc \in (a, b)$, then $\yc$ being isolated implies that $\gc$ has a strict local maximum at $\yc$, with first derivative equal to zero.
At the boundary values $a$ and $b$, this is not necessarily true.
However, if the first derivative is not zero, it then has to be negative for $\yc = a$ and positive for $\yc = b$.
Furthermore recall that $\gc (\yc) = p_< (\yc) = p_> (\yc)$ with the first $\kappa - 1$ derivatives of $p_<$ and $p_>$ coinciding as well.
With these considerations in mind, we now define the smallest non-zero derivative of the piecewise polynomials:
Since $\kappa_\lessgtr$ is the degree of $p_\lessgtr$, the $\kappa_\lessgtr$-st derivative of $p_\lessgtr$ is a non-zero constant and thus there exists a natural number $\alpha_{\yc, \, \lessgtr} \leq \kappa_\lessgtr \leq \kappa$ such that $p_\lessgtr^{(\alpha_{\yc, \, \lessgtr})}(\yc) \neq 0$ and $p_\lessgtr^{(\alpha)}(\yc) = 0$ for all $\alpha \in \{ \alpha' \in \Nbb ~|~ \alpha' < \alpha_{\yc, \, \lessgtr}\}$.
Note that this set is empty, if $\alpha_{\yc, \, \lessgtr} = 1$ (which can only occur if $\yc \in \{ a, b\}$). 
Since the first $\kappa - 1$ derivatives of $p_<$ and $p_>$ coincide in $\yc$, $\alpha_{\yc, \, >} = \alpha_{\yc, \, <} =: \alpha_{\yc}$.
Note that if $\alpha_{\yc} > 1$ (which is always fulfilled, if $\yc \in (a, b)$ and might be fulfilled, if $\yc \in \{a, b\}$), then $p_\lessgtr$ (and $\gc$) has a strict local maximum at $\yc$, 
which implies that $\alpha_{\yc}$ is even with $p_\lessgtr^{(\alpha_{\yc})}(\yc) < 0$ \citep[Theorem~12.62]{laczkovich2015}.
Thus, one of the following applies to $\yc$:
\begin{enumerate}[label=(\roman*)]
\item\label{P2a:isolated_maximum_interior}
$\yc \in (a,b)$ with $\alpha_{\yc}$ even and $p_\lessgtr^{(\alpha_{\yc})}(\yc) < 0$,
\item\label{P2a:isolated_maximum_a_even}
$\yc = a$ with $\alpha_{\yc}$ even and $p_>^{(\alpha_{\yc})}(\yc) < 0$ ($p_<$ not existent),
\item\label{P2a:isolated_maximum_a_1}
$\yc = a$ with $\alpha_{\yc} = 1$ and $p_>^{(\alpha_{\yc})}(\yc) = p_>'(a) < 0$ ($p_<$ not existent),
\item\label{P2a:isolated_maximum_b_even}
$\yc = b$ with $\alpha_{\yc}$ even and $p_<^{(\alpha_{\yc})}(\yc) < 0$ ($p_>$ not existent),
\item\label{P2a:isolated_maximum_b_1}
$\yc = b$ with $\alpha_{\yc} = 1$ and $p_<^{(\alpha_{\yc})}(\yc) = p_<'(b) > 0$ ($p_>$ not existent).
\end{enumerate}
%
As before, we only proceed with the existent case regarding the subscripts ``$_>$'' and  ``$_<$'', i.e., if $\yc = a$ (cases \ref{P2a:isolated_maximum_a_even} and~\ref{P2a:isolated_maximum_a_1}), only consider the objects with subscript ``$_>$'', if $\yc = b$ (cases \ref{P2a:isolated_maximum_b_even} and~\ref{P2a:isolated_maximum_b_1}), only the objects with subscript ``$_<$''.
For $\delta > 0$, set 
$\Ucal_> (\yc , \delta) := [y_{\mathrm{c}}, y_{\mathrm{c}} + \delta)$ 
and 
$\Ucal_< (\yc , \delta) := (y_{\mathrm{c}} - \delta, y_{\mathrm{c}}]$. 
As $p_\lessgtr^{(\alpha_{\yc})}$ is continuous, for any $\yc \in \Ycmpoint$, 
there exists $\deltat_{\yc, \, \lessgtr} > 0$ such that
\begin{align}
p_\lessgtr^{(\alpha_{\yc})}(\yt) \geq p_\lessgtr^{(\alpha_{\yc})}(\yc) - \varepsilon 
&& \text{and} &&
p_\lessgtr^{(\alpha_{\yc})} (\yt) \leq p_\lessgtr^{(\alpha_{\yc})} (\yc) + \varepsilon
\label{P2a:eq_laplace_lower_bound_derivative}
\end{align}
for all $\yt \in \Ucal_\lessgtr (\yc , \deltat_{\yc, \, \lessgtr})$.
Furthermore, there exists $\delta_{\yc} > 0$ with
$\Ucal_\lessgtr (\yc , \delta_{\yc}) 
\subset \Yc$,
$\Ucal_\lessgtr (\yc , \delta_{\yc}) 
\cap 
\Ycmflat (\varepsilon) = \emptyset$,
$\Ucal_\lessgtr (\yc , \delta_{\yc}) \subseteq \Ucal_\lessgtr (\yc , \deltat_{\yc, \lessgtr})$,
and $\delta_{\yc} \leq \max \{ \varepsilon_< (\yc), \varepsilon_> (\yc) \}$.
By Taylor's formula \citep[e.g.,][Theorem 1.97]{laczkovich2017}, for any 
$y \in \Ucal_\lessgtr (\yc , \delta_{\yc})$, we have
\begin{align}
p_\lessgtr(y) = p_\lessgtr(y_{\mathrm{c}}) + 
\frac{1}{\alpha_{\yc}!} \, p_\lessgtr^{(\alpha_{\yc})}(\yt ) \, (y-\yc)^{\alpha_{\yc}}
\label{P2a:eq_laplace_lower_bound_taylor_expansion}
\end{align}
for a $\yt 
$ between $y$ and $\yc$, 
which implies $\yt \in \Ucal_\lessgtr (\yc , \deltat_{\yc, \, \lessgtr})$. 
First consider cases~\ref{P2a:isolated_maximum_interior}-\ref{P2a:isolated_maximum_b_even}.
Then, $(y-\yc)^{\alpha_{\yc}} \geq 0$ (in cases~\ref{P2a:isolated_maximum_interior}, \ref{P2a:isolated_maximum_a_even}, and \ref{P2a:isolated_maximum_b_even} this follows from $\alpha_{\yc}$ being even, while in case~\ref{P2a:isolated_maximum_a_1}, only the case with subscript $_>$ is considered, with $y - \yc \geq 0$ for all $y \in \Ucal_> (\yc , \deltat_{\yc, \, >})$)
and with the left inequality of~\eqref{P2a:eq_laplace_lower_bound_derivative}, we obtain 
\begin{align*}
p_\lessgtr(y) 
\geq p_\lessgtr(\yc) 
+ \frac1{\alpha_{\yc}!} ( p_\lessgtr^{(\alpha_{\yc})}(\yc) - \varepsilon ) (y-\yc)^{\alpha_{\yc}}
\end{align*}
for any $y \in \Ucal_\lessgtr (\yc , \delta_{\yc})$. 
Then,
\begin{align}
& \int_{\Ucal_\lessgtr (\yc , \delta_{\yc})} \exp [ z \, p_\lessgtr ] \, \dlamb \notag
\\ 
\geq ~& 
\exp \left( z \, p_\lessgtr (\yc) \right) \int_{\Ucal_\lessgtr (\yc , \delta_{\yc})} \exp \Bigl( 
z \, \frac1{\alpha_{\yc}!} \left( p_\lessgtr^{(\alpha_{\yc})}(\yc) - \varepsilon \right) (y-\yc)^{\alpha_{\yc}}
\Bigr) \, \dlamb (y) 
.
\label{P2a:eq_laplace_lower_bound_integral}
\end{align}
Note that as we consider cases~\ref{P2a:isolated_maximum_interior}-\ref{P2a:isolated_maximum_b_even}, 
we have $p_\lessgtr^{(\alpha_{\yc})}(\yc) < 0$,
i.e., 
$
- p_\lessgtr^{(\alpha_{\yc})}(\yc) 
> 0$.
Using
$\exp \left( z \, p_\lessgtr (\yc) \right) 
=
\exp \left( z \,  s^{\max} \right)$,
substituting
$\upsilon = [z \, \frac1{\alpha_{\yc}!} \, (- p_\lessgtr^{(\alpha_{\yc})}(\yc) + \varepsilon )]^{\alpha_{\yc}^{-1}} (y-y_{\mathrm{c}}) 
$, 
and setting
\begin{align*}
\Ucal_> (\yc , z, \varepsilon) 
&:=
\Bigl[ 0 \, , \, \delta_{y_{\mathrm{c}}} \, \bigl[z \, \frac1{\alpha_{\yc}!} \, (| p_>^{(\alpha_{\yc})}(\yc)| + \varepsilon ) \bigr]^{\alpha_{\yc}^{-1}} \Bigr) 
\\
\Ucal_< (\yc , z, \varepsilon) 
&:= 
\Bigl( - \delta_{y_{\mathrm{c}}} \, \bigl[z \, \frac1{\alpha_{\yc}!} \, (| p_<^{(\alpha_{\yc})}(\yc)| + \varepsilon )\bigr]^{\alpha_{\yc}^{-1}} \, , \, 0 \Bigr]
\end{align*}
in the right side of \eqref{P2a:eq_laplace_lower_bound_integral} yields
\begin{align*}
\int_{\Ucal_\lessgtr (\yc , \delta_{\yc})} \exp [ z \, p_\lessgtr ] \, \dlamb 
\geq ~& 
\frac{\exp \left( z \,  s^{\max} \right)}{\left(z \, \frac1{\alpha_{\yc}!} \, \left(- p_>^{(\alpha_{\yc})}(\yc) + \varepsilon  \right)\right)^{\alpha_{\yc}^{-1}}} 
\int_{\Ucal_\lessgtr (\yc , z, \varepsilon)} 
\exp \left( -\upsilon^{\alpha_{\yc}} \right) \, \dlamb (\upsilon) .
\end{align*} 

Defining
$
I_\lessgtr (\yc , z, \varepsilon)
:=
\int_{\Ucal_\lessgtr (\yc , z, \varepsilon)} 
\exp \left( -\upsilon^{\alpha_{\yc}} \right) \, \dlamb (\upsilon) 
$ 
and recalling that 
$\gc = p_\lessgtr$ on $\Ucal_\lessgtr (\yc , \delta_{\yc})$,
we obtain
\begin{align}
& \int_{\Ucal_\lessgtr (\yc , \delta_{\yc})} \exp [ z \, \gc ] \, \dlamb 
\geq
\frac{\exp ( z \,  s^{\max} ) \, I_\lessgtr (\yc , z, \varepsilon) \, z^{-\alpha_{\yc}^{-1}}}{\bigl(\frac1{\alpha_{\yc}!} \, 
\left( | p_>^{(\alpha_{\yc})}(\yc)| + \varepsilon  \right) \bigr)^{\alpha_{\yc}^{-1}}} 
.
\label{P2a:eq_laplace_lower_bound_points}
\end{align}
In case~\ref{P2a:isolated_maximum_b_1}, where $\yc = b$ with $\alpha_{\yc} = 1$ and 
$p_<^{(\alpha_{\yc})}(\yc) 
= p_<'(b) > 0$ (and $p_>$ not existent), we have to make small adjustments to obtain~\eqref{P2a:eq_laplace_lower_bound_points}.
As $(y-\yc)^{\alpha_{\yc}} = (y-b) \leq 0$ for all $y \in \Ucal_< (b , \deltat_{b, \, <})$,
we insert the right inequality of~\eqref{P2a:eq_laplace_lower_bound_derivative} into~\eqref{P2a:eq_laplace_lower_bound_taylor_expansion} and obtain 
$ 
p_<(y) 
\geq p_<(b) 
+ ( p_<'(b) + \varepsilon ) (y-b)
$ 
for any $y \in \Ucal_< (b , \delta_{b})$.
Then,
\begin{align}
& \int_{\Ucal_< (b , \delta_{b})} \exp [ z \, p_< ] \, \dlamb \notag
\geq 
\exp \left( z \, p_< (b) \right) \int_{\Ucal_< (b , \delta_{b})} \exp \Bigl( 
z \,  \left( p_<' (b) + \varepsilon \right) (y-b)
\Bigr) \, \dlamb (y) 
.
\label{P2a:eq_laplace_lower_bound_integral}
\end{align}
Substituting 
$\upsilon = z \, (p_<'(b) + \varepsilon) (y - b)$ 
and setting
$I_< (b , z, \varepsilon)
:=
\int_{\Ucal_< (b , z, \varepsilon)} 
\exp \left( \upsilon \right) \, \dlamb (\upsilon) 
$
yields~\eqref{P2a:eq_laplace_lower_bound_points}.
\\
To define the constants $c_{\yc}$ for $\yc \in \Ycmpoint$ in the denominator of the fraction on the left of~\eqref{P2a:eq_laplace_approx_generalized},
we first set $I_< (a, z, \varepsilon) := 0 =: I_> (b, z, \varepsilon)$, $p_<^{(\alpha_{\yc})}(a) := p_>^{(\alpha_{\yc})}(a)$, and $p_>^{(\alpha_{\yc})}(b) := p_<^{(\alpha_{\yc})}(b)$
for the non-existent polynomials in cases~\ref{P2a:isolated_maximum_a_even}-\ref{P2a:isolated_maximum_b_1}.
Furthermore, set $I_< (\yc)
:= \lim_{z \ra \infty} I_< (\yc , z, \varepsilon)$
and 
$I_> (\yc)
:= \lim_{z \ra \infty} I_> (\yc , z, \varepsilon)$.
Then,
\begin{align*}
I_> (\yc)
&=
\begin{cases}
\int_0^{\infty} \exp ( -\upsilon^{\alpha_{\yc}} ) \, \dlamb (\upsilon) 
= \frac1{\alpha_{\yc}} \, \Gamma ( \frac1{\alpha_{\yc}} ) 
&\hspace{3.9cm} \text{\ref{P2a:isolated_maximum_interior},~\ref{P2a:isolated_maximum_a_even} or~\ref{P2a:isolated_maximum_a_1}}
\\
0
&\hspace{3.9cm} \text{\ref{P2a:isolated_maximum_b_even} or~\ref{P2a:isolated_maximum_b_1}}
\end{cases}
\\
I_< (\yc)
&=
\begin{cases}
\int_{-\infty}^0 \exp ( -\upsilon^{\alpha_{\yc}} ) \, \dlamb (\upsilon) = 
\int_0^{\infty} \exp ( -\upsilon^{\alpha_{\yc}} ) \, \dlamb (\upsilon) 
= \frac1{\alpha_{\yc}} \, \Gamma ( \frac1{\alpha_{\yc}} ) 
& \text{\ref{P2a:isolated_maximum_interior} or ~\ref{P2a:isolated_maximum_b_even}}
\\
0
& \text{\ref{P2a:isolated_maximum_a_even} or~\ref{P2a:isolated_maximum_a_1}} \, ,
\\
\int_{-\infty}^0 \exp ( \upsilon ) \, \dlamb (\upsilon) 
= 1 
& \text{\ref{P2a:isolated_maximum_b_1}}
\end{cases}
\end{align*}
with  $\Gamma(x)$ denoting the Gamma function at $x > 0$, where we used \citet[Equation~3.326~2.]{gradshteyn2007} and that $\alpha_{\yc}$ is even.\footnote{Note that $\Gamma ( 1 ) = 1$, i.e., $I_> (\yc) = I_> (a) = 1$ in case~\ref{P2a:isolated_maximum_a_1}.}
Then, for any $\yc \in \Ycmpoint$, both limits $I_> (\yc)$ and $I_< (\yc)$ are nonnegative and finite with at least one limit being positive.
Thus,
\[ 
c_{\yc}
:=
\frac{I_> (\yc)}{[\frac1{\alpha_{\yc}!} \, | p_>^{(\alpha_{\yc})}(\yc)|]^{\alpha_{\yc}^{-1}}} + \frac{I_< (\yc)}{[\frac1{\alpha_{\yc}!} \, | p_<^{(\alpha_{\yc})}(\yc)|]^{\alpha_{\yc}^{-1}}}
\]
is a positive constant for any $\yc \in \Ycmpoint$.\footnote{Note that if $p_> = p_<$ (i.e., if $\yc \in (a, b)$ is not a breackpoint of the piecewise polynomial $\gc$) and if $\alpha_{\yc} < \kappa$, both summands in the definition of $c_{\yc}$ are equal.}
To finally bound the numerator of the fraction on the left of~\eqref{P2a:eq_laplace_approx_generalized} from below, set
$\Yc^{\max} (\varepsilon) := \Ycmflat (\varepsilon) \cup \Ycmpoint (\varepsilon)
$,
where
$\Ycmpoint (\varepsilon) := \bigcup_{\yc \in \Ycmpoint} \left( \Ucal_< (\yc , \delta_{\yc}) \cup \Ucal_> (\yc , \delta_{\yc}) \right)$. 
Then,
\begin{align*}
\int_{\Ycal} \exp [ z \, s ] \, \dmu 
&= \int_{\Ycal_{\mathrm{d}}} \exp [ z \, s] \, \ddel  + 
\int_{\Ycal_{\mathrm{c}}} \exp [ z \, s] \, \dlamb 
\\
&= 
\sum_{d=1}^D w_d \exp [ z \, s(t_d) ]
+ \int_{(\Yc^{\max} (\varepsilon))^\complement} \exp [ z \, s ] \, \dlamb
\\
&\hspace{0.4cm} 
+ \int_{\Ycmflat (\varepsilon)} \exp [ z \, s ] \, \dlamb
+ \int_{\Ycmpoint (\varepsilon)} \exp [ z \, s ] \, \dlamb  
\\
&\geq 
\exp [ z \, s^{\max} ] \, \biggl(
\sum_{d \in \Dcal^{\max}} w_d 
+ \lambda ( \Ycmflat ) 
\\
&\hspace{0.4cm} 
+ \sum_{\yc \in \Ycmpoint} 
\Bigl(
\frac{z^{-\alpha_{\yc}^{-1}} \, I_< (\yc , z, \varepsilon)}{[\frac1{\alpha_{\yc}!} \, (| p_<^{(\alpha_{\yc})}(\yc)| + \varepsilon ) ]^{\alpha_{\yc}^{-1}}}
+
\frac{z^{-\alpha_{\yc}^{-1}} \, I_> (\yc , z, \varepsilon)}{[\frac1{\alpha_{\yc}!} \, (| p_>^{(\alpha_{\yc})}(\yc)| + \varepsilon ) ]^{\alpha_{\yc}^{-1}}}
\Bigr) \biggr),
\end{align*}
where we used~\eqref{P2a:eq_laplace_lower_bound_flat} and~\eqref{P2a:eq_laplace_lower_bound_points} to obtain the last inequality.
Dividing the inequality by 
$\exp ( z \,  s^{\max} ) 
\, [ \,\sum_{d \in \Dcal^{\max}} w_d 
+ \lambda ( \Ycmflat ) 
+ \sum_{y \in \Ycmpoint} c_y \, z^{-\alpha_y^{-1}} ]$,
which is nonzero as at least one of the sets $\Dcal^{\max}$, $\Ycmflat$, and $\Ycmpoint$ is nonempty and thus at least one of the summands in the second factor is positive (recalling $c_y > 0$ for any $y \in \Ycmpoint$),
yields
\begin{align*}
&~ 
\frac{\int_{\Ycal} \exp [ z \, s ] \, \dmu}{\exp [ z \, s^{\max} ] \, 
\bigl( \sum_{d \in \Dcal^{\max}} w_d 
+ \lambda ( \Ycal_{\mathrm{c}, \, -}^{\max} ) 
+ \sum_{y \in \Ycal_{\mathrm{c}, \, \bullet}^{\max}} c_y \, z^{-\alpha_y^{-1}} \bigr)}
\\
\geq &~
\frac{\sum_{d \in \Dcal^{\max}} w_d 
+ \lambda ( \Ycmflat ) 
+ \sum_{\yc \in \Ycmpoint}
\Bigl( \frac{z^{-\alpha_{\yc}^{-1}} \, I_<(\yc , z, \varepsilon)}{[\frac1{\alpha_{\yc}!} \, ( | p_<^{(\alpha_{\yc})}(\yc)| + \varepsilon )]^{\alpha_{\yc}^{-1}}}
+ \frac{z^{-\alpha_{\yc}^{-1}} \, I_>(\yc , z, \varepsilon)}{[\frac1{\alpha_{\yc}!} \, ( | p_>^{(\alpha_{\yc})}(\yc)| + \varepsilon )]^{\alpha_{\yc}^{-1}}} \Bigr)}
{\sum_{d \in \Dcal^{\max}} w_d 
+ \lambda ( \Ycal_{\mathrm{c}, \, -}^{\max} ) 
+ \sum_{\yc \in \Ycal_{\mathrm{c}, \, \bullet}^{\max}} 
\Bigl( \frac{z^{-\alpha_y^{-1}} \, I_<(\yc) }{[\frac1{\alpha_{\yc}!} \, | p_<^{(\alpha_{\yc})}(\yc)|]^{\alpha_{\yc}^{-1}}}
+ \frac{z^{-\alpha_y^{-1}} \, I_>(\yc) }{[\frac1{\alpha_{\yc}!} \, | p_>^{(\alpha_{\yc})}(\yc)|]^{\alpha_{\yc}^{-1}}} \Bigr)}
\,.
\end{align*}
Consider the limit for $z \ra \infty$.
If $\Dcal^{\max} \neq \emptyset$ or $\Ycmflat \neq \emptyset$, the limit in the second line can be computed separately in the numerator and the denominator with the last terms converging to zero, respectively, and thus numerator and denominator are equal.
Otherwise, we multiply numerator and denominator (then both consisting of only their last term) by $z^{\alpha^{-1}}$ for $\alpha = \max_{\yc \in \Ycmpoint} \alpha_{\yc}$ (which exists as $|\Ycmpoint| < \infty$) and again can compute the limit separately for the numerator and the denominator.
Then, all summands with $\alpha_{\yc} < \alpha$ converge to zero and the remaining terms in the numerator only include $z$ in $I_\lessgtr(\yc , z, \varepsilon)$, which converges to $I_\lessgtr(\yc)$ for $z \ra \infty$.
Then, the numerator converges to the denominator for $\varepsilon \ra 0$. 
Thus, in all cases, we have 
$ 
\lim_{z \ra \infty} \frac{\int_{\Ycal} \exp [ z \, s ] \, \dmu}{\exp [ z \, s^{\max} ] \, 
[ \, \sum_{d \in \Dcal^{\max}} w_d 
+ \lambda ( \Ycal_{\mathrm{c}, \, -}^{\max} ) 
+ \sum_{y \in \Ycal_{\mathrm{c}, \, \bullet}^{\max}} c_y \, z^{-\alpha_y^{-1}} ]}
\geq 1.
$ 
%
\\[0.2cm]
Now, we show that the limit is at most $1$, proceeding similarly.
Let $\varepsilon > 0$ be small enough such that 
$\Ycmflat (\varepsilon) \cap \Ycmpoint = \emptyset$,
$p_\lessgtr^{(\alpha_{\yc})} (\yc) + \varepsilon < 0$ for all $\yc \in \Ycmpoint$ with $p_\lessgtr^{(\alpha_{\yc})} (\yc) < 0$ (cases~\ref{P2a:isolated_maximum_interior}-\ref{P2a:isolated_maximum_b_even}),
and $p_<' (b) - \varepsilon > 0$ if $b \in \Ycmpoint$ with $p_<' (b) > 0$ (case~\ref{P2a:isolated_maximum_b_1}).
Note that such an $\varepsilon$ exists as $| \Ycmpoint | < \infty$ and recall for $\yc \in \Ycmpoint$ that if $\yc = a$, we only consider the objects with subscript ``$_>$'', if $\yc = b$, we only consider the objects with subscript ``$_<$''.
%
We have
\begin{align}
\int_{\Ycmflat (\varepsilon)} \exp [ z \, \gc ] \, \dlamb
\leq
\exp [z \,  s^{\max}] \, \bigl( \lambda ( \Ycmflat ) + 2 \, M \, \varepsilon \bigr)
.
\label{P2a:eq_laplace_upper_bound_flat}
\end{align}
Furthermore, for any $\yc \in \Ycmpoint$, 
as 
in the proof of the lower bound above,
there exists $\delta_{\yc} > 0$ such that 
the inequalities~\eqref{P2a:eq_laplace_lower_bound_derivative} hold
for all $\yt \in \Ucal_\lessgtr (\yc , \delta_{\yc})$,
$\Ucal_\lessgtr (\yc , \delta_{\yc}) \subset \Yc$,
$\Ucal_\lessgtr (\yc , \delta_{\yc}) \cap \Ycmflat (\varepsilon) = \emptyset$,
and $\delta_{\yc} \leq \max \{ \varepsilon_< (\yc), \varepsilon_> (\yc) \}$. 
Consider $y \in \Ucal_\lessgtr (\yc , \delta_{\yc})$ and recall that $(y-\yc)^{\alpha_{\yc}} \geq 0$ in cases~\ref{P2a:isolated_maximum_interior}-\ref{P2a:isolated_maximum_b_even} and $(y-\yc)^{\alpha_{\yc}} = y-b \leq 0$ in case~\ref{P2a:isolated_maximum_b_1}. 
Thus, inserting the right inequality of~\eqref{P2a:eq_laplace_lower_bound_derivative} into equality~\eqref{P2a:eq_laplace_lower_bound_taylor_expansion} 
in cases~\ref{P2a:isolated_maximum_interior}-\ref{P2a:isolated_maximum_b_even}
and inserting the left inequality of~\eqref{P2a:eq_laplace_lower_bound_derivative} into~\eqref{P2a:eq_laplace_lower_bound_taylor_expansion} in case~\ref{P2a:isolated_maximum_b_1},
we obtain
\begin{align*}
p_\lessgtr(y) &\leq 
p_\lessgtr(\yc) 
+ \frac1{\alpha_{\yc}!} ( p_\lessgtr^{(\alpha_{\yc})}(\yc) + \varepsilon ) (y-\yc)^{\alpha_{\yc}}
&& \text{in cases~\ref{P2a:isolated_maximum_interior}-\ref{P2a:isolated_maximum_b_even}},
\\
p_<(y) &\leq 
p_<(b) 
+ ( p_<'(b) - \varepsilon ) (y-b)
&& \text{in case~\ref{P2a:isolated_maximum_b_1}}.
\end{align*}

Thus,
%
\begin{align*}
&\int_{\Ucal_\lessgtr (\yc , \delta_{\yc})} \exp [ z \, \gc 
] \, \dlamb 
= \int_{\Ucal_\lessgtr (\yc , \delta_{\yc})} \exp [ z \, p_\lessgtr 
] \, \dlamb 
\notag
\\
\leq ~ & 
\begin{cases}
\exp [ z \, p_\lessgtr (\yc) ] \, \int_{\Ucal_\lessgtr (\yc , \delta_{\yc})} \exp \bigl( z \, \frac1{\alpha_{\yc}!} \left( p_\lessgtr^{(\alpha_{\yc})}(\yc) + \varepsilon \right) (y-\yc)^{\alpha_{\yc}} \bigr) \, \dlamb (y)
& \text{\ref{P2a:isolated_maximum_interior}-\ref{P2a:isolated_maximum_b_even}},
\\[0.2cm]
\exp [ z \, p_< (b) ] \, \int_{\Ucal_< (b , \delta_{b})} \exp \bigl( z \, \left( p_<'(b) - \varepsilon \right) (y-b) \bigr) \, \dlamb (y)
& \text{\ref{P2a:isolated_maximum_b_1}}
\end{cases}
\notag
\\
= ~&
\frac{\exp \left( z \,  s^{\max} \right)}{\left(z \, \frac1{\alpha_{\yc}!} \, (| p_\lessgtr^{(\alpha_{\yc})}(\yc) | - \varepsilon ) \right)^{\alpha_{\yc}^{-1}}} \, I_\lessgtr (\yc , z, \varepsilon)
\notag
\, ,
\end{align*}
where we substitute as in the proof of the lower bound, but changing the sign before~$\varepsilon$ from $+$ to $-$ to obtain the last equality.
Note that 
this also changes the sign before~$\varepsilon$ from $+$ to $-$ in the definition of the sets $\Ucal_\lessgtr (\yc, z, \varepsilon)$, on which the integrals $I_\lessgtr (\yc , z, \varepsilon)$ are defined.
Since $\Ucal_< (\yc, z, \varepsilon) \subset (-\infty , 0]$ and $\Ucal_> (\yc, z, \varepsilon) \subset [0, \infty)$ for all $\yc \in \Ycmpoint$, all $z > 0$, and all $\varepsilon > 0$, we have $I_\lessgtr (\yc , z, \varepsilon) \leq I_\lessgtr (\yc)$ for all $z > 0$ and all $\varepsilon > 0$ and thus
\begin{align}
\int_{\Ucal_\lessgtr (\yc , \delta_{\yc})} \exp [ z \, \gc 
] \, \dlamb 
\leq
\frac{\exp \left( z \,  s^{\max} \right) \, I_\lessgtr (\yc) \, z^{-\alpha_{\yc}^{-1}}}{\left(\frac1{\alpha_{\yc}!} \, ( | p_\lessgtr^{(\alpha_{\yc})}(\yc) | - \varepsilon \right)^{\alpha_{\yc}^{-1}}}
\, .
\label{P2a:eq_laplace_upper_bound_points}
\end{align}

Per construction, there exists a $\xi > 0$ such that $\gc(y) < s^{\max} - \xi$ for all $y \in \Ycm (\varepsilon )^\complement$ and $s(t_d) < s^{\max} - \xi$ for all $d \in (\Dcal^{\max})^\complement$. 
Then, by~\eqref{P2a:eq_laplace_upper_bound_flat} and~\eqref{P2a:eq_laplace_upper_bound_points}, we have
\begin{align*}
& \int_{\Ycal} \exp [ z \, s ] \, \dmu
= \int_{\Ycal_{\mathrm{d}}} \exp [ z \, s ] \, \ddel  + 
\int_{\Ycal_{\mathrm{c}}} \exp [ z \, s ] \, \dlamb
\\
= &~
\sum_{d \in \Dcal^{\max}} w_d \exp [ z \, s(t_d) ]
+ \int_{\Ycmflat (\varepsilon)} \exp [ z \, s ] \, \dlamb 
+ \int_{\Ycmpoint (\varepsilon)} \exp [ z \, s ] \, \dlamb  
\\
&~+ \sum_{d \in (\Dcal^{\max})^\complement} w_d \exp [ z \, s(t_d) ]
+ \int_{(\Yc^{\max} (\varepsilon))^\complement} \exp [ z \, s ] \, \dlamb
\\
\leq &~ 
\exp [ z \, s^{\max} ] \, \biggl(
\sum_{d \in \Dcal^{\max}} w_d 
+ \lambda ( \Ycmflat ) + 2 \, M \, \varepsilon
\\
&\hspace{2.2cm}~
+ \sum_{\yc \in \Ycmpoint} \Bigl(
\frac{z^{-\alpha_{\yc}^{-1}} \, I_< (\yc)}{[\frac1{\alpha_{\yc}!} \, (| p_<^{(\alpha_{\yc})}(\yc) | - \varepsilon )]^{\alpha_{\yc}^{-1}}}
+ \frac{z^{-\alpha_{\yc}^{-1}} \, I_> (\yc)}{[\frac1{\alpha_{\yc}!} \, (| p_<^{(\alpha_{\yc})}(\yc) | - \varepsilon )]^{\alpha_{\yc}^{-1}}}
\Bigr)
\\
&\hspace{2.2cm}~
+ \lambda \left( (\Yc^{\max} (\varepsilon))^\complement \right) \, \exp [ - z \, \xi ]
+ \sum_{d \in (\Dcal^{\max})^\complement} w_d \exp [ - z \, \xi ]
\biggr) 
.
\end{align*}
Dividing the inequality by 
$\exp ( z \,  s^{\max} ) 
\, [ \,\sum_{d \in \Dcal^{\max}} w_d 
+ \lambda ( \Ycal_{\mathrm{c}, \, -}^{\max} ) 
+ \sum_{y \in \Ycal_{\mathrm{c}, \, \bullet}^{\max}} c_y \, z^{-\alpha_y^{-1}} ]$
yields
\begin{align*}
&
\frac{\int_{\Ycal} \exp [ z \, s ] \, \dmu}{\exp [ z \, s^{\max} ] \, 
\bigl( \sum_{d \in \Dcal^{\max}} w_d 
+ \lambda ( \Ycal_{\mathrm{c}, \, -}^{\max} ) 
+ \sum_{y \in \Ycal_{\mathrm{c}, \, \bullet}^{\max}} c_y \, z^{-\alpha_y^{-1}} \bigr)}
\\
\leq &~
\frac{\sum_{d \in \Dcal^{\max}} w_d 
+ \lambda ( \Ycmflat ) + 2 \, M \, \varepsilon
+ \sum_{\yc \in \Ycmpoint} 
\Biggl( 
\frac{z^{-\alpha_{\yc}^{-1}} \, I_<(\yc)}{\bigl[\frac{| p_<^{(\alpha_{\yc})}(\yc) | - \varepsilon }{\alpha_{\yc}!} \bigr]^{\alpha_{\yc}^{-1}}}
+ 
\frac{z^{-\alpha_{\yc}^{-1}} \, I_>(\yc)}{\bigl[\frac{| p_>^{(\alpha_{\yc})}(\yc) | - \varepsilon }{\alpha_{\yc}!} \bigr]^{\alpha_{\yc}^{-1}}}
\Biggr)
}
{\sum_{d \in \Dcal^{\max}} w_d 
+ \lambda ( \Ycal_{\mathrm{c}, \, -}^{\max} ) 
+ \sum_{\yc \in \Ycal_{\mathrm{c}, \, \bullet}^{\max}} 
\Biggl( \frac{z^{-\alpha_y^{-1}} \, I_<(\yc) }{\bigl[\frac{| p_<^{(\alpha_{\yc})}(\yc)|}{\alpha_{\yc}!} \bigr]^{\alpha_{\yc}^{-1}}}
+ \frac{z^{-\alpha_y^{-1}} \, I_>(\yc) }{\bigl[\frac{| p_>^{(\alpha_{\yc})}(\yc)|}{\alpha_{\yc}!} \bigr]^{\alpha_{\yc}^{-1}}} 
\Biggr)}
\\
&~ + 
\frac{\lambda\left( (\Yc^{\max} (\varepsilon))^\mathrm{c} \right) \, \exp [ - z \, \xi]
+ \sum_{d \in (\Dcal^{\max})^\complement} w_d \exp [- z \, \xi ]}{
\sum_{d \in \Dcal^{\max}} w_d 
+ \lambda ( \Ycal_{\mathrm{c}, \, -}^{\max} ) 
+ \sum_{y \in \Ycal_{\mathrm{c}, \, \bullet}^{\max}} c_y \, z^{-\alpha_y^{-1}}}
.
\end{align*}
Consider the limit for $z \ra \infty$.
The limit of the fraction in the last line is zero, which can be seen by expanding the fraction by $\exp [z \,  \xi]$ and using $\lim_{z \ra \infty} \exp [z \,  \xi] \, z^{-\alpha_{\yc}^{-1}} = \infty$ for all $\alpha_{\yc} \in \Nbb$ as $\xi > 0$.
If $\Dcal^{\max} \neq \emptyset$ or $\Ycmflat \neq \emptyset$, the limit of the fraction in the second line can be computed separately in the numerator and the denominator with the last terms converging to zero, respectively.
Otherwise, expand the fraction (where numerator/denominator then only consist of the last term) by $z^{\alpha^{-1}}$ where $\alpha = \max_{\yc \in \Ycmpoint} \alpha_{\yc}$.
Finally,
$ 
\lim_{z \ra \infty} \frac{\int_{\Ycal} \exp [ z \, s ] \, \dmu}{\exp [ z \, s^{\max} ] \, 
[ \, \sum_{d \in \Dcal^{\max}} w_d 
+ \lambda ( \Ycal_{\mathrm{c}, \, -}^{\max} ) 
+ \sum_{y \in \Ycal_{\mathrm{c}, \, \bullet}^{\max}} c_y \, z^{-\alpha_y^{-1}} ]}
\leq 1
$ 
is obtained by considering $\varepsilon \ra 0$.
\qedhere

\end{proof}

\pagebreak

\begin{thm}\label{P2a:thm_loglikelihood_coercive}
Under assumption 
\ref{P2a:assumption_appendix_scalar_product_non_constant},
$\lim_{\Vert \thetaf \Vert_2 \ra \infty} \ell_{\pen} (\thetaf) = - \infty$. 
\end{thm}

\begin{proof}
We first show the result for the unpenalized log-likelihood, i.e., 
$\lim_{\Vert \thetaf \Vert_2 \ra \infty} \ell (\thetaf)$ $= - \infty$.
Let $i \in \{1, \ldots , N\}$.
For $\thetaf \in \Rbb^K \setminus \{ \mathbf{0} \}$, let $z := \Vert \thetaf \Vert_2 > 0$ and $\vf := \frac{\thetaf}{\Vert \thetaf \Vert_2}\in \Sbb^{K-1}$, i.e., $\thetaf = z \, \vf$.
As we later apply Theorem~\ref{P2a:thm_generalized_laplace_approximation} to $s_{\vf, \xf_i}$, we briefly show that the assumptions are fulfilled, i.e., that $s_{\vf, \xf_i, \mathrm{c}}$ is a $(\kappa -1)$ times continuously differentiable piecewise polynomial consisting of finitely many polynomials of degree at most $\kappa$ for some $\kappa \geq 2$.
In the continuous special case, this is straightforward (using property~\ref{P2a:assumption_piecewise_polynomial}),
since for every $\vf \in \Sbb^{K-1}$ and every $\xf \in \Xcal$, 
we have $s_{\vf, \xf, \mathrm{c}} \in \spano (\bft_{\Yc})$. 
In the mixed case, recalling the construction of the basis functions via orthogonal decomposition of $\B$ from appropriate bases for the continuous and discrete special cases via the embeddings~\eqref{P2a:eq_embedding_L20_c} and~\eqref{P2a:eq_embedding_L20_d}, 
we have
$\bft_{\Ycal} = (\tilde{\iota}_{\mathrm{c}} (\bt_{\Yc, \, 1}), \ldots, \tilde{\iota}_{\mathrm{c}} (\bt_{\Yc, \, K_{\Yc}}), \tilde{\iota}_{\mathrm{d}} (\bt_{\Yd^\bullet, \, 1}), \ldots , \tilde{\iota}_{\mathrm{d}} (\bt_{\Yd^\bullet, \, K_{\Yd^\bullet}}))^\top$, where $\tilde{\iota}_{\mathrm{c}} (\bt_{\Yc, \, m})$ is a piecewise polynomial with the desired properties 
on $\Yc \setminus \Yd$ for all $m = 1, \ldots , K_{\Yc}$ and $\tilde{\iota}_{\mathrm{d}} (\bt_{\Yd, \, m})$ is constant on $\Yc \setminus \Yd$ for all $m = 1, \ldots , K_{\Yd^\bullet}$.
Thus, for all $\vf \in \Sbb^{K-1}$ and all $\xf \in \Xcal$, 
as $s_{\vf, \xf} \in \spano (\bft_{\Ycal})$, the function $s_{\vf, \xf, \mathrm{c}}$ is a 
$(\kappa -1)$ times continuously differentiable piecewise polynomial consisting of finitely many polynomials of degree at most $\kappa$ for some $\kappa \geq 2$.
\\
Thus, let $s_{\vf, \xf_i}^{\max} := \sup_{y \in \Ycal} s_{\vf, \xf_i} (y)$ and let $\Dcal^{\max} (s_{\vf, \xf_i})$ and $\Ycm (s_{\vf, \xf_i}) = \Ycmflat (s_{\vf, \xf_i}) \sqcup \Ycmpoint (s_{\vf, \xf_i})$ be the sets as in Theorem~\ref{P2a:thm_generalized_laplace_approximation}.
Furthermore, for $y \in \Ycmpoint (s_{\vf, \xf_i})$ let $c_{\vf, \xf_i , y}>0$ and $\alpha_{\vf, \xf_i , y} \in \Nbb$ with $\alpha_{\vf, \xf_i , y} \leq \kappa$ such that~\eqref{P2a:eq_laplace_approx_generalized} holds.
Then, setting
$\ell_i(\thetaf) := \bft(\xf_i)(y_i)^\top \thetaf 
- \log \int_{\Ycal} \exp \bigl[ \bft(\xf_i)^\top \thetaf \bigr]\, \dmu$,
we have 
$\ell(\thetaf) = \sum_{i=1}^N \ell_i(\thetaf)$
with $\ell_i(\thetaf)$ for $i \in \{1, \ldots , N\}$ equal to
\begin{align*}
%
&~ 
z \, s_{\vf, \xf_i} (y_i) 
\\
& 
- \log \frac{\int_{\Ycal} \exp [ z \, s_{\vf, \xf_i} 
] \, \dmu}
{\exp [ z \, s_{\vf, \xf_i}^{\max} 
] \Bigl( \sum_{d \in \Dcal^{\max} (s_{\vf, \xf_i})} w_d 
+ \lambda \bigl( \Ycmflat (s_{\vf, \xf_i}) \bigr) 
+ \sum_{y \in \Ycmpoint (s_{\vf, \xf_i})} c_{\vf, \xf_i , y} \, z^{-\alpha_{\vf, \xf_i , y}^{-1}} \Bigr) } 
\\
& 
- z \, s_{\vf, \xf_i}^{\max} 
- \log \Bigl( 
\sum_{d \in \Dcal^{\max} (s_{\vf, \xf_i})} w_d 
+ \lambda \bigl( \Ycmflat (s_{\vf, \xf_i}) \bigr) 
+ \sum_{y \in \Ycmpoint (s_{\vf, \xf_i})} c_{\vf, \xf_i , y} \, z^{-\alpha_{\vf, \xf_i , y}^{-1}}
\Bigr)
\, .
\end{align*}
By 
Theorem~\ref{P2a:thm_generalized_laplace_approximation}
and setting 
$C_{\vf, \xf_i} := 
\sum_{d \in \Dcal^{\max} (s_{\vf, \xf_i})} w_d 
+ \lambda ( \Ycmflat (s_{\vf, \xf_i}) ) 
$ 
and 
$c_{\vf}^{\min} := \min \bigl\{ c_{\vf, \xf_i , y} ~|~ y \in \Ycmpoint (s_{\vf, \xf_i}), i \in \{1, \ldots, N\} \bigr\}$,
we obtain
\begin{align*}
&\lim_{\Vert \thetaf \Vert_2 \ra \infty} \ell (\thetaf)
\\
= &~ 
\lim_{z \ra \infty} \sum_{i=1}^N \Bigl( z \, 
[ s_{\vf, \xf_i} (y_i) - s_{\vf, \xf_i}^{\max}]
- \log 
\bigl( C_{\vf, \xf_i} +
\sum_{y \in \Ycmpoint (s_{\vf, \xf_i})} c_{\vf, \xf_i , y} \, z^{-\alpha_{\vf, \xf_i , y}^{-1}}
\bigr)
\Bigr)
\\
\leq &~
\lim_{z \ra \infty} \sum_{i=1}^N \Bigl( z \, 
[ s_{\vf, \xf_i} (y_i) - s_{\vf, \xf_i}^{\max}]
- \log \bigl( 
|\Ycmpoint (s_{\vf, \xf_i})| \, c_{\vf}^{\min} \, z^{-1} \bigr)
\Bigr)
\\
= &~ 
\lim_{z \ra \infty} \biggl( \sum_{i=1}^N \bigl( z \, 
[ s_{\vf, \xf_i} (y_i) - s_{\vf, \xf_i}^{\max}] \bigr)
+ N \, \log  z
\biggr)
- \sum_{i=1}^N 
\log \bigl( |\Ycmpoint (s_{\vf, \xf_i})| \, c_{\vf}^{\min} \bigr) 
.
\end{align*}
We have 
$ 
s_{\vf, \xf_i} (y_i) - s_{\vf, \xf_i}^{\max} \leq 0$ for all $i \in \{1, \ldots , N\}$.
Per assumption~\ref{P2a:assumption_appendix_scalar_product_non_constant}, there exists an $i \in \{ 1, \ldots , N\}$ such that 
$
s_{\vf, \xf_i} (y_i) - s_{\vf, \xf_i}^{\max}< 0$ and thus, 
$\lim_{z \ra \infty} z \, 
[s_{\vf, \xf_i} (y_i) - s_{\vf, \xf_i}^{\max}] = -\infty$. 
Since this linear term dominates the logarithmic term $N \, \log z$ (which tends towards $+ \infty$) for $z \ra \infty$, we obtain 
$\lim_{\Vert \thetaf \Vert_2 \ra \infty} \ell (\thetaf) 
= -\infty$.
As the matrices $\Pf_j$ are positive semi-definite for all $j \in \{1, \ldots , J\}$, also $\Pf$ is positive semi-definite and thus 
$\lim_{\Vert \thetaf \Vert_2 \ra \infty} \ell_{\pen} (\thetaf) 
= \lim_{\Vert \thetaf \Vert_2 \ra \infty} \ell (\thetaf) - \lim_{\Vert \thetaf \Vert_2 \ra \infty} \thetaf ^\top \Pf \thetaf
= -\infty$.
\end{proof}

\begin{thme}\label{P2a:thm_appendix_loglikelihood_level_bounded}
Under assumption~\ref{P2a:assumption_appendix_scalar_product_non_constant}, the function $-\ell_{\mathrm{pen}}$ is level-bounded, i.e., for every $\alpha \in \Rbb$, the set 
$\{ \thetaf \in \Rbb^K ~|~ - \ell_{\pen} (\thetaf) \leq \alpha \}$ is bounded.
\end{thme}

\begin{proof}
Assume $-\ell_{\mathrm{pen}}$ is not level-bounded.
Then, there exists $\alpha \in \Rbb$, such that the set 
$\{ \thetaf \in \Rbb^K 
~|~ - \ell_{\pen} (\thetaf) \leq \alpha \}$ is not bounded.
Thus, there exist 
$\thetaf^{(t)} \in \{ \thetaf \in \Rbb^K 
~|~ - \ell_{\pen} (\thetaf) \leq \alpha \},
t \in \Nbb,$ with $\lim_{t \ra \infty} \Vert \thetaf^{(t)} \Vert_2 = \infty$.
However, we have $\lim_{\Vert \thetaf \Vert_2 \ra \infty} \ell_{\pen} (\thetaf) = -\infty$ by Theorem~\ref{P2a:thm_loglikelihood_coercive}, since assumption 
\ref{P2a:assumption_appendix_scalar_product_non_constant} holds.
Thus, $\lim_{t \ra \infty} - \ell_{\pen} (\thetaf^{(t)}) = \infty$, which is a contradiction to $-\ell_{\pen} (\thetaf^{(t)}) \leq \alpha$ for all $t \in \Nbb$.
Thus, $-\ell_{\mathrm{pen}}$ is level-bounded and $\thetafh$ exists.
\end{proof}

\subsection{Proofs of Theorems and Lemmas from the Main Manu\-script 
(and Formulation and Proof of Lemma~\ref{P2a:lemma_confidence_regions_simultaneos})
}\label{P2a:chapter_proofs}

\begin{proof}[Proof of Theorem~\ref{P2:thm_existence_uniqueness_bayes_PMLE}]
We first show uniqueness of the PMLE.
From Proposition~\ref{P2a:prop_appendix_fisher_informations_positive_definite} \ref{P2a:prop_appendix_bayes_fisher_information_positive_definite}, we have that the Hessian matrix of $\ell_{\pen} (\thetaf)$ is negative definite, which implies that $\ell_{\pen} (\thetaf)$ is strictly concave on $\Rbb^K$.
Per assumption, there exists a PMLE $\thetafh$. 
Now, assume that 
it is not unique, i.e., there are 
$\thetafh_1 \neq \thetafh_2 \in \Rbb^K 
$ such that 
$\ell_{\text{pen}} (\thetafh_1) = \ell_{\text{pen}} (\thetafh_2) = \max_{\thetaf \in \Rbb^K 
} \ell_{\text{pen}} (\thetaf)$.
Due to the strict concavity of $\ell_{\text{pen}}$, we have
$
\ell_{\text{pen}} ( (1-\alpha) \, \thetafh_1 + \alpha \, \thetafh_2 )
> (1-\alpha ) \, \ell_{\text{pen}} (\thetafh_1) + \alpha \, \ell_{\text{pen}} (\thetafh_2) = \max_{\thetaf \in \Rbb^K 
} \ell_{\text{pen}} (\thetaf)
$
for $\alpha \in (0, 1)$, which is a contradiction, i.e., the PMLE is unique, if it exists.
\\
To obtain existence of the PMLE $\thetafh$ (which is also the minimizer of $-\ell_{\mathrm{pen}}$) it suffices to show that $-\ell_{\mathrm{pen}}$ is lower semicontinuous, level-bounded and proper \citep[Theorem 1.9]{rockafellar2009}. 
As $-\ell_{\mathrm{pen}}$ takes values in $\Rbb$, it is proper.
Furthermore, it is continuous, implying it is lower semicontinuous \citep[Exercise 1.12]{rockafellar2009}.
Finally, by Theorem~\ref{P2a:thm_appendix_loglikelihood_level_bounded}, $-\ell_{\mathrm{pen}}$ is level-bounded and thus $\thetafh$ exists. 
\end{proof}

\subsubsection*{Proof of Theorem~\ref{P2:thm_MLE_asymptotic_normal}}\label{P2a:chapter_asymptotic_PMLE}
The following auxiliary Lemmas~\ref{P2a:lem_assumption_smoothing_parameters_implies_convergence_penalty}, \ref{P2a:lem_score_convergesP_zero}, and~\ref{P2a:lem_lower_bound_quadratic_form_Fisher_info} are dedicated to ultimately 
prove consistency and asymptotic normal distribution of the PMLE~$\hat{\thetaf}$, 
i.e., Theorem~\ref{P2:thm_MLE_asymptotic_normal}.
Both of these properties are well-known for ``regular log-likeli\-hoods''.
However, this term is often not specified further in the literature, especially in our case of non-identically distributed (but independent) observations.
Exceptions are \citet{hoadley1971, philippou1975} who both discuss concrete conditions in this case for asymptotic normality of the unpenalized MLE \citep[also for consistency]{hoadley1971}, however, without obtaining the result of asymptotic existence.
Furthermore, one condition of \citet{philippou1975} is consistency, whose proof already requires some ideas needed to show asymptotic normality as well.
Our proof generalizes the one of \citet[Section 33.3]{cramer1946} for a scalar parameter in the i.i.d.\ case for unpenalized estimation to a parameter vector in the i.non-i.d.\ case for penalized estimation.
This includes replacing the weak Law of Large Numbers (Khintchine's Theorem) and the (Lindeberg-L\'{e}vy) Central Limit Theorem with variants not requiring identically distributed data and including the penalty term, as well as
using a multivariate generalization of the intermediate value theorem.

%
We denote convergence in probability with $\xra[N \ra \infty]{\Pbb}$ and convergence in distribution with $\xra[N \ra \infty]{\mathrm{D}}$.
Furthermore, we write $\Zf_N = \ocal_{\Pbb}(a_N)$ to denote $\frac{\Zf_N}{a_N} \xra[N \ra \infty]{\Pbb} 0$ for a sequence of random vectors $(\Zf_N)_{N \in \Nbb}$ and a sequence of real numbers $(a_N)_{N \in \Nbb}$.

\begin{lem}\label{P2a:lem_assumption_smoothing_parameters_implies_convergence_penalty}
For any sequence of real numbers $(a_N)_{N \in \Nbb}$, such that
\begin{align*}
\xi_{\max} := \max \{\xi_{\Xcal, \, 1}, \ldots , \xi_{\Xcal, \, J}, \xi_{\Ycal, \, 1}, \ldots , \xi_{\Ycal, \, J} \} = \ocal_{\Pbb}(a_N),
\end{align*}
we have
$\frac1{a_N} \Vert \Pf \Vert_{\mathrm{op}} 
\xra[N \ra \infty]{\Pbb} 0$
and
$\frac1{a_N} 
\Pf \thetaf 
\xra[N \ra \infty]{\Pbb} \mathbf{0}$ for every $\thetaf \in \Rbb^K$, where $\Vert \Af \Vert_{\mathrm{op}} 
:= \sup_{\wf \in \Rbb^K \setminus\{\mathbf{0}\}} \frac{\Vert \Af \wf \Vert_2}{\Vert \wf \Vert_2}$ denotes the spectral norm of $\Af \in \Rbb^{K \times K}$.
\end{lem}

\begin{proof}
Set 
$P_{\max} := \max
\left\{ \Vert (\Pf_{\Xcal, \, j} \otimes \If_{K_{\Ycal}}) + (\If_{K_{\Xcal, \, j}} \otimes \Pf_{\Ycal}) \Vert_{\mathrm{op}}
~|~
j \in \{1, \ldots , J\} \right\}$.
Since $\Pf$ is block diagonal, we have $\Vert \Pf \Vert_{\mathrm{op}} = \max \{ \Vert \Pf_1 \Vert_{\mathrm{op}}, \ldots , \Vert \Pf_J \Vert_{\mathrm{op}} \}$ \citep[e.g.,][(5.2.12)]{meyer2000}.
Furthermore, for any $j \in \{ 1, \ldots , J\}$, we have
$ 
\Vert \Pf_j \Vert_{\mathrm{op}}
\leq 
\xi_{\max} 
\,
P_{\max} 
$ 
and thus in particular $ \Vert \Pf \Vert_{\mathrm{op}} \leq \xi_{\max} \, P_{\max}$.
Hence, for any $\varepsilon > 0$,
\begin{align*}
\Pbb_{\varthetaf} \left( \Bigl\vert \frac1{a_N} \Vert \Pf \Vert_{\mathrm{op}} - 0 \Bigr\vert \geq \varepsilon \right)
\leq
\Pbb_{\varthetaf} \left( \frac1{a_N} \xi_{\max} \geq \frac{\varepsilon}{P_{\max}} \right)
\xra[N \ra \infty]{} 0,
\end{align*}
and using $\Vert \Pf \thetaf \Vert_2 \leq \Vert \Pf \Vert_{\mathrm{op}} \, \Vert \thetaf \Vert_2$ (by definition of the spectral norm),
\begin{align*}
\Pbb_{\varthetaf} \left( \Bigl\Vert \frac1{a_N} \Pf \thetaf - \mathbf{0} \Bigr\Vert_2 \geq \varepsilon \right)
\leq
\Pbb_{\varthetaf} \left( \frac1{a_N} \xi_{\max} \geq \frac{\varepsilon}{P_{\max} \, \Vert \thetaf \Vert_2} \right)
\xra[N \ra \infty]{} 0,
\end{align*}
for any $\thetaf \in \Rbb^K$, i.e., $\frac1{a_N} \Vert \Pf \Vert_{\mathrm{op}} \xra[N \ra \infty]{\Pbb} 0$
and
$\frac1{a_N} \Vert \Pf \thetaf \Vert_2 \xra[N \ra \infty]{\Pbb} 0$.
\end{proof}

\begin{lem}\label{P2a:lem_score_convergesP_zero}
Let $\nabla \varphi := ( \frac{\partial}{\partial \theta_{1}} \varphi (\thetaf), \ldots ,\frac{\partial}{\partial \theta_{K}} \varphi (\thetaf))^\top$ denote the gradient of a function $\varphi : \Rbb^K \ra \Rbb$ 
and $\Ff (\varthetaf)$ the Fisher information of $\ell (\varthetaf)$.
\begin{enumerate}[label=\arabic*)]
\item\label{P2a:lem_score_convergesP_zero_item_expectation_zero}
$\Ebb_{\varthetaf} \left( \nabla \log f_{\xf_i, \varthetaf} (Y_i) \right) = \mathbf{0}$ for all $i \in \{1, \ldots , N\}$,
\item\label{P2a:lem_score_convergesP_zero_item_variance_fisherinfo}
$\Var_{\varthetaf} ( \nabla \ell (\varthetaf) )
= \Var_{\varthetaf} ( \sum_{i=1}^N \nabla \log f_{\xf_i, \varthetaf} (Y_i) )
= \Ff (\varthetaf)$,
\item\label{P2a:lem_score_convergesP_zero_item_convergesP_unpenalized}
$\frac1{N} \nabla \ell (\varthetaf) 
\xra[N \ra \infty]{\Pbb} \mathbf{0}$,
\item\label{P2a:lem_score_convergesP_zero_item_convergesP}
Assume that
$\xi_{\max} 
= \ocal_{\Pbb}(N)$.
Then, $\frac1{N} \nabla \ell_{\mathrm{pen}} (\varthetaf) \xra[N \ra \infty]{\Pbb} \mathbf{0}$. 
\end{enumerate}
\end{lem}

\begin{rem}
For independent and identically distributed observations, \ref{P2a:lem_score_convergesP_zero_item_convergesP_unpenalized} -- which is also used to show \ref{P2a:lem_score_convergesP_zero_item_convergesP} -- would follow from the weak Law of Large Numbers. 
However, we consider independent, but non-identically distributed data and thus prove the result by showing the variances of the first derivatives of the log-likelihood contribution are bounded and applying Chebyshev's inequality.
\end{rem}

\begin{proof}[Proof of Lemma~\ref{P2a:lem_score_convergesP_zero}]
\begin{enumerate}[label=\arabic*)]
\item
Let $i \in \{ 1, \ldots , N\}$.
Provided that
\begin{align}
\int_{\Ycal} \nabla f_{\xf_i, \varthetaf} \, \dmu
= \nabla \int_{\Ycal} f_{\xf_i, \varthetaf} \, \dmu ,
\label{P2a:eq_exchange_integral_derivative_density}
\end{align}
we directly obtain the result, using that $f_{\xf_i, \thetaf}$ integrates to $1$ for all 
$\thetaf \in \Rbb^K 
$:
\begin{align}
\Ebb_{\varthetaf} ( \nabla \log f_{\xf_i, \varthetaf} (Y_i))
= \int_{\Ycal} \Bigl( \frac1{f_{\xf_i, \varthetaf}} \nabla f_{\xf_i, \varthetaf} \Bigr) f_{\xf_i, \varthetaf} \, \dmu
= \nabla \int_{\Ycal} f_{\xf_i, \varthetaf} \, \dmu
= \mathbf{0}.
\label{P2a:eq_expectation_score_zero}
\end{align}
To show \eqref{P2a:eq_exchange_integral_derivative_density}, we verify the three properties given in Corollary 16.3 in \citet{bauer2001} for the considered function $f_{\xf_i, \thetaf} = \frac{\exp [\bft (\xf_i)^\top \thetaf]}{\int_{\Ycal} \exp [ \bft (\xf_i)^\top \thetaf ]\, \dmu}$. 
First, $f_{\xf_i, \thetaf}$ is $\mu$-inte\-grable for all 
$\thetaf \in \Rbb^K 
$.
Second,
the partial derivative 
$\frac{\partial}{\partial \theta_k} f_{\xf_i, \thetaf} (y_i)$ exists. 
It is given as
\begin{align*}
\frac{\exp [\bft (\xf_i)(y_i)^\top \thetaf] \bigl( \int_{\Ycal} \exp [ \bft (\xf_i)^\top \thetaf ]\, \dmu \, \bft (\xf_i)(y_i)_{[k]} - \int_{\Ycal} \exp [ \bft (\xf_i)^\top \thetaf ] \, \bft (\xf_i)_{[k]}\, \dmu \bigr)}{\bigl( \int_{\Ycal} \exp [ \bft (\xf_i)^\top \thetaf ]\, \dmu\bigr)^2} ,
\end{align*} 
where we used $\frac{\partial}{\partial \theta_k} \int_\Ycal \exp [ \bft (\xf_i)^\top \thetaf ] \, \dmu = \int_\Ycal \frac{\partial}{\partial \theta_k} \exp [ \bft (\xf_i)^\top \thetaf ] \, \dmu$ as shown in the paragraph before \eqref{P2a:eq_first_derivative_C_i}.
For the third required property, note that since the basis functions contained in $\bft_\Ycal$ and $\bfe_{\Xcal}$ are bounded and $\mu$ is finite, $|\frac{\partial}{\partial \theta_k} f_{\xf_i, \thetaf}|$ is bounded on an open ball 
$\upsilon_\varepsilon(\thetaf) := \{\thetaf' \in \Rbb^K 
~|~ \Vert \thetaf' - \thetaf \Vert_2 < \varepsilon \}$ for each 
$\thetaf \in \Rbb^K 
$ and $\varepsilon > 0$ by a nonnegative finite constant, 
which defines a constant and thus $\mu$-integrable function on $\Ycal$.
Thus, \eqref{P2a:eq_exchange_integral_derivative_density} follows with \citet[Corollary 16.3]{bauer2001} and \eqref{P2a:eq_expectation_score_zero} is proven.
\item
From \ref{P2a:lem_score_convergesP_zero_item_expectation_zero}, we have
\begin{align*}
\Var_{\varthetaf} \Bigl( \sum_{i=1}^N \nabla \log f_{\xf_i, \varthetaf} (Y_i) \Bigr)
&= \Ebb_{\varthetaf} \Bigl[ \Bigl( \sum_{i=1}^N \nabla \log f_{\xf_i, \varthetaf} (Y_i) \Bigr) \Bigl( \sum_{j=1}^N \nabla \log f_{\xf_j, \varthetaf} (Y_j) \Bigr)^\top \Bigr] 
,
\end{align*}
whose $(k, k')$-th entry is given as 
$ 
\sum_{i, j=1}^N \Ebb_{\varthetaf} ( \frac{\partial}{\partial \theta_k} \log f_{\xf_i, \varthetaf} (Y_i) \frac{\partial}{\partial \theta_{k'}} \log f_{\xf_j, \varthetaf} (Y_j) )
$ 
for $k, k' \in \{1, \ldots , K\}$.
Since $(Y_1, \Xf_1), \ldots , (Y_N, \Xf_N)$ are independent, $\frac{\partial}{\partial \theta_k} \log f_{\xf_i, \varthetaf} (Y_i)$ and $ \frac{\partial}{\partial \theta_{k'}} \log f_{\xf_j, \varthetaf} (Y_j)$ are independent as well. 
With \eqref{P2a:eq_expectation_score_zero}, we obtain that all summands with $i \neq j$ are zero.
Using \eqref{P2a:eq_first_derivative_C_i}, 
the remaining summands (with $i = j$) are
\begin{align*}
&\, \int_{\Ycal} \frac{\partial}{\partial \theta_k} \log f_{\xf_i, \varthetaf} \, \frac{\partial}{\partial \theta_{k'}} \log f_{\xf_i, \varthetaf} \, f_{\xf_i, \varthetaf} \, \dmu
\\
=& \, \int_{\Ycal} \Bigl( \bft (\xf_i)_{[k]} - \int_\Ycal f_{\xf_i, \varthetaf} \, \bft (\xf_i)_{[k]} \, \dmu \Bigr) 
\, \Bigl( \bft (\xf_i)_{[k']} - \int_\Ycal f_{\xf_i, \varthetaf} \, \bft (\xf_i)_{[k']} \, \dmu \Bigr) \, f_{\xf_i, \varthetaf} \, \dmu
\\
=& \int_{\Ycal} f_{\xf_i, \varthetaf} \, \bft (\xf_i)_{[k]} \, \bft (\xf_i)_{[k']} \, \dmu
- \int_{\Ycal} f_{\xf_i, \varthetaf} \,  \bft (\xf_i)_{[k]} \, \dmu \, \int_\Ycal f_{\xf_i, \varthetaf} \, \bft (\xf_i)_{[k']} \, \dmu
.
\end{align*}
From~\eqref{P2a:eq_hesse_matrix_log_likelihood}
we obtain that this is the $(k, k')$-th element of the $i$-th contribution of the negative Hessian matrix of $\ell (\varthetaf)$ 
and thus,
$\Var_{\varthetaf} ( \sum_{i=1}^N \nabla \log f_{\xf_i, \varthetaf} (Y_i) ) = \Ff (\varthetaf)$. 
\item
Since the basis functions contained in $\bft_\Ycal$ and $\bfe_{\Xcal}$ are bounded and $\mu$ is finite, there exists $\sigma_{\varthetaf}$ 
such that for all $k = 1, \ldots , K$ and all $i = 1, \ldots , N$
\begin{align*}
\Bigl|\frac{\partial}{\partial \theta_k} \log f_{\xf_i, \varthetaf} \Bigr|
= \Bigl| \bft (\xf_i)(y_i)_{[k]} - \int_\Ycal f_{\xf_i, \varthetaf} \, \, (\bft (\xf_i))_{[k]} \, \dmu \Bigr| 
\leq \sigma_{\varthetaf}
\end{align*} 
and thus,
$\Var_{\varthetaf} ( \frac{\partial}{\partial \theta_k} \log f_{\xf_i, \varthetaf} (Y_i) )
= \int_{\Ycal} ( \frac{\partial}{\partial \theta_k} \log f_{\xf_i, \varthetaf} )^2 \, f_{\xf_i, \varthetaf} \, \dmu
\leq \sigma_{\varthetaf}^2$.
Noting that $\{ \Vert \frac1{N} \sum_{i=1}^N  \nabla \log f_{\xf_i, \varthetaf} (Y_i) \Vert_2 \geq \varepsilon \} \subseteq \bigcup_{k=1}^K \{ | \frac1{N} \sum_{i=1}^N \frac{\partial}{\partial \vartheta_k} \log f_{\xf_i, \varthetaf} (Y_i) | \geq \frac{\varepsilon}{\sqrt{K}} \}$ for all $\varepsilon >0$,
%
%
using \ref{P2a:lem_score_convergesP_zero_item_expectation_zero}, and applying Chebyshev's inequality \citep[e.g.,][Corollary 4.7.3]{borovkov2013}, we have for all $\varepsilon >0$,
\begin{align*}
& \Pbb_{\varthetaf} \Bigl( \Bigl\Vert \frac1{N} \sum_{i=1}^N  \nabla \log f_{\xf_i, \varthetaf} (Y_i) - \mathbf{0} \Bigr\Vert_2 \geq \varepsilon \Bigr)
\\
\leq ~& \sum_{k=1}^K \Pbb_{\varthetaf} \Bigl( \Bigl| \frac1{N} \sum_{i=1}^N \frac{\partial}{\partial \theta_k} \log f_{\xf_i, \varthetaf} (Y_i) \Bigr| \geq \frac{\varepsilon}{\sqrt{K}} \Bigr) \notag
\\
\leq ~& \sum_{k=1}^K \frac{ K \,\Var_{\varthetaf} (\frac1{N} \sum_{i=1}^N \frac{\partial}{\partial \theta_k} \log f_{\xf_i, \varthetaf} (Y_i)  )}{\varepsilon^2} \notag
\\
\leq ~& \frac{K^2 N \sigma_{\varthetaf}^2}{N^2 \varepsilon^2}
\xra[N \ra \infty]{} 0,
\end{align*}
where we used that $(Y_i, \Xf_i)$ and thus 
$\frac{\partial}{\partial \theta_k} \log f_{\xf_i, \varthetaf} (Y_i)
$ are independent for $i = 1, \ldots , N$. 
Hence, $\frac1{N} \nabla \ell (\varthetaf) \xra[N \ra \infty]{\Pbb} \mathbf{0}$.
\item
For the proof of 
$ \frac1{N} \nabla \ell_{\mathrm{pen}} (\varthetaf) 
\xra[N \ra \infty]{\Pbb} \mathbf{0}$, 
note that
$
\frac1{N} \nabla \ell_{\mathrm{pen}} (\varthetaf) 
= \frac1{N} \nabla \ell (\varthetaf) - \frac1{N} \, 2 \, \Pf \varthetaf
$.
By~\ref{P2a:lem_score_convergesP_zero_item_convergesP_unpenalized}, the first term converges in probability to $\mathbf{0}$.
Furthermore, using $\xi_{\max} = \ocal_{\Pbb}(N)$, also the second term converges in probability to $\mathbf{0}$ by Lemma~\ref{P2a:lem_assumption_smoothing_parameters_implies_convergence_penalty}.
Let $\varepsilon > 0$.
Noting that
$
\{ \Vert \frac1{N} \nabla \ell (\varthetaf) - \frac1{N} \, 2 \, \Pf \varthetaf \Vert_2 \geq \varepsilon \} 
\subseteq \{ \Vert \frac1{N} \nabla \ell (\varthetaf) \Vert_2 + \Vert \frac1{N} \, 2 \, \Pf \varthetaf \Vert_2 \geq \varepsilon \}
\subseteq \{ \Vert \frac1{N} \nabla \ell (\varthetaf) \Vert_2 \geq \frac{\varepsilon}2 ~ \vee ~ \Vert \frac1{N} \, 2 \, \Pf \varthetaf \Vert_2  \geq \frac{\varepsilon}2\}
$
%
%
we obtain,
\begin{align*}
\Pbb \Bigl( \Bigl\Vert 
\frac1{N} \nabla \ell_{\mathrm{pen}} (\varthetaf) 
- \mathbf{0} \Bigr\Vert_2 \geq \varepsilon \Bigr)
&= \Pbb \Bigl( \Bigl\Vert 
\frac1{N} \nabla \ell (\varthetaf) - \frac1{N} \, 2 \, \Pf \varthetaf \Bigr\Vert_2 \geq \varepsilon \Bigr)
\\
&\leq
\Pbb \Bigl( \Bigl\Vert \frac1{N} \nabla \ell (\varthetaf) \Bigr\Vert_2 \geq \frac{\varepsilon}2 \Bigr) + \Pbb \Bigl( \Bigl\Vert \frac1{N} \, 2 \, \Pf \varthetaf \Bigr\Vert_2  \geq \frac{\varepsilon}2 \Bigr)
\\
&\xra[N \ra \infty]{} 0,
\end{align*}
i.e., $\frac1{N} \nabla \ell_{\mathrm{pen}} (\varthetaf) \xra[N \ra \infty]{\Pbb} \mathbf{0}$.
\qedhere
\end{enumerate}
\end{proof}

\begin{lem}\label{P2a:lem_lower_bound_quadratic_form_Fisher_info}
Assume that 
assumption \ref{P2a:assumption_appendix_stone_1} 
holds.
Then, for each $\thetaf \in \Rbb^K 
$, there exist $M > 0$ (independent of $N$) 
such that for all $\vf \in \Rbb^K$ and all $N \geq N'$,
\begin{align*}
N \, M \, \Vert \vf \Vert_2^2 \leq \vf^\top \Ff_{\mathrm{pen}} (\thetaf) \vf ,
\end{align*}
where $\Ff_{\mathrm{pen}} (\thetaf)$ denotes the Fisher information of the penalized log-likelihood $\ell_{\mathrm{pen}} (\thetaf)$. 
In particular, this shows that there exists $N' \in \Nbb$ such that for all $N \geq N'$, we have
$N \, M \, \Id_K \preceq \Ff_{\mathrm{pen}} (\thetaf)$,
where $\Af \preceq \Bf$ denotes that $\Bf - \Af$ is positive semi-definite for symmetric matrices $\Af, \Bf \in \Rbb^{K \times K}$. 
\end{lem}

\begin{proof}
We first show the statement for the unpenalized setting. 
The proof 
is inspired by \citet[Lemma 11]{stone1991}, but uses some different ideas to obtain the result. 
Let $\thetaf \in \Rbb^K 
$.
Since $\bfe_{\Xcal}$ and $\bft_{\Ycal}$ are bounded, and $\mu$ is finite, $M_{\thetaf} := \inf_{\xf \in \Xcal, \, y \in \Ycal} f_{\xf, \thetaf} (y) > 0$ exists.
For $\vf \in \Rbb^K$, set $\vf_{\Xcal} := (\bft_{\Ycal}^\top \vf_{1,1}, \ldots , \bft_{\Ycal}^\top \vf_{J, K_{\Xcal, \, J}}) \in L_0^2(\mu)^{K_{\Xcal}}$ with $\vf_{j, n} = (v_{j, n, 1} , \ldots , v_{j, n, K_{\Ycal}})^\top$, where $v_{j, n, m}$ for $j = 1, \ldots , J,~n = 1, \ldots , K_{\Xcal, \, j}$, and $m = 1, \ldots , K_{\Ycal}$, corresponds to the respective entry of $\vf$, where the indices are arranged in the same order as the vector $\thetaf \in \Rbb^K$ containing the values $\theta_{j, n, m}$ in Section~2.2. 
It is straightforward to show $\bft (\xf_i)^\top \vf = \bfe_{\Xcal} (\xf_i)^\top \vf_{\Xcal}$ for $i = 1, \ldots , N$.
Let $M' > 0$ and $N' \in \Nbb$ such that~\eqref{P2a:eq_stone_1} holds for $N \geq N'$.
From~\eqref{P2a:eq_quatratic_form_Hessematrix}, 
we obtain for the Fisher information $\Ff (\thetaf)$ of $\ell (\thetaf)$, 
\begin{align*}
\vf^\top \Ff (\thetaf) \vf
&= \sum_{i=1}^N \int_\Ycal f_{\xf_i, \thetaf} \, \Bigl( \bft (\xf_i)^\top \vf - \int f_{\xf_i, \thetaf} \cdot \bft (\xf_i)^\top \vf \, \dmu \Bigr)^2 \, \dmu
\\
&\geq 
M_{\thetaf} \sum_{i=1}^N \int_\Ycal \Bigl( \bft (\xf_i)^\top \vf - \int f_{\xf_i, \thetaf} \cdot \bft (\xf_i)^\top \vf \, \dmu \Bigr)^2 \, \dmu
\\
&\geq M_{\thetaf} \sum_{i=1}^N  \int_\Ycal \left( \bft (\xf_i)^\top \vf \right) ^2 \, \dmu - 2 \, \int_\Ycal \bft (\xf_i)^\top \vf \, \dmu \, \int f_{\xf_i, \thetaf} \cdot \bft (\xf_i)^\top \vf \, \dmu 
\\
&\hspace{-0.1cm}\overset{(\ast_1)}{=} M_{\thetaf} \int_\Ycal \sum_{i=1}^N \left( \bfe_{\Xcal} (\xf_i)^\top \vf_{\Xcal} \right) ^2 \, \dmu
\\
&\hspace{-0.25cm}\overset{\eqref{P2a:eq_stone_1}}{\geq} M_{\thetaf} \, M' \, N \int_\Ycal \int_{\Xcal} \left( \bft (\xf)^\top \vf \right) ^2 \, \dnuf (\xf) \, \dmu
\\
&\hspace{-0.1cm}\overset{(\ast_2)}{\geq} N \, (M_{\thetaf} \, M' \, M_{\bft}) \Vert \vf \Vert_2^2 
\end{align*}
where in $(\ast_1)$, we used $\bft (\xf_i)^\top \vf \in L_0^2(\mu)$ and thus $\int_\Ycal \bft (\xf_i)^\top \vf \, \dmu = 0$.
For $(\ast_2)$, note that $\Vert \vf \Vert_{\bft} := [\int_\Ycal \int_{\Xcal} ( \bft (\xf)^\top \vf ) ^2 \, \dnuf (\xf) \, \dmu]^{\frac12}$ defines a norm on $\Rbb^K$ (which is verified analogously to $\Vert \cdot \Vert_{\bfe_{\Xcal}}$ being a norm in the proof of Lemma~\ref{P2a:lem_assumptions}~\ref{P2a:lem_assumption_stone_implication} via the corresponding inner product, using the linear independence of the functions contained in $\bft_{\Ycal}$ and of the functions contained in $\bfe_{\Xcal}$) and thus $M_{\bft} > 0$ as in $(\ast_2)$ exists 
\citep[e.g.,][Corollary 5.4.5]{horn2012}.
Finally, for $M := M_{\thetaf} \, M' \, M_{\bft}$ we then have $\vf^\top \Ff (\thetaf) \vf - N \, M \, \Vert \vf \Vert_2^2 \geq 0$ and equivalently $\vf^\top (\Ff (\thetaf) - N \, M \, \Id) \vf \geq 0$.

Now consider the penalized case, i.e., $\Ff_{\mathrm{pen}} (\thetaf) = \Ff (\thetaf) + 2 \, \Pf$. 
Since $\Pf$ is positive semi-definite (which follows from the matrices $\Pf_j$ being positive semi-definite for all $j \in \{1, \ldots , J\}$), 
we have $\vf^\top 2 \, \Pf \vf \geq 0$ and thus
$\vf^\top (\Ff_{\mathrm{pen}} (\thetaf) - N \, M \, \Id) \vf \geq \vf^\top (\Ff (\thetaf) - N \, M \, \Id) \vf \geq 0$ and equivalently $\vf^\top (\Ff_{\mathrm{pen}} (\thetaf) - N \, M \, \Id) \vf \geq 0$.
Thus, $N \, M \, \Id \preceq \Ff_{\mathrm{pen}} (\thetaf)$ for $N \geq N'$.
\qedhere
\end{proof}

\begin{proof}[Proof of Theorem~\ref{P2:thm_MLE_asymptotic_normal}] 
\begin{enumerate}[label=\arabic*)]
\item
We showed in Lemma~\ref{P2a:lem_assumptions}~\ref{P2a:lem_stone_assumptions_implication} that assumption \ref{P2a:assumption_appendix_stone_1} implies assumption \ref{P2a:assumption_appendix_covariate_basis_nonsingular} and thus, by Proposition~\ref{P2a:prop_appendix_fisher_informations_positive_definite}, the Hessian matrix $\Hf_{\mathrm{pen}} (\thetaf)$ of $\ell_{\mathrm{pen}} (\thetaf)$ is negative definite for large enough $N$.
Hence, every root of $\nabla \ell_{\mathrm{pen}}$ is a maximizer of $\ell_{\mathrm{pen}}$.
Furthermore, 
$\frac1{N} \, 
\Hf_{\mathrm{pen}} (\varthetaf)
$ is a symmetric (and thus normal) matrix and the spectral theorem \citep[e.g.,][Theorem 2.5.3]{horn2012} yields that there is an orthonormal basis $\vf_1, \ldots , \vf_K$ of $\Rbb^K$ and values $\lambda_1, \ldots , \lambda_K \in \Rbb$ such that $\frac1{N} \, \Hf_{\mathrm{pen}} (\varthetaf) \vf_k = \lambda_k \vf_k$ for all $k = 1, \ldots , K$.
We show: 
\begin{enumerate}
\nitem{(Q)}\label{P2a:claim_solution_in_cube}
For all $\varepsilon, \delta > 0$, there exists an $N_0 \in \Nbb$, such that for all $N \geq N_0$, the equation $\nabla \ell_{\mathrm{pen}}(\thetaf) = \mathbf{0}$ has a solution in the cube 
$Q_\delta (\varthetaf) := \bigtimes_{k=1}^K (\varthetaf - \frac{\delta}{K} \vf_k, \varthetaf + \frac{\delta}{K} \vf_k)$ on a set whose probability exceeds $1 - \varepsilon$. 
\end{enumerate}
In the following, we write $f (\thetaf) 
:= f_{\xf , \thetaf} (y) = \frac{\exp [\bft (\xf)(y)^\top \thetaf]}{\int_{\Ycal} \exp [ \bft (\xf)^\top \thetaf ]\, \dmu}$ for arbitrary $(y, \xf) \in \Ycal \times \Xcal$ and $f_i (\thetaf) 
:= f_{\xf_i , \thetaf} (y_i)$ for $i = 1, \ldots , N$, since we are interested in $\thetaf$ as function argument.
We have 
$\ell_{\mathrm{pen}} (\thetaf) 
= \sum_{i = 1}^N ( \log f_i (\thetaf) + \frac1{N} \pen (\thetaf))
$.
Note that for almost all $(y, \xf) \in \Ycal \times \Xcal$ the derivatives $\nabla ( \log f (\thetaf) + \frac1{N} \pen (\thetaf) ), D^2 (\log f (\thetaf) + \frac1{N} \pen (\thetaf))$, and $D^3 (\log f (\thetaf) + \frac1{N} \pen (\thetaf)) = D^3 \log f (\thetaf)$ exist 
for every $\thetaf \in \Rbb^K 
$. 
Thus, 
we can expand the first derivative of $\log f (\thetaf) + \frac1{N} \pen (\thetaf)$ around the true value $\varthetaf$ using Taylor's formula \citep[e.g.,][Theorem 1.97]
{laczkovich2017} 
for every $\thetaf \in \Rbb^K 
$ as
\begin{align}
\nabla \Bigl( \log f (\thetaf) + \frac1{N} \pen (\thetaf) \Bigr)
&= \nabla \log f (\varthetaf) - \frac1{N} \, 2 \, \Pf \varthetaf \notag
\\
&\hspace{0.5cm} 
+ \Bigl( D^2 \log f (\varthetaf) - \frac1{N} \, 2 \, \Pf \Bigr) (\thetaf - \varthetaf) \notag
\\
&\hspace{0.5cm} + \frac1{2} (\thetaf - \varthetaf)^\top D^3 \log f (\thetaft) (\thetaf - \varthetaf)
, 
\label{P2a:eq_taylor_expansion_score_function}
\end{align}
for a $\thetaft = \thetaft_{(y, \xf)}$ 
in the segment
$\{ \varthetaf + \alpha \, (\thetaf - \varthetaf) ~|~ \alpha \in [0, 1] \}$
connecting $\varthetaf$ and $\thetaf$, 
and $D^3 \log f (\thetaft) := (D^2 \frac{\partial}{\partial \theta_1} \log f (\thetaft), \ldots , D^2 \frac{\partial}{\partial \theta_K} \log f (\thetaft)) \in \Rbb^{K \times K \times K}$.\footnote{For $\vf \in \Rbb^K$ and $\Wcalf = (\Wf_1 , \ldots , \Wf_K) \in \Rbb^{K \times K \times K}$, where $\Wf_k \in \Rbb^{K \times K}$ for $k = 1, \ldots , K$,
$\vf^\top \Wcalf \in \Rbb^{K \times K}$ is the matrix containing $\vf^\top \Wf_k$ in the $k$-th row, $k = 1, \ldots, K$. 
}
Building on~\eqref{P2a:eq_hessian_bayes}, 
where $\frac{\partial^2}{\partial \theta_k \partial \theta_{k'}} C_i(\thetaf) = \frac{\partial^2}{\partial \theta_k \partial \theta_{k'}}\log f_i (\thetaf)$, we obtain for $k, k', k'' \in \{ 1, \ldots , K\}$ that $\frac{\partial^3}{\partial \theta_k \partial \theta_{k'} \partial \theta_{k''}} \log f (\thetaf)$ is a sum of products, whose factors have the form
$ 
\pm \frac{ \int_\Ycal \exp [ \bft (\xf)^\top \thetaf ] \, \prod_{j \in \{k, k', k''\}} (\bft (\xf)_{[j]})^{\delta_j} \, \dmu}{\int_\Ycal \exp [ \bft (\xf)^\top \thetaf ] \, \dmu},
$ 
where $\delta_j \in \{0, 1\}$ and $\bft (\xf)_{[j]}$ denotes the $j$-th entry of the vector $\bft (\xf)$.
Since the basis functions contained in $\bft_\Ycal$ and $\bfe_{\Xcal}$ are bounded and $\mu$ is finite, there thus exists an $M > 0$, such that 
$ 
\sup_{\thetaf \in Q_\delta (\varthetaf)} \Vert D^3 \log f (\thetaf) \Vert_{\mathrm{op}} \leq M 
$ 
for almost all $(y, \xf) \in \Ycal \times \Xcal$, where 
$\Vert \Wcalf \Vert_{\mathrm{op}} 
:= \sup_{\vf \in \Rbb^K \setminus\{\mathbf{0}\}} \frac{\Vert \vf^\top \Wcalf \Vert_{\mathrm{op}}}{\Vert \vf \Vert_2}$ 
denotes the operator norm of $\Wcalf \in \Rbb^{K \times K \times K}$ with $\Vert \Af \Vert_{\mathrm{op}} 
:= \sup_{\wf \in \Rbb^K \setminus\{\mathbf{0}\}} \frac{\Vert \Af \wf \Vert_2}{\Vert \wf \Vert_2}$ the spectral norm of $\Af \in \Rbb^{K \times K}$.
Thus,
\begin{align*}
\Bigl\Vert \frac{D^3 \log f (\thetaft)}{M} \Bigr\Vert_{\mathrm{op}} 
\leq \frac1{M} \sup_{\thetaf \in Q_\delta (\varthetaf)} \left\Vert D^3 \log f (\thetaf) \right\Vert_{\mathrm{op}} 
\leq 1 ,
\end{align*}
and \eqref{P2a:eq_taylor_expansion_score_function} can be written as
\begin{align*}
\nabla \Bigl( \log f (\thetaf) + \frac1{N} \pen (\thetaf) \Bigr)
&= \nabla \log f (\varthetaf) - \frac1{N} \, 2 \, \Pf \varthetaf \notag
\\
&\hspace{0.5cm} 
+ \Bigl( D^2 \log f (\varthetaf) - \frac1{N} \, 2 \, \Pf \Bigr) (\thetaf - \varthetaf) \notag
\\
&\hspace{0.5cm} + \frac{M}{2} (\thetaf - \varthetaf)^\top \Ecalf_{(y, \xf)} (\thetaf - \varthetaf), 
\end{align*}
where $\Ecalf_{(y, \xf)} := \frac{D^3 \log f (\thetaft)}M 
\in \Rbb^{K \times K \times K}$ with $\Vert \Ecalf_{(y, \xf)} \Vert_{\mathrm{op}} \leq 1$.
Thus,
the equation $\nabla \ell_{\mathrm{pen}} (\thetaf) = \mathbf{0}$ is equivalent to
\begin{align}
\mathbf{0} 
&= \varsigmaf_1 + \varsigmaf_2 (\thetaf - \varthetaf) + \frac{M}{2} 
(\thetaf - \varthetaf)^\top \Ecalf (\thetaf - \varthetaf)
,
\label{P2a:eq_score_equation_taylor}
\end{align}
where
\begin{align*}
\varsigmaf_1 
&:= \frac1{N} \sum_{i = 1}^N \nabla \log f_i (\varthetaf) - \frac1{N} \, 2 \, \Pf \varthetaf
= \frac1{N} \nabla \ell_{\mathrm{pen}} (\varthetaf),
\\
\varsigmaf_2 
&:= \frac1{N} \sum_{i = 1}^N D^2 \log f_i (\varthetaf) - \frac1{N} \, 2 \, \Pf = \frac1{N} \, \Hf_{\mathrm{pen}} (\varthetaf) ,
&& \text{and} &&
\Ecalf 
:= \frac1{N} \sum_{i=1}^N \Ecalf_{(y_i, \xf_i)} .
\end{align*}
Note that $\Vert \Ecalf \Vert_{\mathrm{op}} \leq \frac1{N} \sum_{i=1}^N \Vert \Ecalf_{(y_i, \xf_i)} \Vert_{\mathrm{op}} \leq 1$.
Now let $\varepsilon, \delta > 0$.
Since 
$\xi_{\max} = \ocal_{\Pbb}(N)$,
we have from Lemma~\ref{P2a:lem_score_convergesP_zero}~\ref{P2a:lem_score_convergesP_zero_item_convergesP} 
that 
$\lim_{N \ra \infty} \Pbb_{\varthetaf} (\Vert \varsigmaf_1 \Vert_2 \geq \delta^2) = 0$. 
Thus, 
there exists an $N_1 \in \Nbb$ such that for all $N \geq N_1$, we have
$ 
\Pbb_{\varthetaf} (\Vert \varsigmaf_1 \Vert_2 \geq \delta^2) < \varepsilon
.
$ 
Setting
$\Psi_{N} := \{(y_1, \xf_1) , \ldots , (y_{N}, \xf_{N}) ~|~ 
\Vert \varsigmaf_1 \Vert_2 < \delta^2 
\}$,
this yields $\Pbb_{\varthetaf} (\Psi_{N}) \geq 1 - \varepsilon$ for all $N \geq N_1$.
We show that there exists $N_0 \geq N_1$ such that \eqref{P2a:eq_score_equation_taylor} has a solution in the cube $Q_\delta (\varthetaf) 
$
on the set $\Psi_{N}$ (and thus with probability $\geq 1 - \varepsilon$) for $N \geq N_0$. 
For this purpose, we use a multivariate generalization of the intermediate value theorem, the Poincar\'e-Miranda Theorem\footnote{The theorem is (without loss of generality) formulated for the cube $[0, 1]^K$, we use the straightforward generalization for arbitrary cubes. 
} \citep[e.g.,][]{kulpa1997}.
Let $k \in \{1, \ldots, K\}$ and denote the $k$-th opposite faces of $Q_\delta (\varthetaf)$ with 
\begin{align*}
Q_k^\pm &:= 
\Bigl\{ \thetaf = \varthetaf + \frac1{K} \, \sum_{k'=1}^K \alpha_{k'} \vf_{k'} ~\Big|~ \alpha_k = \pm \delta ~\wedge \forall k' \in \{1, \ldots , K\}: \alpha_{k'} \in [-\delta , \delta] \Bigr\}
.
\end{align*} 
We show that
$\langle v_k , \varsigmaf_1 + \varsigmaf_2 (\thetaf - \varthetaf) + \frac{M}{2} (\thetaf - \varthetaf)^\top \Ecalf (\thetaf - \varthetaf)\rangle$ is positive for all $\thetaf \in Q_k^-$ and negative for all $\thetaf \in Q_k^+$.
For $\thetaf \in Q_k^- \cup Q_k^+$, we have 
\begin{align*}
\langle \vf_k , \varsigmaf_2 (\thetaf - \varthetaf) \rangle_2
&= \frac1{K} \, \Bigl\langle \vf_k , \varsigmaf_2 \sum_{k'=1}^K \alpha_{k'} \vf_{k'} \Bigr\rangle_2
= \frac1{K} \, \sum_{k'=1}^K \alpha_{k'} \langle \vf_k , \varsigmaf_2 \vf_{k'} \rangle_2
\\
&= \frac1{K} \, \sum_{k'=1}^K \alpha_{k'} \lambda_{k'} \langle \vf_k , \vf_{k'} \rangle_2
= \frac1{K} \, \alpha_k \lambda_k
\\
&= \frac1{K} \, \begin{cases}
- \delta \lambda_k 
& , \, \thetaf \in Q_k^-
\\
\delta \lambda_k 
& , \, \thetaf \in Q_k^+
\end{cases} \,
.
\end{align*}
From Lemma~\ref{P2a:lem_lower_bound_quadratic_form_Fisher_info}, 
we obtain that there exist $\Lambda > 0$ and $N_2 \in \Nbb$ such that for all $N \geq N_2$, we have $\varsigmaf_2 \preceq -\Lambda \, \If_{K}$ 
which implies $\lambda_k \leq -\Lambda$ for all $k = 1, \ldots , K$. \citep[Corollary 7.7.4 (c)]{horn2012}.
Set $N_0 := \max \{N_1, N_2\}$ and let $N \geq N_0$ in the following.
Then, 
\begin{align*}
\frac{\delta \, \Lambda}{K}
\leq -\frac{\delta \, \lambda_k}{K}
= \begin{cases}
\langle \vf_k , \varsigmaf_2 (\thetaf - \varthetaf) \rangle_2 & , \, \thetaf \in Q_k^-
\\
- \langle \vf_k , \varsigmaf_2 (\thetaf - \varthetaf) \rangle_2 & , \, \thetaf \in Q_k^+
\end{cases}
\, .
\end{align*}
Furthermore, for data in $\Psi_{N}$, $\thetaf \in Q_k^- \cup Q_k^+$ and $\delta < \frac{\Lambda}{K \, (1 + M)}$, using the Cauchy-Schwarz inequality \citep[e.g.,][Theorem 5.1.4]{horn2012} yields\footnote{If $\delta \geq \frac{\Lambda}{K \, (1 + M)}$, consider $\deltat < \frac{\Lambda}{K \, (1 + M)}$ and adapt $N_1$ above if necessary. We then effectively show that \eqref{P2a:eq_score_equation_taylor} has a solution in $Q_{\deltat} (\varthetaf) \subset Q_\delta (\varthetaf)$.}
\begin{align*}
&\, \Bigl\vert \Bigl\langle \vf_k , \varsigmaf_1 + \frac{M}2 (\thetaf - \varthetaf)^\top \Ecalf (\thetaf - \varthetaf) \Bigr\rangle_2 \Bigr\vert
\\
\leq &\, 
\Vert \vf_k \Vert_2 \, 
\Bigl\Vert \varsigmaf_1 + \frac{M}2 (\thetaf - \varthetaf)^\top \Ecalf (\thetaf - \varthetaf) \Bigr\Vert_2 
\\
\leq &\, 
\Vert \varsigmaf_1 \Vert_2 + \frac{M}2 \Bigl\Vert \Bigl(\frac1{K} \, \sum_{k'=1}^K \alpha_{k'} \vf_{k'} \Bigr)^\top \Ecalf \Bigl(\frac1{K} \, \sum_{k'=1}^K \alpha_{k'} \vf_{k'} \Bigr) \Bigr\Vert_2
\\
\overset{(\ast)}{<} &\, \delta^2 + \frac{M}2 \Bigl(\frac1{K} \, \sum_{k'=1}^K | \alpha_{k'} |  \, \Vert \Ecalf \Vert_{\mathrm{op}} \, \frac1{K} \, \sum_{k'=1}^K | \alpha_{k'} | \Bigr)
\\
\leq &\, \delta^2 + M \, \delta^2 
\leq (1 + M \, ) \delta^2
< \frac{\Lambda \, \delta}{K}
\leq 
\begin{cases}
\langle \vf_k , \varsigmaf_2 (\thetaf - \varthetaf) \rangle_2 & , \, \thetaf \in Q_k^-
\\
- \langle \vf_k , \varsigmaf_2 (\thetaf - \varthetaf) \rangle_2 & , \, \thetaf \in Q_k^+
\end{cases}
\, ,
\end{align*}
where we used in $(\ast)$ that 
$\Vert \Wcalf \Vert_{\mathrm{op}} = \sup_{\vf, \wf \in \Rbb^K \setminus\{\mathbf{0}\}} \frac{\Vert \vf^\top \Wcalf \wf \Vert_2}{\Vert \vf \Vert_2 \, \Vert \wf \Vert_2}$ yields 
$\Vert \wf^\top \Wcalf \wf \Vert_2$ $\leq \Vert \wf \Vert_2 \, \Vert \Wcalf \Vert_{\mathrm{op}} \, \Vert \wf \Vert_2$ for all $\Wcalf \in \Rbb^{K \times K \times K}$ and $\wf \in \Rbb^K$.
In particular, the inequality above yields $\langle \vf_k , \varsigmaf_1 + \frac{M}2 (\thetaf - \varthetaf)^\top \Ecalf (\thetaf - \varthetaf) \rangle_2 < - \langle \vf_k , \varsigmaf_2 (\thetaf - \varthetaf) \rangle_2$ for $\thetaf \in Q_k^+$ and
$- \langle \vf_k , \varsigmaf_1 + \frac{M}2 (\thetaf - \varthetaf)^\top \Ecalf (\thetaf - \varthetaf) \rangle_2 < \langle \vf_k , \varsigmaf_2 (\thetaf - \varthetaf) \rangle_2$ for $\thetaf \in Q_k^-$, 
which is equivalent to 
$\langle \vf_k , \varsigmaf_1 + \frac{M}2 (\thetaf - \varthetaf)^\top \Ecalf (\thetaf - \varthetaf) \rangle_2 > - \langle \vf_k , \varsigmaf_2 (\thetaf - \varthetaf) \rangle_2$.
Thus,
\begin{align*}
& \, \Bigl\langle \vf_k , \varsigmaf_1 + \varsigmaf_2 (\thetaf - \varthetaf) + \frac{M}2 (\thetaf - \varthetaf)^\top \Ecalf (\thetaf - \varthetaf) \Bigr\rangle_2
\\
=& \, \Bigl\langle \vf_k , \varsigmaf_1 + \frac{M}2 (\thetaf - \varthetaf)^\top \Ecalf (\thetaf - \varthetaf) \Bigr\rangle_2
+ \langle \vf_k , \varsigmaf_2 (\thetaf - \varthetaf) \rangle_2
~
\begin{cases}
> 0 & , \, \thetaf \in Q_k^-
\\
< 0 & , \, \thetaf \in Q_k^+
\end{cases}
\, ,
\end{align*}
and the Poincar\'e-Miranda Theorem yields that \eqref{P2a:eq_score_equation_taylor} 
has a solution in $Q_\delta (\varthetaf)$. 
This proves statement \ref{P2a:claim_solution_in_cube}.
Since 
$Q_{\delta} (\varthetaf) \subseteq \{\thetaf \in \Rbb^K 
~|~ \Vert \thetaf - \varthetaf \Vert_2 < \delta \}$ this also implies consistency of the PMLE.
%
\item
First note that $\xi_{\max} = \ocal_{\Pbb}(\sqrt{N})$ implies $\xi_{\max} = \ocal_{\Pbb}(N)$ and thus, we have asymptotic existence and consistency of the PMLE $\thetafh$ by \ref{P2:thm_item_PMLE_cosistent}.
To show the asymptotic normality, we consider $\Ff_\alpha (\varthetaf) := \Ff (\varthetaf) + \alpha \, \Pf$ for $\alpha \geq 0$, 
which yields the Fisher information $\Ff_{\mathrm{pen}} (\varthetaf)
$ for $\alpha = 2$.
For any $\alpha \geq 0$,
$\Ff_\alpha (\varthetaf)$ is symmetric and positive definite by Proposition~\ref{P2a:prop_appendix_fisher_informations_positive_definite}. 
Thus, there exists a unique symmetric and positive definite 
matrix $\Ff_\alpha (\varthetaf)^{\nicefrac12} $ such that $\Ff_\alpha (\varthetaf)^{\nicefrac12}  \Ff_\alpha (\varthetaf)^{\nicefrac12}  = (\Ff_\alpha (\varthetaf)^{\nicefrac12} )^\top \Ff_\alpha (\varthetaf)^{\nicefrac12}  = \Ff_\alpha (\varthetaf)$ (\citealp[e.g.,][Theorem 7.2.6]{horn2012}\,\footnote{The Theorem actually assumes a symmetric and positive semi-definite matrix $\Af \in \Rbb^K$ and yields existence of a unique symmetric positive semi-definite square root $\Af^{\nicefrac12}$. However if $\Af$ is even strictly positive definite, one can prove identically that also the unique $\Af^{\nicefrac12}$ is strictly positive definite.}).
Since it is positive definite, its inverse $\Ff_\alpha (\varthetaf)^{-\nicefrac12}  := (\Ff_\alpha (\varthetaf)^{\nicefrac12} )^{-1}$ exists.
Considering~\eqref{P2a:eq_score_equation_taylor} for $\thetaf = \thetafh$, we obtain after 
reorganizing, 
\begin{align*}
(\thetafh - \varthetaf)
&= \Bigl( \Ff_{\mathrm{pen}} (\varthetaf) - \frac{NM}{2} (\thetafh - \varthetaf)^\top \Ecalf \Bigr)^{-1} N\varsigmaf_1 
. 
\end{align*}
Multiplying from the left with $\Ff_{\mathrm{pen}} (\varthetaf)^{\nicefrac12}$, inserting $\If_K = \Ff_{\mathrm{pen}} (\varthetaf)^{\nicefrac12}  \, \Ff_{\mathrm{pen}} (\varthetaf)^{-\nicefrac12}$ before $\varsigmaf_1$, 
and reorganizing yields 
\begin{align}
& \Ff_{\mathrm{pen}} (\varthetaf)^{\nicefrac12}  (\thetafh - \varthetaf) \notag
\\
= ~&\Bigl( \Ff_{\mathrm{pen}} (\varthetaf)^{-\nicefrac12}  \bigl( \Ff_{\mathrm{pen}} (\varthetaf) - \frac{N \, M}{2} (\thetafh - \varthetaf)^\top \Ecalf \bigr) \Ff_{\mathrm{pen}} (\varthetaf)^{-\nicefrac12}  \Bigr)^{-1} 
\, \Ff_{\mathrm{pen}} (\varthetaf)^{-\nicefrac12}  \, N \, \varsigmaf_1 
, \notag
\\ 
%
\intertext{and equivalently}
&\Ff_{\mathrm{pen}} (\varthetaf)^{\nicefrac12}  (\thetafh - \varthetaf) \notag
\\
=~ &\Bigl( 
\If_K - \frac{N \, M}{2} \Ff_{\mathrm{pen}} (\varthetaf)^{-\nicefrac12}  (\thetafh - \varthetaf)^\top \Ecalf \Ff_{\mathrm{pen}} (\varthetaf)^{-\nicefrac12}  
\Bigr)^{-1} 
\Ff_{\mathrm{pen}} (\varthetaf)^{-\nicefrac12}  \, N \, \varsigmaf_1 
\label{P2a:eq_asymptotic_distribution_factorization_zwischenschritt}
.
\end{align}
Here, the matrix $\Jf := \If_K - \frac{N \, M}{2} \Ff_{\mathrm{pen}} (\varthetaf)^{-\nicefrac12}  (\thetafh - \varthetaf)^\top \Ecalf \Ff_{\mathrm{pen}} (\varthetaf)^{-\nicefrac12}$ (and equivalently $\Ff_{\mathrm{pen}} (\varthetaf) - \frac{N \, M}{2} (\thetafh - \varthetaf)^\top \Ecalf$) 
is invertible with probability tending to one for $N \ra \infty$.
To see this, we first show for $\Ff_\alpha (\varthetaf) 
$ and in particular for $\Ff_{\mathrm{pen}} (\varthetaf) = \Ff_2 (\varthetaf)$:
\begin{enumerate}
\nitem{(F)}\label{P2a:claim_inverse_root_Fisher_info_bounded}
For every $\alpha \geq 0$, there exist $N_3 \in \Nbb$ and $C>0$ such that for all $N \geq N_3$:
\[
\Vert \Ff_\alpha (\varthetaf) ^{-\nicefrac12} \Vert_{\mathrm{op}} \leq \frac1{\sqrt{N \, C}}.
\]
\end{enumerate}
Since $\Ff_\alpha (\varthetaf)^{\nicefrac12}$ is symmetric, also $\Ff_\alpha (\varthetaf)^{-\nicefrac12}$ is symmetric and thus, denoting the set of eigenvalues of a matrix $\Af \in \Rbb^{K \times K}$ with $\sigma(\Af)$, $\max \sigma (\Ff_\alpha (\varthetaf)^{-\nicefrac12}) = \Vert \Ff_\alpha (\varthetaf)^{-\nicefrac12} \Vert_{\mathrm{op}}$ \citep[Theorem~11.28~(b)]{rudin1991}. 
Furthermore, 
$\sigma (\Ff_\alpha (\varthetaf)^{-\nicefrac12}) = \{ \frac1{\sqrt{\lambda}} ~|~ \lambda \in \sigma (\Ff_\alpha(\varthetaf)) \}$ due to the spectral mapping theorem \citep[Theorem~10.28~(b)]{rudin1991}.
From Lemma~\ref{P2a:lem_lower_bound_quadratic_form_Fisher_info}, 
there exist $C > 0$ and $N_3 \in \Nbb$ such that for $N \geq N_3$, we have $N \, C \, \Id \preceq \Ff_\alpha (\varthetaf)$ and thus $N \, C \leq \lambda$ for all $\lambda \in \sigma (\Ff_\alpha(\varthetaf))$~\citep[Corollary 7.7.4~(c)]{horn2012}.
This implies $\frac1{\sqrt{\lambda}} \leq \frac1{\sqrt{N \, C}}$ for all $\lambda \in \sigma (\Ff_\alpha(\varthetaf))$ and in particular $\Vert \Ff_\alpha (\varthetaf)^{-\nicefrac12} \Vert_{\mathrm{op}} = \max \sigma (\Ff_\alpha (\varthetaf)^{-\nicefrac12}) \leq \frac1{\sqrt{N \, C}}$, which proves statement~\ref{P2a:claim_inverse_root_Fisher_info_bounded}.
\\
Noting that $\Vert \Af \Bf \Vert_{\mathrm{op}} \leq \Vert \Af  \Vert_{\mathrm{op}} \Vert \Bf \Vert_{\mathrm{op}}$ for $\Af, \Bf \in \Rbb^{K \times K}$ \citep[Theorem~5.6.2~(c)]{horn2012} 
and recalling that $\Vert \Ecalf \Vert_{\mathrm{op}} \leq 1$, 
we then obtain for $N \geq N_3$,
\begin{align}
\Bigl\Vert \frac{N \, M}{2} \Ff_{\mathrm{pen}} (\varthetaf)^{-\nicefrac12}  (\thetafh - \varthetaf)^\top \Ecalf \Ff_{\mathrm{pen}} (\varthetaf)^{-\nicefrac12} \Bigr\Vert_{\mathrm{op}} 
& \leq \frac{N \, M}{2} \frac1{\sqrt{N \, C}} \Vert \thetafh - \varthetaf \Vert_2  \frac1{\sqrt{N \, C}}
\notag
\\
&= \frac{M}{2 \, C} \Vert \thetafh - \varthetaf \Vert_2,
\label{P2a:eq_Abschaetzung_Nenner_asymptotische_NV}
\end{align}
and thus
\begin{align*}
&\Pbb_{\varthetaf} \Bigl( \Bigl\Vert \frac{N \, M}{2} \Ff_{\mathrm{pen}} (\varthetaf)^{-\nicefrac12}  (\thetafh - \varthetaf)^\top \Ecalf \Ff_{\mathrm{pen}} (\varthetaf)^{-\nicefrac12} \Bigr\Vert_{\mathrm{op}} < 1 \Bigr)
\\ \geq ~&
\Pbb_{\varthetaf} \Bigl( \frac{M}{2 \, C} \Vert \thetafh - \varthetaf \Vert_2 < 1 \Bigr)
= \Pbb_{\varthetaf} \Bigl( \Vert \thetafh - \varthetaf \Vert_2 < \frac{2 \, C}{M} \Bigr)
\xra[N \ra \infty]{} 1
,
\end{align*}
since $\thetafh \xra[N \ra \infty]{\Pbb} \varthetaf $.
The invertibility of $\Jf 
$ with probability tending to one 
then follows via convergence of the Neumann series \citep[Statement~5.7]{alt2016}. 
Using $N \, \varsigmaf_1 = \nabla \ell (\varthetaf) 
- 2 \, \Pf \varthetaf$ and multiplying with $\If_K = \Ff (\varthetaf)^{\nicefrac12} \, \Ff (\varthetaf)^{-\nicefrac12}$ 
in \eqref{P2a:eq_asymptotic_distribution_factorization_zwischenschritt}, we obtain the equivalent equation
\begin{align}
\Ff_{\mathrm{pen}} (\varthetaf)^{\nicefrac12}  (\thetafh - \varthetaf)
&= \Jf^{-1} 
\Ff_{\mathrm{pen}} (\varthetaf)^{-\nicefrac12} \, \Ff (\varthetaf)^{\nicefrac12} \, \Ff (\varthetaf)^{-\nicefrac12} \, \nabla \ell (\varthetaf) \notag
\\
&\hspace{0.4cm}- \Jf^{-1} 
\Ff_{\mathrm{pen}} (\varthetaf)^{-\nicefrac12} \, 2 \, \Pf \varthetaf
\label{P2a:eq_asymptotic_distribution_factorization}
\end{align}

We now show
\begin{enumerate}[label=(\alph*)]
\item\label{P2a:item_inverse_matrix_convergence_in_probability}
$\Jf = \If_K - \frac{N \, M}{2} \Ff_{\mathrm{pen}} (\varthetaf)^{-\nicefrac12}  (\thetafh - \varthetaf)^\top \Ecalf \Ff_{\mathrm{pen}} (\varthetaf)^{-\nicefrac12} \xra[N \ra \infty]{\Pbb} \If_K$,
\item\label{P2a:item_Fpen_F_convergence_in_probability}
$\Ff_{\mathrm{pen}} (\varthetaf)^{-\nicefrac12} \, \Ff (\varthetaf)^{\nicefrac12} \xra[N \ra \infty]{\Pbb} \If_K$,
\item\label{P2a:item_convergence_in_distribution}
$\Ff (\varthetaf)^{-\nicefrac12} \, \nabla \ell (\varthetaf)
\xra[N \ra \infty]{\mathrm{D}} ~ N(\mathbf{0}, \If_K )$, 
\item\label{P2a:item_remainder_convergence_in_probability}
$\Jf^{-1} 
\Ff_{\mathrm{pen}} (\varthetaf)^{-\nicefrac12} \, 2 \, \Pf \varthetaf \xra[N \ra \infty]{\Pbb} \mathbf{0}$,
\end{enumerate}
which then with \eqref{P2a:eq_asymptotic_distribution_factorization} and Slutsky's thoerem \citep[e.g.,][Lemma 2.8
]{vaart1998} 
yields 
$\Ff_{\mathrm{pen}} (\varthetaf)^{\nicefrac12} (\thetafh - \varthetaf) 
\xra[N \ra \infty]{\mathrm{D}} 
N(\mathbf{0}, \If_K)$, 
i.e., $\thetafh 
\overset{\text{a}}{\sim} \mathcal{N} (\varthetaf, \Ff_{\mathrm{pen}} (\varthetaf)^{-1})$.
Proof of \ref{P2a:item_inverse_matrix_convergence_in_probability} to \ref{P2a:item_remainder_convergence_in_probability}:

\begin{enumerate}[label=(\alph*)]
\item
Using \eqref{P2a:eq_Abschaetzung_Nenner_asymptotische_NV} 
and $\thetafh \xra[N \ra \infty]{\Pbb} \varthetaf $ we obtain for any $\varepsilon > 0$,
\begin{align*}
&\, \Pbb_{\varthetaf} \Bigl( \Bigl\Vert \If_K - \frac{N \, M}{2} \Ff_{\mathrm{pen}} (\varthetaf)^{-\nicefrac12}  (\thetafh - \varthetaf)^\top \Ecalf \Ff_{\mathrm{pen}} (\varthetaf)^{-\nicefrac12}  - \If_K \Bigr\Vert_{\mathrm{op}} \geq \varepsilon \Bigr)
\\
\leq \, & \, \Pbb_{\varthetaf} \Bigl( \frac{M}{2 \,C} \Vert \thetafh - \varthetaf \Vert_2  \geq \varepsilon \Bigr) 
= \, \Pbb_{\varthetaf} \Bigl( \Vert \thetafh - \varthetaf \Vert_2 \geq \frac{2\varepsilon \, C}{M} \Bigr) 
\xra[N \ra \infty]{} 0
%
.
\end{align*}
\item
With $\Ff_{\mathrm{pen}} (\varthetaf) = \Ff (\varthetaf) + 2 \, \Pf$, we obtain
\begin{align*}
&\, \bigl\Vert \Ff_{\mathrm{pen}} (\varthetaf)^{-\nicefrac12} \, \Ff (\varthetaf)^{\nicefrac12} - \If_K \bigr\Vert_{\mathrm{op}}
\\
= &\, \bigl\Vert \left( (\Ff (\varthetaf) + 2 \, \Pf)^{-\nicefrac12} - \Ff (\varthetaf)^{-\nicefrac12} \right) \, \Ff (\varthetaf)^{\nicefrac12} \bigr\Vert_{\mathrm{op}} \notag
\\
= &\, \Bigl\Vert \Bigl( \int_0^1 \frac{\mathrm{d}}{\dt} \left( (\Ff (\varthetaf) + t \, 2 \, \Pf)^{-\nicefrac12} \right) \, \dt \Bigr) \, \Ff (\varthetaf)^{\nicefrac12} \Bigr\Vert_{\mathrm{op}} \notag
\\
= &\, \Bigl\Vert  \Bigl( \int_0^1 \, (\Ff (\varthetaf) + t \, 2 \, \Pf)^{-\nicefrac32} \, \dt \Bigr)  \, \Pf \, \Ff (\varthetaf)^{\nicefrac12} \Bigr\Vert_{\mathrm{op}} \notag
\\
\leq &\, \max_{t \in [0, 1]} \Vert (\Ff (\varthetaf) + t \, 2 \, \Pf)^{-\nicefrac12} \Vert_{\mathrm{op}}^3  \, \Vert \Pf \Vert_{\mathrm{op}} \, \Vert \Ff (\varthetaf)^{\nicefrac12} \Vert_{\mathrm{op}}.
\end{align*}
By~\ref{P2a:claim_inverse_root_Fisher_info_bounded}, there exists $N_3 \in \Nbb$ and $C > 0$, such that
\begin{align*}
\max_{t \in [0, 1]} \Vert (\Ff (\varthetaf) + t \, 2 \, \Pf)^{-\nicefrac12} \Vert_{\mathrm{op}}^3 \leq \frac{1}{\sqrt{N \, C}^3}
\end{align*} 
for all $N \geq N_3$. 
Furthermore, we have $\Vert \Ff (\varthetaf)^{\nicefrac12} \Vert_{\mathrm{op}} = \Vert \Ff (\varthetaf) \Vert_{\mathrm{op}}^{\nicefrac12}$, which follows from the spectral norm of a symmetric matrix being equal to its maximal eigenvalue \citep[Theorem~11.28~(b)]{rudin1991} and the spectral mapping theorem \citep[Theorem~10.28~(b)]{rudin1991}, and
%
%
\begin{align*}
\Vert \Ff (\varthetaf) \Vert_{\mathrm{op}}
&= \Bigl\Vert \sum_{i=1}^N D^2 \log f_i (\varthetaf) \Bigr\Vert_{\mathrm{op}} 
\leq N \, \max_{i \in \{1, \ldots , N\}} \Vert D^2 \log f_i (\varthetaf) \Vert_{\mathrm{op}} 
.
\end{align*}
Recalling from~\eqref{P2a:eq_hessian_bayes} that
$\frac{\partial^2}{\partial \theta_k \partial \theta_{k'}} \log f (\varthetaf)$
is a sum of products, whose factors have the form
$ 
\pm \frac{ \int_\Ycal \exp [ \bft (\xf)^\top \varthetaf ] \, \prod_{j \in \{k, k'\}} (\bft (\xf)_{[j]})^{\delta_j} \, \dmu}{\int_\Ycal \exp [ \bft (\xf)^\top \varthetaf ] \, \dmu},
$ 
where $\delta_j \in \{ 0, 1\}$ and $\bft (\xf)_{[j]}$ denotes the $j$-the entry of the vector $ \bft (\xf)$,
we obtain that
there exists an $M > 0$ such that $\sup_{\xf \in \Xcal} \Vert D^2 \log f (\varthetaf) 
\Vert_{\mathrm{op}} \leq M$, since the basis functions contained in $\bft_\Ycal$ and $\bfe_{\Xcal}$ are bounded and $\mu$ is finite.
Hence, $\Vert \Ff (\varthetaf)^{\nicefrac12} \Vert_{\mathrm{op}} \leq \sqrt{N \, M}$.
All in all, we have for every $\varepsilon > 0$ and $N \geq N_3$,
\begin{align*}
& \Pbb_{\varthetaf} \Bigl( \Bigl\Vert \Ff_{\mathrm{pen}} (\varthetaf)^{-\nicefrac12} \, \Ff (\varthetaf)^{\nicefrac12} - \If_K \Bigr\Vert_{\mathrm{op}} \geq \varepsilon \Bigr)
\\
\leq ~&
\Pbb_{\varthetaf} \Bigl( \frac{1}{\sqrt{N \, C}^3} \, \Vert \Pf \Vert_{\mathrm{op}}  \, \sqrt{N \, M} \geq \varepsilon \Bigr)
\\
= ~& 
\Pbb_{\varthetaf} \Bigl( \frac{1}{N} \, \Vert \Pf \Vert_{\mathrm{op}} \geq \frac{\sqrt{C^3} \, \varepsilon}{\sqrt{M}} \Bigr)
\xra[N \ra \infty]{} 0,
\end{align*}
by Lemma~\ref{P2a:lem_assumption_smoothing_parameters_implies_convergence_penalty}, since $\xi_{\max} = \ocal_{\Pbb} (N)$.
Thus, $\Ff_{\mathrm{pen}} (\varthetaf)^{-\nicefrac12} \, \Ff (\varthetaf)^{\nicefrac12} \xra[N \ra \infty]{\Pbb} \If_K$.
\item
Recall that $\nabla \ell (\varthetaf) = \sum_{i=1}^N \nabla \log f_i (\varthetaf)$.
We have $\Ebb_{\varthetaf} (\nabla \log f_i (\varthetaf) ) = \mathbf{0}$ for all $i = 1, \ldots ,N$ and
$\Var (\sum_{i=1}^N \nabla \log f_i (\varthetaf)) = \Ff (\varthetaf)$ from Lemma~\ref{P2a:lem_score_convergesP_zero}.
Thus, the multivariate version of the Berry–Esseen bound \citep[e.g.,][Theorem 1.1]{bentkus2005} yields
\begin{align}
&\sup_{A \in \Ccal} \Bigl\vert \Pbb_{\varthetaf} \Bigl( \Ff (\varthetaf)^{-\nicefrac12} \sum_{i=1}^N \nabla \log f_i (\varthetaf) \in A \Bigr) - \Pbb_{\varthetaf} (\Zf \in A) \Bigr\vert \notag
\\
\leq &\, c \, K^{\nicefrac14} \, \sum_{i=1}^N \Ebb_{\varthetaf} (\Vert \Ff (\varthetaf)^{-\nicefrac12}  \nabla \log f_i (\varthetaf) \Vert_2^3),
\label{P2a:eq_berry-esseen-bound_score_function}
\end{align}
where $\Ccal$ is the set of all convex subsets of 
$\Rbb^K 
$, $\Zf$ is a $K$-dimensional multivariate standard normally distributed random vector, 
and $c$ is some positive constant.
We now show
$\lim_{N \ra \infty} \sum_{i=1}^N \Ebb_{\varthetaf} (\Vert \Ff (\varthetaf)^{-\nicefrac12}  \nabla \log f_i (\varthetaf) \Vert_2^3) = 0$, which together with \eqref{P2a:eq_berry-esseen-bound_score_function} yields \ref{P2a:item_convergence_in_distribution}, since $\{ (v_1, \ldots , v_K) \in \Rbb^K 
~|~ v_1 \leq z_1 , \ldots , v_K \leq z_K \}$ is a convex subset of 
$\Rbb^K 
$ for all $(z_1, \ldots , z_K) \in \Rbb^K$. 
%
%
Note that since $\bfe_{\Xcal}$ and $\bft_{\Ycal}$ are bounded and $\mu$ is finite, there exists a $C_2 > 0$ such that for all $\xf_i \in \Xcal$,
\begin{align*}
&\, \Ebb_{\varthetaf} \bigl( \left\Vert \nabla \log f_{\xf_i, \varthetaf} (Y_i) \right\Vert_2^3 \bigr)
\\
= &\, \int \Bigl( \sum_{k=1}^K \Bigl( \bft (\xf_i)_{[k]} - \int_{\Ycal} f_{\xf_i, \varthetaf} \, \bft (\xf_i)_{[k]} \, \dmu \Bigr)^2 \Bigr)^{\frac32} \, f_{\xf_i , \varthetaf} \, \dmu
\leq C_2 ,
\end{align*}
where the partial derivatives are computed using~\eqref{P2a:eq_first_derivative_C_i}.
Using~\ref{P2a:claim_inverse_root_Fisher_info_bounded}, we obtain that there exist $N_3 \in \Nbb$ and $C>0$, such that $\Vert \Ff (\varthetaf)^{-\nicefrac12} \Vert_{\mathrm{op}} \leq \frac1{\sqrt{N \, C}}$ for $N \geq N_3$.
Thus, we have
\begin{align*}
\sum_{i=1}^N \Ebb_{\varthetaf} \bigl( \left\Vert \Ff (\varthetaf)^{-\nicefrac12}  \nabla \log f_i (\varthetaf) \right\Vert_2^3 \bigr)
&\leq
\left\Vert \Ff (\varthetaf)^{-\nicefrac12} \right\Vert_{\mathrm{op}}^3 \, \sum_{i=1}^N \Ebb_{\varthetaf} \bigl( \left\Vert  \nabla \log f_i (\varthetaf) \right\Vert_2^3 \bigr)
\\
&\leq
\frac{1}{\sqrt{N \, C}^3} \, N \, C_2
= \frac{1}{\sqrt{N \, C^3}}\, C_2
\xra[N \ra \infty]{} 0.
\end{align*}
\item
From \ref{P2a:item_inverse_matrix_convergence_in_probability}, we have $\Jf \xra[N \ra \infty]{\Pbb} \If_K$. 
Since the matrix inversion function 
is continuous (compare Lemma~\ref{P2a:lem_matrix_inversion_continuous}), 
we obtain $\Jf^{-1} \xra[N \ra \infty]{\Pbb} \If_K$ by the continuous mapping theorem \citep[e.g.,][Theorem~2.3]{vaart1998}. 
Furthermore, by~\ref{P2a:claim_inverse_root_Fisher_info_bounded}, there exist $N_3 \in \Nbb$ and $C>0$, such that 
for $N \geq N_3$
\begin{align*}
\left\Vert 
\Ff_{\mathrm{pen}} (\varthetaf)^{-\nicefrac12} \, 2 \, \Pf \varthetaf \right\Vert_2
&\leq 
2 \, \left\Vert \Ff_{\mathrm{pen}} (\varthetaf)^{-\nicefrac12} \right\Vert_{\mathrm{op}} 
\, \left\Vert \Pf \varthetaf \right\Vert_2
\leq 
\frac{2}{\sqrt{N \, C}} \left\Vert \Pf \varthetaf \right\Vert_2 ,
\end{align*}
Thus, using $\xi_{\max} = \ocal_{\Pbb} (\sqrt{N})$, we obtain for any $\varepsilon > 0$,
\begin{align*}
\Pbb_{\varthetaf} \left( \left\Vert 
\Ff_{\mathrm{pen}} (\varthetaf)^{-\nicefrac12} \, 2 \, \Pf \varthetaf \right\Vert_2
\geq \varepsilon \right)
\leq
\Pbb_{\varthetaf} \left( 
\frac1{\sqrt{N}} \,
\left\Vert 
\Pf \varthetaf \right\Vert_2
\geq \frac{\sqrt{C} \, \varepsilon}2 \right)
\xra[N \ra \infty]{} 0,
\end{align*}
by 
Lemma~\ref{P2a:lem_assumption_smoothing_parameters_implies_convergence_penalty}, i.e., $\Ff_{\mathrm{pen}} (\varthetaf)^{-\nicefrac12} \, 2 \, \Pf \varthetaf \xra[N \ra \infty]{\Pbb} \mathbf{0}$.
Since convergence in distribution and convergence in probability are equivalent, if the limiting random vector is constant \citep[e.g.,][Theorem 2.7~(iii)]{vaart1998}, this yields
$ 
\Jf^{-1}\Ff_{\mathrm{pen}} (\varthetaf)^{-\nicefrac12} \, 2 \, \Pf \varthetaf \xra[N \ra \infty]{\Pbb} \If_K \mathbf{0} = \mathbf{0}.
$ 
by Slutsky's thoerem \citep[e.g.,][Lemma 2.8
]{vaart1998}.
\qedhere
\end{enumerate}
\end{enumerate}

\end{proof}

\begin{proof}[Proof of Lemma~\ref{P2:lemma_confidence_regions}]
For the matrix $\Af$ as in Lemma~\ref{P2:lemma_confidence_regions}, set $\Vf_\Ycal := \Af \Ff_{\mathrm{pen}}(\varthetaf) ^{-1}
\Af^\top$.
From Theorem~\ref{P2:thm_MLE_asymptotic_normal}, we have 
$\thetafh 
\overset{\text{a}}{\sim} \mathcal{N} (\varthetaf, (\Ff_{\mathrm{pen}}(\varthetaf))^{-1}) 
$, 
which yields
$
\thetafh_\Ycal 
\overset{\mathrm{a}}{\sim} 
\mathcal{N} (\varthetaf_\Ycal, \Vf_\Ycal)
$ 
and thus, $(\thetafh_\Ycal - \varthetaf_\Ycal)^\top \Vf_\Ycal^{-1} (\thetafh_\Ycal - \varthetaf_\Ycal) 
\overset{\mathrm{a}}{\sim} \chi^2( K_\Ycal)$. 
Hence, $CR_{\varthetaf_\Ycal} := \{ \thetaf_\Ycal \in \Rbb^{K_\Ycal} ~|~ (\thetaf_\Ycal - \thetafh_\Ycal)^\top \Vfh_\Ycal^{-1} (\thetaf_\Ycal - \thetafh_\Ycal) \leq \chi^2_{\alpha} (K_\Ycal)\}$ is an asymptotic $\alpha \cdot 100\%$ confidence 
region for $\varthetaf_\Ycal$. 
Since the basis functions contained in the vector $\bft_\Ycal$ are linearly independent, 
$\spano (\bft_\Ycal) = \spano (\tilde{b}_{\Ycal, 1}, \ldots , \tilde{b}_{\Ycal, K_\Ycal})$ is a $K_\Ycal$-dimensional subspace of $L^2_0 (\mu)$ and thus is isomorphic to $\Rbb^{K_\Ycal}$ via the mapping $\Phi : \spano (\bft_\Ycal) \ra \Rbb^{K_\Ycal}, \bft_\Ycal^\top \thetaf_\Ycal 
\mapsto \thetaf_\Ycal$. 
Then, 
$
\{ \bft_\Ycal^\top \thetaf_\Ycal \in \spano (\bft_{\Ycal}) ~|~ \thetaf_\Ycal \in CR_{\varthetaf_\Ycal} \}$ is an $\alpha \cdot 100\%$ confidence 
region for 
\begin{align*}
\bft_\Ycal^\top \varthetaf_\Ycal
&= \bft_\Ycal^\top \Af \varthetaf
= \bft_\Ycal^\top \Bigl(\sum_{j \in \Jcal} \bigl( \bfe_{\Xcal, \, j} (\xf)^\top \otimes \Id_{K_\Ycal} \bigr) \Sf_{j} \Bigr) \varthetaf
\\
&= \sum_{j \in \Jcal} \bft_\Ycal^\top  \bigl( \bfe_{\Xcal, \, j} (\xf)^\top \otimes \Id_{K_\Ycal} \bigr) \varthetaf_{
j}
= \sum_{j \in \Jcal} \bigl( \bfe_{\Xcal, \, j}(\xf) \otimes \bft_{\Ycal} \bigr)^\top \varthetaf_{
j} 
= \clr[h_{\Jcal}(\xf)].
\end{align*} 
Using that the inverse clr transformation is ismorphic, as well, 
we obtain that
$
\{ \bfe_\Ycal^\top \thetaf_\Ycal \in \spano (b_{\Ycal, 1}, \ldots , b_{\Ycal, K_\Ycal}) ~|~ \thetaf_\Ycal \in CR_{\varthetaf_\Ycal} \}$ is an $\alpha \cdot 100\%$ confidence 
region for $\bfe_\Ycal^\top \varthetaf_\Ycal = h_{\Jcal}(\xf)$.

Denoting the probability measure 
corresponding to the $\chi^2(K_{\Ycal})$ distribution by $\Pbb_{\chi^2(K_{\Ycal})}$ 
and considering a random variable $T \sim \chi^2(K_{\Ycal})$, 
we obtain asymptotic p-values for $\varthetaf_{\Ycal} = \mathbf{0}$ and equivalently for $h_{\Jcal}(\xf) = 0 \in \B$ and for $\clr[h_{\Jcal}(\xf)] = 0 \in L^2_0(\mu)$ via
\begin{align*}
\Pbb_{\chi^2(K_{\Ycal})} \left(T \geq (\thetafh_{\Ycal} - \varthetaf_{\Ycal})^\top \Vfh_{\Ycal}^{-1} (\thetafh_{\Ycal} - \varthetaf_{\Ycal}) ~|~ \varthetaf_{\Ycal} = \mathbf{0} \right)
&= 1 - F_{\chi^2(K_{\Ycal})} (\thetafh_{\Ycal}^\top \Vfh_{\Ycal}^{-1} \thetafh_{\Ycal}). 
\hspace{0.3cm}
\qedhere
\end{align*}
\end{proof}

\begin{lem}\label{P2a:lemma_confidence_regions_simultaneos}
Consider the $j$-th partial effect as a function of the covariates, i.e., $h_j : \Xcal \ra \B$ with $h_j = (\bfe_{\Xcal, \, j} \ootimes \bfe_{\Ycal})^\top \varthetaf_j$.
For $K_j := K_{\Xcal, \, j} K_\Ycal$, denote the matrix mapping the vector $\thetafh
$ to its $j$-th subvector $\thetafh_{
j}$ with $\Sf_j \in \Rbb^{K_j \times K}$ (i.e.,
$\varthetaf_j = \Sf_j \varthetaf
$,
$\thetafh_j = \Sf_j \thetafh
$)
and set
$
\Vfh_j := \Sf_j \Ff_{\mathrm{pen}}(\thetafh) ^{-1} \Sf_j^\top$.
Assume~\ref{P2a:assumption_appendix_stone_1} and $\xi_{\max} = \ocal_{\Pbb}(\sqrt{N})$.
Then, 
$CR_{\varthetaf_j} := \{ \thetaf_j \in \Rbb^{K_j} ~|~ (\thetaf_j - \thetafh_j)^\top \Vfh_j^{-1} (\thetaf_j - \thetafh_j) \leq \chi^2_{\alpha} (K_j)\}$,
is an asymptotic $\alpha \cdot 100\%$ confidence 
region for $\varthetaf_j$,
$ 
\{ (\bfe_{\Xcal, \, j} \ootimes \bfe_{\Ycal})^\top \thetaf_j 
~|~ \thetaf_j \in CR_{\varthetaf_j} \} \subseteq \{ \Xcal \ra \B\}
$ 
is an asymptotic $\alpha \cdot 100\%$ confidence 
region for $h_{j}$, and 
$
\{ (\bfe_{\Xcal, \, j} \otimes \bft_{\Ycal})^\top \thetaf_j 
~|~ \thetaf_j \in CR_{\varthetaf_j} \} \subseteq \{ \Xcal \ra L_0^2 (\mu)\}
$ is an asymptotic $\alpha \cdot 100\%$ confidence 
region for 
$
\hti_{j}
:= (\bfe_{\Xcal, \, j} \otimes \bft_{\Ycal})^\top \varthetaf_j$.
Furthermore, $1 - F_{\chi^2(K_j)} (\thetafh_j^\top \Vfh_j^{-1} \thetafh_j)$ is a valid asymptotic p-value for 
$H_0: [\forall \xf \in \Xcal: h_j (\xf) = 0 \in \B]$
and for 
$H_0: [\forall \xf \in \Xcal: \hti_j (\xf) = 0 \in L^2_0(\mu)]$.
\end{lem}

\begin{proof}
For the matrix $\Sf_j$ 
as in Lemma~\ref{P2a:lemma_confidence_regions_simultaneos}, s
Set $\Vf_j := \Sf_j \Ff_{\mathrm{pen}}(\varthetaf) ^{-1} \Sf_j^\top$.
From Theorem~\ref{P2:thm_MLE_asymptotic_normal}, we have 
$\thetafh 
\overset{\text{a}}{\sim} \mathcal{N} (\varthetaf, (\Ff_{\mathrm{pen}}(\varthetaf))^{-1}) 
$, 
yielding 
$
\thetafh_j 
\overset{\mathrm{a}}{\sim} 
\mathcal{N} (\varthetaf_j, \Vf_j)
$ 
and thus, $(\thetafh_j - \varthetaf_j)^\top \Vf_j^{-1} (\thetafh_j - \varthetaf_j) 
\overset{\mathrm{a}}{\sim} \chi^2( K_j)$. 
Hence, $CR_{\varthetaf_j} := \{ \thetaf_j \in \Rbb^{K_j} ~|~ (\thetaf_j - \thetafh_j)^\top \Vfh_j^{-1} (\thetaf_j - \thetafh_j) \leq \chi^2_{\alpha} (K_j)\}$ is an asymptotic $\alpha \cdot 100\%$ confidence 
region for $\varthetaf_j$. 
Since the basis functions contained in the vector $\bfe_{\Xcal, \, j} \otimes \bft_{\Ycal}$ are linearly independent -- which follows from the basis functions contained in $\bfe_{\Xcal, \, j}$ and the basis functions contained in $\bft_{\Ycal}$ being linearly independent, respectively --, 
$\spano (\bfe_{\Xcal, \, j} \otimes \bft_{\Ycal}) = \spano ((\bfe_{\Xcal, \, j} \otimes \bft_{\Ycal})_{[1]}, \ldots , (\bfe_{\Xcal, \, j} \otimes \bft_{\Ycal})_{[K_j]})$ is a $K_j$-dimensional subspace of $\{\Xcal \mapsto L^2_0 (\mu) \}$ and thus is isomorphic to $\Rbb^{K_j}$ via the mapping $\Phi : \spano (\bfe_{\Xcal, \, j} \otimes \bft_{\Ycal}) \ra \Rbb^{K_j}, (\bfe_{\Xcal, \, j} \otimes \bft_{\Ycal})^\top \thetaf_j
\mapsto \thetaf_j$. 
Then, 
$
\{ (\bfe_{\Xcal, \, j} \otimes \bft_{\Ycal})^\top \thetaf_j \in \spano (\bfe_{\Xcal, \, j} \otimes \bft_{\Ycal}) ~|~ \thetaf_j \in CR_{\varthetaf_j} \}$ is an $\alpha \cdot 100\%$ confidence 
region for $(\bfe_{\Xcal, \, j} \otimes \bft_{\Ycal})^\top \varthetaf_j = \clr[h_{j}]$.
Using that the inverse clr transformation is isomorphic, as well, 
we obtain that
$
\{ (\bfe_{\Xcal, \, j} \ootimes \bfe_{\Ycal})^\top \thetaf_j \in \spano ((\bfe_{\Xcal, \, j} \ootimes \bfe_{\Ycal})_{[1]}, \ldots , (\bfe_{\Xcal, \, j} \ootimes \bfe_{\Ycal})_{[K_j]}) ~|~ \thetaf_j \in CR_{\varthetaf_j} \}$ is an $\alpha \cdot 100\%$ confidence 
region for $(\bfe_{\Xcal, \, j} \ootimes \bfe_{\Ycal})^\top \varthetaf_j = h_j$.

Denoting the probability measure 
corresponding to the $\chi^2(K_{j})$ distribution with $\Pbb_{\chi^2(K_{j})}$ 
and considering a random variable $T \sim \chi^2(K_{j})$, 
we obtain p-values for $\varthetaf_{j} = \mathbf{0}$ and equivalently for $h_{j}(\xf) = 0 \in \B$ for all $\xf \in \Xcal$ and for $\clr[h_{j}(\xf)] = 0 \in L^2_0(\mu)$ for all $\xf \in \Xcal$ via
\begin{align*}
\Pbb_{\chi^2(K_{j})} \left(T \geq (\thetafh_{j} - \varthetaf_{j})^\top \Vfh_{j}^{-1} (\thetafh_{j} - \varthetaf_{j}) ~|~ \varthetaf_{j} = \mathbf{0} \right)
&= 1 - F_{\chi^2(K_{j})} (\thetafh_{j}^\top \Vfh_{j}^{-1} \thetafh_{j}).
\qedhere
\end{align*}
\end{proof}

\begin{proof}[Proof of Theorem~\ref{P2:thm_approximation_bayes_multinomial}]
Per construction, we have $\ft_{\xf^{(l)}, \thetaf} = \ft_{\xf_i, \thetaf}$ for all $i \in \Ical^{(l)}$ and 
$\thetaf \in \Rbb^K 
$. 
The (shifted) log-likelihoods of interest are
\begin{align*}
\ell^{\mathrm{mn}}_\Zcal (\thetaf) 
&=\sum_{l=1}^L \sum_{g= 1}^{\Gamma^{(l)}} n_g^{(l)} \Bigl(  \ft_{\xf^{(l)}, \thetaf}(u_g^{(l)}) - \log \underbrace{
\sum_{g'=1}^{\Gamma^{(l)}} \Delta_{g'}^{(l)}\exp [ \ft_{\xf^{(l)}, \thetaf}(u_{g'}^{(l)}) ]}_{=: \Sigma(\Zcal^{(l)})} 
\Bigr)
\\
\ell (\thetaf)
&= \sum_{i=1}^N \bigl(\ft_{\xf_i, \thetaf}(y_i) - \log \int_{\Ycal} \exp(\ft_{\xf_i, \thetaf}) \, \dmu\bigr) .
\end{align*}

Using $\Delta_{g'}^{(l)} = w_{g' - G^{(l)}}$, $u_{g'}^{(l)} = t_{g' - G^{(l)}}$ for $g' = G^{(l)} + 1, \ldots , G^{(l)} + D = \Gamma^{(l)}$, we get
\begin{align}
\Sigma(\Zcal^{(l)})
&= \sum_{g'=1}^{\Gamma^{(l)}} \Delta_{g'}^{(l)}\exp [ \ft_{\xf^{(l)}, \thetaf}(u_{g'}^{(l)}) ]
= \underbrace{\sum_{g'=1}^{G^{(l)}} \Delta_{g'}^{(l)}\exp [ \ft_{\xf^{(l)}, \thetaf}(u_{g'}^{(l)}) ]}_{=: \sigma(\Zcal^{(l)})} 
+ \sum_{g'= 1}^{D} w_{g'} \exp [ \ft_{\xf^{(l)}, \thetaf}(t_{g'}) ] \notag
\\
&= \sigma(\Zcal^{(l)}) + \int_{\Yd} \exp ( \ft_{\xf^{(l)}, \thetaf} ) \, \ddel .
\label{P2a:eq_split_sum}
\end{align}

Let $\Yc = \emptyset$, i.e., we consider the discrete special case with $\Ycal = \Yd = \{t_1, \ldots , t_D\}$ and $\mu = \delta$. 
Then, $G^{(l)} = 0$ 
and $\Gamma^{(l)} = D$
for all $l = 1, \ldots , L$.
Thus,
\begin{align}
\ell^{\mathrm{mn}}_\Zcal (\thetaf) 
&=\sum_{l=1}^L \sum_{g=1}^{D} n_g^{(l)} \Bigl(  \ft_{\xf^{(l)}, \thetaf}(u_g^{(l)}) - \log 
\int_{\Yd} \exp ( \ft_{\xf^{(l)}, \thetaf} ) \, \ddel
\Bigr) \notag
\\
&=\sum_{l=1}^L \sum_{i \in \Ical_{\mathrm{d}}^{(l)}} \sum_{g=1}^{D} \mathbbm{1}_{\{ t_g \} }(y_i) \Bigl(  \ft_{\xf^{(l)}, \thetaf}(t_g) - \log 
\int_{\Ycal} \exp(\ft_{\xf^{(l)}, \thetaf}) \, \dmu
\Bigr) \notag
\\
%
%
&\overset{(\ast_1)}{=} \sum_{l=1}^L \sum_{i \in \Ical_{\mathrm{d}}^{(l)}} \Bigl(  \ft_{\xf^{(l)}, \thetaf}(y_i) - \log 
\int_{\Ycal} \exp(\ft_{\xf_i, \thetaf}) \, \dmu
\Bigr) \notag
\\
&\overset{(\ast_2)}{=} \sum_{i=1}^N \bigl(\ft_{\xf_i, \thetaf}(y_i) - \log \int_{\Ycal} \exp(\ft_{\xf_i, \thetaf}) \, \dmu\bigr) 
= \ell (\thetaf), \label{P2a:equation_approximation_bayes_multinomial_discrete}
\end{align}
where we used in $(\ast_1)$ that for every $l \in \{ 1, \ldots , L \}$ and every $i \in \Ical_{\mathrm{d}}^{(l)}$, there exists a unique $g_i \in \{ 1, \ldots , D \}$ with $y_i = t_{g_i}$, since the $y_i$ are observations of the (distinct) values in $\Yd$. 
In $(\ast_2)$, we used that in the discrete case $\Ical_{\mathrm{d}}^{(l)} = \Ical^{(l)}$ for all $l = 1 \ldots , L$, i.e., the sets $\Ical_{\mathrm{d}}^{(l)}$ form a partition of $\{ 1, \ldots , N\}$ for $l = 1, \ldots , L$.
Adding the penalty term to \ref{P2a:equation_approximation_bayes_multinomial_discrete} yields 
$\ell^{\mathrm{mn}}_{\Zcal, \mathrm{pen}} (\thetaf) = \ell_{\mathrm{pen}} (\thetaf)$, which
directly implies equality of the penalized Fisher informations, i.e.,
$
\Ff^{\mathrm{mn}}_\Zcal (\thetaf) 
= \Ff (\thetaf)
$.
\\[0.2cm]
Now, let $\Yc \neq \emptyset$. 
We show the convergence for the unpenalized log-likelihoods, i.e., 
$\lim_{\Delta \ra 0} \ell^{\mathrm{mn}}_{\Zcal} (\thetaf) = \ell (\thetaf)$,
which immediately implies convergence for the penalized log-likelihoods, since the penalty term is constant regarding $\Delta$.

For this purpose, we need to show that for all $\varepsilon > 0$ there exists a  $\delta > 0$, such that for all partitions $\Zcal^{(1)}, \ldots , \Zcal^{(L)}$ with overall maximal bin width $\Delta < \delta$, we have $\left|  \ell^{\mathrm{mn}}_\Zcal (\thetaf) - \ell (\thetaf) \right| < \varepsilon
$.
For this purpose, let $\varepsilon > 0$ and thus also $\frac{\varepsilon}{N}>0$.
Let $l \in \{1, \ldots , L \}$.
Define $\zeta^{(l)} := \int_{\Yd} \exp(\ft_{\xf^{(l)}, \thetaf}) \, \ddel \geq 0$.
Since the function $\xi \mapsto \log (\xi + \zeta^{(l)})$ is continuous for all $\xi > 0$, there exists an $\tilde{\varepsilon}^{(l)} > 0$, such that for all $\xi > 0$ with $|\xi - \int_{\Yc} \exp(\ft_{\xf^{(l)}, \thetaf}) \, \dlamb| < \tilde{\varepsilon}^{(l)}$, we have $| \log (\xi + \zeta^{(l)}) - \log (\int_{\Yc} \exp(\ft_{\xf^{(l)}, \thetaf}) \, \dlamb  + \zeta^{(l)}) | < \frac{\varepsilon}{N}$.
Set $\tilde{\varepsilon} := \min \{\tilde{\varepsilon}^{(1)}, \ldots , \tilde{\varepsilon}^{(L)} \} > 0$.
Since $\sigma(\Zcal^{(l)}) > 0$ is a Riemann sum 
of the integrable function $\exp ( \ft_{\xf^{(l)}, \thetaf} )$, for $\tilde{\varepsilon} > 0$ there exists a $\delta^{(l)}_{\tilde{\varepsilon}} > 0$ such that every partition $\Zcal^{(l)}$ with maximal bin width
$\Delta^{(l)} 
:= \max \{\Delta_g^{(l)} ~|~ g = 1, \ldots , G^{(l)} \} < \delta^{(l)}_{\tilde{\varepsilon}}$ fulfills $\vert \sigma(\Zcal^{(l)} ) - \int_{\Yc} \exp ( \ft_{\xf^{(l)}, \thetaf} ) \, \dlamb \vert < \tilde{\varepsilon}$ \citep[Theorem 14.23 (iii)]{laczkovich2015}.
Set $\delta_{\tilde{\varepsilon}} := \min \{ \delta^{(1)}_{\tilde{\varepsilon}} , \ldots , \delta^{(L)}_{\tilde{\varepsilon}} \}$.
If there exist $y_i \neq y_j$ with $i, j \in \bigcup_{l = 1}^L \Ical_{\mathrm{c}}^{(l)}$, set $\delta_y := \min \{|y_i - y_j| ~|~ y_i \neq y_j, ~ i, j \in \bigcup_{l = 1}^L \Ical_{\mathrm{c}}^{(l)} \} > 0$, and $\delta := \min \{ \delta_y , \delta_{\tilde{\varepsilon}}\} >0$.
Otherwise, set $\delta := \delta_{\tilde{\varepsilon}} >0$.\footnote{
Note that the latter case corresponds to only having observed one distinct continuous value overall. 
This is included for completeness of the theory, but in practice is highly unlikely unless the number $N$ of observations is very small (in which case regression results and results of statistical analyses. 
}
\\
Now, let $\Zcal^{(1)}, \ldots , \Zcal^{(L)}$ be partitions with $\Delta < \delta$ and let $l \in \{ 1, \ldots, L\}$. 
We have $\Delta^{(l)} \leq \Delta < \delta \leq \delta_{\tilde{\varepsilon}} \leq \delta^{(l)}_{\tilde{\varepsilon}}$, which yields $\vert \sigma(\Zcal^{(l)} ) - \int_{\Yc} \exp ( \ft_{\xf^{(l)}, \thetaf} ) \, \dlamb \vert < \tilde{\varepsilon} \leq \tilde{\varepsilon}^{(l)}$, implying $| \log (\sigma(\Zcal^{(l)} ) + \zeta^{(l)}) - \log (\int_{\Yc} \exp ( \ft_{\xf^{(l)}, \thetaf} ) \, \dlamb + \zeta^{(l)}) | < \frac{\varepsilon}{N}~(\ast_3)$.
Assume there exist 
$g \in \{ 1, \ldots , G^{(l)} \}$ and 
$i, j \in \Ical_{\mathrm{c}}^{(l)}$, such that 
$y_i \neq y_j$ and 
$y_i, y_j \in U_g^{(l)}$. 
Then, we have $|y_i - y_j| \leq a_{g}^{(l)} - a_{g-1}^{(l)} \leq \Delta < \delta \leq \delta_y
$, which is a contradiction.
Thus, for each 
$g \in \{ 1, \ldots , G^{(l)} \}$, there exists at most one $y_{g}^{(l)} \in \{y_i ~|~ i \in \Ical_{\mathrm{c}}^{(l)} \}$ 
such that $y_{g}^{(l)} 
\in U_g^{(l)}$.
Then, for $g \in \{ 1, \ldots , G^{(l)}\}$,
\[
n_g^{(l)} = 
\begin{cases}
| \{i \in \Ical_{\mathrm{c}}^{(l)} ~|~ y_i = y_{g}^{(l)} \} | & , \exists ! \, y_{g}^{(l)} \in \{y_i ~|~ i \in \Ical_{\mathrm{c}}^{(l)} \} : y_{g}^{(l)} \in U_g^{(l)}
\\
0 & , \text{otherwise}
\end{cases} \, .
\]
Recall that if such a unique $y_{g}^{(l)}$ exists, we choose $u_{g}^{(l)} = y_{g}^{(l)}$.
Thus, 
for each 
$i \in \Ical_{\mathrm{c}}^{(l)}$, there exists a unique $g_i \in \{ 1, \ldots , G^{(l)}\}$ such that $y_i = u_{g_i}^{(l)} 
$ (and $n_{g_i}^{(l)} > 0$), since the intervals $U_g^{(l)}$ are disjoint for $g = 1, \ldots , G^{(l)}$, and
as in argument $(\ast_1)$ in \eqref{P2a:equation_approximation_bayes_multinomial_discrete}, for each 
$i \in \Ical_{\mathrm{d}}^{(l)}$, there exists a unique $g_i \in \{ G^{(l)} + 1, \ldots , G^{(l)} + D\}$ such that $y_i = t_{g_i - G^{(l)}} = u_{g_i}^{(l)}$. 
%
Then, we get for the shifted multinomial log-likelihood,
\begin{align*}
\ell^{\mathrm{mn}}_\Zcal (\thetaf) 
%
%
&=
\sum_{l=1}^L \sum_{g=1}^{G^{(l)}} n_g^{(l)} \Bigl( \ft_{\xf^{(l)}, \thetaf}(u_g^{(l)}) - \log  \Sigma(\Zcal^{(l)}) \Bigr) +  
\\
&\hspace{0.5cm} 
\sum_{l=1}^L \sum_{g= G^{(l)} + 1}^{\Gamma^{(l)}} n_g^{(l)} \Bigl( \ft_{\xf^{(l)}, \thetaf}(u_g^{(l)}) - \log  \Sigma(\Zcal^{(l)}) \Bigr)
\\
&= \sum_{l=1}^L \sum_{i \in \Ical_{\mathrm{c}}^{(l)}} \Bigl( \ft_{\xf_i, \thetaf}(y_i) - \log \Sigma(\Zcal^{(l)}) \Bigr)
+ \sum_{l=1}^L \sum_{i \in \Ical_{\mathrm{d}}^{(l)}} \Bigl( \ft_{\xf_i, \thetaf}(y_i) - \log  \Sigma(\Zcal^{(l)}) \Bigr)
\\
&= \sum_{i =1}^N \ft_{\xf_i, \thetaf}(y_i) - \sum_{l=1}^L \sum_{i \in \Ical^{(l)}}  \log \Sigma(\Zcal^{(l)}),
\end{align*}
where we used in the last equation that $\Ical^{(l)}$ form a partition of $\{ 1, \ldots , N\}$ for $l = 1, \ldots , L$, and $\Ical_{\mathrm{c}}^{(l)}$ and $\Ical_{\mathrm{d}}^{(l)}$ form a partition of $\Ical^{(l)}$ for each $l = 1, \ldots , L$.
Finally, 
we then have
\begin{align*}
\left| \ell^{\mathrm{mn}}_\Zcal (\thetaf) - \ell (\thetaf) \right| 
&=
\Big| \sum_{i =1}^N \ft_{\xf_i, \thetaf}(y_i) - \sum_{l=1}^L \sum_{i \in \Ical^{(l)}}  \log \Sigma(\Zcal^{(l)}) \, - 
\\
&\hspace{0.7cm} 
\Bigl( \sum_{i=1}^N \ft_{\xf_i, \thetaf}(y_i) - \sum_{i=1}^N \log \int_{\Ycal} \exp(\ft_{\xf_i, \thetaf}) \, \dmu \Bigr) \Big|
\\
%
%
&=
\Big|  \sum_{l=1}^L \sum_{i \in \Ical^{(l)}} \Bigl( \log \bigl( \sigma(\Zcal^{(l)}) + \int_{\Yd} \exp(\ft_{\xf^{(l)}, \thetaf}) \, \ddel \bigr)  -  
\\
&\hspace{2.6cm} 
\log \bigl( \int_{\Yc} \exp(\ft_{\xf^{(l)}, \thetaf}) \, \dlamb + \int_{\Yd} \exp(\ft_{\xf^{(l)}, \thetaf}) \, \ddel \Bigr) \Big|
\\
&\leq
\sum_{l=1}^L \sum_{i \in \Ical^{(l)}} \Big| \Bigl( \log \bigl( \sigma(\Zcal^{(l)}) + \zeta^{(l)} \bigr) -
\\
&\hspace{2.6cm} 
\log \bigl( \int_{\Yc} \exp(\ft_{\xf^{(l)}, \thetaf}) \, \dlamb + \zeta^{(l)} \bigr) \Bigr) \Big|
\\
&\overset{(\ast_3)}{<}
\sum_{l=1}^L \sum_{i \in \Ical^{(l)}} \frac{\varepsilon}{N}
= N \frac{\varepsilon}{N} 
= \varepsilon ,
\end{align*}
which proves $\lim_{\Delta \ra 0} \ell^{\mathrm{mn}}_\Zcal (\thetaf) = \ell (\thetaf)$ 
and thus
$\lim_{\Delta \ra 0} \ell^{\mathrm{mn}}_{\Zcal, \mathrm{pen}} (\thetaf) = \ell_{\mathrm{pen}} (\thetaf)$.

Now, we show the convergence of the Hessian matrices $\Hf^{\text{mn}}_{\Zcal} (\thetaf)$ of $\ell^{\text{mn}}_{\Zcal} (\thetaf)$ to the Hessian matrix $\Hf (\thetaf)$ of $\ell (\thetaf)$ for decreasing bin widths, which is equivalent to convergence of the (unpenalized) Fisher informations.
For this purpose, we 
use the notation $\gammaf_{[k]}$ to denote the $k$-th entry of a vector $\gammaf$.
As in the proof of Proposition~\ref{P2a:prop_appendix_fisher_informations_positive_definite}, we have
$\ell (\thetaf)
= \sum_{i=1}^N [ \bft (\xf_i)(y_i)^\top \thetaf - C_i(\thetaf) ]
$
with Hessian matrix $\Hf(\thetaf) = \sum_{i=1}^N \Hf_i (\thetaf)$,
where $\Hf_i (\thetaf) \in \Rbb^{K \times K}$ denotes the Hessian matrix of 
$C_i(\thetaf) = - \log \int_{\Ycal} \exp [ \bft (\xf_i)^\top \thetaf ]\, \dmu$ for $i = 1, \ldots , N$.
By \eqref{P2a:eq_hessian_bayes}, the $(k, k')$-th element of $\Hf_i (\thetaf), ~i = 1, \ldots , N$, is
\begin{align*}
\frac{\partial^2}{\partial \theta_k \partial \theta_{k'}} C_i(\thetaf)
&= - \frac{ \int_\Ycal \exp [ \bft (\xf_i)^\top \thetaf ] \, (\bft (\xf_i))_{[k]} \, (\bft (\xf_i))_{[k']} \, \dmu}{\int_\Ycal \exp [ \bft (\xf_i)^\top \thetaf ] \, \dmu} \notag
\\
&\hspace{0.5cm}+ \frac{\int_\Ycal \exp [ \bft (\xf_i)^\top \thetaf ] \cdot (\bft (\xf_i))_{[k]} \, \dmu }{\int_\Ycal \exp [ \bft (\xf_i)^\top \thetaf ] \, \dmu} \, \frac{\int_\Ycal \exp [ \bft (\xf_i)^\top \thetaf ] \cdot (\bft (\xf_i))_{[k']} \, \dmu}{\int_\Ycal \exp [ \bft (\xf_i)^\top \thetaf ] \, \dmu}
\, . 
\end{align*}
Similarly, consider
$
\ell^{\text{mn}}_{\Zcal} (\thetaf) 
=
\sum_{l=1}^L \sum_{g=1}^{\Gamma^{(l)}} n_g^{(l)} [ \bft(\xf^{(l)}) (u_g^{(l)}) ^\top \thetaf 
- C_{\Zcal}^{\text{mn} \, (s)}(\thetaf) ]
$
with Hessian matrix
$\Hf^{\text{mn}}_{\Zcal} (\thetaf)
= \sum_{l=1}^L \sum_{g=1}^{\Gamma^{(l)}} n_g^{(l)} \Hf_{\Zcal}^{\text{mn} \, (s)}(\thetaf) 
$
where $\Hf_{\Zcal}^{\text{mn} \, (s)}(\thetaf)$
denotes the Hessian matrix of
$C_{\Zcal}^{\text{mn} \, (s)}(\thetaf) = - \log \sum_{g'=1}^{\Gamma^{(l)}} \Delta_{g'}^{(l)}\exp [ \bft(\xf^{(l)}) (u_{g'}^{(l)}) ^\top \thetaf ]$ 
for $l = 1, \ldots , L$.
By~\eqref{P2a:eq_hessian_multinomial}, the $(k, k')$-th element of $\Hf_{\Zcal}^{\text{mn} \, (s)}(\thetaf)$ is
\begin{align*}
&~ \frac{\partial^2}{\partial \theta_k \partial \theta_{k'}} C_{\Zcal}^{\text{mn} \, (s)}(\thetaf)
\\
= &~ - \frac{ \sum_{g'=1}^{\Gamma^{(l)}} \bigl[ \Delta_{g'}^{(l)}\exp [ \bft(\xf^{(l)}) (u_{g'}^{(l)}) ^\top \thetaf ] \, (\bft(\xf^{(l)})(u_{g'}^{(l)}))_{[k]} 
\, (\bft (\xf^{(l)})(u_{g'}^{(l)}))_{[k']} \bigr]}{\sum_{g'=1}^{\Gamma^{(l)}} \Delta_{g'}^{(l)}\exp [ \bft(\xf^{(l)}) (u_{g'}^{(l)}) ^\top \thetaf ] } \notag
\\
&~+ \frac{\sum_{g'=1}^{\Gamma^{(l)}} \bigl[ \Delta_{g'}^{(l)}\exp [ \bft(\xf^{(l)}) (u_{g'}^{(l)}) ^\top \thetaf ] \, (\bft(\xf^{(l)})(u_{g'}^{(l)}))_{[k]} \bigr]}{\sum_{g'=1}^{\Gamma^{(l)}} \Delta_{g'}^{(l)}\exp [ \bft(\xf^{(l)}) (u_{g'}^{(l)}) ^\top \thetaf ]} \notag
\\
&~ \cdot \frac{\sum_{g'=1}^{\Gamma^{(l)}} \bigl[ \Delta_{g'}^{(l)}\exp [ \bft(\xf^{(l)}) (u_{g'}^{(l)}) ^\top \thetaf ] \, (\bft(\xf^{(l)})(u_{g'}^{(l)}))_{[k']} \bigr] \,}{\sum_{g'=1}^{\Gamma^{(l)}} \Delta_{g'}^{(l)}\exp [ \bft(\xf^{(l)}) (u_{g'}^{(l)}) ^\top \thetaf ]} 
. 
\end{align*}
As in~\eqref{P2a:eq_split_sum}, we split each of the appearing sums over $g' = 1, \ldots , \Gamma^{(l)}$ into the respecting sums over $g' = 1, \ldots , G^{(l)}$ and over $g' = G^{(l)} + 1, \ldots , \Gamma^{(l)}$. 
In each case, the former sum is a Riemann sum of the corresponding integrable function and thus converges against the function's integral over $\Yc$ with respect to $\lambda$ for $\Delta \ra 0$, while the latter sum equals the integral over $\Yd$ with respect to $\delta$. 
The sum of both integrals corresponds to the integral over $\Ycal$ with respect to $\mu$.
Note that the limiting integral in the denominators, $\int_\Ycal \exp [ \bft (\xf^{(l)})^\top \thetaf ] \, \dmu$, is positive and finite, while the integrals in the numerators are also finite, since all integrands are finite on $\Ycal$ and $\mu(\Ycal) < \infty$.
Thus, sums, products, and quotients are preserved under the limit 
\citep[Theorems~5.12, 5.14 and 5.16]{laczkovich2015} and the limit of \eqref{P2a:eq_hessian_multinomial} for $\Delta \ra 0$ is 
\begin{align*}
& - \frac{ \int_\Ycal \exp [ \bft (\xf^{(l)})^\top \thetaf ] \, (\bft (\xf^{(l)}))_{[k]} \, (\bft (\xf^{(l)}))_{[k']} \, \dmu}{\int_\Ycal \exp [ \bft (\xf^{(l)})^\top \thetaf ] \, \dmu} \notag
\\
&+ \frac{\int_\Ycal \exp [ \bft (\xf^{(l)})^\top \thetaf ] \, (\bft (\xf^{(l)}))_{[k]} \, \dmu }{\int_\Ycal \exp [ \bft (\xf^{(l)})^\top \thetaf ] \, \dmu} \, \frac{\int_\Ycal \exp [ \bft (\xf^{(l)})^\top \thetaf ] \, (\bft (\xf^{(l)}))_{[k']} \, \dmu}{\int_\Ycal \exp [ \bft (\xf^{(l)})^\top \thetaf ] \, \dmu} \notag
.
\end{align*}
Finally, we then have
\begin{align*}
\lim_{\Delta \ra 0} \Hf^{\text{mn}}_{\Zcal} (\thetaf)
= \sum_{l=1}^L |\Ical^{(l)}| \, \lim_{\Delta \ra 0} \Hf_{\Zcal}^{\text{mn} \, (s)}(\thetaf)
= \sum_{i=1}^N \Hf_i (\thetaf)
= \Hf(\thetaf) ,
\end{align*}
where we used $\xf_i = \xf^{(l)}$ for all $i \in \Ical^{(l)}$.
Since the Hessian matrices of the penalized log-likelihoods $\ell^{\mathrm{mn}}_{\Zcal, \mathrm{pen}} (\thetaf)$ and $\ell_{\mathrm{pen}} (\thetaf)$ are given by $\Hf^{\text{mn}}_{\Zcal} (\thetaf) + D^2 \pen (\thetaf)$ and $\Hf(\thetaf) + D^2 \pen (\thetaf)$, where $D^2 \pen (\thetaf)$ is constant in $\Delta$, we obtain $\lim_{\Delta \ra 0} \Ff^{\text{mn}}_{\Zcal, \mathrm{pen}} (\thetaf) = \Ff_{\mathrm{pen}} (\thetaf)$.
\\
Now assume, all involved penalized Fisher informations are invertible. 
Since the matrix inversion function $\inv : \GL(K , \Rbb) \ra \GL(K , \Rbb), \, \Ff \mapsto \Ff^{-1}$ is continuous (compare Lemma~\ref{P2a:lem_matrix_inversion_continuous}), we obtain $\lim_{\Delta \ra 0} (\Ff^{\text{mn}}_{\Zcal, \mathrm{pen}} (\thetaf))^{-1} = (\Ff_{\mathrm{pen}}(\thetaf))^{-1}$. 
\end{proof}

\begin{proof}[Proof of Theorem~\ref{P2:thm_approximation_bayes_multinomial_MLE}]

\begin{enumerate}[label=\arabic*)]
\item
The uniqueness of $\thetafh_{\Zcal}^{\mathrm{mn}}$ follows completely analogously to the proof of the uniqueness of $\thetafh$ in Theorem~\ref{P2:thm_existence_uniqueness_bayes_PMLE}, using Proposition~\ref{P2a:prop_appendix_fisher_informations_positive_definite}~\ref{P2a:prop_appendix_multinomial_fisher_information_positive_definite} to obtain strict concavity of $\ell_{\Zcal, \pen}^{\mathrm{mn}} (\thetaf)$ on $\Rbb^K$.
\item
For $m \in \Nbb$, consider $\Zcal_m \in \Xi^L$ with $\lim_{m \ra \infty} \Delta (\Zcal_m) = 0$.
We use \citet[Theorem 7.33]{rockafellar2009} to obtain asymptotic existence and convergence of the minimizers of $-\ell_{\Zcal_m, \mathrm{pen}}^{\mathrm{mn}}$, which correspond to the PMLEs of $\ell_{\Zcal_m, \mathrm{pen}}^{\mathrm{mn}}$.
The theorem requires that $-\ell_{\Zcal_m, \mathrm{pen}}^{\mathrm{mn}}$ and $-\ell_{\mathrm{pen}}$ are lower semicontinuous and proper and that the sequence of functions $(- \ell_{\Zcal_m, \mathrm{pen}}^{\mathrm{mn}})_{m \in \Nbb}$ is eventually level-bounded and epi-converges to $\ell_{\mathrm{pen}}$.
We already showed in the proof of Theorem~\ref{P2:thm_existence_uniqueness_bayes_PMLE} that $-\ell_{\mathrm{pen}}$ is lower semicontinuous and proper.
For the functions $-\ell_{\Zcal_m, \mathrm{pen}}^{\mathrm{mn}}$ this follows analogously.
Furthermore, by Proposition~\ref{P2a:prop_appendix_fisher_informations_positive_definite}, $-\ell_{\Zcal_m, \mathrm{pen}}^{\mathrm{mn}}$ and $-\ell_{\mathrm{pen}}$ are convex on $\Rbb^K$ (using assumptions 
\ref{P2a:assumption_appendix_covariate_basis_nonsingular} 
and \ref{P2a:assumption_appendix_domain_basis_nonsingular}).
Together with 
$\dom (-\ell_{\mathrm{pen}}) = \Rbb^K$ having nonempty interior
and $\lim_{m \ra \infty} ( - \ell_{\Zcal_m, \mathrm{pen}}^{\mathrm{mn}} (\thetaf) ) = -\ell_{\mathrm{pen}} (\thetaf)$ for all $\thetaf \in \Rbb^K 
$ (Theorem~\ref{P2:thm_approximation_bayes_multinomial}),
this implies that $- \ell_{\Zcal_m, \mathrm{pen}}^{\mathrm{mn}}$ epi-converges to $-\ell_{\mathrm{pen}}$ \citep[Theorem 7.17]{rockafellar2009}.
Moreover, by Theorem~\ref{P2a:thm_appendix_loglikelihood_level_bounded}, $-\ell_{\mathrm{pen}}$ is level-bounded (using assumption~\ref{P2a:assumption_appendix_scalar_product_non_constant}).
Then, the sequence $(- \ell_{\Zcal_m, \mathrm{pen}}^{\mathrm{mn}})_{m \in \Nbb}$ is eventually level-bounded \citep[Exercise 7.32 (c)]{rockafellar2009}.
As $\thetafh$ exists 
and is unique by Theorem~\ref{P2:thm_existence_uniqueness_bayes_PMLE}, 
we obtain from \citet[Theorem 7.33]{rockafellar2009}, that there exists an $M \in \Nbb$ such that the minimizer $\thetafh^{\mathrm{mn}}_{\Zcal_m}$ of $- \ell_{\Zcal_m, \mathrm{pen}}^{\mathrm{mn}} (\thetaf)$ exists (and is unique by~\ref{P2:thm_approximation_multinomial_MLE_unique}, using assumption~\ref{P2a:assumption_appendix_domain_basis_nonsingular}) for all $m \geq M$ and $\lim_{m \ra \infty} \hat{\thetaf}_{\Zcal_m}^{\mathrm{mn}} = \hat{\thetaf}$.
\qedhere
\end{enumerate}
\end{proof}

\begin{proof}[Proof of Theorem~\ref{P2:proposition_equivalence_multniomial_poisson}]
Note that \citet{baker1994} provide a proof of this Theorem for unpenalized estimation, which might be extended to our penalized setting.
When we became aware of their work, we already derived the following proof, which pursues a different strategy. 
While \citet{baker1994} follow a two-step/substitution approach regarding maximization of $\tauf$ and $\thetaf$ and use results of \citet{richards1961}, our proof uses a parameter transformation as proposed by \citet{palmgren1981}.
\\[0.2cm]
We first show that the PMLEs are identical, extending the work of \citet{palmgren1981} for a linear predictor $\eta_g^{(l)} (\thetaf)$ without penalization to an arbitrary predictor with penalized estimation.
Setting $\Pen(\thetaf) := \exp (\pen(\thetaf))$, the penalized likelihood function corresponding to the Poisson distribution assumption~\ref{P2:proposition_equivalence_multniomial_poisson_assumption_poisson} is
\begin{align}
\Lcal_{\mathrm{pen}}^{\mathrm{po}} ( \thetaf, \tauf ) 
&= \prod_{l=1}^L \prod_{g=1}^{\Gamma^{(l)}} 
\biggl( \frac{\lambda_g^{(l)} \left( \thetaf, \tauf \right)^{n_g^{(l)}}}{n_g^{(l)}!} \exp\left(-\lambda_g^{(l)} \left( \thetaf, \tauf \right)\right) \biggr)
\, \Pen(\thetaf). \label{P2a:equation_likelihood}
\end{align}
Consider the transformation 
$\psi: \Rbb^K 
\times \Rbb^L
\ra \Rbb^K 
\times (\Rbb^+)^L
, 
(\thetaf, \tauf) 
\mapsto (\thetaf, \lambdaf) 
,$
%
with $\lambdaf = (\lambda^{(1)}, \ldots , \lambda^{(L)})$ and $\lambda^{(l)} := \sum_{g=1}^{\Gamma^{(l)}} \lambda_g^{(l)} (\thetaf, \tauf) = \sum_{g=1}^{\Gamma^{(l)}} \exp (\eta_g^{(l)} (\thetaf) + \tau^{(l)} )$ for $l = 1, \ldots , L$.
The transformation is bijective with inverse transformation
$\psi^{-1}: \Rbb^K 
\times (\Rbb^+)^L
\ra \Rbb^K 
\times \Rbb ^L
, 
(\thetaf, \lambdaf) 
\mapsto (\thetaf, \tauf),$
with $\tau^{(l)} := \log\lambda^{(l)} - \log (\sum_{g=1}^{\Gamma^{(l)}}  \exp (\eta_g^{(l)} (\thetaf) ) )$ for $l = 1, \ldots , L$.
Thus, $\psi$ 
and $\psi^{-1}$ 
preserve maximizers of the penalized likelihood \eqref{P2a:equation_likelihood}.
Reparameterizing $(\thetaf, \tauf) = \psi^{-1} (\thetaf , \lambdaf)$ 
yields
\begin{align}
\lambda_g^{(l)} (\thetaf , \tauf)
= \lambda_g^{(l)} \bigl( \psi^{-1} (\thetaf , \lambdaf) \bigr)
&= \frac{ \exp \bigl( \eta_g^{(l)} (\thetaf) \bigr) \, \lambda^{(l)}}{\sum_{g'=1}^{\Gamma^{(l)}}  \exp\bigl(\eta_{g'}^{(l)} (\thetaf) \bigr)}
%
\label{P2a:eq_poisson_rate_transformation}
\end{align}
and $\Lcal_{\mathrm{pen}}^{\mathrm{po}} \circ \psi^{-1}$ gives us the penalized likelihood
\begin{align}
\check{\Lcal}_{\mathrm{pen}}^{\mathrm{po}} (\thetaf , \lambdaf) 
&= \prod_{l=1}^L \prod_{g=1}^{\Gamma^{(l)}} \Biggl( \frac{\left( \frac{ \exp ( \eta_g^{(l)} (\thetaf) ) \, \lambda^{(l)}}{\sum_{g'=1}^{\Gamma^{(l)}}  \exp(\eta_{g'}^{(l)} (\thetaf) )} \right)^{n_g^{(l)}}}{n_g^{(l)}!} \, 
\exp \Biggl( \frac{-  \exp \left( \eta_g^{(l)} (\thetaf) \right)\, \lambda^{(l)}}{\sum_{g'=1}^{\Gamma^{(l)}}  \exp\left(\eta_{g'}^{(l)} (\thetaf) \right)} \Biggr) \Biggr) \, \Pen(\thetaf) \notag
\\
&= \prod_{l=1}^L \Biggl( \frac{n^{(l)}!}{\prod_{g=1}^{\Gamma^{(l)}} n_g^{(l)}!} \prod_{g=1}^{\Gamma^{(l)}} \Biggl( \frac{ \exp \bigl( \eta_g^{(l)} (\thetaf) \bigr)}{\sum_{g'=1}^{\Gamma^{(l)}}  \exp\bigl(\eta_{g'}^{(l)} (\thetaf) \bigr)} \Biggr)^{n_g^{(l)}} \Biggr) \, \Pen(\thetaf) \notag \\
&\hspace{0.7cm} \cdot \prod_{l=1}^L \Bigl( \frac{(\lambda^{(l)})^{n^{(l)}}}{n^{(l)}!} \exp(- \lambda^{(l)}) \Bigr), 
\label{P2a:eq_factorized_likelihood}
\end{align}
where we used 
$\sum_{g=1}^{\Gamma^{(l)}} n_g^{(l)} = n^{(l)}$ in the last equation.
Thus, the penalized likelihood $\check{\Lcal}_{\mathrm{pen}}^{\mathrm{po}} (\thetaf , \lambdaf)$ factorizes into two functions, one modeling $\thetaf$ and the other modeling $\lambdaf$, which we denote with $\Lcal_{\mathrm{pen}}^{\mathrm{mn}} (\thetaf )$ and $\Lcal^{\mathrm{po}} (\lambdaf)$. 
To maximize the overall likelihood, we can maximize both parts separately. 
The first part $\Lcal_{\mathrm{pen}}^{\mathrm{mn}} (\thetaf )$ is the penalized likelihood of a multinomial distribution with probabilities 
$p_g^{(l)} (\thetaf) =  \frac{ \exp ( \eta_g^{(l)} (\thetaf) )}{\sum_{g'=1}^{\Gamma^{(l)}}  \exp (\eta_{g'}^{(l)} (\thetaf) )}$ and observations $\mathbf{n}^{(l)}$, 
i.e., maximization yields the 
multinomial PMLE $\hat{\thetaf}^{\mathrm{mn}}$ for $\thetaf$.
The second part $\Lcal^{\mathrm{po}} (\lambdaf)$ yields the Poisson MLEs $\hat{\lambda}^{(l)}$ for $\lambda^{(l)}$ based on one observation $n^{(l)}$, i.e., $\hat{\lambda}^{(l)} = n^{(l)}$ for $l = 1, \ldots , L$. 
Then, due to the invariance of the estimators under $\psi$ 
and $\psi^{-1}$, 
the PMLEs for $\thetaf$ are identical, $\hat{\thetaf}^{\mathrm{po}} = \hat{\thetaf}^{\mathrm{mn}}$. 

We now show the statements on the inverse Fisher informations.
Here, our contribution goes well beyond the work of \citet{palmgren1981}, who, for a linear predictor without penalization, compares the Fisher informations corresponding to $\check{\Lcal}^{\mathrm{po}}$ and $\Lcal^{\mathrm{mn}}$ (denoting the likelihoods $\check{\Lcal}_{\mathrm{pen}}^{\mathrm{po}}$ and $\Lcal_{\mathrm{pen}}^{\mathrm{mn}}$ without penalization).
We take the additional step of considering the Fisher informations corresponding to $\Lcal_{\mathrm{pen}}^{\mathrm{po}} = \check{\Lcal}_{\mathrm{pen}}^{\mathrm{po}} \circ \psi$, which is of great importance for our work, since in practice, our models are estimated on this level.
Denote the (penalized) log-likelihoods corresponding to the (penalized) likelihoods considered so far with $\ell$ and the respective indices etc.\ and set $\thetafh := \thetafh^{\mathrm{po}} = \thetafh^{\mathrm{mn}}$.
Furthermore, denote the differential operators for first and second order derivatives with respect to the full argument vector of a function with $D$ and $D^2$ and with respect to some subvector $\gammaf$ with $D_{\gammaf}$ and $D^2_{\gammaf}$.
Transforming \eqref{P2a:eq_factorized_likelihood} to the level of log-likelihoods yields
\begin{align}\label{P2a:eq_additive_log_likelihood}
\check{\ell}_{\mathrm{pen}}^{\mathrm{po}} (\thetaf , \lambdaf) = \ell_{\mathrm{pen}}^{\mathrm{mn}} (\thetaf ) + \ell^{\mathrm{po}} (\lambdaf) .
\end{align}
Since $\ell_{\mathrm{pen}}^{\mathrm{mn}} (\thetaf )$ is independent of $\lambdaf$, while $\ell^{\mathrm{po}} (\lambdaf)$ is independent of $\thetaf$, 
we obtain a block diagonal form for the Hessian matrix of the penalized log-likelihood $\check{\ell}_{\mathrm{pen}}^{\mathrm{po}} (\thetaf , \lambdaf)$,
\begin{align}
D^2 \check{\ell}_{\mathrm{pen}}^{\mathrm{po}} (\thetaf , \lambdaf) = 
\begin{pmatrix}
D^2 \ell_{\mathrm{pen}}^{\mathrm{mn}} (\thetaf ) & \mathbf{0} \\
\mathbf{0} & D^2 \ell^{\mathrm{po}} (\lambdaf)
\end{pmatrix}.
\label{P2a:eq_block_diagonal_Hessian}
\end{align}
In particular, $D^2_{\thetaf} \check{\ell}_{\mathrm{pen}}^{\mathrm{po}} (\thetaf , \lambdaf) = D^2 \ell_{\mathrm{pen}}^{\mathrm{mn}} (\thetaf)$, i.e., the upper left $K \times K$ submatrix of the Fisher information corresponding to $\check{\ell}_{\mathrm{pen}}^{\mathrm{po}} (\thetaf , \lambdaf)$ is identical to the Fisher information obtained from the multinomial distributional assumption~\ref{P2:proposition_equivalence_multniomial_poisson_assumption_multinom} for all 
$\thetaf \in \Rbb^K 
$. 
To obtain the Hessian matrix of the penalized log-likelihood $\ell_{\mathrm{pen}}^{\mathrm{po}} 
= \check{\ell}_{\mathrm{pen}}^{\mathrm{po}} 
\circ \psi$, we apply the chain rule, which gives
\begin{align}
D \ell_{\mathrm{pen}}^{\mathrm{po}} (\thetaf , \tauf)
&= [ D \check{\ell}_{\mathrm{pen}}^{\mathrm{po}} \circ \psi] (\thetaf , \tauf) \, D \psi (\thetaf , \tauf) , \notag
\\
D^2 \ell_{\mathrm{pen}}^{\mathrm{po}} (\thetaf , \tauf)
&= (D \psi (\thetaf , \tauf))^\top \, [ D^2 \check{\ell}_{\mathrm{pen}}^{\mathrm{po}} \circ \psi] (\thetaf , \tauf) \, D \psi (\thetaf , \tauf) \notag
\\
&\hspace{0.5cm}+ [ D \check{\ell}_{\mathrm{pen}}^{\mathrm{po}} \circ \psi ] (\thetaf , \tauf) \, D^2 \psi (\thetaf , \tauf)
,
\label{P2a:eq_Hessian_poisson_tau}
\end{align}
where $D^2 \psi (\thetaf , \tauf) = (D^2 \psi_1 (\thetaf , \tauf), \ldots , D^2 \psi_{K + L} (\thetaf , \tauf))^\top 
$.
Inserting $(\thetafh, \taufh)$, the last summand is zero since $\psi (\thetafh, \taufh)$ maximizes $\check{\ell}_{\mathrm{pen}}^{\mathrm{po}}$ and thus
\begin{align}
D^2 \ell_{\mathrm{pen}}^{\mathrm{po}} (\thetafh , \taufh)
&= (D \psi (\thetafh , \taufh)^\top \, [ D^2 \check{\ell}_{\mathrm{pen}}^{\mathrm{po}} \circ \psi ] (\thetafh , \taufh) \, D \psi (\thetafh , \taufh) .
\label{P2a:eq_Hessian_poisson_tauhat}
\end{align}
Setting
$\psi_{\lambdaf} : \Rbb^K 
\times \Rbb^L \ra (\Rbb^+)^L$, $(\thetaf, \tauf) \mapsto \lambdaf = (\psi_{K + 1} (\thetaf, \tauf), \ldots , \psi_{K + L} (\thetaf, \tauf))$,
we have
\begin{align}
D \psi (\thetafh , \taufh) =
\begin{pmatrix}
\If_K & \mathbf{0}
\\
D_{\thetaf} \psi_{\lambdaf} (\thetafh , \taufh) & D_{\tauf} \psi_{\lambdaf} (\thetafh , \taufh)
\end{pmatrix},
\label{P2a:eq_Dpsi_hat}
\end{align}
which at $(\thetafh, \taufh)$ (omitting the argument 
in all functions for the sake of readability) yields
\begin{align}
D^2 \ell_{\mathrm{pen}}^{\mathrm{po}} 
&= \begin{pmatrix}
D^2 \ell^{\mathrm{mn}}_{\mathrm{pen}} 
+ (D_{\thetaf} \psi_{\lambdaf} 
)^\top 
\, [D^2 \ell^{\mathrm{po}} \circ \psi_{\lambdaf} 
] \, D_{\thetaf} \psi_{\lambdaf} 
& 
(D_{\thetaf} \psi_{\lambdaf} 
)^\top \, [D^2 \ell^{\mathrm{po}} \circ \psi_{\lambdaf} 
] \, D_{\tauf} \psi_{\lambdaf} 
\\
(D_{\tauf} \psi_{\lambdaf} 
)^\top \, [D^2 \ell^{\mathrm{po}} \circ \psi_{\lambdaf} 
] \, D_{\thetaf} \psi_{\lambdaf} 
&
(D_{\tauf} \psi_{\lambdaf} 
)^\top \, [D^2 \ell^{\mathrm{po}} \circ \psi_{\lambdaf} 
] \, D_{\tauf} \psi_{\lambdaf} 
\end{pmatrix} .
\label{P2a:eq_Hessian_poisson_tauhat_expanded}
\end{align}
In particular, 
$(\Ffh^{\mathrm{po}})_{[K]} 
\neq \Ffh^{\mathrm{mn}}$.
To show the equality for the inverted matrices, 
note that \eqref{P2a:eq_Hessian_poisson_tauhat} yields
\begin{align*}
( D^2 \ell_{\mathrm{pen}}^{\mathrm{po}} (\thetafh , \taufh)
)^{-1}
=
( D \psi (\thetafh , \taufh) )^{-1} 
\, ( [ D^2 \check{\ell}_{\mathrm{pen}}^{\mathrm{po}} \circ \psi ] (\thetafh , \taufh) )^{-1} 
\, ( (D \psi (\thetafh , \taufh))^\top )^{-1}.
\end{align*}
Since the block matrix structure of~\eqref{P2a:eq_block_diagonal_Hessian} is preserved under invertation, we have
\begin{align*}
\left( [ D^2 \check{\ell}_{\mathrm{pen}}^{\mathrm{po}} \circ \psi ] (\thetafh , \taufh) \right)^{-1}
&=
\begin{pmatrix}
(D^2 \ell_{\mathrm{pen}}^{\mathrm{mn}} (\thetafh ) 
)^{-1}
& 
\mathbf{0} 
\\
\mathbf{0} 
& 
([D^2 \ell^{\mathrm{po}} \circ \psi_{\lambdaf}] (\thetafh , \taufh) 
)^{-1}
\end{pmatrix}.
\end{align*}
Furthermore,
\begin{align}
\left( D \psi (\thetafh , \taufh)
\right)^{-1}
&= \begin{pmatrix}
\If_K & \mathbf{0}
\\
- (D_{\tauf} \psi_{\lambdaf} (\thetafh , \taufh))^{-1} \, D_{\thetaf} \psi_{\lambdaf} (\thetafh , \taufh) 
&
(D_{\tauf} \psi_{\lambdaf} (\thetafh , \taufh))^{-1}
\end{pmatrix} .
\label{P2a:eq_Dpsi_hat_inv}
\end{align}
Setting $\Afh = (D_{\tauf} \psi_{\lambdaf} (\thetafh , \taufh) 
)^{-1} \, D_{\thetaf} \psi_{\lambdaf} (\thetafh , \taufh)$ and $\Bfh := ([D^2 \ell^{\mathrm{po}} \circ \psi_{\lambdaf}] (\thetafh , \taufh) 
)^{-1}$, we obtain for the inverse of $D^2 \ell^{\mathrm{po}}_{\mathrm{pen}} (\thetafh , \taufh)$ (again omitting the argument $(\thetafh , \taufh)$ in all functions)
\begin{align*}
(D^2 \ell_{\mathrm{pen}}^{\mathrm{po}} 
)^{-1}
&= \begin{pmatrix}
(D^2 \ell_{\mathrm{pen}}^{\mathrm{mn}} 
)^{-1}
& 
- (D^2 \ell_{\mathrm{pen}}^{\mathrm{mn}} 
)^{-1}
\, \Afh^\top
\\
- \Afh 
\, (D^2 \ell^{\mathrm{mn}}_{\mathrm{pen}})^{-1}
&
\Afh
\, (D^2 \ell^{\mathrm{mn}}_{\mathrm{pen}})^{-1}
\, \Afh ^\top
+
(D_{\tauf} \psi_{\lambdaf} 
)^{-1}
\, \Bfh
\, ((D_{\tauf} \psi_{\lambdaf} 
)^{-1})^\top
\end{pmatrix} ,
\end{align*}
which shows $((\Ffh^{\mathrm{po}})^{-1})_{[K]} = (\Ffh^{\mathrm{mn}})^{-1}$.
Taking the expectation $\Ebb^{\mathrm{po}}_{\thetaf, \tauf}$ of~\eqref{P2a:eq_Hessian_poisson_tau} yields
\begin{align*}
\Ebb^{\mathrm{po}}_{\thetaf, \tauf}(D^2 \ell_{\mathrm{pen}}^{\mathrm{po}} (\thetaf , \tauf))
&= (D \psi (\thetaf , \tauf)^\top \, \Ebb^{\mathrm{po}}_{\thetaf, \tauf}\left( [ D \check{\ell}_{\mathrm{pen}}^{\mathrm{po}} \circ \psi ] (\thetaf , \tauf) \right) \, D^2 \psi (\thetaf , \tauf) \, D \psi (\thetaf , \tauf)
\\
&\hspace{0.5cm} + \Ebb^{\mathrm{po}}_{\thetaf, \tauf}\left( [ D \check{\ell}_{\mathrm{pen}}^{\mathrm{po}} \circ \psi ] (\thetaf , \tauf) \right) \, D^2 \psi (\thetaf , \tauf) ,
\end{align*}
where the last summand is zero. 
To see this, first note that the $(K + 1)$-th to $(K + L)$-th elements in the row vector
$ \Ebb^{\mathrm{po}}_{\thetaf, \tauf} ( [ D \check{\ell}_{\mathrm{pen}}^{\mathrm{po}} \circ \psi ] (\thetaf , \tauf)) 
$, are
\begin{align*}
\Ebb^{\mathrm{po}}_{\thetaf, \tauf} \left( \frac{\partial}{\partial \lambda^{(l)}} \ell^{\mathrm{po}} (\psi_{\lambdaf} (\thetaf , \tauf)) \right)
= \frac{\sum_{g=1}^{\Gamma^{(l)}}\Ebb^{\mathrm{po}}_{\thetaf, \tauf} ( N_g^{(l)} )}{\psi_{K + s} (\thetaf, \tauf) } - 1 
= \frac{\sum_{g=1}^{\Gamma^{(l)}} \lambda_g^{(l)} (\thetaf, \tauf)}{\sum_{g=1}^{\Gamma^{(l)}} \lambda_g^{(l)} (\thetaf, \tauf) } - 1 
= 0
,
\end{align*}
where $l = 1, \ldots , L$.
Furthermore, $D^2 \psi_k (\thetaf , \tauf) = \mathbf{0}$ for $k = 1, \ldots , K$, i.e., the first $K$ elements of $D^2 \psi (\thetaf , \tauf)$ are $\mathbf{0}$.
Together, this yields $\Ebb^{\mathrm{po}}_{\thetaf, \tauf} ( [ D \check{\ell}_{\mathrm{pen}}^{\mathrm{po}} \circ \psi ] (\thetaf , \tauf) ) \, D^2 \psi (\thetaf , \tauf) = \mathbf{0}$. 
Since \eqref{P2a:eq_Dpsi_hat} holds for arbitrary $(\thetaf, \tauf)$ (not only for $(\thetafh, \taufh)$) an equivalent computation as the one leading to~\eqref{P2a:eq_Hessian_poisson_tauhat_expanded} yields
(ommitting the argument $(\thetaf, \tauf)$ in all functions)
\begin{align}
(\Ebb^{\mathrm{po}}_{\thetaf, \tauf} ( D^2 \ell_{\mathrm{pen}}^{\mathrm{po}} 
))_{[K]} 
= 
\Ebb^{\mathrm{po}}_{\thetaf, \tauf} ( D^2 \ell^{\mathrm{mn}}_{\mathrm{pen}} 
)
+ (D_{\thetaf} \psi_{\lambdaf} 
)^\top 
\, \Ebb^{\mathrm{po}}_{\thetaf, \tauf} ( [D^2 \ell^{\mathrm{po}} \circ \psi_{\lambdaf} 
]) \, D_{\thetaf} \psi_{\lambdaf} 
.
\label{P2a:eq_expected_Hessian_poisson_tau_RxR}
\end{align}
Furthermore,
\begin{align*}
&\bigl( \Ebb^{\mathrm{po}}_{\thetaf, \tauf} ( D^2 \ell_{\mathrm{pen}}^{\mathrm{po}} (\thetaf , \tauf)
) \bigr)^{-1}
\\
=
~&( D \psi (\thetaf , \tauf) )^{-1} 
\, \bigl( \Ebb^{\mathrm{po}}_{\thetaf, \tauf} ( [ D^2 \check{\ell}_{\mathrm{pen}}^{\mathrm{po}} \circ \psi ] (\thetaf , \tauf)  ) \bigr)^{-1} 
\, ( (D \psi (\thetaf , \tauf))^\top )^{-1}.
\end{align*}
Since 
\eqref{P2a:eq_Dpsi_hat_inv} holds for arbitrary $(\thetaf, \tauf)$ (not only for $(\thetafh, \taufh)$) and since
\begin{align*}
&\bigl( \Ebb^{\mathrm{po}}_{\thetaf, \tauf} ( [ D^2 \check{\ell}_{\mathrm{pen}}^{\mathrm{po}} \circ \psi ] (\thetaf , \tauf)  ) \bigr)^{-1} 
\\
= ~& 
\begin{pmatrix}
(\Ebb^{\mathrm{po}}_{\thetaf, \tauf} ( D^2 \ell_{\mathrm{pen}}^{\mathrm{mn}} (\thetaf ) 
))^{-1}
& 
\mathbf{0} 
\\
\mathbf{0} 
& 
(\Ebb^{\mathrm{po}}_{\thetaf, \tauf} ( [D^2 \ell^{\mathrm{po}} \circ \psi_{\lambdaf}] (\thetaf , \tauf) 
))^{-1}
\end{pmatrix}
\end{align*}
has again block matrix structure, setting $\Af = (D_{\tauf} \psi_{\lambdaf} (\thetaf , \tauf) 
)^{-1} \, D_{\thetaf} \psi_{\lambdaf} (\thetaf , \tauf)$ and $\Bf := ([D^2 \ell^{\mathrm{po}} \circ \psi_{\lambdaf}] (\thetaf , \tauf) 
)^{-1}$, we obtain
\begin{align*}
&\bigl( \Ebb^{\mathrm{po}}_{\thetaf, \tauf} ( D^2 \ell_{\mathrm{pen}}^{\mathrm{po}} (\thetaf , \tauf)
) \bigr)^{-1}
\\
=& \hspace{-0.1cm} 
\begin{pmatrix}
(\Ebb^{\mathrm{po}}_{\thetaf, \tauf} ( D^2 \ell_{\mathrm{pen}}^{\mathrm{mn}} (\thetaf ) 
))^{-1}
& 
- (\Ebb^{\mathrm{po}}_{\thetaf, \tauf} ( D^2 \ell_{\mathrm{pen}}^{\mathrm{mn}} (\thetaf ) 
))^{-1}
\, \Af^\top
\\
- \Af 
\, (\Ebb^{\mathrm{po}}_{\thetaf, \tauf} ( D^2 \ell_{\mathrm{pen}}^{\mathrm{mn}} (\thetaf ) 
))^{-1}
&
\Af
\, (\Ebb^{\mathrm{po}}_{\thetaf, \tauf} ( D^2 \ell_{\mathrm{pen}}^{\mathrm{mn}} (\thetaf ) 
))^{-1}
\, \Af ^\top
+
(D_{\tauf} \psi_{\lambdaf} 
)^{-1}
\, \Bf
\, ((D_{\tauf} \psi_{\lambdaf} 
)^{-1})^\top
\end{pmatrix} 
\hspace{-0.1cm}.
\end{align*}
Thus, $( ( \Ebb^{\mathrm{po}}_{\thetaf, \tauf} ( \Ff^{\mathrm{po}} 
) )^{-1})_{[K]}
=
(\Ebb^{\mathrm{po}}_{\thetaf, \tauf} ( \Ff^{\mathrm{mn}} 
))^{-1}$.
We now show that on the parameter space restricted by $\lambdaf = \nf$, this is identical to
$(\Ebb^{\mathrm{mn}}_{\thetaf} ( \Ff^{\mathrm{mn}} 
))^{-1}$.
We have
\begin{align*}
D^2_{\thetaf} \ell_{\mathrm{pen}}^{\mathrm{mn}} (\thetaf ) = \sum_{l=1}^L \sum_{g=1}^{\Gamma^{(l)}} n_g^{(l)} D^2_{\thetaf} \log (p_g^{(l)}(\thetaf) ) + D^2_{\thetaf} \pen(\thetaf).
\end{align*}
Taking the expectations $\mathbb{E}^{\mathrm{po}}_{\thetaf, \tauf}$ and $\mathbb{E}^{\mathrm{mn}}_{\thetaf}$ 
(replacing the observations $n_g^{(l)}$ with the random variables $N_g^{(l)}$),
$D^2_{\thetaf} \log (p_g^{(l)}(\thetaf) )$ and $D^2_{\thetaf} \pen(\thetaf)$ are constants and thus the expectations $\mathbb{E}^{\mathrm{po}}_{\thetaf, \tauf} (D^2_{\thetaf} \ell_{\mathrm{pen}}^{\mathrm{mn}} (\thetaf ) )$ and $\mathbb{E}^{\mathrm{mn}}_{\thetaf} (D^2_{\thetaf} \ell_{\mathrm{pen}}^{\mathrm{mn}} (\thetaf ) )$ only differ in $\mathbb{E}^{\mathrm{po}}_{\thetaf, \tauf} (N_g^{(l)}) = \lambda_g^{(l)} (\thetaf, \tauf)$ and  $\mathbb{E}^{\mathrm{mn}}_{\thetaf} (N_g^{(l)}) = n^{(l)} \, p_g^{(l)}(\thetaf)$ for $l = 1, \ldots , L, \, g = 1, \ldots , \Gamma^{(l)}$.
Under the restriction $\lambda^{(l)} = n^{(l)}$ for all $l = 1, \ldots , L$ (corresponding to the estimator $\lambdah^{(l)}$ of \eqref{P2a:eq_factorized_likelihood}), \eqref{P2a:eq_poisson_rate_transformation} gives us $\lambda_g^{(l)} (\thetaf, \tauf) = n^{(l)} \, p_g^{(l)}(\thetaf)$.
Thus, $\mathbb{E}^{\mathrm{po}}_{\thetaf, \tauf} (D^2_{\thetaf} \ell_{\mathrm{pen}}^{\mathrm{mn}} (\thetaf )  \, \vert \, \lambdaf = \nf) = \mathbb{E}^{\mathrm{mn}}_{\thetaf} (D^2_{\thetaf} \ell_{\mathrm{pen}}^{\mathrm{mn}} (\thetaf ) )$, which yields $((\Ebb^{\mathrm{po}}_{\thetaf, \tauf} ( \Ff^{\mathrm{po}} \, \vert \, \lambdaf = \nf))^{-1})_{[K]}
= (\Ebb^{\mathrm{mn}}_{\thetaf} ( \Ff^{\mathrm{mn}}))^{-1}$.
Furthermore, together with~\eqref{P2a:eq_expected_Hessian_poisson_tau_RxR} we obtain $(\Ebb^{\mathrm{po}}_{\thetaf, \tauf} ( \Ff^{\mathrm{po}} \, \vert \, \lambdaf = \nf 
))_{[K]} 
\neq
\Ebb^{\mathrm{mn}}_{\thetaf} ( \Ff^{\mathrm{mn}} 
)$.
\end{proof}

\subsection*{Formulation and Proof of Lemma~\ref{P2a:lemma_splitting_histograms}}
Recall the notation introduced to construct multinomial data in Section~\ref{P2:chapter_multinomial_regression},
in particular the sets of indices
$\Ical^{(l)} := \{ i \in \{ 1, \ldots , N \} ~|~ \xf_i = \xf^{(l)} \}$,
$\Ical_{\mathrm{d}}^{(l)} := \{ i \in \Ical^{(l)} ~|~ y_i \in \Yd 
\}$ 
and 
$\Ical_{\mathrm{c}}^{(l)} := \{ i \in \Ical^{(l)} ~|~  y_i \in \Yc \setminus \Yd \} 
$
for $l = 1, \ldots , L$.
To formalize the statement that finer partitions of $\Ical^{(l)}$, $l \in \{ 1, \ldots , L\}$ can be used for count data construction, without changing estimation results in Lemma~\ref{P2a:lemma_splitting_histograms} below, we extend this notation:
For $l \in \{1, \ldots, L\}$ let $\Ical_1^{(l)}, \ldots , \Ical_{\rho^{(l)}}^{(l)}$ be a partition of $\Ical^{(l)}$, where $\rho^{(l)} \in \Nbb$ and $\rho^{(l)} \leq n^{(l)}$.
For $r \in \{ 1, \ldots, \rho^{(l)}\}$, let $\Ical_{\mathrm{d} \, r}^{(l)} := \{ i \in \Ical_r^{(l)} ~|~ y_i \in \Yd\}$ and $\Ical_{\mathrm{c} \, r}^{(l)} := \Ical_r^{(l)} \setminus \Ical_{\mathrm{d} \, r}^{(l)}$. 
Furthermore,
for $l \in \{ 1, \ldots, L \}$ and $r \in \{ 1, \ldots, \rho^{(l)} \}$
set 
\[
\Ical_{g}^{(l)} := 
\begin{cases}
\Ical_{\mathrm{c}}^{(l)} & \text{for}~ g = 1, \ldots , G^{(l)}
\\
\Ical_{\mathrm{d}}^{(l)} & \text{for}~ g = G^{(l)} + 1, \ldots , \Gamma^{(l)}
\end{cases} 
~\text{and}~
\Ical_{g \, r}^{(l)} := 
\begin{cases}
\Ical_{\mathrm{c} \, r}^{(l)} & \text{for}~ g = 1, \ldots , G^{(l)}
\\
\Ical_{\mathrm{d} \, r}^{(l)} & \text{for}~ g = G^{(l)} + 1, \ldots , \Gamma^{(l)}
\end{cases} 
\, .
\]
Recall that $\Af_{[K]}$ denotes the upper left $K \times K$ submatrix of a matrix~$\Af$.

\begin{lem}\label{P2a:lemma_splitting_histograms}
For $g = 1, \ldots , \Gamma^{(l)}$, $l = 1, \ldots, L$, and $r = 1, \ldots, \rho^{(l)}$, let $n_g^{(l)} = 
\sum_{i \in \Ical_{g}^{(l)}} \mathbbm{1}_{U_g^{(l)}}(y_i)$
and
$n_{g \, r}^{(l)} 
= \sum_{i \in \Ical_{g \, r}^{(l)}} \mathbbm{1}_{U_g^{(l)}} (y_i) 
$.
For $n_g^{(l)}$, consider the distributional assumption as in Theorem~\ref{P2:proposition_equivalence_multniomial_poisson}~\ref{P2:proposition_equivalence_multniomial_poisson_assumption_poisson}, i.e.,
view $n_g^{(l)}$ as realizations of Poisson variables $N_g^{(l)} \sim Po (\lambda_g^{(l)} ( \thetaf, \tauf))$, independent for $g = 1, \ldots , \Gamma^{(l)}, l = 1, \ldots, L$, where the Poisson rates are modeled via 
$\lambda_g^{(l)} ( \thetaf, \tauf) = \exp ( \eta_g^{(l)} (\thetaf) + \tau^{(l)} )$ 
with a predictor $\eta_g^{(l)} (\thetaf) \in \Rbb$ for an unknown coefficient vector 
$\thetaf \in 
\Rbb^{K}$ and 
$\tauf = (\tau^{(1)}, \ldots , \tau^{(L)} )^\top \in \Rbb^L $.
Set $\lambda^{(l)} := \sum_{g=1}^{\Gamma^{(l)}} \lambda_g^{(l)} (\thetaf , \tauf)$, $n^{(l)} := \sum_{g=1}^{\Gamma^{(l)}} n_g^{(l)}$,  $\lambdaf = (\lambda^{(1)} , \ldots , \lambda^{(L)})$, and $\nf = (n^{(1)}, \ldots , n^{(L)})$ for $l = 1, \ldots , L$.
Given an additive penalty $\pen(\thetaf)$ on the log-likelihood, denote the PMLE for $(\thetaf, \tauf)$ with $(\thetafh, \taufh)$,
the 
penalized Fisher information at $(\thetafh, \taufh)$
with $\Ffh_{\mathrm{pen}}$
and the 
expected penalized Fisher information (with respect to the Poisson distributional assumption given $(\thetaf, \tauf)$) 
on the restricted parameter space 
$\{ (\thetaf , \tauf) \in \Rbb^K 
\times \Rbb^L ~|~ \lambdaf = \nf \}$ with $\Ebb_{\thetaf, \tauf} ( \Ff_{\mathrm{pen}} ~|~ \lambdaf = \nf)$. 
\begin{enumerate}[label=\arabic*)]
\item\label{P2a:lemma_splitting_histograms_wo_weights}
Viewing $n_{g \, r}^{(l)}$ as realizations of Poisson variables $N_{g \, r}^{(l)} \sim Po (\lambda_g^{(l)} ( \thetaf, \tauf))$,  independent for $g = 1, \ldots , \Gamma^{l)}, r = 1, \ldots , \rho^{(l)}, l = 1, \ldots , L$,
given the same penalty term $\pen(\thetaf)$ as above,
the PMLE for $\thetaf$ equals $\thetafh$, the upper left $K \times K$ submatrix of the inverse penalized Fisher information at the PMLEs for $(\thetaf, \tauf)$ 
equals $((\Ffh_{\mathrm{pen}})^{-1})_{[K]}$
and the upper left $K \times K$ submatrix of the inverse expected penalized Fisher information equals $((\Ebb_{\thetaf, \tauf} ( \Ff_{\mathrm{pen}} \, \vert \, \lambdaf = \nf))^{-1})_{[K]}$. 
\item\label{P2a:lemma_splitting_histograms_with_weights}
Viewing $n_{g \, r}^{(l)}$ as realizations of Poisson variables $N_{g \, r}^{(l)} \sim Po (\omega_{g \, r}^{(l)} \, \lambda_g^{(l)} ( \thetaf, \tauf)$, independent for $g = 1, \ldots , \Gamma^{l)}, r = 1, \ldots , \rho^{(l)}, l = 1, \ldots , L$, where $\omega_{g \, r}^{(l)} > 0$ are arbitrary weights with $\sum_{r = 1}^{\rho^{(l)}} \omega_{g \, r}^{(l)} = 1$, yields a penalized likelihood function which is proportional to the one of distributional assumption
on $n_g^{(l)}$,
i.e., in particular the same PMLE $(\thetafh, \taufh)$ for $(\thetaf, \tauf)$ and identical penalized Fisher informations $\Ffh_{\mathrm{pen}}$ and $\Ebb_{\thetaf, \tauf} ( \Ff_{\mathrm{pen}} ~|~ \lambdaf = \nf)$.
\end{enumerate}
\end{lem} 

\begin{rem}
\begin{enumerate}[label = \arabic*)]
\item
Distributional assumption~\ref{P2a:lemma_splitting_histograms_wo_weights} ``only'' yields equal results with respect to $\thetaf$.
By adjusting the Poisson rates via weights in \ref{P2a:lemma_splitting_histograms_with_weights}, we obtain the same estimation results with respect to the whole vector of unknown parameters $(\thetaf, \tauf)$.
As we are only interested in estimating $\thetaf$, it is sufficient to use distributional assumption~\ref{P2a:lemma_splitting_histograms_wo_weights}.
\item
By Theorem~\ref{P2:proposition_equivalence_multniomial_poisson}, the estimation results with respect to $\thetaf$ obtained by both distributional assumptions~\ref{P2a:lemma_splitting_histograms_wo_weights} and~\ref{P2a:lemma_splitting_histograms_with_weights} on $n_{g \, r}^{(l)}$ are also identical to the ones obtained by the multinomial distributional assumption on $\mathbf{n}^{(l)} = (n_1^{(l)}, \ldots , n_{\Gamma^{(l)}}^{(l)}) \in \Nbb^{\Gamma^{(l)}}$ as in Theorem~\ref{P2:proposition_equivalence_multniomial_poisson}~\ref{P2:proposition_equivalence_multniomial_poisson_assumption_multinom}.
\item
In the most extreme case, $\rho^{(l)} = n^{(l)}$ for all $l = 1, \ldots, L$, the construction in the lemma yields one vector of counts per observation (with one count of $1$ and all others $0$).
While this is usually not desirable with respect to computation time in practice, since it inflates the size of the data set artificially, it shows the estimation does not depend on aggregation of observations and the presented approach is thus overall coherent.
\end{enumerate}
\end{rem}

\begin{proof}[Proof of Lemma~\ref{P2a:lemma_splitting_histograms}]
\begin{enumerate}[label = \arabic*)]
\item
First, note that per construction, the sets $\Ical_{\mathrm{c} \, 1}^{(l)}, \ldots , \Ical_{\mathrm{c} \, \rho^{(l)}}^{(l)}$ form a partition of $\Ical_{\mathrm{c}}^{(l)}$ for $l = 1, \ldots , L$ (analogously for $\Ical_{\mathrm{d}}^{(l)}$).
This yields 
\begin{align*}
\sum_{r=1}^{\rho^{(l)}} \left( n_{g \, r}^{(l)} \right)
= \sum_{r=1}^{\rho^{(l)}} \sum_{i \in \Ical_{g \, r}^{(l)}} \mathbbm{1}_{U_g^{(l)}} (y_i) 
= \sum_{i \in \Ical_{g}^{(l)}} \mathbbm{1}_{U_g^{(l)}} (y_i) 
= n_g^{(l)},
\end{align*}
for $g = 1, \ldots , \Gamma^{(l)}$.
With $\Pen(\thetaf) := \exp (\pen(\thetaf))$, the penalized likelihood function given independent observations $n_{g \, r}^{(l)}$ of $N_{g \, r}^{(l)} \sim Po ( \lambda_g^{(l)} ( \thetaf, \tauf))$ is
\begin{align*}
\prod_{l=1}^L \prod_{g=1}^{\Gamma^{(l)}} \prod_{r=1}^{\rho^{(l)}} 
\Biggl( \frac{\lambda_g^{(l)} \left( \thetaf, \tauf \right)^{n_{g \, r}^{(l)}}}{n_{g \, r}^{(l)}!} \exp\left(-\lambda_g^{(l)} \left( \thetaf, \tauf \right)\right) \Biggr)
\, \Pen(\thetaf)
. 
\end{align*}
As in 
the Proof of Theorem~\ref{P2:proposition_equivalence_multniomial_poisson}, we reparameterize $(\thetaf, \tauf) = \psi^{-1} (\thetaf, \lambdaf)$ to obtain $\lambda_g^{(l)} (\thetaf , \tauf)
= \lambda_g^{(l)} \left( \psi^{-1} (\thetaf , \lambdaf) \right)
= [ \exp ( \eta_g^{(l)} (\thetaf) ) \, \lambda^{(l)}] \, / \, [\sum_{g'=1}^{\Gamma^{(l)}}  \exp (\eta_{g'}^{(l)} (\thetaf) ) ]$.
Then, analogously to \eqref{P2a:eq_factorized_likelihood}, the penalized likelihood is proportional to
\begin{align*}
&\prod_{l=1}^L \Biggl( \frac{n^{(l)}!}{\prod_{g=1}^{\Gamma^{(l)}} n_g^{(l)}!} \prod_{g=1}^{\Gamma^{(l)}} \Biggl( \frac{ \exp \bigl( \eta_g^{(l)} (\thetaf) \bigr)}{\sum_{g'=1}^{\Gamma^{(l)}}  \exp\bigl(\eta_{g'}^{(l)} (\thetaf) \bigr)} \Biggr)^{n_g^{(l)}} \Biggr) \, \Pen(\thetaf) \notag
\\
& 
\cdot \prod_{l=1}^L \Bigl( \frac{(\rho^{(l)} \, \lambda^{(l)})^{n^{(l)}}}{n^{(l)}!} \exp(- \rho^{(l)} \, \lambda^{(l)}) \Bigr). 
\end{align*}
The first part, modeling $\thetaf$, is identically to the first part in the factorized penalized likelihood~\eqref{P2a:eq_factorized_likelihood} and thus yields $\thetafh$ as PMLE for $\thetaf$. 
The second part, which models $\lambda^{(l)}$ (and indirectly $\tau^{(l)}$) for $l = 1, \ldots , L$, differs by the additional factor $\rho^{(l)}$.
Since the structure of the factorization is equal to the one in~\eqref{P2a:eq_factorized_likelihood}, the proof that the upper $K \times K$ submatrices of the inverse penalized Fisher information at the PMLE for $(\thetaf, \tauf)$ and of the inverse of the expected version, i.e., the submatrices corresponding to $\thetaf$, are identical to $((\Ffh_{\mathrm{pen}})^{-1})_{[K]}$ and $((\Ebb^{\mathrm{po}}_{\thetaf, \tauf} ( \Ff_{\mathrm{pen}} \, \vert \, \lambdaf = \nf))^{-1})_{[K]}$, 
can be carried out analogously to the respective part of the Proof of Theorem~\ref{P2:proposition_equivalence_multniomial_poisson}.
\item
Using the adjusted Poisson rates, i.e., viewing $n_{g \, r}^{(l)}$ as independent observations of $N_{g \, r}^{(l)} \sim Po ( \omega_{g \, r}^{(l)} \, \lambda_g^{(l)} ( \thetaf, \tauf))$, where $r = 1, \ldots , \rho^{(l)}$, $g = 1, \ldots , \Gamma^{(l)}$, $l = 1, \ldots , L$, the penalized likelihood function is
\begin{align*}
&\prod_{l=1}^L \prod_{g=1}^{\Gamma^{(l)}} \prod_{r=1}^{\rho^{(l)}} 
\Biggl( \frac{\omega_{g \, r}^{(l)} \, \lambda_g^{(l)} \left( \thetaf, \tau^{(l)} \right)^{n_{g \, r}^{(l)}}}{n_{g \, r}^{(l)}!} \exp\left(-\omega_{g \, r}^{(l)} \, \lambda_g^{(l)} \left( \thetaf, \tau^{(l)} \right)\right) \Biggr)
\, \Pen(\thetaf)
\\
\propto ~&
\prod_{l=1}^L \prod_{g=1}^{\Gamma^{(l)}} 
\Bigl( \lambda_g^{(l)} \left( \thetaf, \tau^{(l)} \right)^{n_{g}^{(l)}} \exp\left(-\lambda_g^{(l)} \left( \thetaf, \tau^{(l)} \right)\right) \Bigr)
\, \Pen(\thetaf) 
= \Lcal_{\mathrm{pen}}^{\mathrm{po}} ( \thetaf, \tauf ),
\end{align*}
where $\Lcal_{\mathrm{pen}}^{\mathrm{po}} ( \thetaf, \tauf )$ is the penalized likelihood~\eqref{P2a:equation_likelihood} given independent observations $n_g^{(l)}$ of $N_g^{(l)} \sim Po ( \lambda_g^{(l)} ( \thetaf, \tauf) )$. 
Hence, the penalized likelihoods are proportional and thus yield the same estimation results.
\qedhere
\end{enumerate}
\end{proof}

\section{Transforming effects to desired reference coding}\label{P2a:chapter_reference_coding_application}
The function \texttt{densreg()} used for fitting density regression models contained in our \texttt{R} package \texttt{DensityRegression} (developer version on \url{https://github.com/Eva2703/DensityRegression})
uses the function \texttt{gam()} of the well-known package \texttt{mgcv} \citep{wood2017} for Poisson regression. To the best of the author's knowledge, the reference coding used in 
\citet{maier2021} cannot be specified directly using \texttt{mgcv::gam()}, since it does not seem intended to specify a reference value for a continuous covariate like \emph{year}.
Thus, also the reference \emph{year} cannot be specified directly in \texttt{densreg()}.
Instead, we transform the received estimated effects, appropriately.
The model actually estimated via \texttt{densreg()} is
\begin{align}
f_{\text{\emph{West\_East, c\_age, year}}} &= \betat_0 \oplus \betat_{\text{\emph{West\_East}}} 
\oplus \betat_{c\_age} \oplus \betat_{\text{\emph{c\_age, West\_East}}} \notag \\
&\hspace{0.45 cm} \oplus \tilde{g}_0(year) \oplus \tilde{g}_{\text{\emph{West\_East}}} (year) \oplus \tilde{g}_{c\_age} (year) 
\notag \\ 
&\hspace{0.45 cm} \oplus \tilde{g}_{\text{\emph{c\_age, West\_East}}}(year)
, \label{P2a:soep_model_effect}
\end{align}
with reference categories \emph{West} for \emph{West\_East} and \emph{other} for \emph{c\_age}.\footnote{Regarding the interaction effects for \emph{West\_East} and \emph{c\_age}, note that it is not possible to use the argument \texttt{by} of the smoother \texttt{mgcv::ti()} with more than one variable.
Thus, we combine \emph{c\_age} and \emph{West\_East} in a new covariate for the corresponding interaction terms, where all combinations containing either $\text{\emph{West\_East}} = \text{\emph{West}}$ or $\text{\emph{c\_age}} = \text{\emph{other}}$ are set to the same (reference) category, such that they all have an effect of zero, in accordance with the reference coding.}
The effects as discussed in Section~\ref{P2a:chapter_application}, 
i.e., with reference category $1991$ for \emph{year}, are obtained by adding the respective smooth effect for $1991$ for (group specific) intercepts and subtracting it for (group specific) smooth effects:
\begin{align*}
\beta_0 &= \betat_0 \oplus \tilde{g}_0(1991)
\\
\beta_{\text{\emph{West\_East}}} &= \betat_{\text{\emph{West\_East}}} 
\oplus \tilde{g}_{\text{\emph{West\_East}}} (1991) 
\\
\beta_{\text{\emph{c\_age}}} &= \betat_{\text{\emph{c\_age}}} 
\oplus \tilde{g}_{\text{\emph{c\_age}}}(1991) 
\\
\beta_{\text{\emph{c\_age, West\_East}}} &= \betat_{\text{\emph{c\_age, West\_East}}} 
\oplus \tilde{g}_{\text{\emph{c\_age, West\_East}}} (1991) 
\\
g_0(\text{\emph{year}}) &= \tilde{g}_0 (\text{\emph{year}}) \ominus \tilde{g}_0(1991) 
\\
g_{\text{\emph{West\_East}}}(\text{\emph{year}}) &= \tilde{g}_{\text{\emph{West\_East}}} (year) 
\ominus \tilde{g}_{West\_East}(1991) 
\\
g_{\text{\emph{c\_age}}} (\text{\emph{year}}) &= \tilde{g}_{\text{\emph{c\_age}}} (\text{\emph{year}}) 
\ominus \tilde{g}_{\text{\emph{c\_age}}}(1991) 
\\
g_{\text{\emph{c\_age, West\_East}}}(\text{\emph{year}}) &= \tilde{g}_{\text{\emph{c\_age, West\_East}}}(\text{\emph{year}}) 
\ominus \tilde{g}_{\text{\emph{c\_age, West\_East}}} (1991) 
\end{align*}
and the resulting model is equivalent to model~(\ref{P2a:soep_model_effect}), i.e.,
\begin{align*}
f_{\text{\emph{West\_East, c\_age, year}}} 
&= \beta_0 \oplus \beta_{\text{\emph{West\_East}}} 
\oplus \beta_{c\_age} \oplus \beta_{\text{\emph{c\_age, West\_East}}} \notag \\
&\hspace{0.45 cm} \oplus g_0(year) \oplus g_{\text{\emph{West\_East}}} (year) \oplus g_{c\_age} (year) 
\notag \\
&\hspace{0.45 cm} \oplus g_{\text{\emph{c\_age, West\_East}}}(year)
\notag \\
&= \betat_0 \oplus \betat_{West\_East} 
\oplus \betat_{c\_age} \oplus \betat_{c\_age, \, West\_East} \notag \\
&\hspace{0.45 cm} \oplus \tilde{g}_0 (year) \oplus \tilde{g}_{\text{\emph{West\_East}}} (year) \oplus \tilde{g}_{c\_age} (year) 
\notag \\
&\hspace{0.45 cm} \oplus \tilde{g}_{c\_age, \, West\_East}(year). \notag 
\end{align*}

To obtain confidence regions for the transformed effects, note that all transformations are of the form 
$\hat{h} 
= \bigoplus_{j \in \Jcal} \varsigma_j \odot \hat{h}_{j} (\xf_{i_j})
= \bigoplus_{j \in \Jcal} \varsigma_j \odot ( \bfe_{\Xcal, \, j}(\xf_{i_j}) \ootimes \bfe_{\Ycal} )^\top \thetafh_{\mathrm{po}, \, j}$, 
where $\varsigma_j \in \{ -1, 1\}, ~ \Jcal \subseteq \{1, \ldots , J\}$, and $i_j \in \{ 1, \ldots , N \}$ for $j \in \Jcal$, where $\hat{h}_{j} (\xf_{i_j})$ denote estimated partial effects, compare Section~\ref{P2:chapter_bayes_space_regression}.
Confidence regions can be constructed as in Lemma~\ref{P2:lemma_confidence_regions}, but letting the covariate $\xf$ vary with $j \in \Jcal$ and adding the signs $\varsigma_j$ to the definition of the transformation matrix, i.e., setting 
$\Af
= \sum_{j \in \Jcal} \varsigma_j \cdot (\bfe_{\Xcal} (\xf_{i_j})^\top \otimes \Id_{K_\Ycal}) \Sf_{j}$.


\section{Including observation weights}\label{P2a:chapter_weights}
The SOEP data analyzed in Section~\ref{P2a:chapter_application} contains sampling weights for all individual observations.
Denote the sampling weight of the $i$-th observation with $\gamma_i$, $i = 1, \ldots , N$.
They would reasonably be included by replacing the counts $n_g^{(l)}$ in Section~\ref{P2:chapter_multinomial_regression} with weighted counts, i.e., setting 
$n_g^{(l)} := \sum_{i \in \Ical_{g}^{(l)}} \gamma_i \mathbbm{1}_{U_g^{(l)}}(y_i)$ for $g = 1, \ldots , \Gamma^{(l)}, l = 1, \ldots , L$,
where $\Ical_{g }^{(l)} := \Ical_{\mathrm{c}}^{(l)}$, if $g = 1, \ldots , G^{(l)}$, and $\Ical_{g}^{(l)} := \Ical_{\mathrm{d}}^{(l)}$, if $g = G^{(l)} + 1, \ldots , \Gamma^{(l)}$.
However, this is not feasible, since the weights $\gamma_i$ included in the SOEP data are not natural numbers, but positive real numbers.\footnote{Actually, there are even weights of exactly zero. The corresponding observations were removed from the data set, with $151,838$ observations having positive real weights remaining.}
Thus, also the weighted counts $n_g^{(l)}$ are not natural numbers and theoretically cannot serve as response observations for a Poisson model.
Consequently, trying to estimate such a model via the function \texttt{gam()} of the \texttt{R} package \texttt{mgcv} \citep{wood2017}, which our newly developed \texttt{R} package \href{P2a:https://github.com/Eva2703/DensityRegression}{\texttt{DensityRegression}} 
uses to fit Poisson models, produces an error.
In the following, we thus present a work-around, which yields the desired likelihood of a weighted counts Poisson model, but is specifiable in \texttt{gam()}.
For this purpose, we 
use the (unweighted) counts $n_g^{(l)}$ as defined in Section~\ref{P2:chapter_multinomial_regression} as observed response values, but appropriately weight the log-likelihood contributions (specifiable via the \texttt{gam()}-argument \texttt{weights}) and include offsets (specifiable in \texttt{gam()} via an \texttt{offset} in the \texttt{formula} of the model) that cancel the undesired terms. 
In the function \texttt{densreg()}, which is used to fit regression models with our approach in our \texttt{R} package \href{P2a:https://github.com/Eva2703/DensityRegression}{\texttt{DensityRegression}}, the sample weights $\gamma_1, \ldots , \gamma_N$ can be specified via the argument \texttt{sample\_weights}. 
The weights for the log-likelihood contributions and the offset are then constructed appropriately by \texttt{densreg()}.
%

The weights for the log-likelihood contributions are
\begin{align*}
v_g^{(l)} 
&:= \begin{cases}
\frac{\sum_{i \in \Ical_{g}^{(l)}} \gamma_i \mathbbm{1}_{U_g^{(l)}}(y_i)}{\sum_{i \in \Ical_{g}^{(l)}} \mathbbm{1}_{U_g^{(l)}}(y_i)} 
= \frac{\sum_{i \in \Ical_{g}^{(l)}} \gamma_i \mathbbm{1}_{U_g^{(l)}}(y_i)}{n_g^{(l)}}
& , n_g^{(l)} > 0 
\\
1 & , n_g^{(l)} = 0
\end{cases}
, \quad l = 1, \ldots , L, ~g = 1, \ldots , \Gamma^{(l)}.
\end{align*}
The Poisson rates then are modeled with offsets $-\log v_g^{(l)}$, i.e., $\tilde{\lambda}_g^{(l)} := \exp(\eta_g^{(l)} - \log v_g^{(l)}) = \frac{\lambda_g^{(l)}}{v_g^{(l)}}$, where $\lambda_g^{(l)} := \exp(\eta_g^{(l)})$ is the Poisson rate of the corresponding model without offset, $l = 1, \ldots , L, ~g = 1, \ldots , \Gamma^{(l)}$.
Then, the $l$-th log-likelihood contribution in the Poisson model with weighted log-likelihood contributions and offsets is (omitting additive constants via $\propto$)
\begin{align*}
\sum_{g=1}^{\Gamma^{(l)}} v_g^{(l)} \cdot \left( n_g^{(l)} \log \tilde{\lambda}_g^{(l)} - \tilde{\lambda}_g^{(l)} \right)
&= \sum_{g=1}^{\Gamma^{(l)}} \left( v_g^{(l)} n_g^{(l)} \log \lambda_g^{(l)} - v_g^{(l)} n_g^{(l)} \log v_g^{(l)} - v_g^{(l)} \frac{\lambda_g^{(l)}}{v_g^{(l)}}\right)
\\
&\propto \sum_{g=1}^{\Gamma^{(l)}} \left( v_g^{(l)} n_g^{(l)} \log \lambda_g^{(l)} - \lambda_g^{(l)} \right),
\end{align*}
which corresponds to the $l$-th log-likelihood contribution in a Poisson model with rate $\lambda_g^{(l)}$ and observations 
\begin{align*}
v_g^{(l)} n_g^{(l)} 
&= \frac{\sum_{i \in \Ical_{g}^{(l)}} \gamma_i \mathbbm{1}_{U_g^{(l)}}(y_i)}{\sum_{i \in \Ical_{g}^{(l)}} \mathbbm{1}_{U_g^{(l)}}(y_i)} \sum_{i \in \Ical_{g}^{(l)}} \mathbbm{1}_{U_g^{(l)}}(y_i) 
= \sum_{i \in \Ical_{g}^{(l)}} \gamma_i \mathbbm{1}_{U_g^{(l)}}(y_i) , ~ g = 1, \ldots , \Gamma^{(l)},
\end{align*}
i.e., to the $l$-th log-likelihood contribution of a Poisson model using the weighted counts as observations as desired.

\pagebreak
\section{Application}\label{P2a:chapter_application}
In this section, we provide further material related to the application presented in Section~\ref{P2a:chapter_application}.

\subsection{Smoothing Parameters}\label{P2a:chapter_application_smoothing_parameters}
Table~\ref{P2a:table_smoothing_parameters} shows the smoothing parameters 
(rounded to three digits) corresponding to the (group-specific) smooth effects -- modeled via cubic B-splines with second order difference penalty -- obtained by REML 
optimization \citep[Section~6.5.2]{wood2017}.
Note that for all unpenalized marginal basis functions (i.e., $\bfe_{\Ycal}$ and all basis functions $\bfe_j$ modeling non-smooth covariate effects), the corresponding smoothing parameters are zero and not listed in the table.

\begin{table}[H]
\scriptsize
\begin{center}
\begin{tabular}{|l|l|}
\hline 
partial effect $h_j (\xf)$ & 
$\xi_{\Xcal, \, j}$
\\
\hline
$g(year)$ &							  $70.979$ \\
$g_{\text{\emph{East}}} (year)$ &        $302.005$ \\
$g_{\text{\emph{0-6}}} (year)$ &        $8.564$ \\
\hline 
\end{tabular} 
\hspace{0.15cm}
\begin{tabular}{|l|l|}
\hline 
partial effect $h_j (\xf)$ & 
$\xi_{\Xcal, \, j}$
\\
\hline
$g_{\text{\emph{7-18}}} (year)$ &        $28.475$ \\
$g_{\text{\emph{0-6, East}}} (year)$ &   $6.399$ \\
$g_{\text{\emph{7-18, East}}} (year)$ &   $18.311$ \\
\hline 
\end{tabular} 
\caption{Smoothing parameters for (group-specific) smooth covariate effects.\label{P2a:table_smoothing_parameters}}
\end{center}
\end{table}

\subsection{Significance and p-values}\label{P2a:chapter_application_pvalues}
%

Table~\ref{P2a:table_pvalues_simultan} shows the p-values (rounded to three digits) for the null hypothesis that a partial effects equals zero as in Lemma~\ref{P2a:lemma_confidence_regions_simultaneos}, which are simultaneous in the covariates.\footnote{Note that the effects considered here are actually the ones obtained by model~\eqref{P2a:soep_model_effect}, i.e., without transformation to the reference year 1991 as described in appendix~\ref{P2a:chapter_reference_coding_application}, which does not change whether a covariate has an influence or not.}
For most effects, the obtained p-values are smaller than $0.01$ and in particular significant for the significance level $\alpha = 5 \%$.
The only two exceptions are the smooth interaction effects of \emph{year} with \emph{West\_East} and of \emph{year} with \emph{c\_age} and \emph{West\_East}, whose p-values are very large with values close to $1$.


\begin{table}[H]
\scriptsize
\begin{center}
\begin{tabular}{|l|l|}
\hline 
partial effect & p-value
\\
\hline
$\beta_0$ &                               $0.000$ \\
$\beta_{\text{\emph{West\_East}}}$ &          $0.000$ \\
$\beta_{\text{\emph{c\_age}}}$ &          $0.000$ \\
$\beta_{\text{\emph{c\_age, West\_East}}}$ &     $0.000$ \\
\hline
\end{tabular}
\hspace{0.15cm}
\begin{tabular}{|l|l|}
\hline 
partial effect & p-value
\\
\hline
$g(year)$ &							  $0.000$ \\
$g_{\text{\emph{West\_East}}} (year)$ &        $0.999$ \\
$g_{\text{\emph{c\_age}}} (year)$ &        $0.000$ \\
$g_{\text{\emph{c\_age, West\_East}}} (year)$ &   $0.070$ \\
\hline 
\end{tabular} 
\caption{p-values for $H_0: h_j(\xf) = 0$ for the different partial effects $h_j(\xf)$ (simultaneous in covariates).\label{P2a:table_pvalues_simultan}}
\end{center}
\end{table}

%
Table~\ref{P2a:table_pvalues_pointwise} shows the p-values pointwise regarding covariates for each partial effect (in the notation of Lemma~\ref{P2:lemma_confidence_regions}, $\Jcal = \{ j \}$ for $j \in \{ 1, \ldots , J\}$) and the respectively relevant covariate values (rounded to three digits).
Effects, which are zero due to reference coding are not included in the table.
Note that most p-values are very small. 
For a significance level $\alpha = 5 \%$, almost all effects are significant with the only exceptions 
the smooth interaction effect of \emph{year}, \emph{0-6}, and \emph{East} for $1992$ to $1994$, and the smooth interaction effect of \emph{year}, \emph{7-18}, and \emph{East} for $1992$ to $1997$.
Note that all of those effects evaluated at the reference year $1991$ are exactly zero, thus it is no surprise that the years following this one (in \emph{East}, there are no earlier observations) are not significant/have larger p-values.

\begin{table}[H]
\scriptsize
\begin{center}
\begin{tabular}{|l|l|}
\hline 
partial effect & p-value
\\
\hline
$\beta_0$ &                               $0.000$ \\
$\beta_{\text{\emph{East}}}$ &          $0.000$ \\
$\beta_{\text{\emph{0-6}}}$ &           $0.000$ \\
$\beta_{\text{\emph{7-18}}}$ &          $0.000$ \\
$\beta_{\text{\emph{0-6, East}}}$ &     $0.000$ \\
$\beta_{\text{\emph{7-18, East}}}$ &    $0.000$ \\
$g(1984)$ &                                $0.000$ \\
$g(1985)$ &                                $0.000$ \\
$g(1986)$ &                                $0.000$ \\
$g(1987)$ &                                $0.000$ \\
$g(1988)$ &                                $0.000$ \\
$g(1989)$ &                                $0.000$ \\
$g(1990)$ &                                $0.000$ \\
$g(1992)$ &                                $0.000$ \\
$g(1993)$ &                                $0.000$ \\
$g(1994)$ &                                $0.000$ \\
$g(1995)$ &                                $0.000$ \\
$g(1996)$ &                                $0.000$ \\
$g(1997)$ &                                $0.000$ \\
$g(1998)$ &                                $0.000$ \\
$g(1999)$ &                                $0.000$ \\
$g(2000)$ &                                $0.000$ \\
$g(2001)$ &                                $0.000$ \\
$g(2002)$ &                                $0.000$ \\
$g(2003)$ &                                $0.000$ \\
$g(2004)$ &                                $0.000$ \\
$g(2005)$ &                                $0.000$ \\
$g(2006)$ &                                $0.000$ \\
$g(2007)$ &                                $0.000$ \\
$g(2008)$ &                                $0.000$ \\
$g(2009)$ &                                $0.000$ \\
$g(2010)$ &                                $0.000$ \\
$g(2011)$ &                                $0.000$ \\
$g(2012)$ &                                $0.000$ \\
$g(2013)$ &                                $0.000$ \\
$g(2014)$ &                                $0.000$ \\
$g(2015)$ &                                $0.000$ \\
$g(2016)$ &                                $0.000$ \\
$g_{\text{\emph{East}}} (1992)$ &        $0.019$ \\
$g_{\text{\emph{East}}} (1993)$ &        $0.012$ \\
$g_{\text{\emph{East}}} (1994)$ &        $0.008$ \\
$g_{\text{\emph{East}}} (1995)$ &        $0.005$ \\
$g_{\text{\emph{East}}} (1996)$ &        $0.003$ \\
$g_{\text{\emph{East}}} (1997)$ &        $0.002$ \\
$g_{\text{\emph{East}}} (1998)$ &        $0.002$ \\
\hline
\end{tabular}
\hfill
\begin{tabular}{|l|l|}
\hline 
partial effect & p-value
\\
\hline
$g_{\text{\emph{East}}} (1999)$ &        $0.001$ \\
$g_{\text{\emph{East}}} (2000)$ &        $0.001$ \\
$g_{\text{\emph{East}}} (2001)$ &        $0.001$ \\
$g_{\text{\emph{East}}} (2002)$ &        $0.000$ \\
$g_{\text{\emph{East}}} (2003)$ &        $0.000$ \\
$g_{\text{\emph{East}}} (2004)$ &        $0.000$ \\
$g_{\text{\emph{East}}} (2005)$ &        $0.000$ \\
$g_{\text{\emph{East}}} (2006)$ &        $0.000$ \\
$g_{\text{\emph{East}}} (2007)$ &        $0.000$ \\
$g_{\text{\emph{East}}} (2008)$ &        $0.000$ \\
$g_{\text{\emph{East}}} (2009)$ &        $0.000$ \\
$g_{\text{\emph{East}}} (2010)$ &        $0.000$ \\
$g_{\text{\emph{East}}} (2011)$ &        $0.000$ \\
$g_{\text{\emph{East}}} (2012)$ &        $0.000$ \\
$g_{\text{\emph{East}}} (2013)$ &        $0.000$ \\
$g_{\text{\emph{East}}} (2014)$ &        $0.000$ \\
$g_{\text{\emph{East}}} (2015)$ &        $0.000$ \\
$g_{\text{\emph{East}}} (2016)$ &        $0.000$ \\
$g_{\text{\emph{0-6}}} (1984)$ &         $0.000$ \\
$g_{\text{\emph{0-6}}} (1985)$ &         $0.000$ \\
$g_{\text{\emph{0-6}}} (1986)$ &         $0.000$ \\
$g_{\text{\emph{0-6}}} (1987)$ &         $0.000$ \\
$g_{\text{\emph{0-6}}} (1988)$ &         $0.000$ \\
$g_{\text{\emph{0-6}}} (1989)$ &         $0.000$ \\
$g_{\text{\emph{0-6}}} (1990)$ &         $0.001$ \\
$g_{\text{\emph{0-6}}} (1992)$ &         $0.016$ \\
$g_{\text{\emph{0-6}}} (1993)$ &         $0.007$ \\
$g_{\text{\emph{0-6}}} (1994)$ &         $0.002$ \\
$g_{\text{\emph{0-6}}} (1995)$ &         $0.001$ \\
$g_{\text{\emph{0-6}}} (1996)$ &         $0.001$ \\
$g_{\text{\emph{0-6}}} (1997)$ &         $0.000$ \\
$g_{\text{\emph{0-6}}} (1998)$ &         $0.000$ \\
$g_{\text{\emph{0-6}}} (1999)$ &         $0.000$ \\
$g_{\text{\emph{0-6}}} (2000)$ &         $0.000$ \\
$g_{\text{\emph{0-6}}} (2001)$ &         $0.000$ \\
$g_{\text{\emph{0-6}}} (2002)$ &         $0.000$ \\
$g_{\text{\emph{0-6}}} (2003)$ &         $0.000$ \\
$g_{\text{\emph{0-6}}} (2004)$ &         $0.000$ \\
$g_{\text{\emph{0-6}}} (2005)$ &         $0.000$ \\
$g_{\text{\emph{0-6}}} (2006)$ &         $0.000$ \\
$g_{\text{\emph{0-6}}} (2007)$ &         $0.000$ \\
$g_{\text{\emph{0-6}}} (2008)$ &         $0.000$ \\
$g_{\text{\emph{0-6}}} (2009)$ &         $0.000$ \\
$g_{\text{\emph{0-6}}} (2010)$ &         $0.000$ \\
~ & ~ \\
\hline
\end{tabular}
\hfill
\begin{tabular}{|l|l|}
\hline 
partial effect & p-value
\\
\hline
$g_{\text{\emph{0-6}}} (2011)$ &         $0.000$ \\
$g_{\text{\emph{0-6}}} (2012)$ &         $0.000$ \\
$g_{\text{\emph{0-6}}} (2013)$ &         $0.000$ \\
$g_{\text{\emph{0-6}}} (2014)$ &         $0.000$ \\
$g_{\text{\emph{0-6}}} (2015)$ &         $0.000$ \\
$g_{\text{\emph{0-6}}} (2016)$ &         $0.000$ \\
$g_{\text{\emph{7-18}}} (1984)$ &        $0.001$ \\
$g_{\text{\emph{7-18}}} (1985)$ &        $0.000$ \\
$g_{\text{\emph{7-18}}} (1986)$ &        $0.000$ \\
$g_{\text{\emph{7-18}}} (1987)$ &        $0.000$ \\
$g_{\text{\emph{7-18}}} (1988)$ &        $0.000$ \\
$g_{\text{\emph{7-18}}} (1989)$ &        $0.000$ \\
$g_{\text{\emph{7-18}}} (1990)$ &        $0.000$ \\
$g_{\text{\emph{7-18}}} (1992)$ &        $0.000$ \\
$g_{\text{\emph{7-18}}} (1993)$ &        $0.000$ \\
$g_{\text{\emph{7-18}}} (1994)$ &        $0.000$ \\
$g_{\text{\emph{7-18}}} (1995)$ &        $0.000$ \\
$g_{\text{\emph{7-18}}} (1996)$ &        $0.000$ \\
$g_{\text{\emph{7-18}}} (1997)$ &        $0.000$ \\
$g_{\text{\emph{7-18}}} (1998)$ &        $0.000$ \\
$g_{\text{\emph{7-18}}} (1999)$ &        $0.000$ \\
$g_{\text{\emph{7-18}}} (2000)$ &        $0.000$ \\
$g_{\text{\emph{7-18}}} (2001)$ &        $0.000$ \\
$g_{\text{\emph{7-18}}} (2002)$ &        $0.000$ \\
$g_{\text{\emph{7-18}}} (2003)$ &        $0.000$ \\
$g_{\text{\emph{7-18}}} (2004)$ &        $0.000$ \\
$g_{\text{\emph{7-18}}} (2005)$ &        $0.000$ \\
$g_{\text{\emph{7-18}}} (2006)$ &        $0.000$ \\
$g_{\text{\emph{7-18}}} (2007)$ &        $0.000$ \\
$g_{\text{\emph{7-18}}} (2008)$ &        $0.000$ \\
$g_{\text{\emph{7-18}}} (2009)$ &        $0.000$ \\
$g_{\text{\emph{7-18}}} (2010)$ &        $0.000$ \\
$g_{\text{\emph{7-18}}} (2011)$ &        $0.000$ \\
$g_{\text{\emph{7-18}}} (2012)$ &        $0.000$ \\
$g_{\text{\emph{7-18}}} (2013)$ &        $0.000$ \\
$g_{\text{\emph{7-18}}} (2014)$ &        $0.000$ \\
$g_{\text{\emph{7-18}}} (2015)$ &        $0.000$ \\
$g_{\text{\emph{7-18}}} (2016)$ &        $0.000$ \\
$g_{\text{\emph{0-6, East}}} (1992)$ &   $0.609$ \\
$g_{\text{\emph{0-6, East}}} (1993)$ &   $0.511$ \\
$g_{\text{\emph{0-6, East}}} (1994)$ &   $0.210$ \\
$g_{\text{\emph{0-6, East}}} (1995)$ &   $0.014$ \\
$g_{\text{\emph{0-6, East}}} (1996)$ &   $0.000$ \\
$g_{\text{\emph{0-6, East}}} (1997)$ &   $0.000$ \\
~ & ~ \\
\hline
\end{tabular}
\hfill
\begin{tabular}{|l|l|}
\hline 
partial effect & p-value
\\
\hline
$g_{\text{\emph{0-6, East}}} (1998)$ &   $0.000$ \\
$g_{\text{\emph{0-6, East}}} (1999)$ &   $0.000$ \\
$g_{\text{\emph{0-6, East}}} (2000)$ &   $0.000$ \\
$g_{\text{\emph{0-6, East}}} (2001)$ &   $0.000$ \\
$g_{\text{\emph{0-6, East}}} (2002)$ &   $0.000$ \\
$g_{\text{\emph{0-6, East}}} (2003)$ &   $0.000$ \\
$g_{\text{\emph{0-6, East}}} (2004)$ &   $0.000$ \\
$g_{\text{\emph{0-6, East}}} (2005)$ &   $0.000$ \\
$g_{\text{\emph{0-6, East}}} (2006)$ &   $0.000$ \\
$g_{\text{\emph{0-6, East}}} (2007)$ &   $0.001$ \\
$g_{\text{\emph{0-6, East}}} (2008)$ &   $0.000$ \\
$g_{\text{\emph{0-6, East}}} (2009)$ &   $0.000$ \\
$g_{\text{\emph{0-6, East}}} (2010)$ &   $0.000$ \\
$g_{\text{\emph{0-6, East}}} (2011)$ &   $0.000$ \\
$g_{\text{\emph{0-6, East}}} (2012)$ &   $0.000$ \\
$g_{\text{\emph{0-6, East}}} (2013)$ &   $0.000$ \\
$g_{\text{\emph{0-6, East}}} (2014)$ &   $0.000$ \\
$g_{\text{\emph{0-6, East}}} (2015)$ &   $0.000$ \\
$g_{\text{\emph{0-6, East}}} (2016)$ &   $0.000$ \\
$g_{\text{\emph{7-18, East}}} (1992)$ &  $0.415$ \\
$g_{\text{\emph{7-18, East}}} (1993)$ &  $0.303$ \\
$g_{\text{\emph{7-18, East}}} (1994)$ &  $0.206$ \\
$g_{\text{\emph{7-18, East}}} (1995)$ &  $0.148$ \\
$g_{\text{\emph{7-18, East}}} (1996)$ &  $0.098$ \\
$g_{\text{\emph{7-18, East}}} (1997)$ &  $0.051$ \\
$g_{\text{\emph{7-18, East}}} (1998)$ &  $0.021$ \\
$g_{\text{\emph{7-18, East}}} (1999)$ &  $0.005$ \\
$g_{\text{\emph{7-18, East}}} (2000)$ &  $0.001$ \\
$g_{\text{\emph{7-18, East}}} (2001)$ &  $0.000$ \\
$g_{\text{\emph{7-18, East}}} (2002)$ &  $0.000$ \\
$g_{\text{\emph{7-18, East}}} (2003)$ &  $0.000$ \\
$g_{\text{\emph{7-18, East}}} (2004)$ &  $0.000$ \\
$g_{\text{\emph{7-18, East}}} (2005)$ &  $0.000$ \\
$g_{\text{\emph{7-18, East}}} (2006)$ &  $0.000$ \\
$g_{\text{\emph{7-18, East}}} (2007)$ &  $0.001$ \\
$g_{\text{\emph{7-18, East}}} (2008)$ &  $0.002$ \\
$g_{\text{\emph{7-18, East}}} (2009)$ &  $0.002$ \\
$g_{\text{\emph{7-18, East}}} (2010)$ &  $0.003$ \\
$g_{\text{\emph{7-18, East}}} (2011)$ &  $0.005$ \\
$g_{\text{\emph{7-18, East}}} (2012)$ &  $0.006$ \\
$g_{\text{\emph{7-18, East}}} (2013)$ &  $0.006$ \\
$g_{\text{\emph{7-18, East}}} (2014)$ &  $0.004$ \\
$g_{\text{\emph{7-18, East}}} (2015)$ &  $0.003$ \\
$g_{\text{\emph{7-18, East}}} (2016)$ &  $0.004$ \\
~ & ~ \\
\hline 
\end{tabular} 
\caption{p-values for partial effects evaluated at all possible covariate values (pointwise in covariates).\label{P2a:table_pvalues_pointwise}} 
\end{center}
\end{table}

\subsection{Estimated effects}\label{P2a:chapter_application_effects}
This section contains all estimated effects of model~\eqref{P2:soep_model} with a sample of their $95\%$ confidence region (sample size: $100$). 
They are structured similar to Figures~\ref{P2:fig_estimated_West_East} to~\ref{P2:fig_estimated_year}, however with one panel per covariate value to avoid overplotting. 
For interaction effects, we display the sum 
of the intercept, the respective effect, and the corresponding main effects for the estimated conditional densities (left side of the figures).
Note that for some effects, (some) functions sampled from the confidence regions spread out for $s \ra 0$ and $s \ra 1$ with $s \in (0, 1)$, inflating the range of function values. 
In these cases we restrict the limits of the y-axis to an interval, which allows to reasonably view the estimated functions, cutting off some of the functions sampled from the confidence regions.

\begin{figure}[H]
\begin{center}
\includegraphics[width=0.49\textwidth]{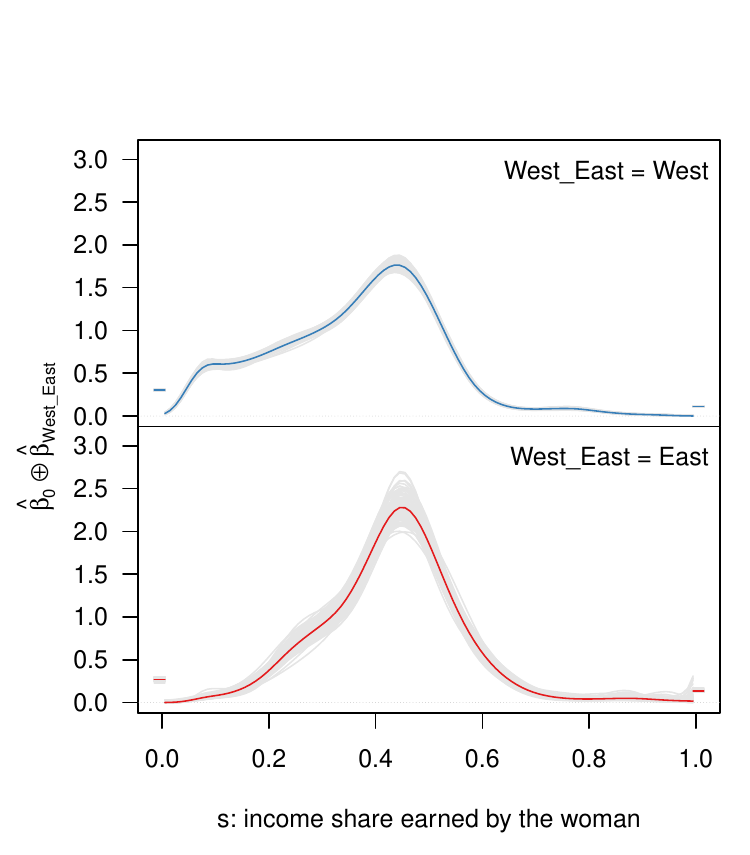}
\includegraphics[width=0.49\textwidth]{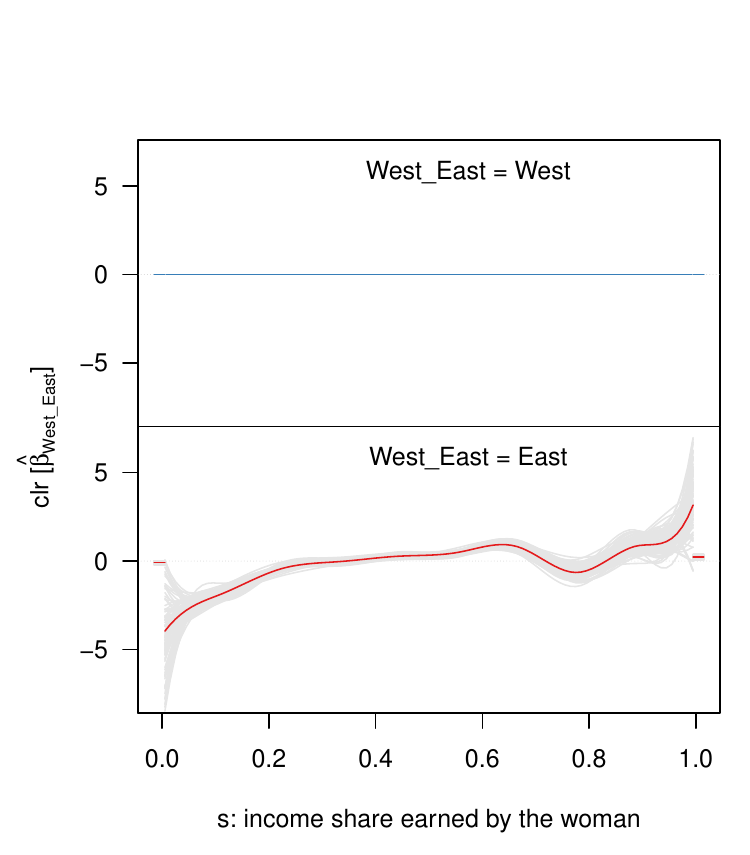}
\end{center}
\vspace{-0.5cm}
\caption{Estimated conditional densities for couples without minor children living in \emph{West} 
vs.\ \emph{East} Germany 
in 1991 [left] and clr transformed estimated effect of \emph{West\_East} [right] with $100$ draws sampled uniformly from the respective $95\%$ simultaneous (over $[0, 1]$) confidence region.}
\end{figure}

%

\begin{figure}[H]
\begin{center}
\includegraphics[width=0.49\textwidth]{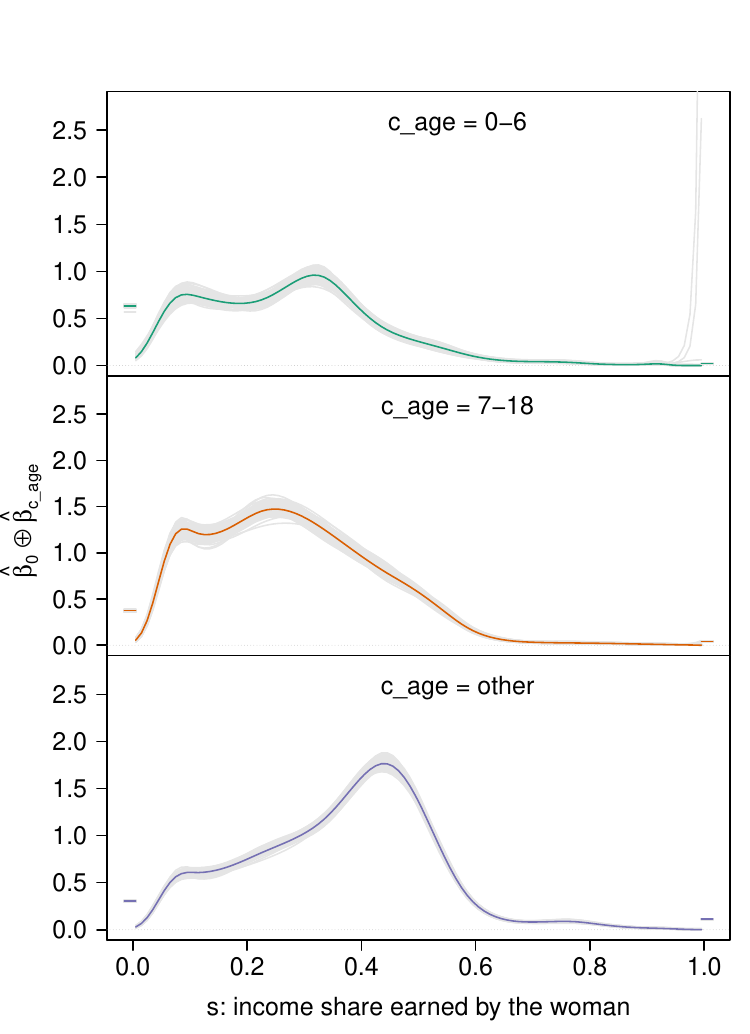}
\includegraphics[width=0.49\textwidth]{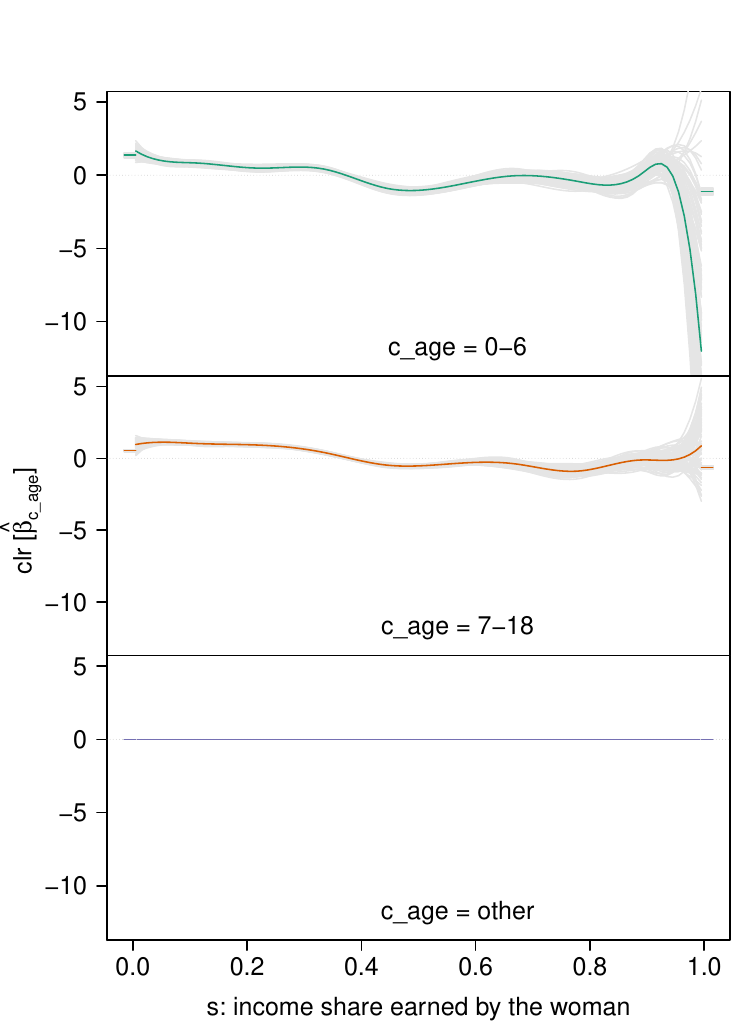}
\end{center}
\vspace{-0.5cm}
\caption{Estimated conditional densities for couples living in \emph{West} Germany in 1991 for all three values of \emph{c\_age} [left] and clr transformed estimated effects of \emph{c\_age} [right] with $100$ draws sampled uniformly from the respective $95\%$ simultaneous (over $[0, 1]$) confidence region.\label{P2a:fig_appendix_estimated_c_age}}
\end{figure}

%

\begin{figure}[H]
\begin{center}
\includegraphics[width=0.49\textwidth]{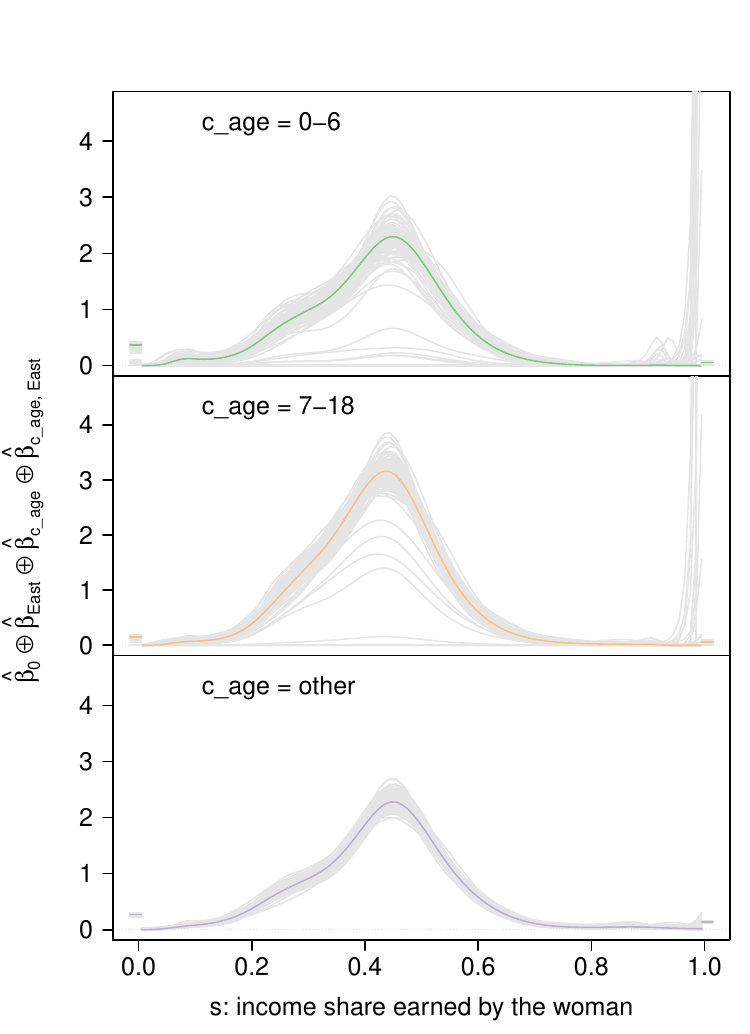}
\includegraphics[width=0.49\textwidth]{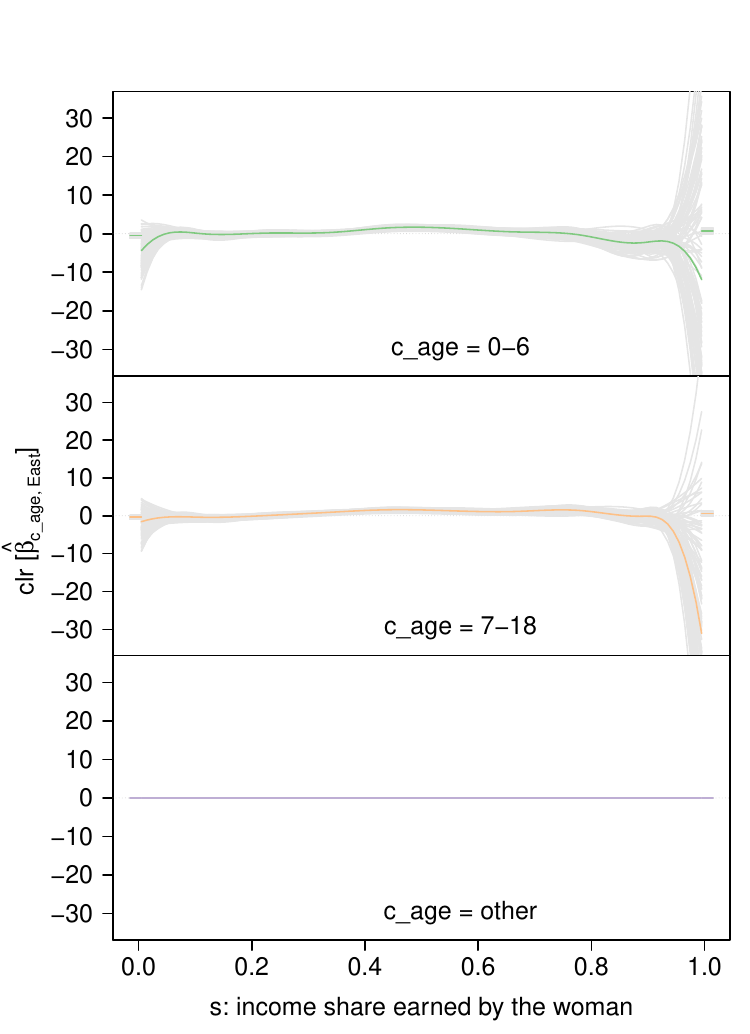}
\end{center}
\vspace{-0.5cm}
\caption{Estimated conditional densities for couples living in \emph{East} Germany in 1991 for all three values of \emph{c\_age} [left] and clr transformed estimated interaction effects of \emph{c\_age} and \emph{East} [right] with $100$ draws sampled uniformly from the respective $95\%$ simultaneous (over $[0, 1]$) confidence region. Note that since \emph{West} is reference category, all interaction effects are zero and thus not illustrated here. The respective estimated conditional densities are shown in Figure~\ref{P2a:fig_appendix_estimated_c_age}.\label{P2a:fig_appendix_estimated_West_East_c_age}}
\end{figure}

\begin{figure}[H]
\vspace{-2cm}
\begin{center}
\includegraphics{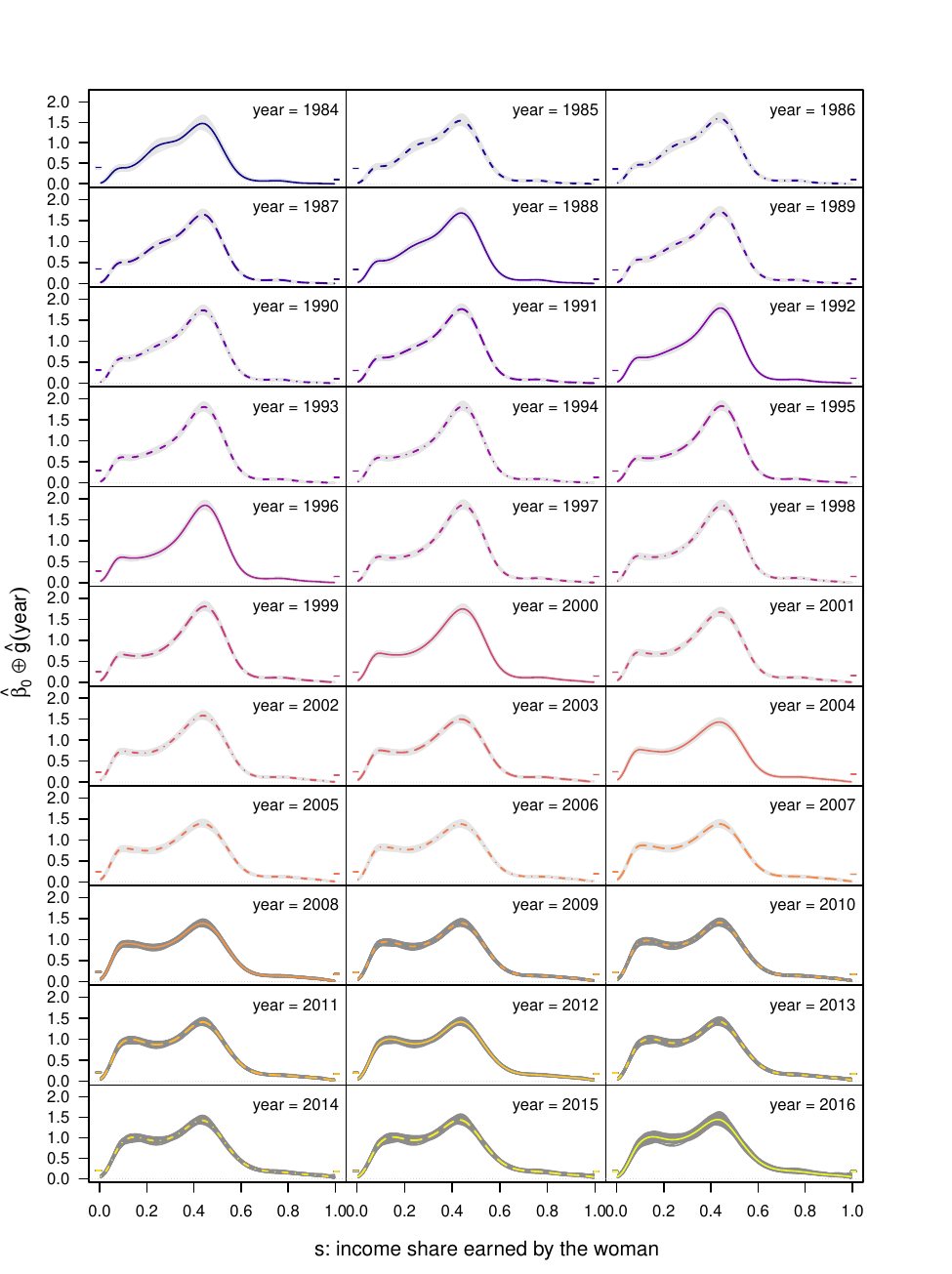}
\end{center}
\vspace{-0.5cm}
\caption{Estimated conditional densities for couples without minor children living in \emph{West} Germany in different \emph{years} with $100$ draws sampled uniformly from the respective $95\%$ simultaneous (over $[0, 1]$) confidence region.\label{P2a:fig_appendix_estimated_year}}
\end{figure}

\begin{figure}[H]
\vspace{-2cm}
\begin{center}
\includegraphics{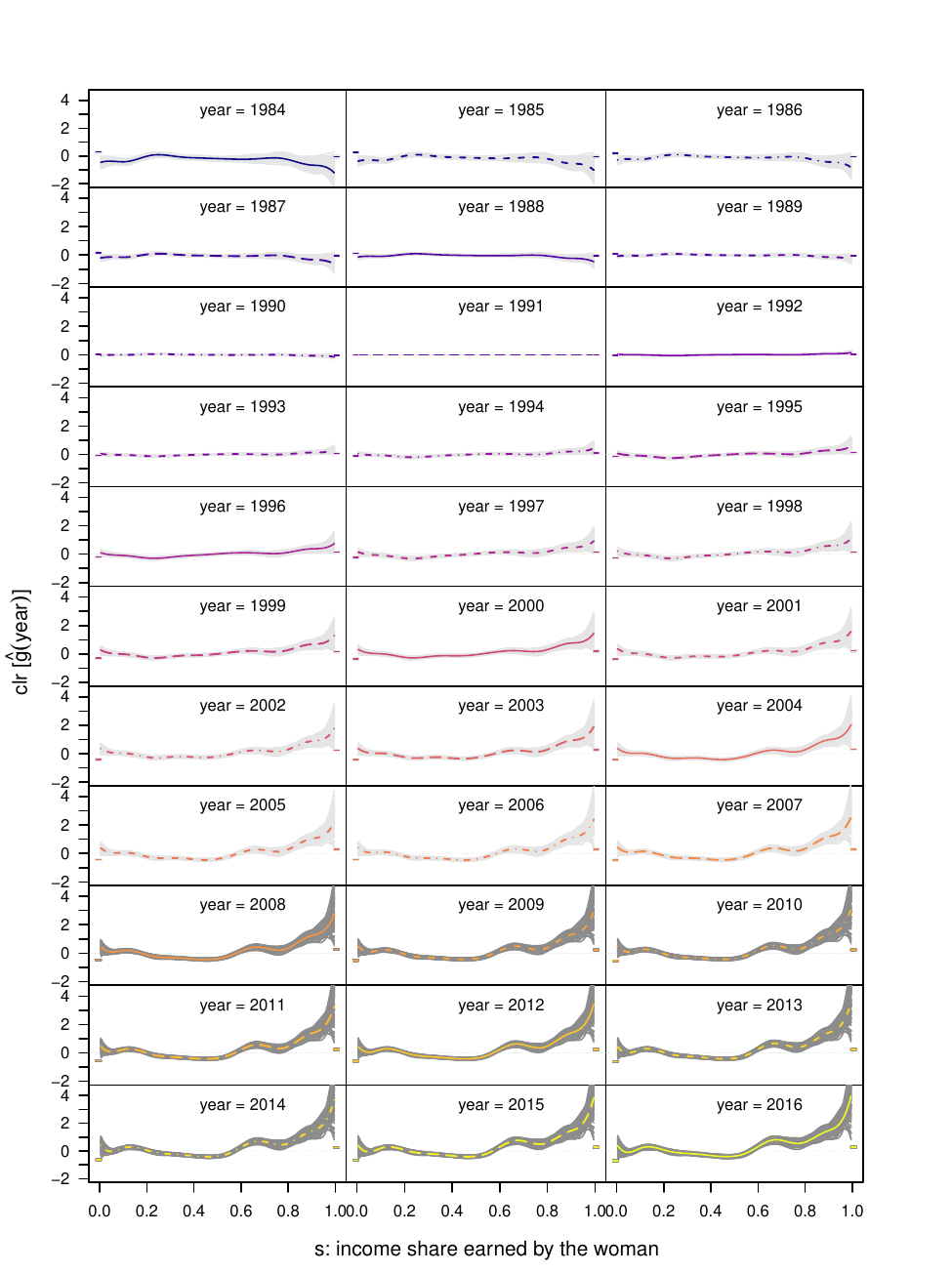}
\end{center}
\vspace{-0.5cm}
\caption{Estimated clr transformed effect of \emph{year} with $100$ draws sampled uniformly from the respective $95\%$ simultaneous (over $[0, 1]$) confidence region.
\label{P2a:fig_appendix_estimated_year_clr}}
\end{figure}

\begin{figure}[H]
\vspace{-2cm}
\begin{center}
\includegraphics{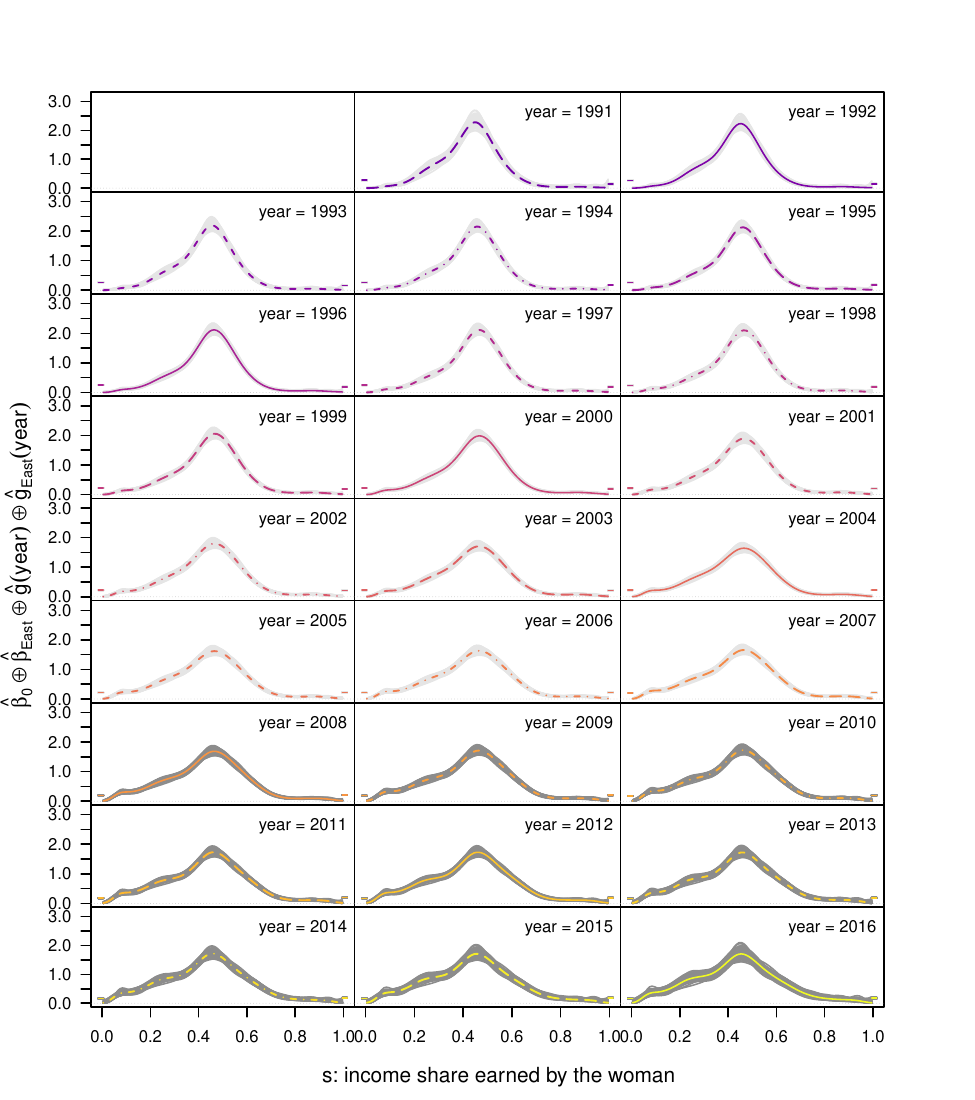}
\end{center}
\vspace{-0.5cm}
\caption{Estimated conditional densities for couples without minor children living in \emph{East} Germany in different \emph{years} with $100$ draws sampled uniformly from the respective $95\%$ simultaneous (over $[0, 1]$) confidence region. Note that the effect of $\gh_{\text{\emph{West\_East}}}$ is not significant.}
\end{figure}

\begin{figure}[H]
\vspace{-2cm}
\begin{center}
\includegraphics{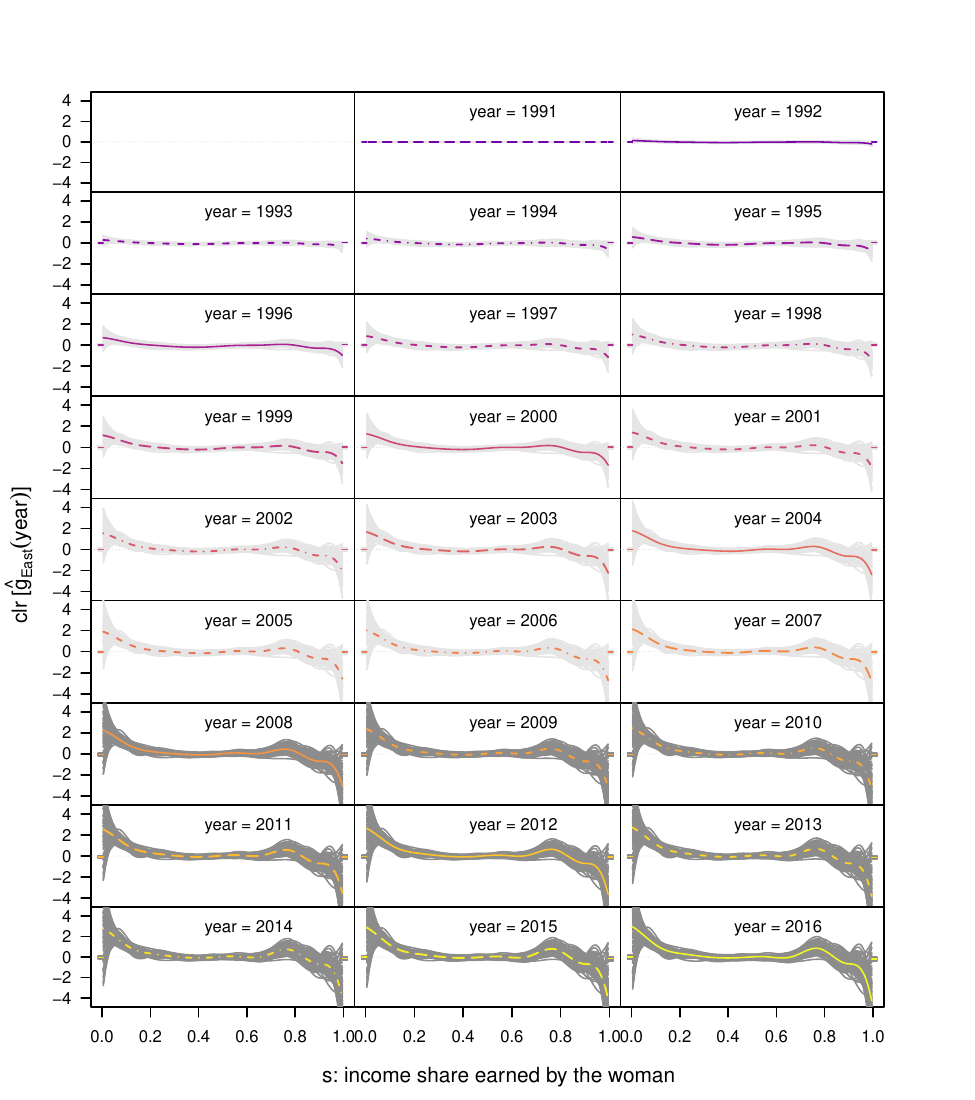}
\end{center}
\vspace{-0.5cm}
\caption{Estimated clr transformed group-specific smooth effect of \emph{year} for \emph{East} with $100$ draws sampled uniformly from the respective $95\%$ simultaneous (over $[0, 1]$) confidence region. 
Note 
that the effect of $\clr [\gh_{\text{\emph{West\_East}}}]$ is not significant.}
\end{figure}

\begin{figure}[H]
\vspace{-2cm}
\begin{center}
\includegraphics{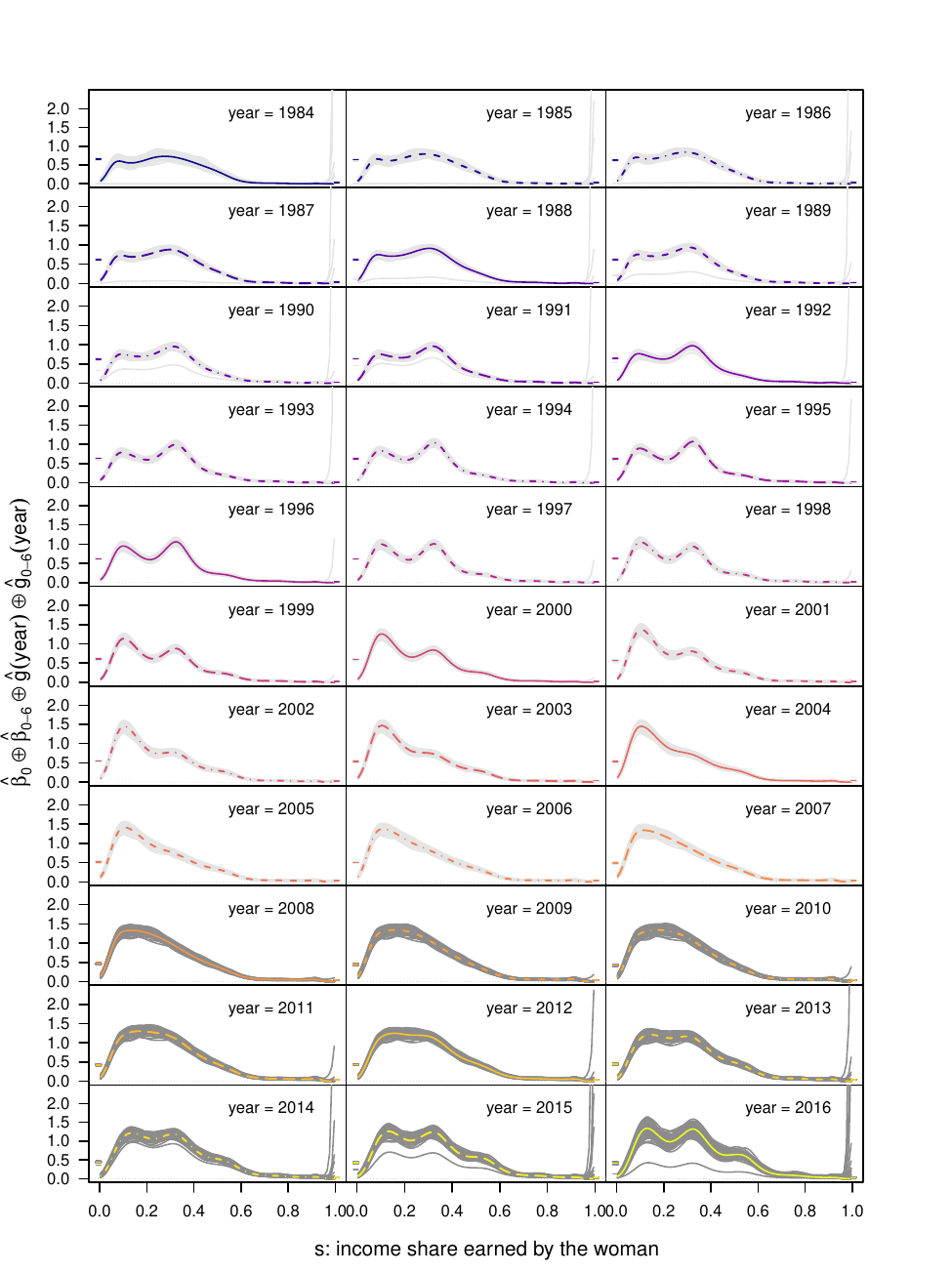}
\end{center}
\vspace{-0.5cm}
\caption{Estimated conditional densities for couples whose youngest child is aged zero to six living in \emph{West} Germany in different \emph{years} with $100$ draws sampled uniformly from the respective $95\%$ simultaneous (over $[0, 1]$) confidence region.}
\end{figure}

\begin{figure}[H]
\vspace{-2cm}
\begin{center}
\includegraphics{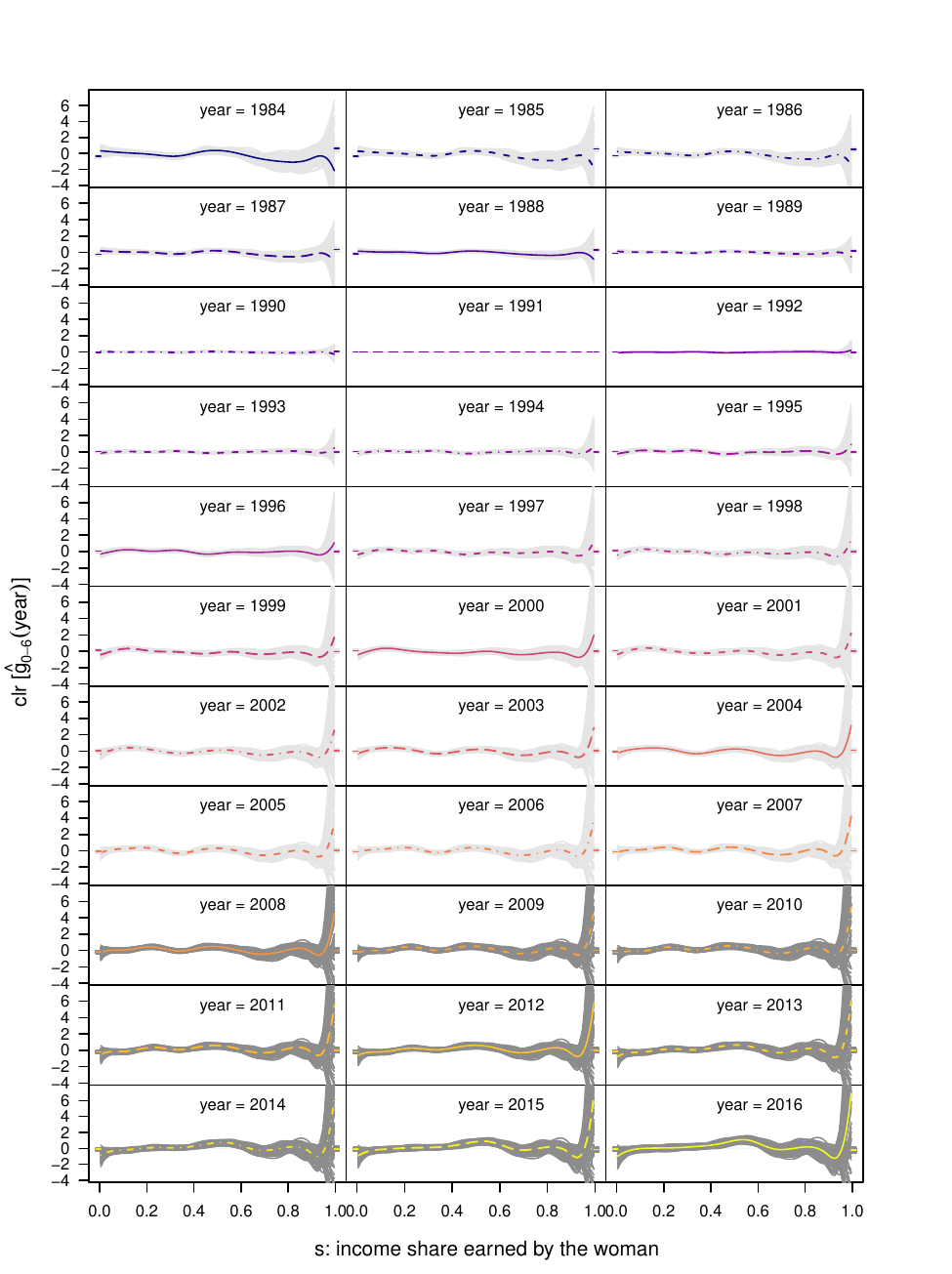}
\end{center}
\vspace{-0.5cm}
\caption{Estimated clr transformed group-specific smooth effect of \emph{year} for \emph{0-6} with $100$ draws sampled uniformly from the respective $95\%$ simultaneous (over $[0, 1]$) confidence region. 
}
\end{figure}

\begin{figure}[H]
\vspace{-2cm}
\begin{center}
\includegraphics{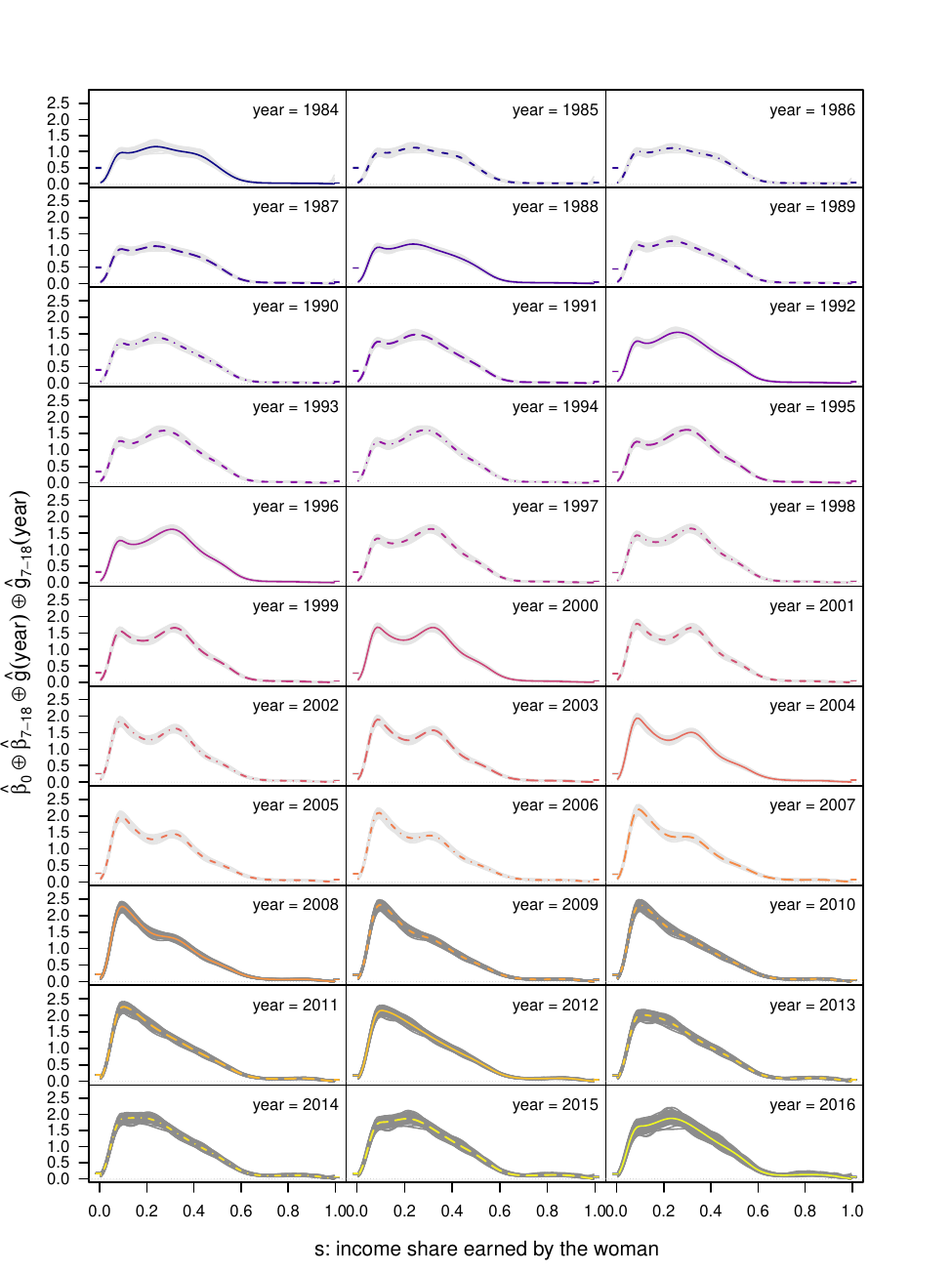}
\end{center}
\vspace{-0.5cm}
\caption{Estimated conditional densities for couples whose youngest child is aged seven to 18 living in \emph{West} Germany in different \emph{years} with $100$ draws sampled uniformly from the respective $95\%$ simultaneous (over $[0, 1]$) confidence region.}
\end{figure}

\begin{figure}[H]
\vspace{-2cm}
\begin{center}
\includegraphics{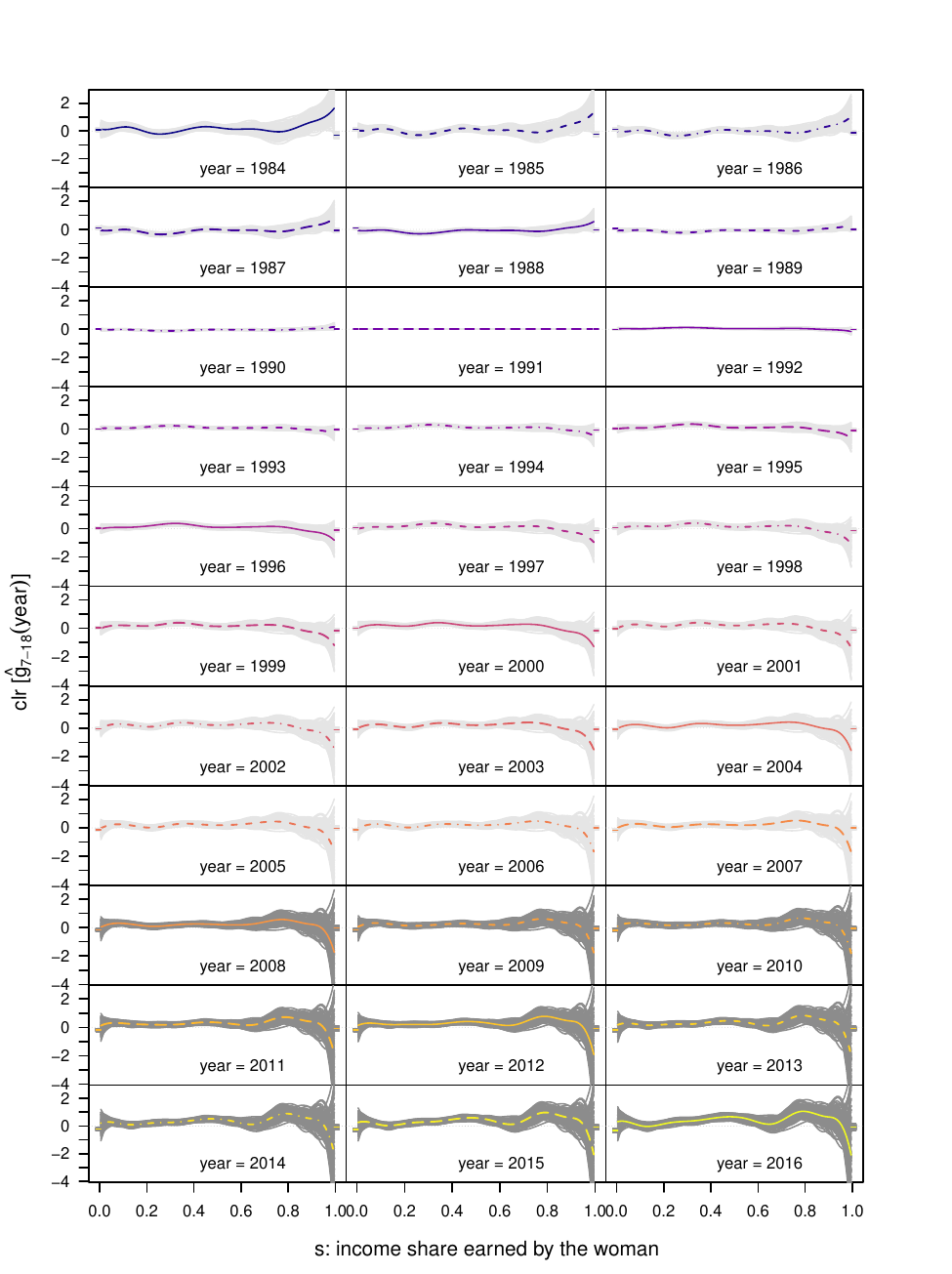}
\end{center}
\vspace{-0.5cm}
\caption{Estimated clr transformed group-specific smooth effect of \emph{year} for \emph{7-18} with $100$ draws sampled uniformly from the respective $95\%$ simultaneous (over $[0, 1]$) confidence region. 
}
\end{figure}

\begin{figure}[H]
\vspace{-2cm}
\begin{center}
\includegraphics{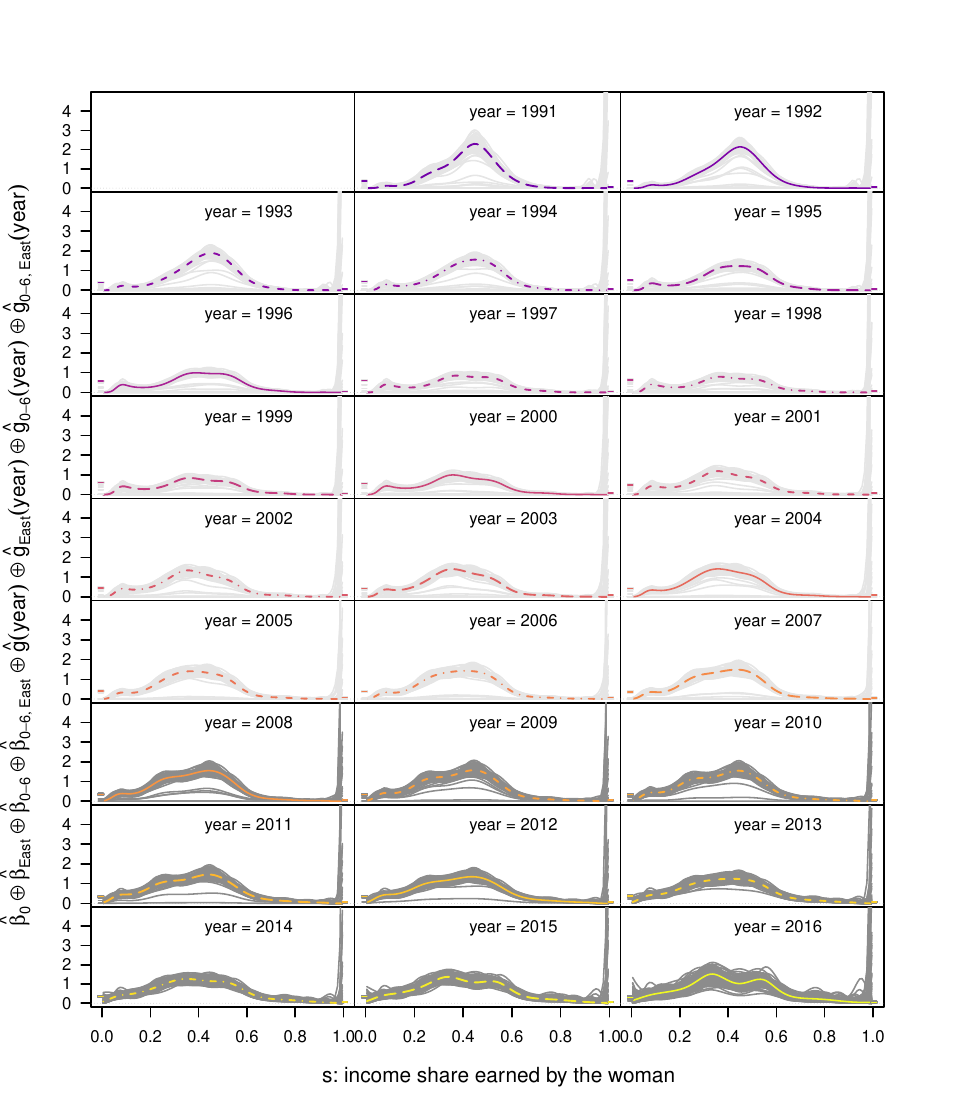}
\end{center}
\vspace{-0.5cm}
\caption{Estimated conditional densities for couples whose youngest child is aged zero to six living in \emph{East} Germany in different \emph{years} with $100$ draws sampled uniformly from the respective $95\%$ simultaneous (over $[0, 1]$) confidence region. Note that the effect of $\gh_{\text{\emph{c\_age, West\_East}}}$ is not significant.}
\end{figure}

\begin{figure}[H]
\vspace{-2cm}
\begin{center}
\includegraphics{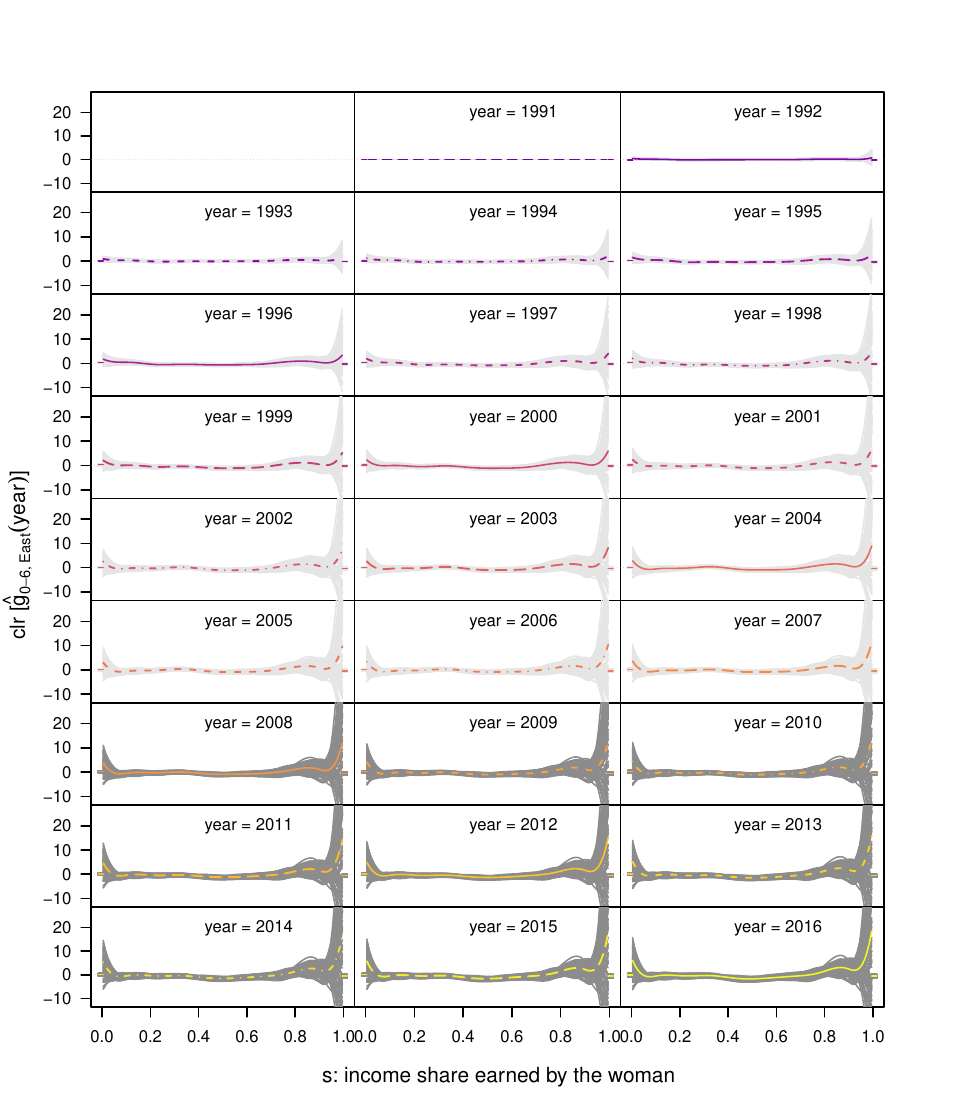}
\end{center}
\vspace{-0.5cm}
\caption{Estimated clr transformed group-specific smooth interaction effect of \emph{year} for \emph{0-6} and \emph{East} with $100$ draws sampled uniformly from the respective $95\%$ simultaneous (over $[0, 1]$) confidence region. 
Note 
that the effect of $\clr [\gh_{\text{\emph{c\_age, West\_East}}}]$ is not significant.}
\end{figure}

\begin{figure}[H]
\vspace{-2cm}
\begin{center}
\includegraphics{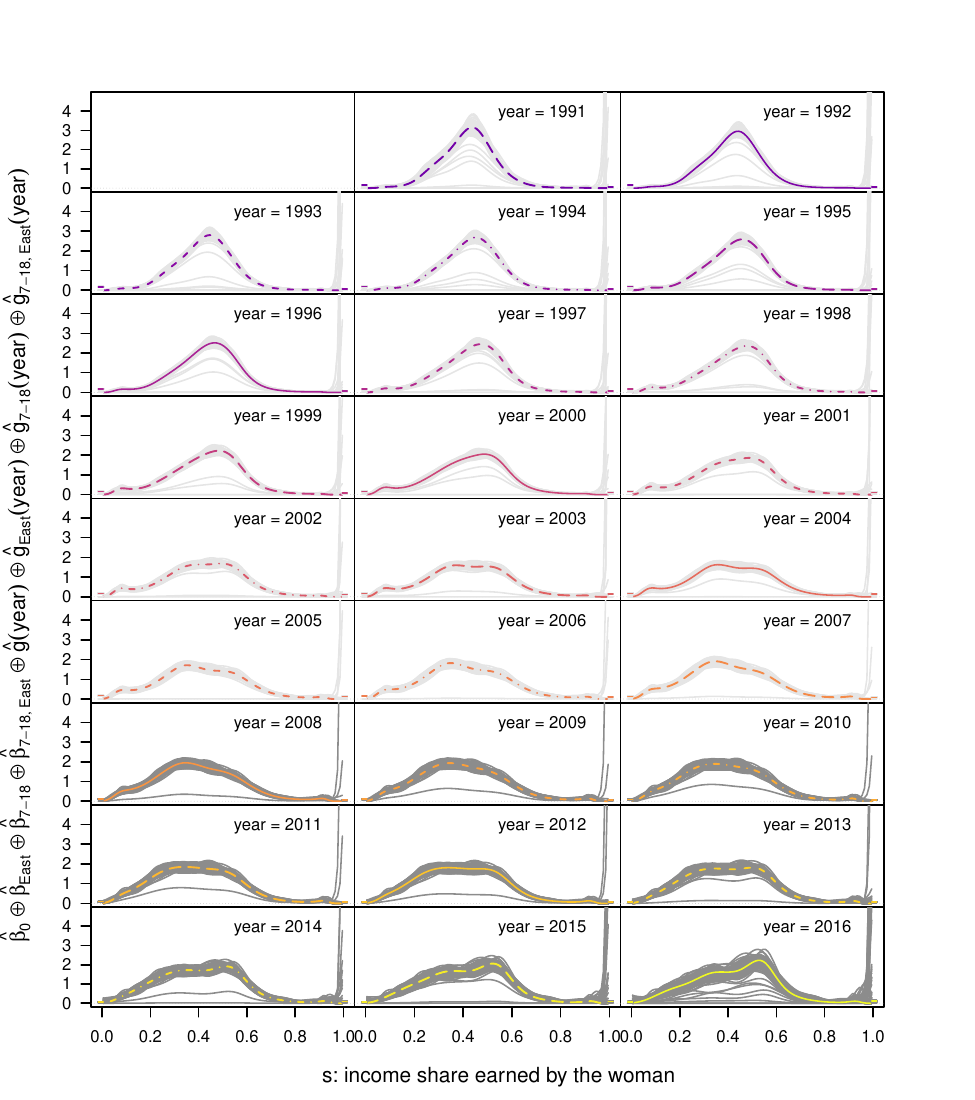}
\end{center}
\vspace{-0.5cm}
\caption{Estimated conditional densities for couples whose youngest child is aged seven to 18 living in \emph{East} Germany in different \emph{years} with $100$ draws sampled uniformly from the respective $95\%$ simultaneous (over $[0, 1]$) confidence region. Note that the effect of $\gh_{\text{\emph{c\_age, West\_East}}}$ is not significant.}
\end{figure}

\begin{figure}[H]
\vspace{-2cm}
\begin{center}
\includegraphics{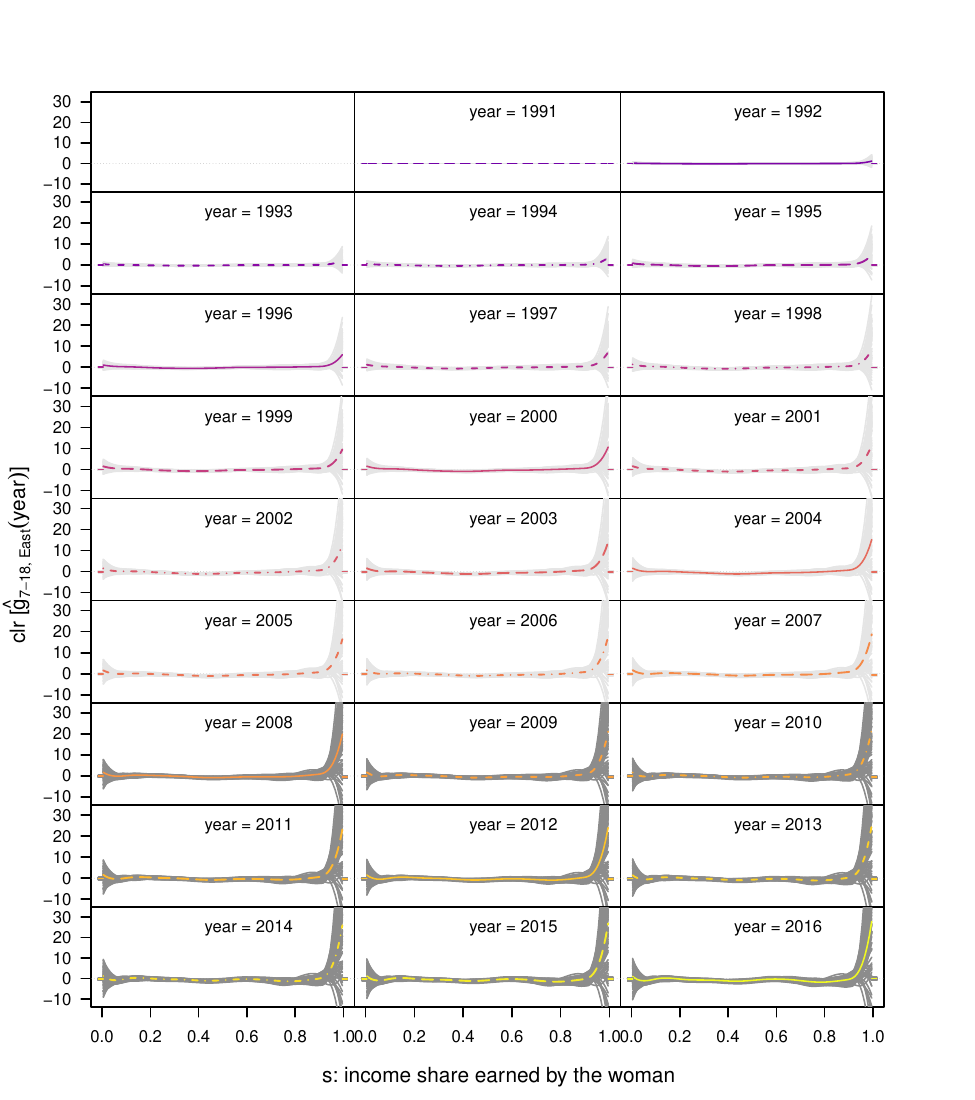}
\end{center}
\vspace{-0.5cm}
\caption{Estimated clr transformed group-specific smooth interaction effect of \emph{year} for \emph{7-18} and \emph{East} with $100$ draws sampled uniformly from the respective $95\%$ simultaneous (over $[0, 1]$) confidence region. 
Note 
that the effect of $\clr [\gh_{\text{\emph{c\_age, West\_East}}}]$ is not significant.}
\end{figure}

\subsection{Estimated conditional densities (predictions)}\label{P2a:chapter_application_predictions}

\vspace{-0.7cm}
\begin{figure}[H]
\begin{minipage}{0.89\textwidth}
\includegraphics[width=\textwidth]{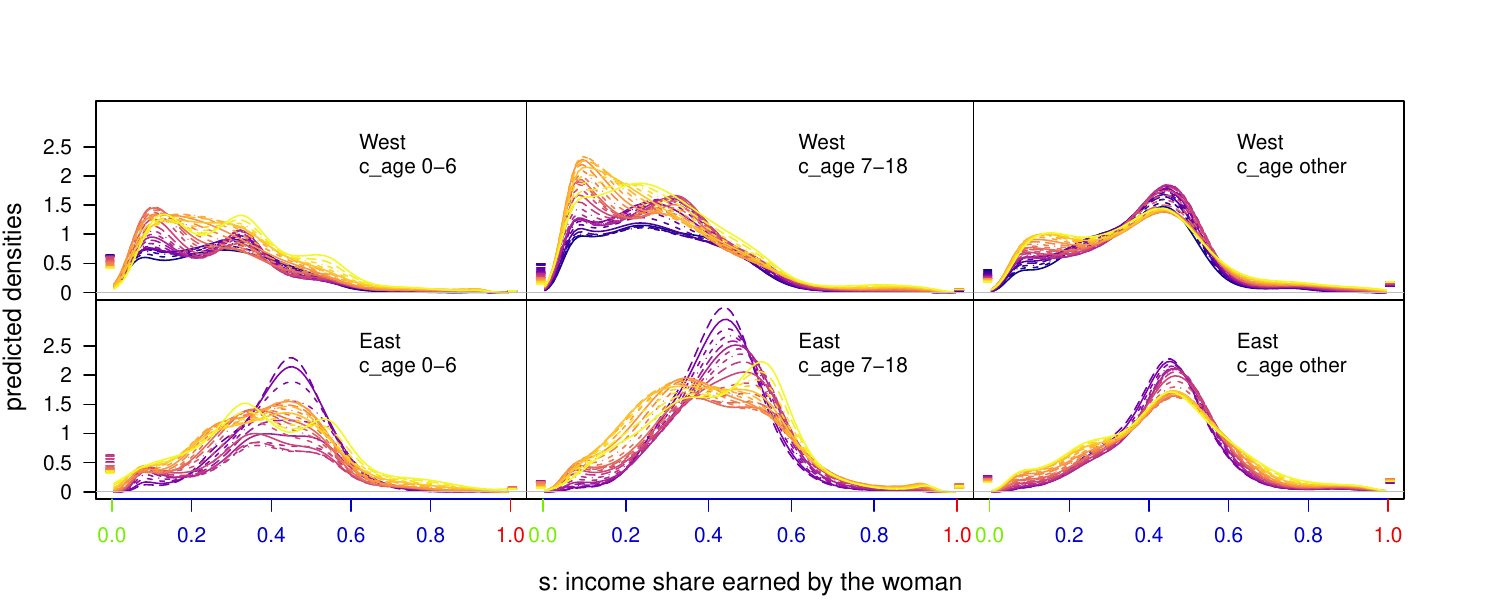} \\
\hspace*{0.75cm}
\includegraphics[width=0.87\textwidth]{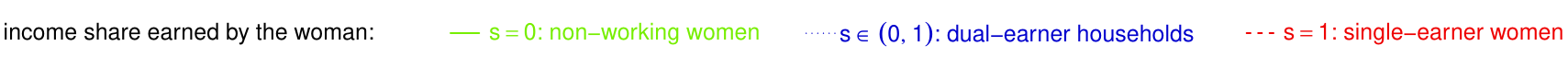} \\[-10mm]
\includegraphics[width=\textwidth]{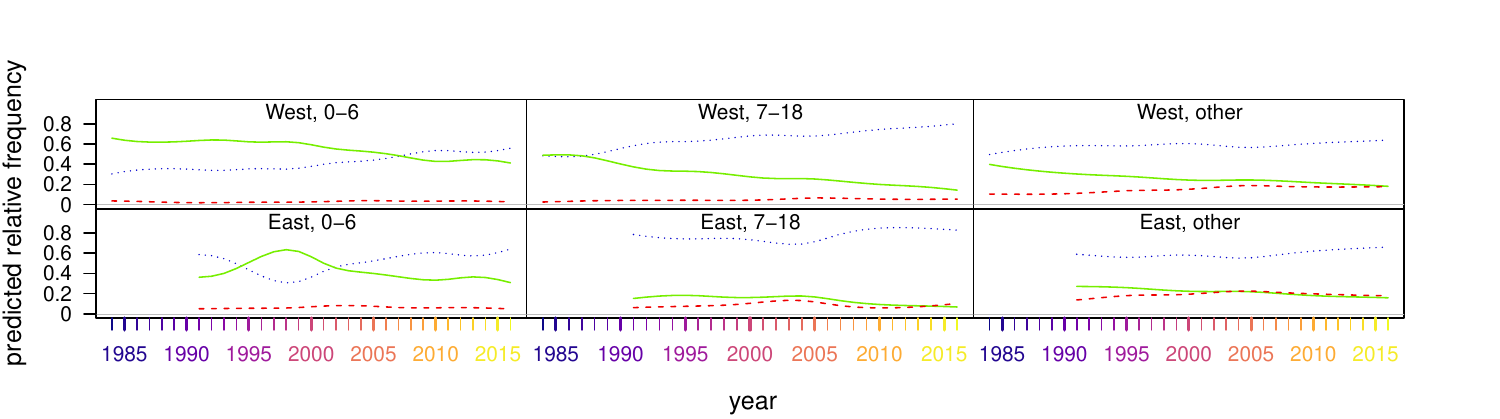}
\end{minipage}
\hspace{-0.3cm}
\begin{minipage}{0.059\textwidth}
\includegraphics[width=1.2\textwidth]{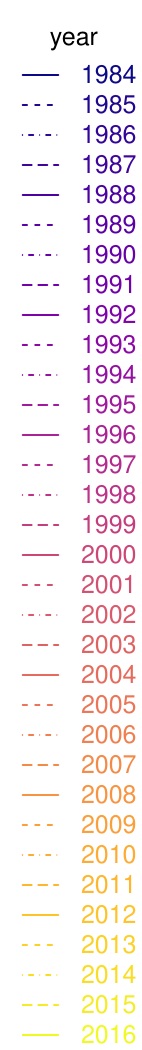}
\end{minipage}
\caption{Estimated conditional densities [upper $2 \times 3$ panels] and corresponding relative frequencies [lower $2 \times 3$ panels] for \emph{West} vs.\ \emph{East} Germany [rows] and all three values of \emph{c\_age} [columns].\label{P2a:figure_predictions}}
\end{figure}

Figure~\ref{P2a:figure_predictions} shows the estimated conditional densities obtained from model~\eqref{P2:soep_model} for the different combinations of covariate values in the upper part. 
The different \emph{years} are illustrated in one panel, respectively, distinguishable from each other via a color gradient and different line types.
For the discrete values with positive probability mass, $0$ and~$1$, the density values are visualized as dashes and shifted slightly outwards for better distinction.
Their smoothly estimated development over time, as well as the one of dual-earner households (corresponding to the Lebesgue integrals of the densities), is shown explicitly in the lower part of the figure.
%
We briefly point out some interesting results visible in this figure.
In \emph{West} Germany in all three categories of \emph{c\_age}, the share of nonworking women ($s = 0$) is decreasing over time while the share of dual-earner households ($s \in (0, 1)$) is increasing.
Simultaneously, the probability mass for lower positive income shares increases, indicating an increase of previously nonworking women in part-time employment.
Without considering mixed densities for single- and double-earner couples, but only continuous densities for double-earner couples on $(0, 1)$, this might be misinterpreted as a shift from larger to smaller income shares, nicely showing the necessity of considering mixed densities.
Furthermore, there is more probability mass at lower positive income shares for couples with minor children (\emph{0-6} and \emph{7-18}) compared to those without minor children (\emph{other}) in \emph{West} Germany, reflecting that women are rather working part-time in the presence of children. 
In \emph{East} Germany the shapes of the densities are less dependent on the age of the youngest child.
Especially for couples without minor children (\emph{other}), the densities are close to symmetric with a mode around $s = 0.45$.
For couples with minor children (\emph{0-6} and \emph{7-18}), this is only the case for the earlier \emph{years}, while over time, probability mass tends to shift towards smaller income shares.
Again, this might include previously nonworking women switching to part-time employment even though there is not such a clear decline of nonworking women as in \emph{West} Germany.
Actually, for couples with small children (\emph{0-6}), the frequency of nonworking women even increases from 1991 to 1998 in \emph{East} Germany, reaching a similar level as in \emph{West} Germany, before decreasing again.
All in all, accompanying the findings of \citet{maier2021}, our observations reflect the different social norms established in the two parts of Germany before reunification in 1990: In \emph{East} Germany, comprehensive child care was available and women were expected to work, while in \emph{West} Germany, it was more common for mothers to be unemployed and instead take care of the children themselves.

\section{Simulation study}\label{P2a:chapter_simulation}
This section contains further figures related to the simulation study presented in Section~\ref{P2:chapter_simulation}.
Figure~\ref{P2a:fig_relMSE_all} is the extended version of Figure~\ref{P2:fig_relMSE_main}, including interaction effects.
Figure~\ref{P2a:fig_empCR_smt_all} shows the empirical coverage rates for all partial effects based on Lemma~\ref{P2a:lemma_confidence_regions_simultaneos}, i.e., simultaneous in covariates.
Compared to Figure~\ref{P2:fig_empCR_smt_main}, interaction effects are added, while excluding prediction.
Instead, in this section, $\rMSE (\fh)$, which is computed based on Lemma~\ref{P2:lemma_confidence_regions}, is shown in Figure~\ref{P2a:fig_empCR_pw_all} together with the empirical coverage rates for all partial effects, also based on Lemma~\ref{P2:lemma_confidence_regions}, i.e., point-wise in the covariates, since conceptually they belong together. 

\begin{figure}[H]
\begin{center}
\includegraphics[width=0.7\textwidth]{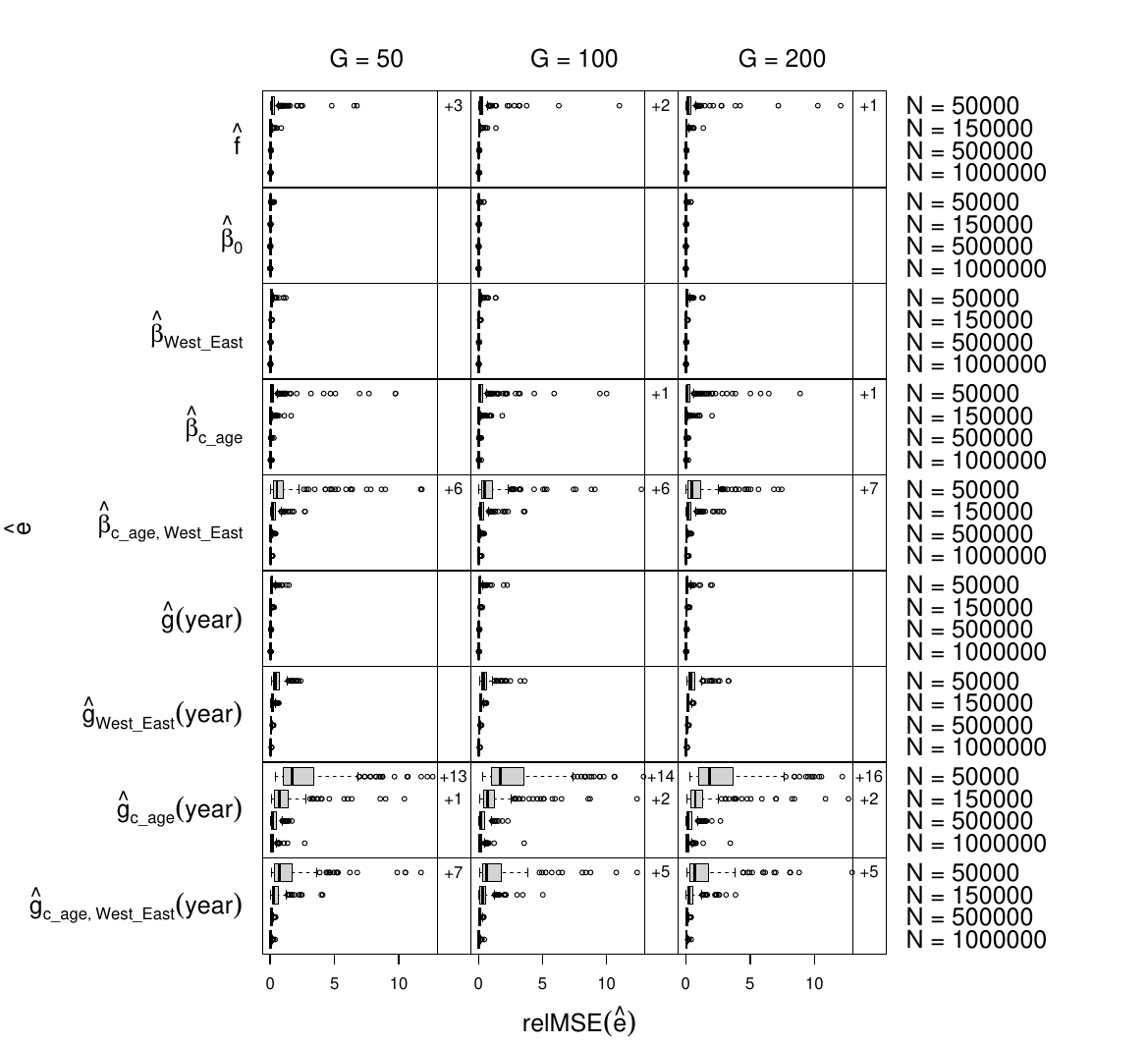}
\end{center}
\vspace{-0.5cm}
\caption{RelMSE for prediction 
and main effects for different values of $G$ [columns] and $N$ [rows]. Numbers following a plus-sign at the right of a panel give the number of values larger than $13$ (vertical line).
Empty space means all values are visible.\label{P2a:fig_relMSE_all}}
\end{figure}

\begin{figure}[H]
\begin{center}
\includegraphics[width=0.75\textwidth]{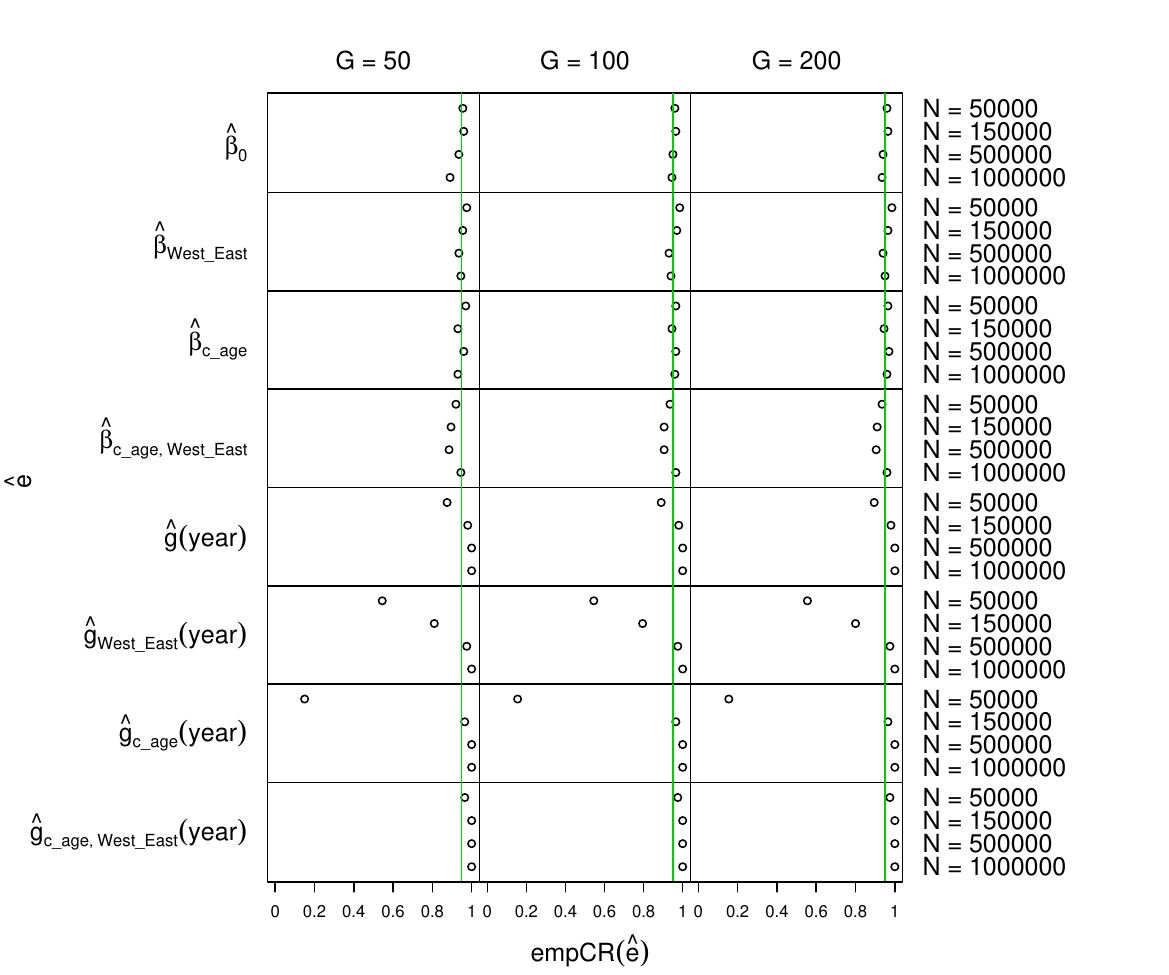} 
\end{center}
\vspace{-0.5cm}
\caption{Empirical coverage rates (empCR) for 
all partial effects (simultaneous in covariates) for different values of $G$ [columns] and $N$ [rows].
\label{P2a:fig_empCR_smt_all}}
\end{figure}

\begin{figure}[H]
\begin{center}
\includegraphics[width=0.75\textwidth]{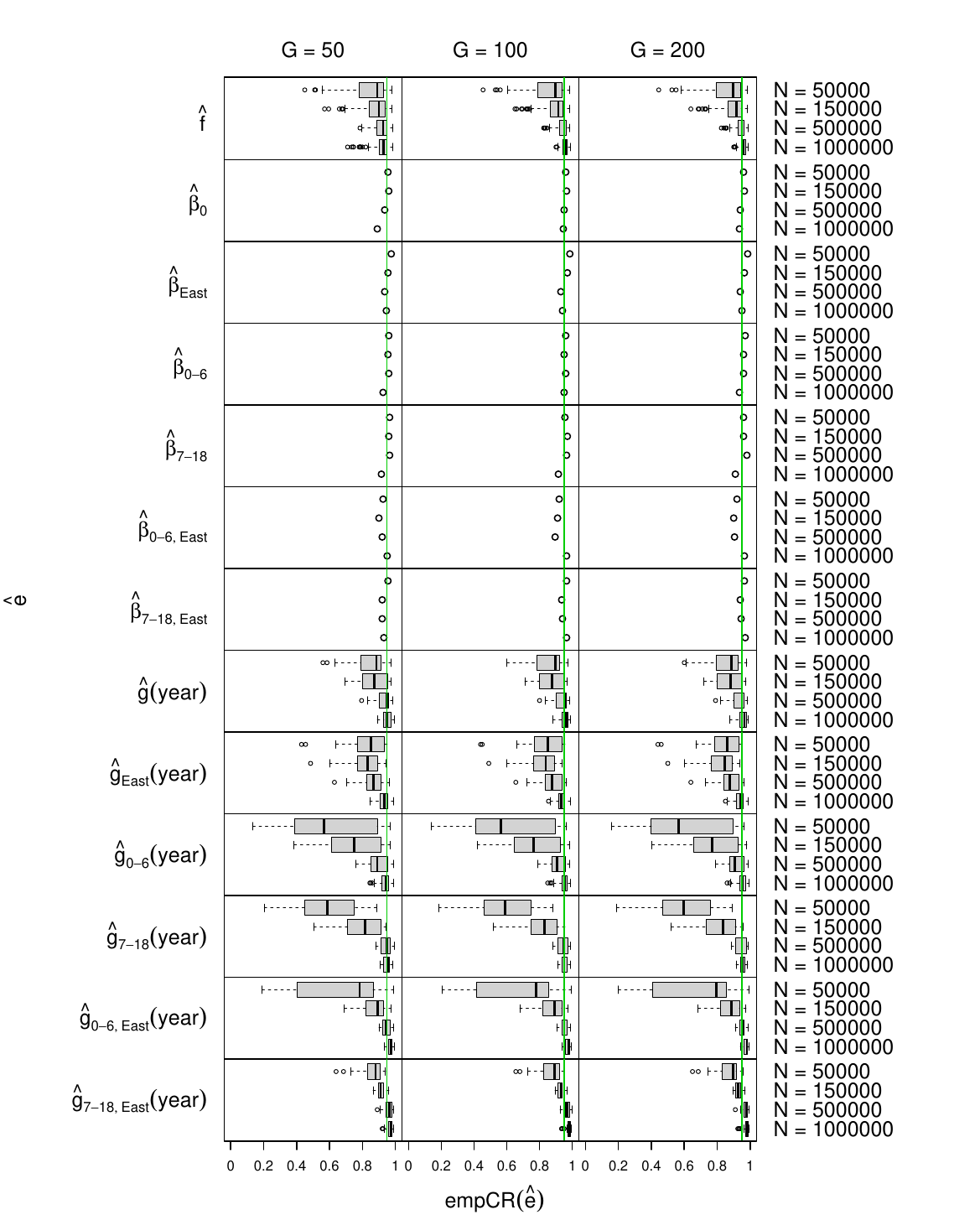}
\end{center}
\vspace{-0.5cm}
\caption{Empirical coverage rates (empCR) for prediction 
and all partial effects (point-wise regarding covariates) for different values of $G$ [columns] and $N$ [rows].
The boxplots for predictions and smooth effects summarize the coverage rates for the unique covariate values ($177$ for prediction, $33$ for effects involving \emph{year}, but not \emph{East}, and $26$ for interaction effects containing \emph{year} and \emph{East}).\label{P2a:fig_empCR_pw_all}}
\end{figure}

%